\documentclass[11pt]{article}
\usepackage[english]{babel}
\usepackage[a4paper,
    top=2.5cm,
    bottom=3cm,
    left=3.1cm,
    right=2.9cm,
    marginparwidth=1.75cm]{geometry}
\usepackage[T1]{fontenc}

\usepackage{amsmath, amssymb, amsthm}
\usepackage{graphicx}
\graphicspath{{./}{figures/}}

\usepackage{subcaption}
\usepackage{cite}
\usepackage{hyperref}
\hypersetup{
    pdftitle={Asymptotic analysis of a Fredholm determinant occurring %
        in the description of the dynamical correlation functions %
        of the Lieb--Liniger Bose gas},
    pdfauthor={Frank G\"{o}hmann, Karol K. Kozlowski, Mikhail D. Minin},
}

\newtheorem{theorem}{Theorem}
\newtheorem*{conjecture}{Conjecture}
\newtheorem{proposition}{Proposition}
\newtheorem{RHp}{Riemann--Hilbert Problem}
\newtheorem*{assumptions}{Assumptions}

\theoremstyle{definition}
\newtheorem*{remark}{Remark}
\newtheorem*{remarks}{Remarks}
\newtheorem{lemma}{Lemma}

\usepackage{enumitem}

\numberwithin{equation}{section}

\usepackage{physics}

\DeclareMathOperator*{\res}{res}
\DeclareMathOperator*{\sgn}{sgn}

\DeclareMathOperator{\Real}{Re}
\DeclareMathOperator{\Imag}{Im}
\DeclareMathOperator{\Li}{Li}
\DeclareMathOperator{\Ext}{Ext}

\usepackage{dsfont}
\newcommand{\indicator}[1]{\mathds{1}_{#1}}

\def\id{\mathrm{id}}
\def\rmi{\mathrm{i}}
\def\rme{\mathrm{e}}
\def\Ocal{\mathrm{O}}
\def\ocal{\mathrm{o}}

\def\rmA{\mathrm{A}}
\def\rmC{\mathrm{C}}
\def\rmD{\mathrm{D}}
\def\rmG{\mathrm{G}}
\def\rmI{\mathrm{I}}
\def\rmL{\mathrm{L}}
\def\rmM{\mathrm{M}}
\def\rmP{\mathrm{P}}
\def\rmR{\mathrm{R}}
\def\rmS{\mathrm{S}}
\def\rmV{\mathrm{V}}

\def\Jcal{\mathcal{J}}
\def\Acal{\mathcal{A}}
\def\Ccal{\mathcal{C}}
\def\Lcal{\mathcal{L}}
\def\Pcal{\mathcal{P}}
\def\Rcal{\mathcal{R}}
\def\Scal{\mathcal{S}}
\def\Ucal{\mathcal{U}}
\def\Vcal{\mathcal{V}}

\allowdisplaybreaks

\title{Asymptotic analysis of a Fredholm determinant occurring\\
    in the description of the dynamical correlation functions\\
    of the Lieb--Liniger Bose gas}

\date{}

\newcommand{\inputfig}[1]{\includegraphics{#1.pdf}}

\begin{document}

\maketitle
\vspace{-50pt}

\begin{center}
    Frank G\"{o}hmann\textsuperscript{1$\star$},
    Karol K. Kozlowski\textsuperscript{2$\dagger$},
    Mikhail D. Minin\textsuperscript{1$\ddagger$}
\end{center}

\begin{center}
    1 Fakult\"{a}t für Mathematik und Naturwissenschaften,
    Bergische Universit\"{a}t Wuppertal, 42097 Wuppertal,
    Germany
    \\
    2 Univ Lyon, ENS de Lyon, Univ Claude Bernard, CNRS,
    Laboratoire de Physique, F-69342 Lyon, France
    \\[\baselineskip]
    $\star$
    \href{mailto:goehmann@uni-wuppertal.de}{%
        \small goehmann@uni-wuppertal.de}%
    \,,\quad
    $\dagger$
    \href{mailto:karol.kozlowski@ens-lyon.fr}{%
        \small karol.kozlowski@ens-lyon.fr}%
    \,,\\
    $\ddagger$
    \href{mailto:minin@uni-wuppertal.de}{\small minin@uni-wuppertal.de}
\end{center}

\begin{abstract}
    We perform a large-$x$ asymptotic analysis of the Fredholm
    determinant of an integrable integral operator with generalized
    sine kernel, where $x$ controls the strength of the oscillations
    along the integration contour of the operator and will play the
    role of the distance variable in applications to the correlation
    functions of integrable quantum systems. Our generalized sine kernel
    involves a number of functional parameters that will allow us to
    adapt it to the analysis of the dynamical correlation functions
    of the Lieb--Liniger model at finite temperature and for all
    positive values of the coupling constant. It will also allow
    us to consider a class of equilibrium correlation functions that
    are governed by generalized Gibbs ensembles. Our work is based
    on the analysis of a matrix Riemann--Hilbert problem that is canonically
    connected with our integrable integral operator.
\end{abstract}

\tableofcontents
\clearpage

\section{Introduction}
Dynamical two-point correlation functions of local operators in
many-body quantum systems describe how local perturbations spread
out in space and time. These functions or their Fourier transforms
are what is typically measured in physics laboratories. Correlation
functions are defined as ensemble averages of quantum mechanical
expectation values of products of local operators that act on a Hilbert
space. Such entities are hard to control or to evaluate. Often a first
step in attempts to study their dependence on parameters, such as the
spatial or temporal separation of the two points involved, is to
expand them into form-factor series. In the case of integrable
many-body systems and quantum field theories form-factor series
may take rather explicit forms, typically as series of multiple
integrals with integrands of rapidly oscillating type. In our
understanding the form-factor series of integrable many-body systems
and quantum field theories define a class of special functions that
ought to be studied. For us a particularly interesting aspect is
the study of the asymptotic behaviour of the correlation functions,
since this is probably the only way to make them, at least to some
extent, explicit. Moreover, the asymptotic behaviour may be argued
to be `generic' or `universal', i.e., to a certain extent independent
of the precise form of the short-range interactions that define the
Hamiltonian of the system under consideration.

An important sub-class of the integrable many-body systems and
quantum field theories are those which are spectrally equivalent
to systems of free Fermions. It includes, in particular, the impenetrable
Bose gas \cite{Lenard64} and the Heisenberg XY model \cite{LSM61}.
For such models some of their two-point correlation functions can
be described by Fredholm determinants of integrable integral
operators \cite{IIKS-90,Sakhnovich-68,Deift-99}. The latter are
connected in a canonical way to a certain matrix Riemann--Hilbert
problem, which makes them accessible to a systematic asymptotic
analysis by a `nonlinear steepest descent method' \cite{DZ-93}.

The kernels of the integrable operators that are relevant for the
various correlation functions of the impenetrable Bose gas and of the
XY model and its variants are all generalizations of the so-called
sine kernel which also appears prominently in random matrix theory
\cite{Gaudin1961}. The original sine kernel applies to the case
of time-independent quantum correlation functions at zero temperature.
Generalizations were required in order to incorporate the time and temperature
dependencies. For the impenetrable Bose gas those were constructed in
\cite{IIKS-90} and for the isotropic XY model in \cite{CIKT92,CIKT93}.
The subsequent asymptotic long-time large-distance analysis at finite
temperature was carried out in the seminal paper \cite{IIKV-92} for
the impenetrable Bose gas and also, in parts, for the isotropic XY model
in \cite{IIKS93b,Jie98,GKS20b}.

Integrable models that do not map to non-interacting Fermions are
less well understood. In particular, very little is known about their
dynamical correlation functions in the experimentally most relevant
equilibrium settings, when the systems are in contact with a heat
bath at temperature $T$ or have relaxed to a state described by
a generalized Gibbs ensemble. A general method for the asymptotic
analysis of some of the two-point correlation functions of the
Lieb--Liniger model at finite coupling was devised in
\cite{KMS11a,KoTe11,Kozlowski15b}. It exploits the fact that the
form-factor series in the finite coupling case can be seen as
deformations of Fredholm series. Its most essential input is the
explicit large-$x$ asymptotics of the Fredholm determinant
of an auxiliary integral operator that depends on certain functional
parameters and has the interpretation as a generating function of
the form-factor series.

The auxiliary integral operator is defined by its kernel and its 
integration contour. The kernels required to study the correlation
functions of the Lieb--Liniger model at finite coupling are again
generalized sine kernels that now depend on additional functional
parameters. Their precise form, the number of the functional parameters
involved, and the choice of the integration contour depend on the details
of the respective correlation functions under consideration. The
analysis that facilitates to calculate the asymptotics of the static
ground-state correlation functions was performed in \cite{KKMST-09-RHp},
the finite-temperature static case was treated in \cite{S-10}, while
the asymptotic analysis underlying the study of the dynamical ground-state
(i.e., zero temperature) correlation functions was worked out in \cite{K-11}.

Our work generalizes the three aforementioned cases. It is devoted to
the asymptotic analysis of the Fredholm determinant of an integrable
integral operator with an even more general kernel. The motivation
for this work lies in subsequent applications to the Lieb--Liniger Bose
gas at finite coupling. However, we will be dealing with an independent
clearly defined mathematical problem that is of some interest in itself.
It can be solved with full rigour and will have many more applications,
some of which will be pointed out below.

We shall apply the nonlinear steepest descent method of Deift and
Zhou \cite{DZ-93} to the Fredholm determinant of an integrable integral
operator whose kernel depends on four functional parameters: a phase
function $u$ that allows us to adapt the dispersion relation, e.g.,
to the dressed finite-temperature case, a filling fraction $\vartheta$
used to model different equilibrium states, and two functions $g$ and
$\nu$ that are needed in the formalism of \cite{Kozlowski15b} to
connect the Fredholm determinant with the underlying form-factor series
at finite coupling. All functional parameters will be characterized
by their analytic properties. Certain aspects of our analysis were
already developed in \cite{IS-99}.

Our main results are summarized in three theorems in
Section~\ref{sec:main_results}. In these theorems we describe the
leading large-$x$ asymptotics of the Fredholm determinant. We
provide explicit expressions for the leading exponential decay,
the logarithmic corrections, and the constant term as functionals of
the parameter functions. The form of the constant and of the
sub-leading corrections (that are worked out explicitly in one case)
depends on the number of poles on the real axis of a certain combination
of the functional parameters. Sections~\ref{sec:the-initial-RHp}-%
\ref{sec:n-poles-on-the-real-axis}
are devoted to the derivation of the main results with some of the
more technical details and some auxiliary material deferred to
several appendices. In our summary Section~\ref{sec:discussion}
we point out more applications and connections with existing works.

\section{Statement of the problem and main results} \label{sec:main_results}
\subsection{The Fredholm determinant}
We shall study the large-$x$ asymptotic behaviour of the Fredholm
determinant
\begin{equation} \label{eq:Fred}
    \det_{\mathbb R} (\id + \rmV)
\end{equation}
of an integral operator $\rmV$ acting on $L^2 ({\mathbb R})$ and
depending parametrically on $x \in {\mathbb R}$. The kernel
$V: {\mathbb R}^2 \rightarrow {\mathbb C}$ that defines the
integral operator is of the form
\begin{equation}
    \label{eq:kernel-V-as-scalar-product}
    V (\lambda, \mu)
    = \frac{
        \mathbf{E}_L^\intercal (\lambda) \cdot \mathbf{E}_R (\mu) }{
        \lambda - \mu },
\end{equation}
where $\intercal$ denotes the transposition and where $\mathbf{E}_L,
\mathbf{E}_R: {\mathbb R} \rightarrow {\mathbb C}^2$ depend on
four functional parameters $u, g, \nu$, and $\vartheta$ and on
$x \in {\mathbb R}$. The function $u$ is supposed to depend on 
an additional parameter $\lambda_0 \in {\mathbb R}_+$. We will
sometimes make this explicit by writing $u(\cdot|\lambda_0)$, but
most of the time suppress the dependence on $\lambda_0$ in our
notation.

The four functional parameters $u, g, \nu$ and $\vartheta$ are defined
by their analytic properties in a strip
\begin{equation}
    \Omega = \bigl\{z \in \mathbb{C}\big| \abs{\Imag z} < w\bigr\},
\end{equation}
wherein $u, g$ and $\nu$ are holomorphic, whereas $\vartheta$ is
meromorphic. Further properties of these functions will be
specified below.

Throughout this work our initial functional parameters will
mostly appear in certain combinations. We, first of all, define
a function $e:\Omega \rightarrow {\mathbb C}$ by
\begin{equation}
\label{eq:function-e}
    e (\lambda)
    = \exp\biggl\{
        - \frac{\rmi x}{2} u (\lambda|\lambda_0)
        - \frac{g (\lambda)}{2} \biggr\},
\end{equation}
which is at the heart of our asymptotic analysis, since it contains
the parameter $x$. This function will be complemented by
a function $E: {\mathbb R} \rightarrow {\mathbb C}$,
\begin{equation}
\label{eq:function-E}
    E (\lambda)
    = 
    - e (\lambda) \mathrm{C}_- (\lambda)
    + \frac{ e^{-1} (\lambda) }{
        \exp(- 2 \pi \mathrm{i} \nu (\lambda)) - 1},
\end{equation}
where
\begin{equation}
    \label{eq:Cauchy-transform}
    \mathrm{C} (\lambda)
    = \int\limits_{\mathbb R} \frac{\dd \mu}{2 \pi \mathrm{i}}
    \frac{ e^{-2} (\mu) }{ \mu - \lambda}
\end{equation}
is the Cauchy transform of $e^{-2}$ and $\rmC_-$ its boundary value from
the right hand side of $\mathbb R$ conceived a contour oriented from
$- \infty$ to $+ \infty$.

The functions $\mathbf{E}_L$ and $\mathbf{E}_R$ in the expression
\eqref{eq:kernel-V-as-scalar-product} for the kernel can then be
introduced as
\begin{equation}
    \label{eq:vectors-E}
    \mathbf E_L (\lambda)
    = \sin[ \pi \nu (\lambda) ]
    \begin{pmatrix}
        - e (\lambda) \\ E (\lambda)
    \end{pmatrix},
    \qquad
    \mathbf E_R (\lambda)
    = \frac{
        4 \vartheta (\lambda) \sin[ \pi \nu (\lambda) ] }{
        2 \pi \mathrm{i} }
    \begin{pmatrix}
        E (\lambda) \\ e (\lambda)
    \end{pmatrix}.
\end{equation}
With this the kernel \eqref{eq:kernel-V-as-scalar-product} is formally
defined. The assumptions on the functional pa\-ra\-me\-ters below will
guarantee that the integral defining the Cauchy transform in 
\eqref{eq:Cauchy-transform} exists and that the series defining
the Fredholm determinant \eqref{eq:Fred} is absolutely convergent (see
Appendix~\ref{app:existence_of_Fred}).

\subsection{Main notations and assumptions}
\label{sec:assumptions}
In the asymptotic analysis carried out below, and consequentially in our
main results as well, the functions $\nu$ and $\vartheta$ will often occur
in specific combinations defining the functions
\begin{equation}
\label{eq:definition-of-Lcal-ell-and-Lcal-r}
    \Lcal_\ell (\lambda)
    = - \frac{1}{ 2 \pi \rmi }
    \ln\big[
        1 + \vartheta (\lambda)
        \big( \rme^{2 \pi \rmi \nu (\lambda)} - 1 \big)
    \big],
    \quad
    \Lcal_r (\lambda)
    = \frac{1}{ 2 \pi \rmi }
    \ln\big[
        1 + \vartheta (\lambda)
        \big( \rme^{- 2 \pi \rmi \nu (\lambda)} - 1 \big)
    \big],
\end{equation}
where we understand the logarithms as
\begin{equation}
\label{eq:complex-logarithms}
    \ln\big[
        1 + \vartheta (\lambda)
        \big( \rme^{\pm 2 \pi \rmi \nu (\lambda)} - 1 \big)
    \big]
    = \int\limits_{- \infty}^\lambda
    \dd{\ln\big[
        1 + \vartheta (\mu)
        \big( \rme^{\pm 2 \pi \rmi \nu (\mu)} - 1 \big)
    \big]}.
\end{equation}
Here the integral is along a curve $\Ccal_{\lambda_0}$
that agrees with the real axis, except in small vicinities
of the poles of the integrand on the real axis
that are circumvented on small semicircles,
see Figure~\ref{fig:contour-Ccal-lambda-0},
which have to be chosen in
a way compatible with the assumptions \ref{item:assumption-5}
and \ref{as:number6} below.
\begin{figure}[t]
    \centering
    \inputfig{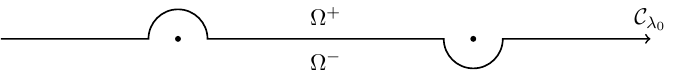}
    \caption{
        Sketch of the integration contour $\Ccal_{\lambda_0}$ in the strip
        $\Omega = \bigl\{z \in \mathbb{C}\big| \abs{\Imag z} < w\bigr\}$.
        The contour goes along the real
        axis and circumvents possible poles of $\Lcal_\ell'$
        and $\Lcal_r'$ on small semicircles
        (here there are two such poles on the real axis),
        while equations \eqref{eq:condition-tau-at-lambda-0} and 
        \eqref{eq:monodromy-condition} are satisfied.
        This implies that the precise form of the contour
        depends on the location of $\lambda_0$.
        Equation~\eqref{eq:saddle_point_not_on_pole} guarantees
        that $\lambda_0 \in \Ccal_{\lambda_0}$. The contour
        $\Ccal_{\lambda_0}$ divides $\Omega$ into two parts: a part
        $\Omega^+$ above the contour and a part $\Omega^-$ below the contour.
        }
    \label{fig:contour-Ccal-lambda-0}
\end{figure}
For later convenience we further
introduce the functions
\begin{subequations}
\label{eq:definition-of-tau-Lcal-and-Lcal-prime}
\begin{align}
    \label{eq:definition-of-tau}
    \tau (\lambda) & = \Lcal_\ell (\lambda) - \Lcal_r (\lambda), \\[1ex]
    \label{eq:definition-of-Lcal}
    \Lcal (\lambda | \lambda_0)
    & = \Lcal_\ell (\lambda)
    \cdot \indicator{\Real(\lambda - \lambda_0) < 0} (\lambda)
    + \Lcal_r (\lambda)
    \cdot \indicator{\Real(\lambda - \lambda_0) > 0} (\lambda),
    \\[1ex]
    \label{eq:definition-of-Lcal-prime}
    \Lcal' (\lambda | \lambda_0)
    & = \Lcal_\ell' (\lambda)
    \cdot \indicator{\Real(\lambda - \lambda_0) < 0} (\lambda)
    + \Lcal_r' (\lambda)
    \cdot \indicator{\Real(\lambda - \lambda_0) > 0} (\lambda),
\end{align}
\end{subequations}
where $\Lcal_\ell (\lambda)$ and $\Lcal_r (\lambda)$
are given by~\eqref{eq:definition-of-Lcal-ell-and-Lcal-r},
and $\indicator{}$ is the indicator function. We stress that
$\Lcal'(\lambda|\lambda_0)$ is not a distributional derivative,
but a point-wise derivative away from $\lambda_0$.
\begin{assumptions}
The functions
$u$, $\vartheta$, $\nu$, and $g$
have the following properties in $\Omega$:
\begin{enumerate}[label=(\arabic*)]
    \item
    The function $u(\cdot|\lambda_0)$ is real valued on $\mathbb{R}$.
    There exist $c, \varepsilon, \eta > 0$ such that
    \begin{equation}
    \label{eq:asymptotic-of-u}
        \Im u (\lambda \mp \rmi y|\lambda_0)
        \geq c y (\pm \lambda)^\eta
    \end{equation}
    for $\mathbb{R} \ni \lambda \to \pm \infty$ and $0 < y < \varepsilon$.
    \label{item:assumption-1}
    \item
    The point $\lambda = \lambda_0$ is the unique solution of the
    equation $u' (\lambda|\lambda_0) = 0$ in $\Omega$. We shall call
    it `the saddle point'. We assume that $u'' (\lambda_0|\lambda_0) < 0$.
    
    There is a function $p$, holomorphic in $\Omega$ such that
    \begin{equation}
        \lim_{\lambda_0 \rightarrow + \infty} u(\lambda|\lambda_0)
        = p(\lambda)
    \end{equation}
    pointwise in $\lambda$. There exist $c, \varepsilon > 0$ such that
    $\Im p(\lambda + \rmi y) \geq c y$ for all $\lambda \in {\mathbb R}$
    and $0 \leq y < \varepsilon$.
    
    We further assume that the saddle point satisfies
    \begin{equation} \label{eq:saddle_point_not_on_pole}
        1 + \vartheta (\lambda_0)
        (\rme^{\pm 2 \pi \rmi \nu (\lambda_0)} - 1)
        \neq 0.
    \end{equation}
    \label{item:assumption-2}
    \item
    \label{item:assumption-on-nu-and-g}
    The functions $\nu$ and $g$ along with their derivatives
    $\nu'$, $g'$ and $g''$ are bounded on $\Omega$.
    \item
    \label{item:assumption-on-vartheta}
    The meromorphic function $\vartheta (\lambda)$ is real, bounded
    and non-negative on $\mathbb{R}$. Furthermore, 
    $\vartheta^\frac{1}{4} \in L^1 ({\mathbb R})$, $\vartheta (\lambda)
    \rightarrow 0$ for $\Re \lambda \rightarrow \pm \infty$, and there
    exists $c > 0$ such that for all $\lambda \in {\mathbb R}$
    \begin{equation} \label{eq:theta_bounds_u}
        \vartheta^\frac{1}{4} (\lambda) \times
        \sup_{\nu \in [-\sigma, \sigma]}
        \bigl\{
            \bigl|u'(\lambda + \nu | \lambda_0)\bigr|
            + \bigl|u'(\lambda + \nu | \lambda_0)^2\bigr|
            + \bigl|u''(\lambda + \nu | \lambda_0)\bigr|\bigr\}
        \leq c,
    \end{equation}
    uniformly in $\lambda_0$ for some small fixed $\sigma > 0$.
    \item
    The function $\tau$ satisfies
    \begin{equation}
    \label{eq:condition-tau-at-lambda-0}
        \bigl|\Real \tau (\lambda)\bigr| < \flatfrac{1}{2}
    \end{equation}
    for all $\lambda \in \Ccal_{\lambda_0}$ with
    $\Real (\lambda - \lambda_0) \ge 0$.
    \label{item:assumption-5}
    \item \label{as:number6}
    All zeroes of the functions
    $1 + \vartheta (\mu) \big( \rme^{\pm 2 \pi \rmi \nu (\mu)} - 1 \big)$
    in $\Omega$ are of first order, and the `mono\-dromy conditions'
    \begin{equation}
    \label{eq:monodromy-condition}
        \int\limits_{\Ccal_{\lambda_0}}
        \dd{\ln\big[
            1 + \vartheta (\mu)
            \big( \rme^{\pm 2 \pi \rmi \nu (\mu)} - 1 \big)
        \big]}
        = 0
    \end{equation}
    are satisfied.
\end{enumerate}
\end{assumptions}
\begin{figure}[t]
    \centering
    \inputfig{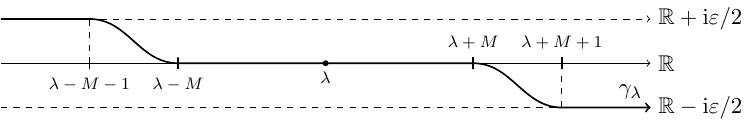}
    \caption{
        We define the Cauchy transform \eqref{eq:Cauchy-transform}
        as $\rmC (\lambda) = \lim_{M \rightarrow + \infty}
        \int_{\gamma_\lambda} \frac{\dd{\mu}}{2 \pi \mathrm{i}}
        \frac{ e^{-2} (\mu) }{ \mu - \lambda}$, where $\gamma_\lambda$
        is the sketched contour. Owing to Assumption~\ref{item:assumption-1}
        the integral over $\gamma_\lambda$ is absolutely convergent
        for given $M$ and $\varepsilon$. Hence, $\rmC (\lambda)$ is a
        smooth function on $\mathbb R$. Since $e$ is holomorphic in
        $\Omega$ the value of the integral over $\gamma_\lambda$ does
        not depend on $M$, and the integral over $\mathbb R$ in
        \eqref{eq:Cauchy-transform} is well defined.
        }
    \label{fig:contour-for-Cauchy-transform}
\end{figure}%
\begin{figure}[t]
    \centering
    \inputfig{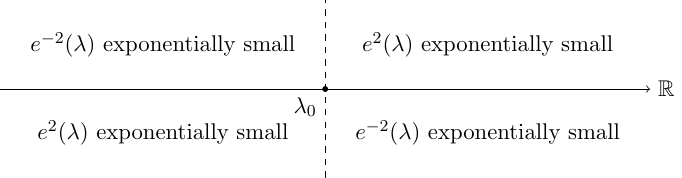}
    \caption{
        Sketch of the behaviour of the functions $e^{\pm 2} (\lambda)$
        in the complex plane as functions of the large parameter $x$.
        By virtue of Assumptions~\ref{item:assumption-1}-\ref{item:assumption-on-nu-and-g}
        the function $e^2$ becomes small for large $x$ northeast and
        southwest of the saddle point $\lambda_0$ and large northwest
        and southeast of it. The behaviour of $e^{-2}$
        is complementary. Since $\Im( u (\lambda) ) = 0$ for 
        $\lambda \in \mathbb{R}$ both, $e^2$ and $e^{-2}$, become
        rapidly oscillating functions of $\lambda \in {\mathbb R}$,
        once $x$ is large.
    }
    \label{fig:behaviour-of-function-e}
\end{figure}
\begin{remarks}
\begin{enumerate}
    \item 
    Assumption~\ref{item:assumption-1} guarantees that the
    Cauchy transform \eqref{eq:Cauchy-transform} is well defined
    and that its minus boundary value is a smooth function on
    $\mathbb R$ (see Figure~\ref{fig:contour-for-Cauchy-transform}).
    \item 
    Due to the properties of the functions $u$, $g$, $\nu$, and
    $\vartheta$ we may deform the contour in the definition
    \eqref{eq:Fred} of the Fredholm determinant and in the
    definition~\eqref{eq:Cauchy-transform} of the Cauchy transform
    from $\mathbb R$ to $\Ccal_{\lambda_0}$.
    \item
    Assumption~\ref{item:assumption-2} implies typical
    saddle-point behaviour of the functions $e^2$ and $e^{- 2}$
    in the vicinity of $\lambda_0$, see
    Figure~\ref{fig:behaviour-of-function-e}. However, since the
    two functions behave in a complementary way, a simple saddle-point
    analysis is not possible.
    \item
    Assumption~\ref{item:assumption-on-vartheta} allows us
    to prove the convergence of the series defining the
    Fredholm determinant \eqref{eq:Fred} by direct estimates
    and to connect our analysis to the `static limit' considered in
    \cite{KKMST-09-RHp,S-10} (see Appendix~\ref{app:existence_of_Fred}).
    \item
    The assumption in \ref{as:number6}, that all zeroes of 
    $1 + \vartheta (\mu) \big( \rme^{\pm 2 \pi \rmi \nu (\mu)} - 1 \big)$
    be of first order, might be relaxed, but this would require to
    consider an additional case when solving the system of linear
    algebraic equations that determine the pole contributions in
    Section~\ref{sec:n-poles-on-the-real-axis}.
\end{enumerate}
\end{remarks}

\subsection{Main results}
As we shall see, those poles of $\Lcal'(\lambda|\lambda_0)$ 
in $\Omega$ that are not poles of $\vartheta$ contribute to
the asymptotic expansion of the Fredholm determinant,
see Figure~\ref{fig:poles-contributing-to-the-AE}.
Their contribution is exponentially small in $x$,
unless they are located on the real axis. We shall denote
the set of those poles
\begin{equation} \label{def:Scal}
    \Scal = \bigl\{
        \lambda \in \Omega \big|
       1/\Lcal' (\lambda|\lambda_0) = 0;
       1/\vartheta(\lambda) \ne 0
    \bigr\}
\end{equation}
and simply refer to them as `the poles'.
Due to the Assumptions \ref{item:assumption-on-nu-and-g}
and~\ref{item:assumption-on-vartheta},
the poles are situated in a compact subset of the strip $\Omega$,
see Figure~\ref{fig:configuration-of-the-poles}.
\begin{figure}[t]
    \centering
    \inputfig{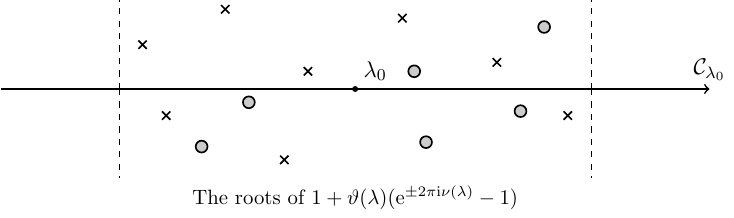}
    \caption{
        Those poles of $\Lcal_\ell' (\lambda)$
        that are roots of
        $1 + \vartheta (\lambda) (\rme^{2 \pi \rmi \nu (\lambda)} - 1)$
        are marked by crosses,
        while those poles of $\Lcal_r' (\lambda)$
        that are roots of
        $1 + \vartheta (\lambda) (\rme^{-2 \pi \rmi \nu (\lambda)} - 1)$
        are marked by circles.
        Here we sketch the case when
        all poles of $\Lcal_\ell' (\lambda)$ and $\Lcal_r' (\lambda)$
        are away from the real axis. Note that the poles are
        finite in number and located in a compact subset of $\Omega$.
        }
    \label{fig:configuration-of-the-poles}
\end{figure}

In order to formulate our results we split the set $\Scal$ of
poles into two subsets, $\Scal^+$ and $\Scal^-$, according to
the definition
\begin{equation}
\label{eq:definition-Scal-pm}
    \Scal^\pm
    = \bigl\{
        \lambda \in \Scal \cap \Omega^\pm
        \big|
        \Real \lambda < \lambda_0
    \bigr\}
    \cup
    \bigl\{
        \lambda \in \Scal \cap \Omega^\mp
        \big|
        \Real \lambda > \lambda_0
    \bigr\}.
\end{equation}
We denote the poles in these sets as $s_j^\pm$ for $j = 1, \dots, n^\pm$,
where $n^\pm$ are the cardinalities of the sets $\Scal^\pm$. Then
\begin{equation}
\label{eq:definition-Scal-pm-set}
    \Scal^\pm = \{ s_j^\pm \}_{j = 1}^{n^\pm}.
\end{equation}
\begin{figure}[t]
    \centering
    \inputfig{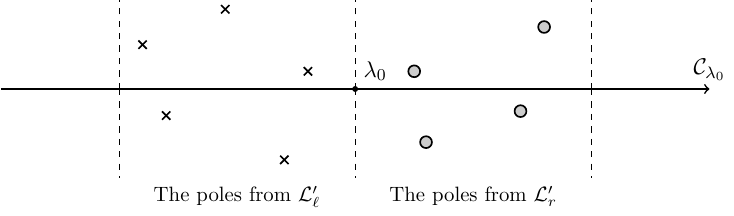}
    \caption{
        `The poles',
        i.e., those points in $\Omega$ that are poles of
        $\Lcal' (\cdot|\lambda_0)$ but not of $\vartheta$.
    }
    \label{fig:poles-contributing-to-the-AE}
\end{figure}

Anticipating the asymptotic analysis in Section~\ref{sec:transform_RH}
(cf.\ \eqref{eq:varkappa}) we introduce the function 
\begin{equation}
\label{eq:definition-of-varkappa-at-lambda-0}
    \varkappa (\lambda_0 | \lambda_0)
    = \exp\Bigg\{
        - \int\limits_{\Ccal_{\lambda_0}}
        \dd{\lambda}
        \Lcal' (\lambda | \lambda_0)
        \ln \left[
            (\lambda_0 - \lambda)
            \cdot \sgn (\Real (\lambda_0 - \lambda))
        \strut\right]
    \Bigg\}
\end{equation}
that will appear in several places in the theorems below.

\begin{theorem}
\label{thm:Fredholm-det-AE-no-poles}
    Let Assumptions \ref{item:assumption-1}--\ref{as:number6}
    be fulfilled, and let $\Scal \cap {\mathbb R} = \emptyset$.
    Then the Fredholm determinant
    of the integrable integral operator $\rmV$
    with kernel~\eqref{eq:kernel-V-as-scalar-product}
    has the large-$x$ asymptotic expansion
    \begin{multline}
    \label{eq:thm:Fredholm-det-AE-no-poles}
        \det_{\mathbb R} (\id + \rmV)
        = 
        \Acal (\lambda_0)
        \ x^{ - \flatfrac{\tau^2 (\lambda_0)}{2}}
        \\
        \times
        \exp\Bigg\{
            - \int\limits_{\Ccal_{\lambda_0}}
            \dd{\lambda} \Lcal (\lambda | \lambda_0)
            \left( \rmi x u' (\lambda) + g' (\lambda) \right)
        \Bigg\}
        \cdot
        \Big(
            1
            + \Ocal \Big( \frac{ (\ln x)^2 }{x} \Big)
        \Big),
    \end{multline}
    where $\tau$ was defined
    in~\eqref{eq:definition-of-tau},
    and where the constant $\Acal (\lambda_0)$ is
    \begin{multline}
    \label{eq:thm:integration-constant}
        \Acal (\lambda_0)
        =
        \frac{
            G (\tau (\lambda_0) + 1) }{
            (2 \pi)^{ \flatfrac{\tau (\lambda_0)}{2}} }
        \left( \rmi u'' (\lambda_0) \right)^{
            - \flatfrac{\tau^2 (\lambda_0)}{2} }
        \left(
            \varkappa (\lambda_0 | \lambda_0)
        \strut\right)^{\tau (\lambda_0)}
        \\
        \times
        \exp\Bigg\{
            \int\limits_{\Ccal_{\lambda_0}}
            \dd{\lambda}
            \Lcal (\lambda | \lambda_0)
            \cdot
            \partial_\lambda
            \ln \left[
                \vartheta (\lambda) \sin^2 (\pi \nu (\lambda))
            \right]
        \Bigg\}
        \\
        \times
        \exp\Bigg\{
            \frac12
            \int\limits_{\Ccal_{\lambda_0}}
            \dd{\lambda}
            \int\limits_{\Ccal_{\lambda_0}}
            \dd{\mu}
            \frac{
                \Lcal' (\lambda | \lambda_0) \Lcal (\mu | \lambda_0)
                - \Lcal (\lambda | \lambda_0) \Lcal' (\mu | \lambda_0) }{
                \lambda - \mu}
        \Bigg\}.
    \end{multline}
    The function $G$ is the Barnes $G$-function,
    the functions $\Lcal (\cdot | \lambda_0)$
    and $\Lcal' (\cdot | \lambda_0)$
    were defined in~\eqref{eq:definition-of-tau-Lcal-and-Lcal-prime},
    and $\varkappa (\lambda_0 | \lambda_0)$
    is given by~\eqref{eq:definition-of-varkappa-at-lambda-0}.
\end{theorem}
\begin{remark}
    In Theorem~\ref{thm:Fredholm-det-AE-no-poles} the integration
    contour is $\Ccal_{\lambda_0} = {\mathbb R}$, since there are
    no poles on the real axis. This allows us to simplify the integrand
    in the expression~\eqref{eq:definition-of-varkappa-at-lambda-0},
    using that
    \begin{equation}
        \ln \left[
            (\lambda_0 - \lambda)
            \cdot \sgn (\Real (\lambda_0 - \lambda))
        \strut\right]
        = \ln \abs{\lambda_0 - \lambda}
        \qquad
        \forall \lambda \in \mathbb{R}.
    \end{equation}
\end{remark}

Next, we would like to formulate our results
which include possible contributions of the poles
to the asymptotics of the Fredholm determinant.
The poles in the set $\Scal^+$ enter the asymptotic formulae
through certain coefficients $h_j^+$ , $j = 1, \dots, n^+$,
defined as
\begin{subequations}
\label{eq:definition-of-h-plus-and-h-minus}
\begin{equation}
    h_j^+ =
    - \dfrac{
        e^{- 2} (s_j^+) (1 - \vartheta (s_j^+)) }{
        2 \pi \rmi \alpha^2 (s_j^+) }
    \cdot
    \begin{cases}
        \dfrac{
            \exp\big[ 2 \pi \rmi \Lcal_\ell (s_j^+) \big] }{
            \Lcal_\ell' (s_j^+) },
        & \Real(s_j^+) < \lambda_0,
        \\
        \dfrac{
            \exp\big[- 2 \pi \rmi \Lcal_r (s_j^+)\big] }{
            \Lcal_r' (s_j^+) },
        & \lambda_0 < \Real(s_j^+),
    \end{cases}
\end{equation}
whereas the contributions from the poles in the set $\Scal^-$
enter in terms of the coefficients $h_j^-$ for $j = 1, \dots, n^-$,
\begin{equation}
    h_j^- =
    - \dfrac{
        4 e^{2} (s_j^-)
        \alpha^2 (s_j^-)
        \vartheta (s_j^-)
        \sin^2 [\pi \nu (s_j^-)] }{
        2 \pi \rmi }
    \cdot
    \begin{cases}
        \dfrac{
            \exp\big[ 2 \pi \rmi \Lcal_\ell (s_j^-) \big] }{
            \Lcal_\ell' (s_j^-) },
        & \Real(s_j^-) < \lambda_0,
        \\
        \dfrac{
            \exp\big[ - 2 \pi \rmi \Lcal_r (s_j^-) \big] }{
            \Lcal_r' (s_j^-) },
        & \lambda_0 < \Real(s_j^-).
    \end{cases}
\end{equation}
\end{subequations}
Here the function $\alpha (\lambda)$ is given by
\begin{equation}
\label{eq:definition-of-alpha}
    \alpha (\lambda)
    = \exp\Bigg\{
        \int\limits_{\Ccal_{\lambda_0}}
        \dd{\mu}
        \frac{ \Lcal (\mu | \lambda_0) }{\mu - \lambda}
    \Bigg\}.
\end{equation}
In addition to the functions that we already introduced
we also need the following coefficients:
\begin{equation}
\label{eq:definition-of-b12-and-b21}
    b_{12} (\lambda_0)
    = \frac{
        \rmi \sqrt{\pi} 
        \rme^{\flatfrac{\pi \rmi }{4}}
        2^{\tau (\lambda_0)} }{
        \rme^{\flatfrac{ \pi \rmi \tau (\lambda_0) }{ 2 }}
        n (\lambda_0) \Gamma (\tau (\lambda_0))
    },
    \qquad
    b_{21} (\lambda_0)
    = \frac{ \rmi \tau (\lambda_0) }{2 b_{12} (\lambda_0)},
\end{equation}
where the function $n (\lambda_0)$ is given by
\begin{equation}
\label{eq:definition-of-n}
    n (\lambda_0)
    = - 4
    e^2 (\lambda_0)
    \big[ - \flatfrac{ x u'' (\lambda_0) }{2} \big]^{ - \tau (\lambda_0)}
    \varkappa^2 (\lambda_0 | \lambda_0)
    \vartheta (\lambda_0)
    \sin^2 [\pi \nu (\lambda_0)].
\end{equation}
The function $\varkappa (\lambda_0 | \lambda_0)$ was defined
in~\eqref{eq:definition-of-varkappa-at-lambda-0}.

Now we formulate our second main result, describing
the special case of two poles on the real axis.
All possible configurations of those poles
are illustrated in Figure~\ref{fig:4-configurations-of-the-poles}.
\begin{figure}[t]
    \centering
    \begin{subfigure}[]{0.5\textwidth}
        \centering
        \inputfig{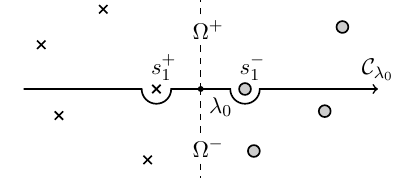}
        \caption{}
    \end{subfigure}%
    \begin{subfigure}[]{0.5\textwidth}
        \centering
        \inputfig{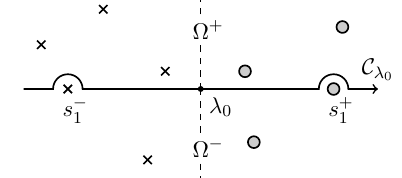}
        \caption{}
    \end{subfigure}%
    \newline%
    \begin{subfigure}[]{0.5\textwidth}
        \centering
        \inputfig{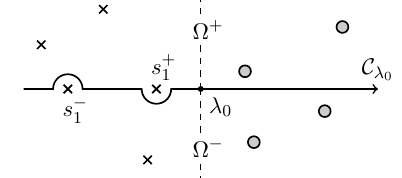}
        \caption{}
        \label{fig:sl-1}
    \end{subfigure}%
    \begin{subfigure}[]{0.5\textwidth}
        \centering
        \inputfig{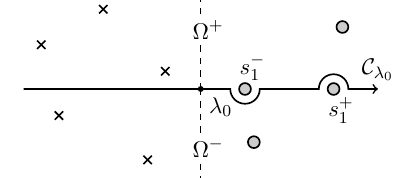}
        \caption{}
        \label{fig:sl-2}
    \end{subfigure}%
    \caption{
        Schematic illustration of the contour $\Ccal_{\lambda_0}$
        for the case of two poles on the real axis
        considered in Theorem~\ref{thm:Fredholm-det-AE-two-poles}.
        In Figures~\ref{fig:sl-1} and~\ref{fig:sl-2},
        the relative position of the poles
        ($s_1^- < s_1^+$ or $s_1^- > s_1^+$) does not matter
        (given that conditions \eqref{eq:condition-tau-at-lambda-0}, 
        \eqref{eq:monodromy-condition} are fulfilled).
    }
    \label{fig:4-configurations-of-the-poles}
\end{figure}

\begin{theorem}
\label{thm:Fredholm-det-AE-two-poles}
    Let Assumptions \ref{item:assumption-1}--\ref{as:number6}
    be satisfied.
    Let the sets $\Scal^\pm \cap {\mathbb R}$ both have cardinality one
    and set $\{s_1^\pm\} = \Scal^\pm \cap {\mathbb R}$,
    see Figure~\ref{fig:4-configurations-of-the-poles},
    and $(s_1^+ - s_1^-)^2 + h_1^+ h_1^- \neq 0$.
    Then the Fredholm determinant of the integrable integral operator
    $\rmV$ with kernel \eqref{eq:kernel-V-as-scalar-product}
    has the large-$x$ asymptotic expansion
    \begin{multline}
        \label{eq:thm:Fredholm-det-AE-two-poles}
        \det_{\mathbb R} (\id + \rmV)
        = 
        \Acal (\lambda_0)
        \ x^{ - \flatfrac{\tau^2 (\lambda_0)}{2}}
        \exp\Bigg\{
            - \int\limits_{\Ccal_{\lambda_0}}
            \dd{\lambda} \Lcal (\lambda | \lambda_0)
            \left( \rmi x u' (\lambda) + g' (\lambda) \right)
        \Bigg\}
        \\
        \times
        \Bigg\{
            1 + \frac{ h_1^+ h_1^- }{ (s_1^+ - s_1^-)^2 }
            + \frac{ 1 + \ocal(1) }{ x^{1/2} }
            \frac{ \sqrt{2} }{ \sqrt{- u'' (\lambda_0)} }
            \Bigg(
                \frac{ h_1^+ b_{21} (\lambda_0) }{ (\lambda_0 - s_1^+)^2 }
                + \frac{ h_1^- b_{12} (\lambda_0) }{ (\lambda_0 - s_1^-)^2 }
            \Bigg)
        \Bigg\},
    \end{multline}
    where the functions $\tau$
    and $\Lcal (\cdot| \lambda_0)$
    were defined by~\eqref{eq:definition-of-tau-Lcal-and-Lcal-prime}
    and the constant $\Acal (\lambda_0)$
    by~\eqref{eq:thm:integration-constant}.
    The coefficients $h_1^\pm$, $b_{12} (\lambda_0)$, and $b_{21} (\lambda_0)$
    were defined in~\eqref{eq:definition-of-h-plus-and-h-minus}
    and~\eqref{eq:definition-of-b12-and-b21}.
\end{theorem}

Finally, we formulate our third main result that covers
the case of $n$ poles in $\Omega$, possibly located
on the real axis, in the so-called `static limit'
$\lambda_0 \to + \infty$. We introduce two
matrices $\rmA^\pm$ with elements
\begin{equation}
\label{eq:matrices-A-pm}
    A_{jk}^- = \frac{h_k^-}{s_j^+ - s_k^-},
    \quad
    A_{kj}^+ = \frac{h_j^+}{s_k^- - s_j^+},
    \quad
    j = 1, \dots, n^+,
    \quad
    k = 1, \dots, n^-,
\end{equation}
where the coefficients $h_k^\pm$ for $k = 1, \dots, n^\pm$
were defined in~\eqref{eq:definition-of-h-plus-and-h-minus}.
For the $n \times n$ unit matrix we shall employ the
notation $\rmI_n$.
\begin{theorem}
    \label{thm:Fredholm-det-AE-static}
    Let Assumptions \ref{item:assumption-1}--\ref{as:number6}
    be fulfilled.
    Assume that the matrix $\rmI_{n^+} - \rmA_\infty$ with 
    $\rmA_\infty = \lim_{\lambda_0 \rightarrow + \infty} \rmA^- \rmA^+$
    and $\rmA^\pm$ given by~\eqref{eq:matrices-A-pm} is invertible.
    Then, as $\lambda_0 \to + \infty$, the Fredholm determinant
    of the integrable integral operator $\rmV$
    with kernel~\eqref{eq:kernel-V-as-scalar-product}
    has the large-$x$ asymptotic expansion
    \begin{multline}
        \label{eq:thm:Fredholm-det-AE-static}
        \lim_{\lambda_0 \rightarrow + \infty} \det_{\mathbb R} (\id + \rmV)
        =
        \Acal_\infty
        \det\big[ \mathrm{I}_{n^+} - \rmA_\infty \big]
        \\
        \times
        \exp\Bigg\{
            - \int\limits_{\Ccal_{\infty}}
            \dd{\lambda} \Lcal_\ell (\lambda)
            \left( \rmi x p' (\lambda) + g' (\lambda) \right)
        \Bigg\}
        \big( 1 + \Ocal (\rme^{-m x}) \big)
    \end{multline}
    for some $m > 0$.
    The constant $\Acal_\infty$ is given by
    \begin{multline}
        \Acal_\infty
        = \exp\Bigg\{
            \int\limits_{\Ccal_\infty}
            \dd{\lambda}
            \Lcal_\ell (\lambda) \partial_\lambda
            \ln\big( \rme^{-2 \pi \rmi \nu (\lambda)} - 1 \big)
        \Bigg\}
        \\
        \times
        \exp\Bigg\{
            \frac12
            \int\limits_{\Ccal_\infty}
            \dd{\lambda}
            \int\limits_{\Ccal_\infty}
            \dd{\mu}
            \frac{
                \Lcal_\ell' (\lambda) \Lcal_\ell (\mu)
                - \Lcal_\ell (\lambda) \Lcal_\ell' (\mu) }{
                \lambda - \mu}
        \Bigg\}.
    \end{multline}
    The function $\Lcal_\ell$
    was defined by~\eqref{eq:definition-of-Lcal-ell-and-Lcal-r}.
    If $n^+ = 0$ or $n^- = 0$,
    then the determinant on the right-hand side of
    \eqref{eq:thm:Fredholm-det-AE-static} is equal to $1$ by
    definition.
\end{theorem}

\begin{remarks}
\begin{enumerate}
    \item
    Theorem~\ref{thm:Fredholm-det-AE-static}
    is a generalization of~\cite[Theorem 2.2]{S-10}
    to the case when the `momentum function' $p$ is
    different from $p (\lambda) = \lambda$.
    The contribution of those poles,
    that are away from the real axis,
    see Figure~\ref{fig:poles-contributing-to-the-AE},
    is exponentially small and can be neglected.

    \item
    An important example to which our above theorems apply is
    the impenetrable Bose gas. Its finite temperature dynamical
    correlation functions were studied in \cite{IIKS-90,IIKV-92}.
    In order to apply our theorems we have to set
    \begin{subequations}
    \begin{align}
        u(\lambda|\lambda_0) & = \lambda - \frac{\lambda^2}{2 \lambda_0}, \\
        g (\lambda) & = 0, \\
        \nu(\lambda) & = \frac12, \\ \label{eq:Fermi_distribution}
        \vartheta (\lambda) & = \frac{1}{1 + \rme^{\frac{\lambda^2 - h}{T}}},
    \end{align}
    \end{subequations}
    where $T > 0$ and $h \in {\mathbb R}$ are the temperature and
    the chemical potential, two parameters that control the
    thermal equilibrium in the grand canonical ensemble. It is
    straightforward to verify that the above functional parameters
    satisfy Assumptions~\ref{item:assumption-1}--\ref{as:number6}
    if the contour $\Ccal_{\lambda_0}$ is chosen appropriately.
    We will consider the example of the impenetrable Bose gas in
    more detail in separate work.

    \item
    When we apply our analysis to the impenetrable Bose gas
    in thermal equilibrium \cite{IIKS-90,IIKV-92} we will be
    dealing with either no pole on the real axis or with two
    poles on the real axis, corresponding to negative or
    positive chemical potential ($h < 0$ or $h > 0$ in 
    \eqref{eq:Fermi_distribution}), respectively. This is the
    reason why we considered these physically most interesting
    cases in Theorems~\ref{thm:Fredholm-det-AE-no-poles}
    and~\ref{thm:Fredholm-det-AE-two-poles}.

    \item
    In Section~\ref{sec:conjecture}
    we formulate a conjecture on the case of $n$ poles
    for ${\mathbb R} \ni \lambda_0 < + \infty$
    which generalizes Theorem~\ref{thm:Fredholm-det-AE-static}.
\end{enumerate}
\end{remarks}

The remainder of this work is devoted to the proofs of
Theorems~\ref{thm:Fredholm-det-AE-no-poles}--%
\ref{thm:Fredholm-det-AE-static}.

\section{The initial matrix Riemann--Hilbert problem
    and its relation to the Fredholm determinant}
\label{sec:the-initial-RHp}

We note that the functions $\mathbf{E}_L$ and $\mathbf{E}_R$ in the
definition of the kernel function \eqref{eq:kernel-V-as-scalar-product}
have the property that
\begin{equation}
\label{eq:orthogonality-of-vectors-E}
    \mathbf{E}_L^\intercal (\lambda) \cdot \mathbf{E}_R (\lambda)
    = 0.
\end{equation}
For this reason the kernel is non-singular at $\lambda = \mu$.
Integral operators of the form~\eqref{eq:kernel-V-as-scalar-product}
with $\mathbf{E}_L$, $\mathbf{E}_R$ satisfying the
condition~\eqref{eq:orthogonality-of-vectors-E} are called
integrable integral operators
or completely integrable operators~\cite{Sakhnovich-68,IIKS-90,Deift-99}.
One of the important properties of such operators
is that one can construct their resolvents
from the solution to an associated matrix Riemann--Hilbert
problem. This allows one to study the Fredholm determinants
of such operators analytically.

If $\det_{\Ccal_{\lambda_0}} (\id + \rmV) \neq 0$, one can define functions
$\mathbf{F}_L, \mathbf{F}_R: {\mathbb C} \rightarrow {\mathbb C}^2$
as the unique solutions of the linear integral equations
\begin{subequations}
\label{eq:vectors-F}
\begin{align}
    \label{eq:vector-F-L}
    \mathbf{F}_L (\lambda)
    + \int\limits_{\Ccal_{\lambda_0}} \dd{\mu}
    V (\lambda, \mu) \mathbf{F}_L (\mu)
    = \mathbf{E}_L (\lambda),
    \\
    \label{eq:vector-F-R}
    \mathbf{F}_R (\lambda)
    + \int\limits_{\Ccal_{\lambda_0}} \dd{\mu}
    \mathbf{F}_R (\mu) V (\mu, \lambda)
    = \mathbf{E}_R (\lambda).
\end{align}
\end{subequations}
In terms of these functions the kernel of the resolvent $\rmR$ of $\rmV$
can be expressed as~\cite{IIKS-90}
\begin{equation}
\label{eq:kernel-R}
    R (\lambda, \mu)
    = \frac{
        \mathbf{F}_L^\intercal (\lambda)
        \cdot \mathbf{F}_R (\mu) }{
        \lambda - \mu },
    \qquad
    \mathbf{F}_L^\intercal (\lambda) \cdot \mathbf{F}_R (\lambda)
    = 0,
\end{equation}
meaning that the resolvent of an integrable operator is an
integrable operator.

If we define, on the other hand, a matrix $\chi (\lambda)$ and
calculate it inverse as
\begin{equation}
\label{eq:chi-and-chi-inverse}
    \chi (\lambda)
    = \mathrm{I}_2
    - \int\limits_{\Ccal_{\lambda_0}} \dd{\mu}
    \frac{
        \mathbf{F}_R (\mu) \cdot \mathbf{E}_L^\intercal (\mu) }{
        \mu - \lambda},
    \qquad
    \chi^{-1} (\lambda)
    = \mathrm{I}_2
    + \int\limits_{\Ccal_{\lambda_0}} \dd{\mu}
    \frac{
        \mathbf{E}_R (\mu) \cdot \mathbf{F}_L^\intercal (\mu) }{
        \mu - \lambda},
\end{equation}
then $\mathbf{F}_L$ and $\mathbf{F}_R$ are determined algebraically
by the relations
\begin{equation}
\label{eq:F-via-E-and-chi}
    \mathbf{F}_R (\lambda)
    = \chi (\lambda) \mathbf{E}_R (\lambda),
    \qquad
    \mathbf{F}_L^\intercal (\lambda)
    = \mathbf{E}_L^\intercal (\lambda) \chi^{-1} (\lambda).
\end{equation}

The matrix $\chi (\lambda)$, in turn, is a solution of the following
matrix Riemann--Hilbert problem which is guaranteed to exists,
if $\det_{\Ccal_{\lambda_0}} (\id + \rmV) \neq 0$.
\begin{RHp}
\label{RHp:chi}
    Determine $\chi (\lambda) \in \mathbb{C}^{2 \times 2}$ such that
    \begin{enumerate}
        \item
        $\chi (\lambda)$ is analytic
        in $\mathbb{C} \backslash \mathcal{C}_{\lambda_0}$
        and has continuous $\pm$ boundary values on $\mathcal{C}_{\lambda_0}$
        (cf.\ Figure~\ref{fig:contour-Ccal-lambda-0}).
        
        \item
        On the contour $\mathcal{C}_{\lambda_0}$
        the boundary values $\chi_\pm (\lambda)$
        satisfy the jump condition
        $\chi_- (\lambda) = \chi_+ (\lambda) \rmG_\chi (\lambda)$
        with jump matrix $\rmG_\chi (\lambda)$ given by
        \begin{multline}
        \label{eq:jump-matrix-chi}
            \rmG_\chi (\lambda)
            = \mathrm{I}_2
            + 2 \pi \mathrm{i} \ \mathbf{E}_R (\lambda)
            \cdot \mathbf{E}_L^\intercal (\lambda)
            \\
            = \mathrm{I}_2
            + 4 \vartheta (\lambda)
            \sin^2 (\pi \nu (\lambda))
            \begin{pmatrix}
                - e (\lambda) E(\lambda) &  E^2 (\lambda)\\
                -  e^2 (\lambda) &  e (\lambda) E (\lambda)
            \end{pmatrix}.
        \end{multline}

        \item
        $\chi (\lambda)
            = \mathrm{I}_2
            + \lambda^{-1}
            \Ocal \big(
                \begin{smallmatrix}1 & 1 \\ 1 & 1\end{smallmatrix}
            \big)$
        as $\lambda \to \infty$
        up to tangential direction to $\mathcal{C}_{\lambda_0}$.
    \end{enumerate}
\end{RHp}
We recall that $\chi_\pm$ denote the boundary values
from the ``$\pm$'' side of the oriented contour.
The positive (negative) side of the contour
is the one to the left (right) of the contour,
when moving in the direction of the contour.
For example,
\begin{equation}
    \chi_\pm (\lambda)
    = \lim\limits_{\substack{\mu \to \lambda\\ \mu \in \Omega^\pm}}
    \chi (\mu),
\end{equation}
see Figure~\ref{fig:contour-Ccal-lambda-0}.
Also, we emphasize that the big $\Ocal$ symbol of a matrix $\rmM$,
$\Ocal (\rmM)$, should be understood entrywise,
\begin{equation}
    \Ocal (\rmM)
    = \begin{pmatrix}
        \Ocal (M_{11}) & \Ocal (M_{12})
        \\
        \Ocal (M_{21}) & \Ocal (M_{22})
    \end{pmatrix}.
\end{equation}

We note as well that $\det \rmG_\chi (\lambda) = 1$ on
$\Ccal_{\lambda_0}$ and that this fact together with the
asymptotic behaviour of $\chi (\lambda)$ implies that
\begin{equation}
    \det \chi (\lambda) = 1
\end{equation}
for all $\lambda \in {\mathbb C} \setminus \Ccal_{\lambda_0}$,
which then allows one to prove the uniqueness of the solution
of the Riemann--Hilbert Problem~\ref{RHp:chi} (cf.~\cite{IIKS-90}).

For our work it is of particular importance that the logarithmic
derivative of the Fredholm determinant of an integrable integral
operator with respect to its parameters can be directly expressed
in terms of the solution of the associated Riemann--Hilbert
Problem~\ref{RHp:chi} (see~\cite{KKMST-09-RHp, K-11,%
DeiftItsZhouSineKernelOnUnionOfIntervals}).
\begin{proposition} \cite{KKMST-09-RHp,K-11}
    \label{prop:log-der}
    Let $\eta > 0$
    and $\Gamma (\mathcal{C}_{\lambda_0})$ be a loop in $\Omega$
    surrounding the contour $\mathcal{C}_{\lambda_0}$
    in positive direction,
    see Figure~\ref{fig:contour-Gamma}.
    If $\det_{\Ccal_{\lambda_0}} (\id + \rmV) \neq 0$, then
    \begin{multline}
        \label{eq:prop:log-der}
        \partial_\beta \ln \det_{\mathcal{C}_{\lambda_0}}
        (\id + \mathrm{V})
        \\
        =
        - \int\limits_{\Gamma (\mathcal{C}_{\lambda_0})}
        \eval{
            \frac{ \dd{\lambda} }{2 \pi \rmi}
            \tr\left\{
                \chi^\prime (\lambda)
                \begin{pmatrix}
                    1 &  2 \mathrm{C} (\lambda) \\ 0 & -1
                \end{pmatrix}
                \chi^{-1} (\lambda)
            \right\}
            d_\beta (\lambda)
            \rme^{- \eta \lambda^2}
        }_{\eta = 0+},
    \end{multline}
    where
    \begin{equation}
    \label{eq:function-d}
        d_\beta (\lambda) = \partial_\beta d (\lambda), \qquad
        d (\lambda)
        = - \frac{\rmi x}{2} u (\lambda|\lambda_0) - \frac{1}{2} g (\lambda),
    \end{equation}
    for $\beta = x, \lambda_0$.
    $\mathrm{C} (z)$ is the Cauchy transform of $e^{-2} (\lambda)$,
    defined in equation~\eqref{eq:Cauchy-transform},
    $\chi (\lambda)$ is the unique solution of the Riemann--Hilbert
    Problem~\ref{RHp:chi}.
\end{proposition}
The proof of this proposition is along the same lines
as in~\cite[Lemma 3.1]{KKMST-09-RHp}
and~\cite[Proposition 3.1]{K-11}, see~\cite[Appendix A]{M-25}.

Proposition~\ref{prop:log-der} reduces the large-$x$ asymptotic
analysis of the original Fredholm determinant to the large-$x$
asymptotic analysis of the solution $\chi$ of the matrix
Riemann--Hilbert Problem~\ref{RHp:chi} and a subsequent calculation
of the anti-derivatives involved in~\eqref{eq:prop:log-der}.
The asymptotic analysis of the Riemann--Hilbert problem can be
performed rather systematically within the so-called nonlinear
steepest descent method~\cite{DZ-93} and will be the subject
of the following section.

\begin{figure}[t]
    \centering
    \inputfig{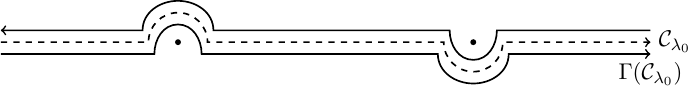}
    \caption{
        The integration contour $\Ccal_{\lambda_0}$ (dashed line)
        and the loop $\Gamma (\Ccal_{\lambda_0})$ (solid line)
        surrounding it in positive direction.
    }
    \label{fig:contour-Gamma}
\end{figure}

\section{Transforming the initial Riemann--Hilbert problem}
\label{sec:transform_RH}
In this section we employ the nonlinear steepest descent
method~\cite{BealsCoifmanScatteringInFirstOrderSystemsEquivalenceRHPSingIntEqnMention,%
BealsCoifmanScatteringInFirstOrderSystemsEquivalenceRHPPRoofRHPUniqueness,DZ-93}
for the large-$x$ asymptotic analysis of
the matrix Riemann--Hilbert Problem~\ref{RHp:chi}. This means
that we construct a chain of explicit bijections that maps $\chi (\lambda)$
to a product of two matrices $S(\lambda) \Pi(\lambda)$ which is
suitable for a direct asymptotic analysis through the Beals--Coifman
procedure~\cite{BealsCoifmanScatteringInFirstOrderSystemsEquivalenceRHPSingIntEqnMention,%
BealsCoifmanScatteringInFirstOrderSystemsEquivalenceRHPPRoofRHPUniqueness}.
While $\Pi$ satisfies a singular
integral equation that can be solved by iteration, if $x$ is large
enough, and that produces a convergent large-$x$ asymptotic series,
the matrix $S: \Omega \rightarrow {\mathbb C}^{2 \times 2}$ is a
meromorphic function with only first order poles whose residues
can be determined from a system of linear algebraic equations.

We then insert the explicit chain of transformation into the
integral on the right-hand side of equation~\eqref{eq:prop:log-der} in
Proposition~\ref{prop:log-der} and derive a representation of the
logarithmic derivatives of the Fredholm determinant involving
only $S$, $\Pi$ and explicit functions. This representation,
see Proposition~\ref{prop:log-der-Fredholm-for-AE}, is suitable
for completing the large-$x$ asymptotic analysis of the Fredholm
determinant \eqref{eq:Fred} in the following sections.

Finally,
we analyse the contribution of the poles
to the asymptotic expansion
of the logarithmic derivative of the Fredholm determinant
in Section~\ref{sec:analysis-of-pole-contribution},
which is formulated in Proposition~\ref{prop:contribution-of-poles}.

\subsection{First transformation}
Let
\begin{equation} \label{def:Pauli}
    \sigma^x =
    \begin{pmatrix}
        0 & 1 \\ 1 & 0
    \end{pmatrix}, \qquad
    \sigma^y =
    \begin{pmatrix}
        0 & - \rmi \\ \rmi & \mspace{14.mu} 0
    \end{pmatrix}, \qquad
    \sigma^z =
    \begin{pmatrix}
        1 & \mspace{12.mu} 0 \\ 0 & -1
    \end{pmatrix}
\end{equation}
denote the Pauli matrices and let
$\sigma^\pm = \frac12(\sigma^x \pm \rmi \sigma^y)$.
Let the $2 \times 2$ unit matrix be denoted by ${\rm I}_2$.

The transformation~\cite{IIKV-92},
\begin{equation}
    \label{eq:chi-tilde}
    \widetilde \chi (\lambda)
    = \chi (\lambda) (\mathrm{I}_2 - \mathrm{C} (\lambda) \sigma^+),
    \qquad
    \lambda \in \mathbb{C} \backslash \Ccal_{\lambda_0},
\end{equation}
removes the Cauchy transform from the right-hand side
of expression~\eqref{eq:prop:log-der} and from the jump 
matrix $\rmG_\chi$, see equation~\eqref{eq:jump-matrix-chi},
while leaving the determinant invariant. Define
\begin{multline}
\label{eq:jump-matrix-chi-tilde}
    \rmG_{\tilde \chi} (\lambda)
    = (\mathrm{I}_2 + \rmC_+ (\lambda) \sigma^+)
    \rmG_{\chi} (\lambda)
    (\mathrm{I}_2 - \rmC_- (\lambda) \sigma^+)
    \\
    = \begin{pmatrix}
        1 + \vartheta (\lambda)
        (\rme^{- 2 \pi \rmi \nu (\lambda)} - 1)
        & e^{-2} (\lambda) (1 - \vartheta (\lambda))
        \\
        - 4
        e^{2} (\lambda) \vartheta (\lambda)
        \sin^2 (\pi \nu (\lambda))
        & 1 + \vartheta (\lambda)
        (\rme^{2 \pi \rmi \nu (\lambda)} - 1)
    \end{pmatrix}.
\end{multline}
Then $\widetilde{\chi}$ is the unique solution of the following
matrix Riemann--Hilbert problem.
\begin{RHp}
\label{RHp:chi-tilde}
    Determine $\widetilde{\chi} (\lambda) \in \mathbb{C}^{2 \times 2}$
    such that
    \begin{enumerate}
        \item
        $\widetilde{\chi} (\lambda)$ is analytic
        in $\mathbb{C} \backslash \Ccal_{\lambda_0}$
        and has continuous $\pm$ boundary values on $\Ccal_{\lambda_0}$,
        see Figure~\ref{fig:contour-Ccal-lambda-0}.
        
        \item
        On the contour $\Ccal_{\lambda_0}$ the boundary values
        $\widetilde{\chi}_\pm (\lambda)$
        satisfy the jump condition
        $\widetilde{\chi}_- (\lambda)
            = \widetilde{\chi}_+ (\lambda) \rmG_{\tilde{\chi}} (\lambda)$
        with jump matrix $\rmG_{\tilde{\chi}} (\lambda)$
        according to~\eqref{eq:jump-matrix-chi-tilde}.

        \item
        $\widetilde{\chi} (\lambda)
            = \mathrm{I}_2
            + \lambda^{-1}
            \Ocal \big(
                \begin{smallmatrix}1 & 1 \\ 1 & 1\end{smallmatrix}
            \big)$
        as $\lambda \to \infty$
        up to tangential direction to $\Ccal_{\lambda_0}$.
    \end{enumerate}
\end{RHp}

The matrix defining the transformation~\eqref{eq:chi-tilde} is
invertible with determinant equal to one. The existence and
uniqueness of the solution of the Riemann--Hilbert
Problem~\ref{RHp:chi-tilde} hence follow from the corresponding
properties of $\chi$, and
\begin{equation}
    \det \widetilde{\chi} (\lambda) = 1.
\end{equation}
This will also hold for the subsequent transformations. They all
preserve the unimodularity of the original matrix-valued function
$\chi$.

\subsection{Scalar Riemann--Hilbert problem}
\label{sec:scalar-RHP}
The next transformation will map
the matrix Riemann--Hilbert Problem~\ref{RHp:chi-tilde}
to a matrix Riemann--Hilbert problem with unimodular
jump matrix of the form
\begin{equation}
    \label{eq:jump-matrix-for-factorization}
    \rmG (\lambda)
    = \begin{cases}
        \begin{pmatrix}
            * & *\\
            * & 1
        \end{pmatrix},
        & \Real (\lambda - \lambda_0) < 0,
        \\[3ex]
        \begin{pmatrix}
            1 & *\\
            * & *
        \end{pmatrix},
        & \Real (\lambda - \lambda_0) > 0.
    \end{cases}
\end{equation}
This can be achieved by a diagonal transformation
\begin{equation}
\label{eq:Xi}
    \Xi (\lambda)
    = \widetilde{\chi} (\lambda) \big[\alpha (\lambda) \big]^{\sigma^z},
\end{equation}
where $\alpha$ is a scalar function that is
analytic in $\mathbb{C}\backslash \Ccal_{\lambda_0}$.
The transformed matrix $\Xi (\lambda)$ has a jump discontinuity
across the contour $\Ccal_{\lambda_0}$ that is encoded by the jump matrix
\begin{equation}
    \label{eq:jump-matrix-Xi-with-alpha}
    \rmG_\Xi (\lambda)
    = \alpha_+^{- \sigma^z} (\lambda)
    \rmG_{\tilde \chi} (\lambda)
    \alpha_-^{\sigma^z} (\lambda)
    =
    \begin{pmatrix}
        \dfrac{ \alpha_- (\lambda) }{ \alpha_+ (\lambda)  }
        \left( \rmG_{\tilde \chi} \right)_{11} (\lambda)
        &
        \dfrac{ \left( \rmG_{\tilde \chi} \right)_{12} (\lambda) }
            { \alpha_- (\lambda) \alpha_+ (\lambda)  }
        \\
        \alpha_- (\lambda) \alpha_+ (\lambda)
        \left( \rmG_{\tilde \chi} \right)_{21} (\lambda)
        &
        \dfrac{ \alpha_+ (\lambda) }{ \alpha_- (\lambda)  }
        \left( \rmG_{\tilde \chi} \right)_{22} (\lambda)
    \end{pmatrix}.
\end{equation}
Thus, $G_\Xi$ assumes the intended
form~\eqref{eq:jump-matrix-for-factorization}, if
$\alpha$ is the solution of the following scalar Riemann--Hilbert
problem.
\begin{RHp}
\label{RHp:alpha}
    Determine $\alpha (\lambda) \in \mathbb{C}$ such that
    \begin{enumerate}
        \item
        $\alpha (\lambda)$ is analytic
        in $\mathbb{C} \backslash \Ccal_{\lambda_0}$
        and has continuous $\pm$ boundary values on
        $\Ccal_{\lambda_0} \setminus \{\lambda_0\}$.
        
        \item
        On the contour $\Ccal_{\lambda_0} \backslash\{\lambda_0\}$
        the boundary values
        $\alpha_\pm (\lambda)$
        satisfy the jump condition
        \begin{equation}
        \label{eq:jump-condition-alpha}
            \frac{\alpha_+ (\lambda)}{ \alpha_- (\lambda)}
            =
            \begin{cases}
                \left[
                    1 + \vartheta (\lambda)
                    \left(
                        \rme^{2 \pi \rmi \nu (\lambda)} - 1
                    \right)
                \right]^{-1},
                & \Real(\lambda - \lambda_0) < 0,
                \\[1ex]
                1 + \vartheta (\lambda)
                \left( \rme^{- 2 \pi \rmi \nu (\lambda)} - 1\right),
                & \Real(\lambda - \lambda_0) > 0.
            \end{cases}
        \end{equation}

        \item
        As $\lambda \to \lambda_0$
        \begin{equation}
        \label{eq:asymptotics-of-alpha-at-lambda-0}
            \alpha (\lambda)
            = \alpha_0 \cdot
            (\lambda - \lambda_0)^{\tau (\lambda)}
        \end{equation}
        for a piecewise constant $\alpha_0$
        and the function $\tau (\lambda)$,
        defined in~\eqref{eq:definition-of-tau}.

        \item
        $\alpha (\lambda) = 1 + \mathrm{O} (\lambda^{-1})$
        as $\lambda \to \infty$
        up to tangential direction to $\Ccal_{\lambda_0}$.
    \end{enumerate}
\end{RHp}
The solution of this scalar Riemann--Hilbert problem is the
origin of the functions $\Lcal_\ell$, $\Lcal_r$ and
$\Lcal (\cdot|\lambda_0)$ introduced
in~\eqref{eq:definition-of-Lcal-ell-and-Lcal-r}
and~\eqref{eq:definition-of-Lcal}. In fact, it is
straightforward to see that
\begin{equation}
    \alpha (\lambda)
    =
    \exp\Bigg\{
        \int\limits_{\Ccal_{\lambda_0}}
        \dd{\mu}
        \frac{ \Lcal (\mu | \lambda_0) }{\mu - \lambda}
    \Bigg\}
\end{equation}
is the unique solution of the Riemann--Hilbert Problem~\ref{RHp:alpha}.

The motivation for constructing this second transformation was that unimodular
matrices of the form \eqref{eq:jump-matrix-for-factorization} can be
factorized into products of upper and lower triangular matrices.

\subsubsection[Extracting the singular part
    of \texorpdfstring{$\alpha$}{alpha}]{\boldmath Extracting the singular part
    of $\alpha$}
The function $\alpha$ has a singularity at $\lambda_0$. For our subsequent
calculations we have to extract it explicitly. For $\varepsilon > 0$
we have the decomposition
\begin{multline}
    \alpha (\lambda)
    =
    \exp\Bigg\{
        \int\limits_{\Ccal_{\lambda_0}}
        \dd{\mu}
        \frac{
            \Lcal (\mu | \lambda_0)
            - \Lcal (\lambda | \lambda_0)
            \cdot \indicator{
                (\lambda_0 - \varepsilon, \lambda_0 + \varepsilon)} (\mu)
            }{\mu - \lambda}
    \Bigg\}
    \\
    \times
    \exp\Bigg\{
        \Lcal_\ell (\lambda)
        \int\limits_{\lambda_0 - \varepsilon}^{\lambda_0}
        \frac{\dd \mu}{\mu - \lambda}
        + \Lcal_r (\lambda)
        \int\limits_{\lambda_0}^{\lambda_0 + \varepsilon}
        \frac{\dd \mu}{\mu - \lambda}
    \Bigg\}.
\end{multline}
The last exponent has a cut and can be written as
\begin{multline}
    \exp\Bigg\{
        \Lcal_\ell (\lambda)
        \int\limits_{\lambda_0 - \varepsilon}^{\lambda_0}
        \frac{\dd \mu}{\mu - \lambda}
        + \Lcal_r (\lambda)
        \int\limits_{\lambda_0}^{\lambda_0 + \varepsilon}
        \frac{\dd \mu}{\mu - \lambda}
    \Bigg\}
    = \left(
        \frac{\lambda - \lambda_0}{ \lambda - \lambda_0 + \varepsilon}
    \right)^{\Lcal_\ell (\lambda)}
    \left(
        \frac{\lambda_0 - \lambda + \varepsilon}{ \lambda_0 - \lambda}
    \right)^{\Lcal_r (\lambda)}
    \\
    = \frac{
        \left(
            \lambda_0 - \lambda + \varepsilon
        \right)^{\Lcal_r (\lambda)}
        \left( \lambda - \lambda_0 \right)^{\Lcal_\ell (\lambda)}
    }{
        \left(
            \lambda - \lambda_0 + \varepsilon
        \right)^{\Lcal_\ell (\lambda)}
        \left( \lambda_0 - \lambda \right)^{\Lcal_r (\lambda)}
    }
    = \frac{
        \left(
            \lambda_0 - \lambda + \varepsilon
        \right)^{\Lcal_r (\lambda)}
        \left( \lambda - \lambda_0 \right)^{\Lcal_\ell (\lambda)}
    }{
        \left(
            \lambda - \lambda_0 + \varepsilon
        \right)^{\Lcal_\ell (\lambda)}
        \left( \lambda - \lambda_0 \right)^{\Lcal_r (\lambda)}
    }
    \rme^{\pi \rmi \Lcal_r (\lambda) \sgn \Imag (\lambda)}.
\end{multline}
The function $\alpha (\lambda)$ can then be expressed as
\begin{equation}
    \label{eq:alpha-factorized}
    \alpha (\lambda)
    = \varkappa (\lambda | \lambda_0)
    \cdot
    (\lambda - \lambda_0)^{\tau (\lambda)} \
    \rme^{ \pi \rmi \sgn (\Imag (\lambda)) \Lcal_r (\lambda)},
\end{equation}
where the function $\tau$ is given by \eqref{eq:definition-of-tau}
and where $\varkappa (\lambda | \lambda_0)$, defined as
\begin{equation}
\label{eq:varkappa}
    \varkappa (\lambda | \lambda_0)
    = \frac{
        (\lambda_0 - \lambda + \varepsilon)^{\Lcal_r (\lambda)} }{
        (\lambda - \lambda_0 + \varepsilon)^{\Lcal_\ell (\lambda)} }
    \exp\Bigg\{
        \int\limits_{\Ccal_{\lambda_0}}
        \dd \mu \frac{
            \Lcal (\mu | \lambda_0)
            - \Lcal (\lambda | \lambda_0)
            \cdot
            \indicator{
                (\lambda_0 - \varepsilon, \lambda_0 + \varepsilon)
            } (\mu)
        }{
            \mu - \lambda }
    \Bigg\},
\end{equation}
is holomorphic in the vicinity of $\lambda_0$.
We note that the representation \eqref{eq:varkappa}
does not depend on the regularization parameter $\varepsilon > 0$,
which can be seen, for example,
after taking the derivative with respect to $\varepsilon$
and showing that it is zero.

\subsection{Factorization of the jump matrix}
\label{sec:factorization-of-the-jump-matrix}
The transformation~\eqref{eq:Xi} is singular at $\lambda_0$,
cf.~\eqref{eq:alpha-factorized}, hence induces a
singularity of the transformed matrix $\Xi (\lambda)$ at
this point. Therefore, $\Xi$ is the solution of the following 
Riemann--Hilbert problem with an additional condition
at $\lambda_0$.
\begin{RHp}
\label{RHp:Xi}
    Determine $\Xi (\lambda) \in \mathbb{C}^{2 \times 2}$ such that
    \begin{enumerate}
        \item
        $\Xi (\lambda)$ is analytic
        in $\mathbb{C} \backslash \Ccal_{\lambda_0}$
        and has continuous $\pm$ boundary values on
        $\Ccal_{\lambda_0} \backslash \{\lambda_0\}$.

        \item
        On the contour $\Ccal_{\lambda_0} \backslash \{\lambda_0\}$
        the boundary values
        $\Xi_\pm (\lambda)$
        satisfy the jump condition
        $\Xi_- (\lambda) = \Xi_+ (\lambda) \rmG_\Xi (\lambda)$
        with jump matrix $\rmG_\Xi (\lambda)$
        given by~\eqref{eq:jump-matrix-Xi-with-alpha}.

        \item
        $\Xi (\lambda) = \mathrm{I}_2
            + \lambda^{-1}
            \Ocal \big(
                \begin{smallmatrix}1 & 1 \\ 1 & 1\end{smallmatrix}
            \big)$
        as $\lambda \to \infty$
        up to tangential direction to $\Ccal_{\lambda_0}$.

        \item
        As $\lambda \to \lambda_0$
        \begin{equation}
        \label{eq:asymptotics-of-Xi-at-lambda-0}
            \Xi (\lambda)
            = \Big[
                \Xi_0
                + (\lambda - \lambda_0)
                \Ocal \big(
                    \begin{smallmatrix}
                        1 & 1
                        \\
                        1 & 1
                    \end{smallmatrix}
                \big)
            \Big]
            (\lambda - \lambda_0)^{\sigma^z \tau (\lambda)}
        \end{equation}
        for a piecewise constant matrix $\Xi_0 \in \mathbb{C}^{2 \times 2}$.
    \end{enumerate}
\end{RHp}

Next we will turn to the factorization
of the jump matrix $\rmG_\Xi$, separately for the two cases
$\Real(\lambda - \lambda_0) < 0$ and $\Real(\lambda - \lambda_0) > 0$.

\subsubsection[Factorization for
\texorpdfstring{$\Real (\lambda - \lambda_0) < 0$}{
Re (lambda - lambda0) < 0}]{\boldmath Factorization for $\Real (\lambda - \lambda_0) < 0$}
Substituting the jump condition~\eqref{eq:jump-condition-alpha}
into~\eqref{eq:jump-matrix-Xi-with-alpha},
we obtain for $\Real (\lambda - \lambda_0) < 0$
\begin{equation}
    \rmG_\Xi (\lambda)
    =
    \begin{pmatrix}
        1 - 4 \sin^2 (\pi \nu (\lambda))
        \vartheta (\lambda) (1 - \vartheta (\lambda))
        &
        \dfrac{
            e^{-2} (\lambda) (1 - \vartheta (\lambda)) }{
            \alpha_+^2 (\lambda)
            \left[1 + \vartheta (\lambda)
            (\rme^{2 \pi \rmi \nu (\lambda)} - 1)\right] }
        \\
        - \dfrac{
            4
            \alpha_-^2 (\lambda) e^{2} (\lambda) \vartheta (\lambda)
            \sin^2 (\pi \nu (\lambda))
        }{
            \left[1 + \vartheta (\lambda)
            (\rme^{2 \pi \rmi \nu (\lambda)} - 1)\right] }
        & 1
    \end{pmatrix}.
\end{equation}
Here we expressed the elements (1,2) and (2,1)
in terms of $\alpha_+$ and $\alpha_-$, respectively,
which leads to an upper-lower factorization of the jump
matrix of the form
\begin{equation}
\label{eq:factorization-left}
    \rmG_\Xi (\lambda)
    = \rmM^+_\ell (\lambda) \rmM^-_\ell (\lambda)
\end{equation}
with matrices $\rmM^+_\ell (\lambda)$
and $\rmM^-_\ell (\lambda)$ given by
\begin{equation}
\label{eq:matrices-M-left}
\begin{aligned}
    \rmM^+_\ell (\lambda)
    &= \mathrm{I}_2
    + e^{-2} (\lambda) Q_\ell^+ (\lambda) \sigma^+,
    \quad
    & Q_\ell^+ (\lambda)
    &= \frac{ 1 - \vartheta (\lambda) }{
        \alpha_+^2 (\lambda)
        [1 + \vartheta (\lambda)
        (\rme^{2 \pi \rmi \nu (\lambda)} - 1)] },
    \\
    \rmM^-_\ell (\lambda)
    &= \mathrm{I}_2
    + e^{2} (\lambda) Q_\ell^- (\lambda) \sigma^-,
    \quad
    & Q_\ell^- (\lambda)
    &=
    -
    \frac{ 4 \alpha_-^2 (\lambda) \vartheta (\lambda)
        \sin^2 (\pi \nu (\lambda))
    }{
        [1 + \vartheta (\lambda)
        (\rme^{2 \pi \rmi \nu (\lambda)} - 1)] }.
\end{aligned}
\end{equation}
We note that the matrix $\rmM_\ell^+$ (resp.\ $\rmM_\ell^-$)
admits an analytic continuation
to the region above (below) the integration contour $\Ccal_{\lambda_0}$
for $\Real (\lambda - \lambda_0) < 0$,
where it becomes exponentially close to identity
due to the factor $e^{-2} (\lambda)$
(resp.\ $e^{2} (\lambda)$).

\subsubsection[Factorization for \texorpdfstring{$\Real (\lambda - \lambda_0) > 0$}{
Re (lambda - lambda0) > 0}]{\boldmath Factorization for $\Real (\lambda - \lambda_0) > 0$}
Similarly, for $\Real (\lambda - \lambda_0) > 0$, substituting the
jump condition~\eqref{eq:jump-condition-alpha}
into~\eqref{eq:jump-matrix-Xi-with-alpha} leads to
\begin{equation} \label{eq:jump-matrix_Xi_larger}
    \rmG_\Xi (\lambda)
    =
    \begin{pmatrix}
        1
        &
        \dfrac{
            e^{-2} (\lambda) (1 - \vartheta (\lambda)) }{
            \alpha_-^2 (\lambda)
            [1 + \vartheta (\lambda)
            (\rme^{-2 \pi \rmi \nu (\lambda)} - 1)]}
        \\
        \dfrac{
            - 4
            \alpha_+^2 (\lambda) e^{2} (\lambda) \vartheta (\lambda)
            \sin^2 (\pi \nu (\lambda))
        }{
            [1 + \vartheta (\lambda)
            (\rme^{-2 \pi \rmi \nu (\lambda)} - 1)] }
        & 
        1 - 4 \sin^2 (\pi \nu (\lambda))
        \vartheta (\lambda) (1 - \vartheta (\lambda))
    \end{pmatrix}.
\end{equation}
Here we expressed the elements (1,2) and (2,1)
in terms of $\alpha_-$ and $\alpha_+$,
respectively. The representation \eqref{eq:jump-matrix_Xi_larger}
implies a lower-upper factorization of the jump matrix,
\begin{equation}
\label{eq:factorization-right}
    \rmG_\Xi (\lambda)
    = \rmM^+_r (\lambda) \rmM^-_r (\lambda),
\end{equation}
with the matrices $\rmM^+_r (\lambda)$
and $\rmM^-_r (\lambda)$ given by
\begin{equation}
\label{eq:matrices-M-right}
\begin{aligned}
    \rmM^+_r (\lambda)
    & =
    \mathrm{I}_2
    + e^{2} (\lambda) Q_r^+ (\lambda) \sigma^-,
    \quad
    & Q_r^+ (\lambda)
    & =
    - \frac{
        4
        \alpha_+^2 (\lambda) \vartheta (\lambda)
        \sin^2 (\pi \nu (\lambda))
    }{
        [1 + \vartheta (\lambda)
        (\rme^{- 2 \pi \rmi \nu (\lambda)} - 1)]},
    \\
    \rmM^-_r (\lambda)
    & =
    \mathrm{I}_2
    + e^{-2} (\lambda) Q_r^- (\lambda) \sigma^+,
    \quad
    & Q_r^- (\lambda)
    & = \frac{
        1 - \vartheta (\lambda) }{
        \alpha_-^2 (\lambda)
        \left[
            1 + \vartheta (\lambda)
            (\rme^{- 2 \pi \rmi \nu (\lambda)} - 1)
        \right] }.
\end{aligned}
\end{equation}
We note that again the matrix $\rmM_r^+$ (resp.\ $\rmM_r^-$)
admits an analytic continuation
to the region above (below) the integration contour $\Ccal_{\lambda_0}$
for $\Real (\lambda - \lambda_0) > 0$,
where it becomes exponentially close to identity
due to the factor $e^{2} (\lambda)$ 
(resp.\ $e^{-2} (\lambda)$).

\subsection{Jump matrix close to identity}
\label{sec:Upsilon}
Our next transformation will split the jump contour
in the matrix Riemann--Hilbert problem. The jump
matrix with respect to the new split contour will be
exponentially close to the identity matrix for 
$\lambda$ uniformly away from the saddle point $\lambda_0$.

We introduce new oriented contours
$\Gamma_\ell^\pm \subset \Omega^\pm$
and $\Gamma_r^\pm \subset \Omega^\pm$ that meet in
the saddle point and, together with the contour
$\Ccal_{\lambda_0}$, divide the strip $\Omega$ into six
disjoint regions as sketched in Figure~\ref{fig:Upsilon-def}.
A matrix $\Upsilon$ is defined separately in each of these
six regions, see Figure~\ref{fig:Upsilon-def}, as a product
of $\Xi$ with either the unit matrix or with
one of the matrices $\rmM_{\ell/r}^\pm$ or their inverses,
in such a way that the jump of $\Xi$ on the initial contour
$\Ccal_{\lambda_0}$ is removed. This creates a new jump
discontinuity on the contour
\begin{equation}
\label{eq:contour-Gamma}
    \Gamma_\Upsilon
    = \Gamma_\ell^+ \cup \Gamma_\ell^-
    \cup \Gamma_r^+ \cup \Gamma_r^-.
\end{equation}
The precise form of the contours $\Gamma_\ell^\pm$ and $\Gamma_r^\pm$
in the vicinity of the saddle point will be specified
in Section~\ref{sec:parametrix} below.
\begin{figure}[ht]
    \centering
    \begin{subfigure}{0.5\textwidth}
        \centering
        \resizebox{\textwidth}{!}{%
            \inputfig{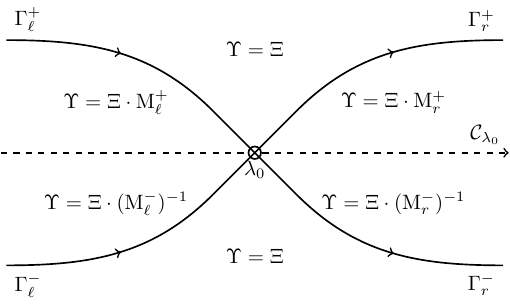}%
        }%
        \caption{}%
        \label{fig:Upsilon-def}
    \end{subfigure}%
    \begin{subfigure}{0.5\textwidth}
        \centering
        \resizebox{\textwidth}{!}{%
            \inputfig{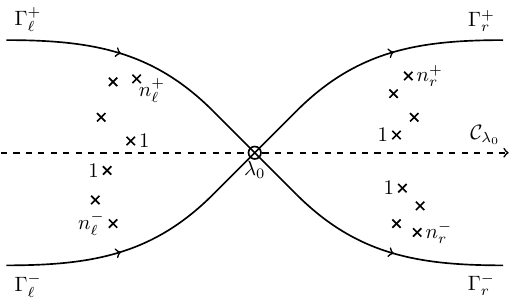}%
        }%
        \caption{}%
        \label{fig:Upsilon-poles}%
    \end{subfigure}%
    \caption{
        (a) Definition of the matrix $\Upsilon$
        in terms of $\Xi$ and $\rmM_{\ell / r}^\pm$;
        (b) the enumerated poles of $\Upsilon$
        stemming from the matrices $\rmM_{\ell / r}^\pm$
        in the corresponding regions.}
    \label{fig:Upsilon}
\end{figure}

The matrices $\rmM^\pm_{\ell}$ and $\rmM^\pm_{r}$ (or, equivalently,
their inverses) are meromorphic and may have poles originating from
their off-diagonal elements, cf.~\eqref{eq:matrices-M-left}
and~\eqref{eq:matrices-M-right}. We choose $\Gamma_\ell^\pm$
and $\Gamma_r^\pm$ in such a way that none of the poles of $\rmM_{\ell/r}^\pm$
is on these contours. We denote the set of those poles of $\rmM_{\ell/r}^+$
that are between $\Ccal_{\lambda_0}$ and $\Gamma_{\ell/r}^+$ by
$\Lcal^+/\Rcal^+$. Similarly, the sets of roots of $\rmM_{\ell/r}^{-1}$
between $\Gamma_{\ell/r}^-$ and $\Ccal_{\lambda_0}$ will be
denoted $\Lcal^-/\Rcal^-$, see Figure~\ref{fig:Upsilon-poles}.
We shall use the notation
\begin{subequations}
\label{eq:sets-of-poles}
\begin{align}
    \Lcal^+ = \{\ell_1^+, \dots, \ell_{n_\ell^+}^+\},
    \qquad
    \Rcal^+ = \{r_1^+, \dots, r_{n_r^+}^+\},\\
    \Lcal^- = \{\ell_1^-, \dots, \ell_{n_\ell^-}^-\},
    \qquad
    \Rcal^- = \{r_1^-, \dots, r_{n_r^-}^-\},
\end{align}
\end{subequations}
where $n_\ell^\pm$ and $n_r^\pm$ are the cardinalities of
the corresponding sets. We note that
\begin{equation}
\label{eq:set-of-poles}
    \Scal
    = \Lcal^+ \cup \Lcal^-
    \cup \Rcal^+ \cup \Rcal^-
\end{equation}
is the set of poles introduced in \eqref{def:Scal}.

These poles are poles of the matrix $\Upsilon$, by construction.
They are all simple by Assumption~\ref{as:number6}. The matrix
$\Upsilon$ can then be described as the unique solution of the
\begin{RHp}
\label{RHp:Upsilon}
    Determine $\Upsilon (\lambda) \in \mathbb{C}^{2 \times 2}$
    such that
    \begin{enumerate}
        \item
        $\Upsilon (\lambda)$ is analytic in
        $\mathbb{C} \backslash (\Gamma_\Upsilon \cup \mathcal{S})$
        and has continuous $\pm$ boundary values on
        $\Gamma_\Upsilon \backslash\{\lambda_0\}$.
        
        \item
        On the contour $\Gamma_\Upsilon \backslash\{\lambda_0\}$
        the boundary values $\Upsilon_\pm (\lambda)$
        satisfy the jump condition
        $\Upsilon_- (\lambda)
        = \Upsilon_+ (\lambda) \rmG_\Upsilon (\lambda)$
        with jump matrix $\rmG_\Upsilon (\lambda)$ given by
        \begin{equation}
        \label{eq:jump-matrix-Upsilon}
            \rmG_\Upsilon (\lambda)
            = \begin{cases}
                \rmM_r^+ (\lambda),
                \quad & \lambda \in \Gamma_r^+,
                \\
                \rmM_r^- (\lambda),
                \quad & \lambda \in \Gamma_r^-,
                \\
                \rmM_\ell^- (\lambda),
                \quad & \lambda \in \Gamma_\ell^-,
                \\
                \rmM_\ell^+ (\lambda),
                \quad & \lambda \in \Gamma_\ell^+.
            \end{cases}
        \end{equation}
        
        \item
        $\Upsilon (\lambda) = \mathrm{I}_2
            + \lambda^{-1}
            \Ocal \big(
                \begin{smallmatrix}1 & 1 \\ 1 & 1\end{smallmatrix}
            \big)$
        as $\lambda \to \infty$
        up to tangential direction to $\Gamma_\Upsilon$.
        
        \item
        As $\lambda \to \lambda_0$
        \begin{equation}
        \label{eq:asymptotics-of-Upsilon-at-lambda-0}
            \Upsilon (\lambda)
            = \Big[
                \Upsilon_0
                + (\lambda - \lambda_0) 
                \Ocal \big(
                    \begin{smallmatrix}
                        1 & 1
                        \\
                        1 & 1
                    \end{smallmatrix}
                \big)
            \Big]
            (\lambda - \lambda_0)^{\tau (\lambda) \sigma^z}
        \end{equation}
        for a piecewise constant matrix
        $\Upsilon_0 \in \mathbb{C}^{2 \times 2}$.
        
        \item
        $\Upsilon$ satisfies the following regularity conditions
        at the poles $\lambda \in \mathcal{S}$:
        \begin{equation}
        \begin{aligned}
            \Upsilon (\lambda) \cdot (\rmM_\ell^+)^{-1} (\lambda)
            \text{ is regular at }
            \lambda = \ell_j^+,
            \quad
            j = 1, \dots, n_\ell^+,
            \\
            \Upsilon (\lambda) \cdot (\rmM_r^+)^{-1} (\lambda)
            \text{ is regular at }
            \lambda = r_j^+,
            \quad
            j = 1, \dots, n_r^+,
            \\
            \Upsilon (\lambda) \cdot \rmM_\ell^- (\lambda)
            \text{ is regular at }
            \lambda = \ell_j^-,
            \quad
            j = 1, \dots, n_\ell^-,
            \\
            \Upsilon (\lambda) \cdot \rmM_r^- (\lambda)
            \text{ is regular at }
            \lambda = r_j^-,
            \quad
            j = 1, \dots, n_r^-.
        \end{aligned}
        \end{equation}
    \end{enumerate}
\end{RHp}
We emphasize that the jump matrix $\rmG_\Upsilon$ goes
to zero pointwise on $\Gamma_\Upsilon \setminus \{\lambda_0\}$ as
$x \rightarrow + \infty$. This suggests that, up to possible pole
contributions, the large-$x$ behaviour of $\Upsilon$ (and hence
of our original function $\chi$) should be dominated by its behaviour
close to the saddle point. The behaviour close to the saddle point,
in turn, should be determined by the asymptotic
condition~\eqref{eq:asymptotics-of-Upsilon-at-lambda-0} and by the
behaviour of the jump matrices $\rmM_{\ell/r}^\pm$ close to $\lambda_0$.

In the following Sections~\ref{sec:parametrix}--\ref{sec:contribution-of-poles}
we separately treat the behaviour of $\Upsilon$ in the vicinity
of the saddle point $\lambda_0$ and the contribution of the poles
$\lambda \in \Scal$.

\subsection{Parametrix}
\label{sec:parametrix}
In this section we construct the parametrix~---
the solution of the matrix Riemann--Hilbert problem
in the vicinity of the saddle point $\lambda_0$.

In the vicinity of $\lambda_0$ the behaviour of the jump matrices
$\rmM_{\ell/r}^\pm$, see~\eqref{eq:matrices-M-left},
\eqref{eq:matrices-M-right}, is dominated by the factors
$e^{\pm 2} (\lambda)$ which vary on a scale $\sqrt{x}$
compared to which all other terms are localized at $\lambda_0$.
Close to $\lambda_0$ we further have to take into account
the singular behaviour exhibited by the terms
$(\lambda - \lambda_0)^{\pm 2 \tau(\lambda)}$ stemming from
the factor $\alpha^2$ in~\eqref{eq:matrices-M-left},
\eqref{eq:matrices-M-right}.

As in the saddle point analysis of contour integrals in the
complex plane we introduce local variables in which the part
of the exponents that multiplies $\pm \rmi x$ is quadratic.
Recall that the function $u$ is holomorphic in an open
neighbourhood $\Ucal_{\lambda_0}$ of $\lambda_0$ and that
$u'(\lambda_0) = 0$ and $u'' (\lambda_0) < 0$. Hence, there
exists a holomorphic function $\omega:
\Ucal_{\lambda_0} \rightarrow {\mathbb C}$ such that
\begin{equation}
\label{eq:local-parametrization-omega}
    u (\lambda | \lambda_0)
    = u (\lambda_0 | \lambda_0) - \omega^2 (\lambda - \lambda_0 | \lambda_0)
\end{equation}
and
\begin{equation}
    \omega'(0|\lambda_0) = \sqrt{- u'' (\lambda_0)/2} > 0.
\end{equation}
The latter inequality implies that $\omega$ is invertible for $\lambda$
in a vicinity of $\lambda_0$ with holomorphic inverse $\omega^{-1}$,
mapping a vicinity of $0$ to a vicinity of $\lambda_0$. In other words,
defining $D_{0, \epsilon} = \{z \in {\mathbb C}|\; |z| < \epsilon\}$,
there exists $\epsilon > 0$ such that $\omega^{-1}: D_{0, \epsilon}
\rightarrow \omega^{-1} (D_{0, \epsilon})$ is a biholomorphism. We
shall set $\Ucal_{\lambda_0} = \omega^{-1} (D_{0, \epsilon})$. Then
$\Ucal_{\lambda_0}$ is open and contains an open disc around $\lambda_0$.

The function $\omega$ furnishes local coordinates for $\Upsilon$.
We use it to provide the precise definition of the oriented  contours
$\Gamma_\ell^\pm$ and $\Gamma_r^\pm$, that were introduced in the previous
subsection, close to $\lambda_0$. Inside $\Ucal_{\lambda_0}$ they
are defined as the pre-images under $\omega$ of the rays emanating from
$\lambda_0$ at angles of $\pi / 4$ relative to the real axis, see 
Figure~\ref{fig:local-parametrization}.
\begin{figure}[ht]
    \centering
    \inputfig{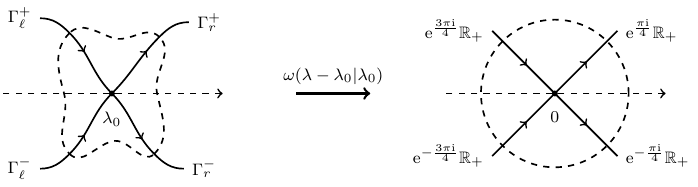}
    \caption{
        The images of the contours $\Gamma_{\ell/r}^\pm \cap \Ucal_{\lambda_0}$
        are line segments that form an angle of $\flatfrac{ \pi }{ 4 }$ with
        the real axis.}
    \label{fig:local-parametrization}
\end{figure}

We `separate fast and slow scales' by setting
\begin{equation}
\label{eq:definition-zeta}
    \zeta (\lambda)
    = \sqrt{x} \cdot \omega (\lambda - \lambda_0 | \lambda_0).
\end{equation}
Then the matrices $\rmM_\ell^\pm$ and $\rmM_r^\pm$
occurring in the jump condition for $\Upsilon$,
see~\eqref{eq:jump-matrix-Upsilon}, take the form
\begin{equation}
\begin{aligned}
\label{eq:matrices-M-local-variables}
    \rmM_\ell^+ (\lambda)
    & = \mathrm{I}_2
    +
    \frac{
        m (\lambda)
        \rme^{-\rmi \zeta^2 (\lambda) } }{
        [\zeta (\lambda)]^{2 \tau (\lambda)} }
    \rme^{2 \pi \rmi \tau (\lambda)}
    \sigma^+,
    \qquad
    \rmM_r^+ (\lambda)
    = \mathrm{I}_2
    +
    n (\lambda)
    \rme^{\rmi \zeta^2 (\lambda)}
    [\zeta (\lambda)]^{2 \tau (\lambda)}
    \sigma^-,
    \\
    \rmM_\ell^- (\lambda)
    & = \mathrm{I}_2
    +
    n (\lambda)
    \rme^{\rmi \zeta^2 (\lambda) }
    [\zeta (\lambda)]^{2 \tau (\lambda)}
    \rme^{2 \pi \rmi \tau (\lambda)}
    \sigma^-,
    \qquad
    \rmM_r^- (\lambda)
    = \mathrm{I}_2 +
    \frac{
        m (\lambda)
        \rme^{-\rmi \zeta^2 (\lambda) } }{
        [\zeta (\lambda)]^{2 \tau (\lambda)} }
    \sigma^+.
\end{aligned}
\end{equation}
Here we introduced the functions $m (\lambda)$ and $n (\lambda)$,
\begin{equation}
\label{eq:m-and-n}
\begin{aligned}
    m (\lambda)
    & =
    \rme^{ \rmi x u (\lambda_0) + g (\lambda)}
    \varkappa_\text{reg}^{- 2} (\lambda | \lambda_0)
    (1 - \vartheta (\lambda))
    x^{\tau (\lambda)},
    \\[1ex]
    n (\lambda)
    & =
    - 4
    \rme^{ - \rmi x u (\lambda_0) - g (\lambda)}
    \varkappa_\text{reg}^2 (\lambda | \lambda_0)
    \vartheta (\lambda)
    \sin^2 [\pi \nu (\lambda)]
    x^{- \tau (\lambda)},
\end{aligned}
\end{equation}
where $\varkappa_\text{reg}$ is the part of the function
$\alpha$ that is regular at $\lambda_0$. It is given by
\begin{equation}
\label{eq:varkappa-reg}
    \varkappa_{\text{reg}} (\lambda | \lambda_0)
    = \left[
        \frac{
            \lambda - \lambda_0 }{
            \omega (\lambda - \lambda_0 | \lambda_0) }
    \right]^{\tau (\lambda)}
    \varkappa (\lambda | \lambda_0)
\end{equation}
with $\varkappa$ defined in equation~\eqref{eq:varkappa}.

In \eqref{eq:m-and-n} the `slow variables' $m(\lambda)$,
$n(\lambda)$ and $\tau(\lambda)$ are separated from the
`fast variable' $\zeta(\lambda)$. Note furthermore that the
slow variables are not independent, but satisfy the relation
\begin{equation} \label{eq:rel_mn_tau}
    m(\lambda) n(\lambda) = \rme^{- 2 \pi \rmi \tau(\lambda)} - 1.
\end{equation}

The above preparations allow us to solve the following
local version of the matrix Riemann--Hilbert Problem~\ref{RHp:Upsilon}
for $\Upsilon$ explicitly.
\begin{RHp}
\label{RHp:Pcal}
    Determine $\mathcal{P} (\lambda) \in \mathbb{C}^{2 \times 2}$
    such that
    \begin{enumerate}
        \item
        $\mathcal{P} (\lambda)$ is analytic
        in $\mathcal{U}_{\lambda_0} \backslash \Gamma_\Upsilon$
        and has continuous $\pm$ boundary values on
        $\Gamma_\Upsilon \cap \Ucal_{\lambda_0} \setminus \{\lambda_0\}$.

        \item
        On the contour
        $\Gamma_\Upsilon \cap \Ucal_{\lambda_0} \setminus \{\lambda_0\}$
        the boundary values
        $\mathcal{P}_\pm (\lambda)$
        satisfy the jump condition
        $\mathcal{P}_- (\lambda)
            = \mathcal{P}_+ (\lambda) \rmG_\mathcal{P} (\lambda)$
        with jump matrix
        \begin{equation}
            \rmG_\mathcal{P} (\lambda)
            = \begin{cases}
                \rmM_r^+ (\lambda),
                \quad
                & \lambda \in \Gamma_r^+
                \cap \mathcal{U}_{\lambda_0},
                \\
                \rmM_r^- (\lambda),
                \quad
                & \lambda \in \Gamma_r^-
                \cap \mathcal{U}_{\lambda_0},
                \\
                \rmM_\ell^- (\lambda),
                \quad
                & \lambda \in \Gamma_\ell^-
                \cap \mathcal{U}_{\lambda_0},
                \\
                \rmM_\ell^+ (\lambda),
                \quad
                & \lambda \in \Gamma_\ell^+
                \cap \mathcal{U}_{\lambda_0}.
            \end{cases}
        \end{equation}

        \item
        $\mathcal{P} (\lambda)
            = \mathrm{I}_2
            + x^{ \rho - \frac12}
            \Ocal \big(
                \begin{smallmatrix}1 & 1 \\ 1 & 1\end{smallmatrix}
            \big)$
        uniformly for $\lambda \in \partial \mathcal{U}_{\lambda_0}$, where 
        $\rho = \max_{\lambda \in \Ucal_{\lambda_0}} \bigl|\Real \tau (\lambda)\bigr| < \frac12$.

        \item
        As $\lambda \to \lambda_0$
        \begin{equation}
        \label{eq:asymptotics-of-Parametrix-at-lambda-0}
            \mathcal{P} (\lambda)
            = \Big[
                \mathcal{P}_0
                + (\lambda - \lambda_0)
                \Ocal
                \big(
                    \begin{smallmatrix}
                        1 & 1
                        \\
                        1 & 1
                    \end{smallmatrix}
                \big)
            \Big]
            \cdot \big[ \zeta (\lambda) \big]^{\sigma^z \tau (\lambda)}
        \end{equation}
        for a piecewise constant matrix
        $\mathcal{P}_0 \in \mathbb{C}^{2 \times 2}$.
    \end{enumerate}
\end{RHp}

Utilizing the differential equation method~\cite{Its-81, IS-99},
the Riemann--Hilbert Problem~\ref{RHp:Pcal} can be solved exactly
and explicitly in terms of parabolic cylinder functions. We recall
the method in Appendix~\ref{app:parametrix-construction}, where
we derive the following expression for the solution $\mathcal{P}$
of the Riemann--Hilbert Problem~\ref{RHp:Pcal}:
\begin{equation}
    \label{eq:parametrix}
    \mathcal{P} (\lambda)
    = \Psi (\lambda)
    \rmL (\lambda)
    \rme^{ \flatfrac{ \rmi \zeta^2 (\lambda) \sigma^z }{2} }
    \big[ \zeta (\lambda) \big]^{\tau (\lambda) \sigma^z}.
\end{equation}
Here the matrix $\Psi (\lambda)$ is expressed in terms of the
parabolic cylinder function $D_\tau (\lambda)$,
defined in~\eqref{eq:parabolic-cylinder-function-def}
in Appendix~\ref{app:parabolic-cylinder-functions},
\begin{equation}
    \label{eq:Psi}
    \Psi (\lambda)
    = \begin{pmatrix}
        D_{-\tau (\lambda)}
        ( \sqrt2 \rme^{\frac{\pi \rmi }{4}} \zeta (\lambda) )
        &
        (1 - \rmi)
        b_{12} (\lambda)
        D_{\tau (\lambda) - 1}
        ( \sqrt2 \rme^{-\frac{\pi \rmi }{4}} \zeta (\lambda) )
        \\
        (1 + \rmi)
        b_{21} (\lambda)
        D_{-\tau (\lambda) - 1}
        ( \sqrt2 \rme^{\frac{\pi \rmi }{4}} \zeta (\lambda) )
        &
        D_{\tau (\lambda)}
        ( \sqrt2 \rme^{-\frac{\pi \rmi }{4}} \zeta (\lambda) )
        \\
    \end{pmatrix}.
\end{equation}
The functions $b_{12}$ and $b_{21}$ are given by
\begin{equation}
\label{eq:b12-and-b21}
    b_{12} (\lambda)
    = \frac{
        \rmi \sqrt{\pi} 
        \rme^{\flatfrac{\pi \rmi }{4}}
        \ 2^{\tau (\lambda)} }{
        \rme^{\flatfrac{ \pi \rmi \tau (\lambda) }{ 2 }}
        n (\lambda) \Gamma (\tau (\lambda))
    },
    \qquad
    b_{21} (\lambda)
    = \frac{
        \rme^{\flatfrac{ \pi \rmi \tau (\lambda) }{2} }
        n(\lambda )
        \Gamma (\tau (\lambda) + 1)
    }{
        \sqrt{\pi}
        \rme^{\flatfrac{\pi \rmi }{4}}
        2^{\tau (\lambda) + 1}
    },
\end{equation}
and, finally, the matrix $\rmL (\lambda)$ is defined piecewise
in regions $\Vcal_1$, $\Vcal_2$, $\Vcal_3$, and $\Vcal_4^\pm$,
see Figure~\ref{fig:L-matrix},
\begin{subequations}
    \label{eq:piece-wise-matrix-L}
    \begin{align}
        \rmL (\lambda)
        &= 
        \rme^{\frac{\pi \rmi \tau(\lambda)}{4}}
        2^{\frac{\tau(\lambda) \sigma^z}{2} }
        \begin{pmatrix}
            1 & 0\\
            - n(\lambda) & 1
        \end{pmatrix},
        & \lambda \in \Vcal_1,
        \\
        \rmL (\lambda)
        &= \rme^{\frac{\pi \rmi \tau(\lambda)}{4}}
        2^{\frac{\tau(\lambda) \sigma^z}{2} }
        \begin{pmatrix}
            1 & 0\\ 0 & 1
        \end{pmatrix},
        & \lambda \in \Vcal_2,
        \\
        \rmL (\lambda)
        & =
        \rme^{\frac{\pi \rmi \tau(\lambda)}{4}}
        2^{\frac{\tau(\lambda) \sigma^z}{2} }
        \begin{pmatrix}
            1 & m (\lambda)\\
            0 & 1
        \end{pmatrix},
        & \lambda \in \Vcal_3,
        \\
        \rmL (\lambda)
        & =
        \rme^{\frac{\pi \rmi \tau(\lambda)}{4}}
        2^{\frac{\tau(\lambda) \sigma^z}{2} }
        \begin{pmatrix}
            1 & 0\\
            - n(\lambda) \exp(2 \pi \rmi \tau (\lambda)) & 1
        \end{pmatrix},
        & \lambda \in \Vcal_4^-,
        \\
        \rmL (\lambda)
        & =
        \rme^{\frac{\pi \rmi \tau(\lambda)}{4}}
        2^{\frac{\tau(\lambda) \sigma^z}{2} }
        \begin{pmatrix}
            1 & m (\lambda) \exp(2 \pi \rmi \tau (\lambda))\\
            0 & 1
        \end{pmatrix},
        & \lambda \in \Vcal_4^+.
    \end{align}
\end{subequations}
We note as well that it follows from~\eqref{eq:b12-and-b21} that
\begin{equation}
    \label{eq:product-of-b12-and-b21}
    2 b_{12} (\lambda) b_{21} (\lambda)
    = \rmi \tau (\lambda).
\end{equation}
\begin{figure}[ht]
    \centering
    \inputfig{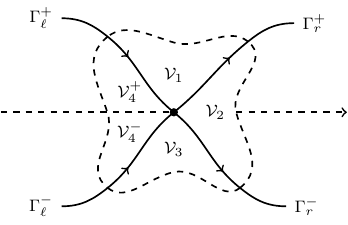}
    \caption{
        Regions $\Vcal_1$, $\Vcal_2$, $\Vcal_3$, $\Vcal_4^\pm$.}
    \label{fig:L-matrix}
\end{figure}

\subsection{Global solution}
\label{sec:Phi}
Finally, taking advantage of the fact that $\Pcal$ and $\Upsilon$
satisfy the same jump condition inside $\Ucal_{\lambda_0}$, we can
construct a Riemann--Hilbert problem with a jump matrix that is
small in~$x$ uniformly on its whole jump contour. Let
\begin{equation}
    \Phi (\lambda)
    =
    \begin{cases}
        \Upsilon (\lambda),
        & \lambda \in \mathbb{C} \backslash \mathcal{U}_{\lambda_0},
        \\[1ex]
        \Upsilon (\lambda) \mathcal{P}^{-1} (\lambda),
        & \lambda \in \mathcal{U}_{\lambda_0},
    \end{cases}
\end{equation}
and
\begin{equation}
    \Gamma_\Phi = (- \partial \Ucal_{\lambda_0})
    \cup (\Gamma_\Upsilon \setminus \Ucal_{\lambda_0}),
\end{equation}
see Figure~\ref{fig:Phi}. For convenience we also
introduce the contours
$\widetilde{\Gamma}_\ell^\pm
    = \Gamma_\ell^\pm \setminus \Ucal_{\lambda_0}$
and $\widetilde{\Gamma}_r^\pm
    = \Gamma_r^\pm \setminus \Ucal_{\lambda_0}$.
Then the matrix $\Phi$ is the unique solution
of the following matrix Riemann--Hilbert problem.
\begin{figure}[th]
    \centering
    \inputfig{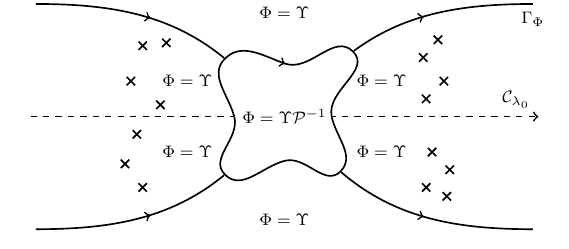}
    \caption{Definition of the matrix $\Phi$
        and of the oriented contour $\Gamma_\Phi$.}
    \label{fig:Phi}
\end{figure}
\begin{RHp}
    \label{RHp:Phi}
    Determine $\Phi (\lambda) \in \mathbb{C}^{2 \times 2}$ such that
    \begin{enumerate}
        \item
        $\Phi (\lambda)$ is analytic in
        $\mathbb{C} \backslash
            \left( \Gamma_\Phi \cup \mathcal{S} \right)$
        and has continuous $\pm$ boundary values on $\Gamma_\Phi$.

        \item
        On the contour $\Gamma_\Phi$ the boundary values
        $\Phi_\pm (\lambda)$ satisfy the jump condition
        $\Phi_- (\lambda) = \Phi_+ (\lambda) \rmG_\Phi (\lambda)$
        with jump matrix $\rmG_\Phi (\lambda)$ given by
        \begin{equation}
        \label{eq:jump-matrix-Phi}
            \rmG_\Phi (\lambda)
            = \begin{cases}
                \mathcal{P}^{-1} (\lambda),
                \quad & \lambda \in - \partial \mathcal{U}_{\lambda_0},
                \\
                \rmM_\ell^+ (\lambda),
                \quad & \lambda \in \widetilde\Gamma_\ell^+,
                \\
                \rmM_r^+ (\lambda),
                \quad & \lambda \in \widetilde\Gamma_r^+,
                \\
                \rmM_\ell^- (\lambda),
                \quad & \lambda \in \widetilde\Gamma_\ell^-,
                \\
                \rmM_r^- (\lambda),
                \quad & \lambda \in \widetilde\Gamma_r^-.
            \end{cases}
        \end{equation}

        \item
        $\Phi (\lambda)
            = \mathrm{I}_2
            + \lambda^{-1}
            \Ocal \big(
                \begin{smallmatrix}1 & 1 \\ 1 & 1\end{smallmatrix}
            \big)$
        as $\lambda \to \infty$
        up to tangential direction to $\Gamma_\Phi$.

        \item
        $\Phi$ satisfies the following regularity conditions
        at the poles $\lambda \in \mathcal{S}$,
        \begin{equation}
        \label{eq:regularity-condition-for-Phi}
            \begin{aligned}
                \Phi (\lambda) \cdot (\rmM_\ell^+)^{-1} (\lambda)
                \text{ is regular at }
                \lambda = \ell_j^+,
                \quad
                j = 1, \dots, n_\ell^+,
                \\
                \Phi (\lambda) \cdot (\rmM_r^+)^{-1} (\lambda)
                \text{ is regular at }
                \lambda = r_j^+,
                \quad
                j = 1, \dots, n_r^+,
                \\
                \Phi (\lambda) \cdot \rmM_\ell^- (\lambda)
                \text{ is regular at }
                \lambda = \ell_j^-,
                \quad
                j = 1, \dots, n_\ell^-,
                \\
                \Phi (\lambda) \cdot \rmM_r^- (\lambda)
                \text{ is regular at }
                \lambda = r_j^-,
                \quad
                j = 1, \dots, n_r^-.
            \end{aligned}
        \end{equation}
    \end{enumerate}
\end{RHp}

We divide the solution of the Riemann--Hilbert
Problem~\ref{RHp:Phi} in two separate steps.
First we deal only with the jump condition
and ignore the conditions at the poles.
That is to say, we seek for a matrix $\Pi (\lambda)$
which is the solution of the following Riemann--Hilbert problem.
\begin{RHp}
    \label{RHp:Pi}
    Determine $\Pi (\lambda) \in \mathbb{C}^{2 \times 2}$ such that
    \begin{enumerate}
        \item
        $\Pi (\lambda)$ is analytic
        in $\mathbb{C} \backslash \Gamma_\Phi$
        and has continuous $\pm$ boundary values on $\Gamma_\Phi$.
        
        \item
        On the contour $\Gamma_\Phi$ the boundary values
        $\Pi_\pm (\lambda)$ satisfy the jump condition
        $\Pi_- (\lambda) = \Pi_+ (\lambda) \rmG_\Phi (\lambda)$
        with jump matrix $\rmG_\Phi (\lambda)$,
        see equation~\eqref{eq:jump-matrix-Phi}.

        \item
        $\Pi (\lambda) = \mathrm{I}_2
            + \lambda^{-1}
            \Ocal \big(
                \begin{smallmatrix}1 & 1 \\ 1 & 1\end{smallmatrix}
            \big)$
        as $\lambda \to \infty$
        up to tangential direction to $\Gamma_\Phi$.
    \end{enumerate}
\end{RHp}
If we multiply the unique solution $\Pi (\lambda)$ of the above
Riemann--Hilbert problem by a matrix $\rmS (\lambda)$ from the
left, which is meromorphic everywhere in the complex plane
except at the poles $\lambda \in \Scal$, then the product
\begin{equation}
\label{eq:Phi}
    \Phi (\lambda)
    = \rmS (\lambda) \Pi (\lambda)
\end{equation}
is the unique solution of the Riemann--Hilbert Problem~\ref{RHp:Phi},
provided it has the required asymptotics for $\lambda \rightarrow
\infty$ and the regularity conditions are fulfilled. As we shall
see below in Section~\ref{sec:contribution-of-poles}, these
requirements determine the form of the matrix $\rmS (\lambda)$
and imply a system of linear equations for its residua that
fixes it completely.

Before continuing with the derivation of these linear equations,
following Beals and Coifman
\cite{BealsCoifmanScatteringInFirstOrderSystemsEquivalenceRHPPRoofRHPUniqueness},
we derive a singular integral equation for the matrix 
$\Pi (\lambda)$ in Section~\ref{sec:Pi} that can be solved by
iteration.

\subsection{Singular integral equation}
\label{sec:Pi}
The ultimate goal behind the chain of transformations that
we have applied to the original Riemann--Hilbert problem for $\chi$
was to construct an equivalent Riemann--Hilbert problem with
jump matrix that is uniformly close to the identity matrix for large $x$. This
has been achieved with the matrix Riemann--Hilbert Problem~\ref{RHp:Pi}
which has a jump matrix that is exponentially close to identity
on $\Gamma_\Upsilon \setminus \Ucal_{\lambda_0}$ and uniformly
close to identity as $\mathrm{I}_2 + \Ocal (x^{-1/2 + \rho})$
on $- \partial \Ucal_{\lambda_0}$. It is one of the essential
ideas of the nonlinear steepest descent method that the large-$x$
asymptotic solution of the Riemann--Hilbert problem can
then be found by solving an equivalent singular integral equation
in terms of its Neumann series~%
\cite{BealsCoifmanScatteringInFirstOrderSystemsEquivalenceRHPSingIntEqnMention,
BealsCoifmanScatteringInFirstOrderSystemsEquivalenceRHPPRoofRHPUniqueness}.

For $f \in L^p (\Gamma_\Phi)$, $1 \le p < + \infty$ let
\begin{equation}
    \rmC_{\Gamma_\Phi} [f] (\lambda) = 
      \int\limits_{\Gamma_\Phi} \frac{\dd{\mu}}{ 2 \pi \rmi }
      \frac{f(\mu)}{\mu - \lambda}
\end{equation}
be the Cauchy transform of $f$ with respect to $\Gamma_\Phi$.
Since $\Gamma_\Phi$ is a union of oriented Lipschitz
curves, $\rmC_{\Gamma_\Phi} [f]$ admits continuous $\pm$
boundary values $\rmC_{\Gamma_\Phi}^{(\pm)} [f]$
on $\Gamma_\Phi$ and defines a continuous map \cite{Calderon77}
$\rmC_{\Gamma_\Phi} [f]: L^p (\Gamma_\Phi) \rightarrow
L^p (\Gamma_\Phi)$, \textit{i.e.}, there exists some
$c > 0$ such that
\begin{equation} \label{eq:bv_cauchy_bv_continuous}
    \bigl\|\rmC_{\Gamma_\Phi}^{(\pm)} [f]\bigr\|_{L^p(\Gamma_\Phi)}
       \le c \|f\|_{L^p(\Gamma_\Phi)}.
\end{equation}
Further observe that, given the ${\rm Mat}_{2 \times 2}
( L^p (\Gamma_\Phi))$ norm
\begin{equation} \label{eq:Lp_matrix_norm}
    \| \rmM \|_p
    = \biggl(
        \sum_{a,b=1,2}
        \| \rmM_{ab}\|_{L^p (\Gamma_\Phi)}^p
    \biggr)^{1/p},
\end{equation}
it readily follows from the explicit form of 
$\rmG_\Phi - \rmI_2$ that, for all $p \ge 1$ up
to $p = + \infty$ and for some $c_p > 0$
\begin{equation}
    \|\rmG_\Phi - \rmI_2\|_p \le c_p x^{\rho - 1/2},
\end{equation}
where $\rho - 1/2 < 0$
as in the Riemann--Hilbert Problem~\ref{RHp:Pcal}.

Consider the singular integral equation
on ${\rm Mat}_{2 \times 2} ( L^p (\Gamma_\Phi))$,
$1 < p < + \infty$,
\begin{equation}
    Q(\lambda) = - \rmC_{\Gamma_\Phi}^{(+)} [\rmG_\Phi - \rmI_2] (\lambda)
       - R[Q] (\lambda),
\end{equation}
where
\begin{equation}
    R[Q] (\lambda) = \rmC_{\Gamma_\Phi}^{(+)} [Q(\rmG_\Phi - \rmI_2)] (\lambda).
\end{equation}
Then \eqref{eq:bv_cauchy_bv_continuous} and \eqref{eq:Lp_matrix_norm}
imply that
\begin{equation}
    \bigl\|\rmC_{\Gamma_\Phi}^{(+)} [\rmG_\Phi - \rmI_2]\bigr\|_p \le
       c \|\rmG_\Phi - \rmI_2\|_p,
\end{equation}
\textit{i.e.}, $\rmC_{\Gamma_\Phi}^{(+)} [\rmG_\Phi - \rmI_2] \in
{\rm Mat}_{2 \times 2} ( L^p (\Gamma_\Phi))$.

Moreover, for every $p > 1$, there exists $c > 0$ such that
\begin{multline}
    \Bigl(
        \bigl\|R[Q]\bigr\|_p
    \Bigr)^p
    = \sum_{a,b=1}^2 
    \Bigl(
        \bigl\|
            \rmC_{\Gamma_\Phi}^{(+)} [(Q(\rmG_\Phi - \rmI_2))_{ab}]
        \bigr\|_{L^p(\Gamma_\Phi)}
    \Bigr)^p \\
    \le c \sum_{a,b=1}^2
    \Bigl(
        \bigl\|
            (Q(\rmG_\Phi - \rmI_2))_{ab}
        \bigr\|_{L^p(\Gamma_\Phi)}
    \Bigr)^p
    \le c \sum_{a,b=1}^2
      \biggl(\sum_{c=1}^2\bigl\|Q_{ac}(\rmG_\Phi - \rmI_2)_{cb}\bigr\|_{L^p(\Gamma_\Phi)}\biggr)^p \\
   \le 2 c \Bigl(\max_{c,b} \bigl\|(\rmG_\Phi - \rmI_2)_{cb}\bigr\|_{L^\infty (\Gamma_\Phi)}\Bigr)^p
       \sum_{a=1}^2
          \biggl(\sum_{c=1}^2\bigl\|Q_{ac}\bigr\|_{L^p(\Gamma_\Phi)}\biggr)^p \\
       \le c\, 2^{p+1} \Bigl(\max_{c,b} \bigl\|(\rmG_\Phi - \rmI_2)_{cb}\bigr\|_{L^\infty (\Gamma_\Phi)}\Bigr)^p
       \sum_{a,c=1}^2
          \biggl(\bigl\|Q_{ac}\bigr\|_{L^p(\Gamma_\Phi)}\biggr)^p.
\end{multline}
Here we have used \eqref{eq:bv_cauchy_bv_continuous} in the
second relation and the Minkowski inequality in the third
relation. Using the bounds on $(\rmG_\Phi - \rmI_2)_{ab}$
we conclude that there exists $c > 0$ such that
\begin{equation}
    \bigl\|R[Q]\bigr\|_p \le \frac{c}{x^{1/2 - \rho}} \|Q\|_p,
\end{equation}
where $\rho < 1/2$. Thus, for $x$ large enough, $\id + R$ is
invertible on ${\rm Mat}_{2 \times 2} ( L^p (\Gamma_\Phi))$,
its inverse computable by the Neumann series. In particular,
there exists a unique solution $Q \in {\rm Mat}_{2 \times 2}
(L^p (\Gamma_\Phi))$ to $(\id + R)[Q] = - 
\rmC_{\Gamma_\Phi}^{(+)} [\rmG_\Phi - \rmI_2]$, and it follows
that $\Pi_+ = \rmI_2 + Q$ satisfies the linear singular integral
equation
\begin{equation}
\label{eq:singular-integral-equation-Pi-plus}
    \Pi_+ (\lambda)
    = \mathrm{I}_2
    - \rmC_{\Gamma_\Phi}^{(+)}
    \left[
        \Pi_+ \left(\rmG_\Phi - \mathrm{I}_2\right)
    \right] (\lambda)
\end{equation}
in which all is well defined, since 
$\rmC_{\Gamma_\Phi}^{(+)} [\rmG_\Phi - \rmI_2] \in
{\rm Mat}_{2 \times 2} (L^p (\Gamma_\Phi))$.

It then follows from standard reasoning that
\begin{equation}
\label{eq:singular-integral-equation}
    \Pi (\lambda)
    = \mathrm{I}_2
    - \int\limits_{\Gamma_\Phi}
    \frac{\dd{\mu}}{ 2 \pi \rmi }
    \frac{
        \Pi_+ (\mu) \left(\rmG_\Phi (\mu)
        - \mathrm{I}_2\right) }{
        \mu - \lambda },
    \qquad
    \lambda \in \mathbb{C} \backslash \Gamma_\Phi,
\end{equation}
is the unique solution of the Riemann--Hilbert Problem~\ref{RHp:Pi}.

We will derive a recursive expression for an asymptotic solution $\Pi$
in terms of the parametrix $\Pcal$ in Section~\ref{sec:asymptotics-Pi}.

\subsection{Contribution of the poles}
\label{sec:contribution-of-poles}
In order for the Riemann--Hilbert Problem~\ref{RHp:Phi}
to have a solution $\Phi (\lambda)$ of the form~\eqref{eq:Phi},
the matrix $\rmS (\lambda)$ must be of the form
\begin{equation}
    \label{eq:matrix-S}
    \rmS (\lambda)
    = \mathrm{I}_2
    + \sum\limits_{\mu \in \Scal} \frac{\rmS_\mu}{\lambda - \mu},
\end{equation}
where the $2 \times 2$ matrices $\rmS_\mu$
do not depend on $\lambda$.
The regularity conditions~\eqref{eq:regularity-condition-for-Phi}
lead to a linear system of equations
for these matrices, which we derive
in Appendix~\ref{app:derivation-of-the-linear-system}.

We note that 
\begin{equation}
\label{eq:det-S}
    \det \rmS (\lambda) = 1.
\end{equation}
To see this, first notice that $\det \Phi (\lambda) = 1$,
since all transformation matrices connecting $\Phi$ with the
solution $\chi$ of the original Riemann--Hilbert
Problem~\ref{RHp:chi} have determinant equal to one, which also holds
for $\chi$ itself. Moreover, $\det \Pi (\lambda) = 1$,
because, as for $\chi$, we have $\det \rmG_\Phi (\lambda)
= 1$ for $\lambda \in \Gamma_\Phi$, and the asymptotic
condition on $\Pi (\lambda)$ for $\lambda \to \infty$. Thus,
$\det \rmS (\lambda) = 1$ by~\eqref{eq:Phi}.

The system of linear equations that determines the matrices $\rmS_\mu$
takes a simpler form, if we combine the poles in the regions
above (resp.\ below) the contour $\Ccal_{\lambda_0}$
to the left of the saddle point with the poles below
(resp.\ above) the contour $\Ccal_{\lambda_0}$
to the right of the saddle point in one set,
which we denote $\Scal^+$ (resp.\ $\Scal^-$).
In other words,
\begin{equation}
    \Scal^\pm
    = \Lcal^\pm \cup \Rcal^\mp,
\end{equation}
see the definitions of the sets~\eqref{eq:sets-of-poles}.
The cardinality of the sets $\Scal^\pm$
are denoted as $n^\pm = n_\ell^\pm + n_r^\mp$.

Then the regularity conditions~\eqref{eq:regularity-condition-for-Phi}
imply the following representation of the matrix $\rmS (\lambda)$:
\begin{equation}
\label{eq:matrix-S-new}
    \rmS (\lambda)
    = \mathrm{I}_2
    + \sum\limits_{j = 1}^{n^+}
    \frac{ \sigma_j^+ }{\lambda - s_j^+}
    (\mathbf{0}, \mathbf{y}_j) \Pi^{-1} (s_j^+)
    + \sum\limits_{j = 1}^{n^-}
    \frac{ \sigma_j^- }{\lambda - s_j^-}
    (\mathbf{x}_j, \mathbf{0}) \Pi^{-1} (s_j^-),
\end{equation}
where $s_j^\pm$ for $j = 1, \dots, n^\pm$
are the poles belonging to the set $\Scal^\pm$,
and the two-dimensional vectors $\mathbf{x}$ and $\mathbf{y}$
are the solutions of a linear system of $n^+ + n^-$ 
algebraic equations having the $2 \times 2$ block form
\begin{equation}
\label{eq:SLE-new}
\left\{
\begin{aligned}
    \mathbf{y}_j
    &= \mathbf{w}_j
    + \sum\limits_{\substack{k = 1\\k \neq j}}^{n^+}
    \frac{
        \sigma_k^+ \big( P_{kj}^{++} \big)_{21} }{
        s_j^+ - s_k^+ }
    \mathbf{y}_k
    + \sum\limits_{k = 1}^{n^-}
    \frac{
        \sigma_k^- \big( P_{kj}^{-+} \big)_{11} }{
        s_j^+ - s_k^- }
    \mathbf{x}_k,
    &&
    \qquad
    j = 1, \dots, n^+,
    \\
    \mathbf{x}_j
    &= \mathbf{v}_j
    + \sum\limits_{k = 1}^{n^+}
    \frac{
        \sigma_k^+ \big( P_{kj}^{+-} \big)_{22} }{
        s_j^- - s_k^+ }
    \mathbf{y}_k
    + \sum\limits_{\substack{k = 1\\k \neq j}}^{n^-}
    \frac{
        \sigma_k^- \big( P_{kj}^{--} \big)_{12} }{
        s_j^- - s_k^- }
    \mathbf{x}_k,
    &&
    \qquad
    j = 1, \dots, n^-.
\end{aligned}
\right.
\end{equation}
Here the vectors $\mathbf{w}_j$ and $\mathbf{v}_j$
are defined by
\begin{equation}
\label{eq:vectors-v-and-w}
    \mathbf{w}_j
    = \begin{pmatrix}
        \Pi_{11} (s_j^+) \\ \Pi_{21} (s_j^+)
    \end{pmatrix},
    \qquad
    \mathbf{v}_j
    = \begin{pmatrix}
        \Pi_{12} (s_j^-) \\ \Pi_{22} (s_j^-)
    \end{pmatrix},
\end{equation}
and the coefficients $\sigma^\pm$ by
\begin{equation}
\label{eq:coefficients-sigma}
    \sigma_j^+
    = \frac{ h_j^+ }{
        1 - h_{j}^+
        \big( Q_{jj}^{++} \big)_{21} },
    \qquad
    \sigma_j^-
    = \frac{ h_j^- }{
        1 - h_{j}^-
        \big( Q_{jj}^{--} \big)_{12} },
\end{equation}
where we introduced two types of coefficients $h^\pm$
which are (up to their sign)
the corresponding residues of the off-diagonal matrix elements
of the matrices $\rmM_\ell^\pm$ and $\rmM_r^\pm$,
see equations~\eqref{eq:matrices-M-left} and~\eqref{eq:matrices-M-right},
\begin{equation}
\label{eq:residues-h}
\begin{aligned}
    h_j^+
    =
    e^{- 2} (s_j^+)
    & \times
    \begin{cases}
        \res\limits_{\lambda = s_j^+}
        Q_\ell^+ (\lambda),
        \qquad & s_j^+ \in \Lcal^+,
        \\
        - \res\limits_{\lambda = s_j^+}
        Q_r^- (\lambda),
        \qquad & s_j^+ \in \Rcal^-,
    \end{cases}
    \\
    h_j^-
    =
    e^2 (s_j^-)
    & \times
    \begin{cases}
        - \res\limits_{\lambda = s_j^-}
        Q_\ell^- (\lambda),
        \qquad & s_j^- \in \Lcal^-,
        \\
        \res\limits_{\lambda = s_j^-}
        Q_r^+ (\lambda),
        \qquad & s_j^- \in \Rcal^+.
    \end{cases}
\end{aligned}
\end{equation}
Above we used the following short-hand notation for the elements
of products of the matrix $\Pi^{-1} (\lambda)$ with $\Pi (\lambda)$,
$\Pi' (\lambda)$:
\begin{equation}
\label{eq:matrix-elements-P-and-Q}
\begin{aligned}
    \big( P_{jk}^{\epsilon \epsilon'} \big)_{m n}
    & = \big(
        \Pi^{-1} (s_j^\epsilon) \Pi (s_k^{\epsilon'})
    \big)_{mn},
    \qquad
    && m, n = 1, 2,
    \quad
    && \epsilon, \epsilon' = \pm,
    \\
    \big( Q_{jk}^{\epsilon \epsilon'} \big)_{mn}
    & = \big(
        \Pi^{-1} (s_j^\epsilon) \Pi' (s_k^{\epsilon'})
    \big)_{mn},
    \qquad
    && m, n = 1, 2,
    \quad
    && \epsilon, \epsilon' = \pm.
\end{aligned}
\end{equation}

\subsection{Preparation for the asymptotic analysis}
\label{sec:deformation-of-the-contour}
Now we turn back to Proposition~\ref{prop:log-der},
see equation~\eqref{eq:prop:log-der},
and derive expressions
for the logarithmic derivatives of the Fredholm determinant
that are directly suitable for the asymptotic analysis. For this
purpose we substitute the chain of transformations that
connects the solution $\chi$ of the original Riemann--Hilbert
Problem~\ref{RHp:chi} with that of the final Riemann--Hilbert
Problem~\ref{RHp:Phi} into equation~\eqref{eq:prop:log-der}
while deforming the integration contour appropriately
in order to finally arrive at
\begin{proposition}
\label{prop:log-der-Fredholm-for-AE}
    If $\det_{\Ccal_{\lambda_0}} (\id + \rmV) \neq 0$,
    then the logarithmic derivative of the Fredholm determinant
    of the integrable integral operator $\mathrm{V}$ with
    kernel~\eqref{eq:kernel-V-as-scalar-product}
    with respect to the parameters $\beta = x, \lambda_0$
    admits the representation
    \begin{multline}
    \label{eq:prop:log-der-Fredholm-for-AE}
        \partial_\beta \ln \det_{\Ccal_{\lambda_0}}
        \big( \id + \mathrm{V} \big)
        = a_\beta (x, \lambda_0)
        \\
        - \int\limits_{\gamma_0}
        \frac{\dd{\lambda}}{ 2 \pi \rmi }
        \tr\{ \Pi' (\lambda) \sigma^z \Pi^{-1} (\lambda)\}
        d_\beta (\lambda)
        - \int\limits_{\gamma_0}
        \frac{\dd{\lambda}}{ 2 \pi \rmi }
        \tr\{
            \rmS' (\lambda) \Pi (\lambda)
            \sigma^z
            \Pi^{-1} (\lambda) \rmS^{-1} (\lambda)
        \}
        d_\beta (\lambda)
        \\
        - \sum\limits_{\mu \in \mathcal{S}}
        \res\limits_{\lambda = \mu}
        \left(
            \tr\{
                \rmS' (\lambda) \Pi (\lambda)
                \sigma^z
                \Pi^{-1} (\lambda) \rmS^{-1} (\lambda)
            \}
            d_\beta (\lambda)
        \strut\right)
        + \Ocal (\rme^{- m x}),
    \end{multline}
    where
    \begin{equation}
    \label{eq:prop:function-a}
        a_\beta (x, \lambda_0)
        = 2 \int\limits_{\Ccal_{\lambda_0}}
        \dd{\lambda}
        \Lcal (\lambda | \lambda_0) d_\beta' (\lambda)
    \end{equation}
    by definition. Here
    $\Lcal(\cdot|\lambda_0)$ was defined in~\eqref{eq:definition-of-Lcal}
    and $d_\beta$ in~\eqref{eq:function-d}.
    The matrix $\Pi$ is the unique solution
    of the Riemann--Hilbert Problem~\ref{RHp:Pi},
    and $\rmS (\lambda)$ is the matrix accounting for the contribution
    of the poles in the set $\Scal$ to the asymptotics.
    It is expressed in terms of the solution
    of a corresponding system of linear algebraic equations 
    derived in Section~\ref{sec:contribution-of-poles}.
    The integration contour $\gamma_0$
    is shown in Figure~\ref{fig:contour-gamma-0-thm}.
    The corrections are determined by
    \begin{equation}
        m = \min\{
            m_\ell^+, m_r^+, m_\ell^-, m_r^-
        \},
    \end{equation}
    where
    \begin{equation}
    \begin{aligned}
        m_\ell^+
        & = \inf\limits_{\lambda \in \gamma_\ell^+} \Imag( u (\lambda) ),
        \qquad
        & m_r^+
        & = \inf\limits_{\lambda \in \gamma_r^+} \Imag( - u (\lambda) ),
        \\
        m_\ell^-
        & = \inf\limits_{\lambda \in \gamma_\ell^-} \Imag( - u (\lambda) ),
        \qquad
        & m_r^-
        & = \inf\limits_{\lambda \in \gamma_r^-} \Imag( u (\lambda) ),
    \end{aligned}
    \end{equation}
    see Figure~\ref{fig:contour-gamma-0-thm}.
\end{proposition}
\begin{figure}[ht]
    \centering
    \inputfig{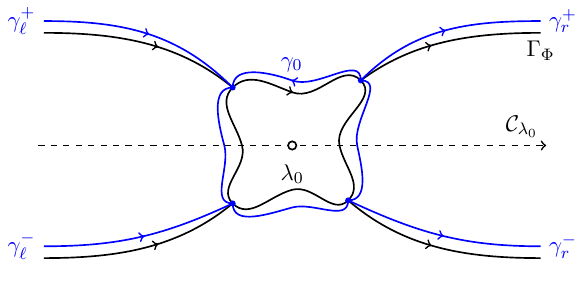}
    \caption{The initial integration contour $\Ccal_{\lambda_0}$ (dashed),
        the jump contour $\Gamma_\Phi$ (black)
        and the integration contour $\gamma_0$ (blue)
        and contours $\gamma_\ell^\pm$ and $\gamma_r^\pm$ (also blue).}
    \label{fig:contour-gamma-0-thm}
\end{figure}

\subsubsection{Proof of Proposition~\ref{prop:log-der-Fredholm-for-AE}}
In order to derive \eqref{eq:prop:log-der-Fredholm-for-AE}, we successively
insert the chain of transformations of the solution $\chi$ of the initial
Riemann--Hilbert problem that was constructed in the previous sections into
\eqref{eq:prop:log-der} and deform the integration contour 
$\Gamma (\Ccal_{\lambda_0})$ to our convenience. In
Figure~\ref{fig:expression-for-chi} we show how in the different regions
$\chi$ is expressed in terms of:
\begin{enumerate}
    \item
    the Cauchy transform of $e^{-2} (\lambda)$
    along the contour $\Ccal_{\lambda_0}$;

    \item
    the solution $\alpha$
    of the scalar Riemann--Hilbert Problem~\ref{RHp:alpha};
    
    \item
    the upper- and lower-triangular matrices $\rmM_\ell^\pm$ and $\rmM_r^\pm$,
    see equations~\eqref{eq:matrices-M-left} and~\eqref{eq:matrices-M-right};

    \item
    the parametrix $\Pcal$,
    the solution of the local Riemann--Hilbert Problem~\ref{RHp:Pcal};

    \item
    the solution $\Pi$ of the Riemann--Hilbert Problem~\ref{RHp:Pi}
    or, equivalently, of the linear singular integral
    equation~\eqref{eq:singular-integral-equation};

    \item
    the matrix $\rmS$
    accounting for the contribution of the poles,
    see expression~\eqref{eq:matrix-S-new},
    which is determined by the linear system of equations~\eqref{eq:SLE-new}.
\end{enumerate}
\begin{figure}[ht]
    \centering
    \inputfig{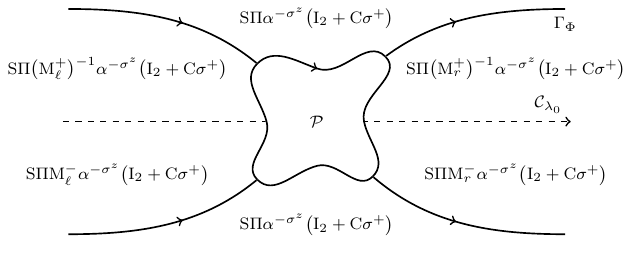}
    \caption{Expressions for the solution $\chi$
        of the initial Riemann--Hilbert Problem~\ref{RHp:chi}
        in terms of the data
        of the transformed Riemann--Hilbert problems
        in all regions
        except for the vicinity of the saddle point $\lambda_0$,
        where the parametrix is constructed.}
    \label{fig:expression-for-chi}
\end{figure}

\paragraph{$\bullet$ Contribution of the scalar Riemann--Hilbert problem}
We first substitute equation~\eqref{eq:chi-tilde} into \eqref{eq:prop:log-der}.
Then
\begin{equation}
\label{eq:log-der-Fredholm-Xi-tilde}
    \partial_\beta \ln \det_{\mathcal{C}_{\lambda_0}}
    (\id + \mathrm{V})
    =
    - \int\limits_{\Gamma (\mathcal{C}_{\lambda_0})}
    \eval{
        \frac{ \dd{\lambda} }{2 \pi \rmi}
        \tr\left\{
            \widetilde{\chi}' (\lambda)
            \sigma^z
            \widetilde{\chi}^{-1} (\lambda)
        \right\}
        d_\beta (\lambda)
        \rme^{- \eta \lambda^2}
    }_{\eta = 0+}.
\end{equation}

Taking into account equation~\eqref{eq:Xi}, the trace under the
integral~\eqref{eq:log-der-Fredholm-Xi-tilde} takes the form
\begin{equation}
    \tr\left\{
        \widetilde{\chi}' (\lambda)
        \sigma^z
        \widetilde{\chi}^{-1} (\lambda)
    \right\}
    \\
    =
    - 2 \partial_\lambda \ln \alpha (\lambda)
    + \tr\left\{
        \Xi' (\lambda)
        \sigma^z
        \Xi^{-1} (\lambda)
    \right\}.
\end{equation}
The integral over the first term on the right-hand side
of this equation can be evaluated as
\begin{multline}
\label{eq:function-a-def}
    2 \int\limits_{\Gamma(\Ccal_{\lambda_0})}
    \frac{ \dd{\lambda} }{2 \pi \rmi}
    \eval{
        \bigl(\partial_\lambda \ln  \alpha (\lambda)\bigr)
        d_\beta (\lambda)
        \rme^{- \eta \lambda^2}
    }_{\eta = 0+} = -
    2 \int\limits_{\Gamma(\Ccal_{\lambda_0})}
    \frac{ \dd{\lambda} }{2 \pi \rmi}
    \eval{
        \ln  \alpha (\lambda)\,
        d_\beta' (\lambda)
        \rme^{- \eta \lambda^2}
    }_{\eta = 0+}
    \\
    =
    2 \int\limits_{\Ccal_{\lambda_0}}
    \frac{ \dd{\lambda} }{2 \pi \rmi}
    \eval{
        \ln \biggl(\frac{ \alpha_+ (\lambda) }{ \alpha_- (\lambda)}\biggr)
        d_\beta' (\lambda)
        \rme^{- \eta \lambda^2}
    }_{\eta = 0+}
    = 2
    \int\limits_{\Ccal_{\lambda_0}}
    \dd{\lambda}
    \Lcal (\lambda | \lambda_0)
    \, d_\beta' (\lambda) =
    a_\beta (x, \lambda_0).
\end{multline}
Here we have integrated by parts in the first equation, have
used the jump condition~\eqref{eq:jump-condition-alpha} and
the definition~\eqref{eq:definition-of-Lcal}
of the function $\Lcal(\cdot|\lambda_0)$
in the third equation,
and have used the definition \eqref{eq:prop:function-a} of
$a_\beta (x,\lambda_0)$ in the fourth equation. It follows that
\begin{equation}
\label{eq:log-der-Fredholm-Xi}
    \partial_\beta \ln \det_{\Ccal_{\lambda_0}}
    \big( \id + \rmV \big)
    = a_\beta (x, \lambda_0)
    - \int\limits_{\Gamma(\Ccal_{\lambda_0})}
    \eval{
        \frac{ \dd{\lambda} }{2 \pi \rmi}
        \tr\left\{
            \Xi' (\lambda) \sigma^z \Xi^{-1} (\lambda)
        \right\}
        d_\beta (\lambda) \rme^{- \eta \lambda^2}
    }_{\eta = 0+}.
\end{equation}

Since the matrix $\Xi (\lambda)$ is analytic
for $\lambda \in \mathbb{C} \backslash \Ccal_{\lambda_0}$,
we can deform the integration contour $\Gamma (\Ccal_{\lambda_0})$
in the integral involving $\Xi (z)$
on the right-hand side of equation~\eqref{eq:log-der-Fredholm-Xi}
into a wider contour $\Gamma' (\Ccal_{\lambda_0})$
as shown in Figure~\ref{fig:contour-Gamma'}.
\begin{figure}[ht]
    \centering
    \resizebox{\textwidth}{!}{
        \inputfig{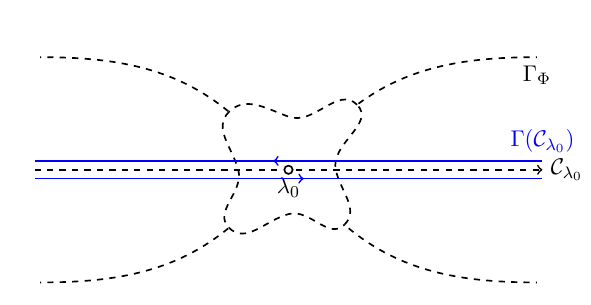}%
        \inputfig{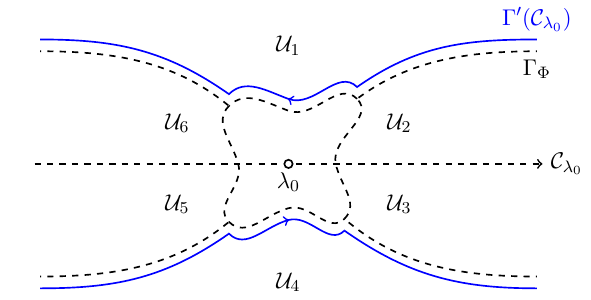}%
    }
    \caption{Deformation of the contour
        $\Gamma (\Ccal_{\lambda_0})$ into $\Gamma' (\Ccal_{\lambda_0})$
        for the integral involving the matrix $\Xi$.}
    \label{fig:contour-Gamma'}
\end{figure}

Therefore, we get
\begin{equation}
    \label{eq:log-der-Fredholm-Xi-Gamma-prime}
    \partial_\beta \ln \det_{\Ccal_{\lambda_0}}
    \big( \id + \rmV \big)
    = a_\beta (x, \lambda_0)
    - \int\limits_{\Gamma' (\Ccal_{\lambda_0})}
     \eval{
        \frac{ \dd{\lambda} }{2 \pi \rmi}
        \tr\left\{
            \Xi' (\lambda) \sigma^z \Xi^{-1} (\lambda)
        \right\}
        d_\beta (\lambda) \rme^{- \eta \lambda^2}
    }_{\eta = 0+}.
\end{equation}
Now we substitute $\Xi \to \Upsilon \to \Phi$,
see Sections~\ref{sec:Upsilon}--\ref{sec:Phi}.
These substitutions are trivial in the regions
containing the integration contour $\Gamma'$,
see Figures~\ref{fig:Upsilon-def}, \ref{fig:Phi}
and~\ref{fig:contour-Gamma'}.
Hence,
\begin{equation}
    \label{eq:log-der-Fredholm-Phi}
    \partial_\beta \ln \det_{\Ccal_{\lambda_0}}
    (\id + \rmV)
    = a_\beta (x, \lambda_0)
    - \int\limits_{\Gamma' (\Ccal_{\lambda_0})}
    \eval{
        \frac{ \dd{\lambda} }{2 \pi \rmi}
        \tr\left\{
            \Phi' (\lambda) \sigma^z \Phi^{-1} (\lambda)
        \right\}
        d_\beta (\lambda) \rme^{- \eta \lambda^2}
    }_{\eta = 0+}.
\end{equation}
We recall that the matrix $\Phi (z)$
does not have a jump across the contour $\Ccal_{\lambda_0}$,
but has poles in the regions $\Ucal_2$, $\Ucal_3$, $\Ucal_5$,
and $\Ucal_6$, see Figure~\ref{fig:contour-Gamma'}.

\paragraph{$\bullet$ Contribution of the poles
    and the solution of the singular integral equation}
Now we add and subtract to the expression~\eqref{eq:log-der-Fredholm-Phi}
the following integrals around the poles,
\begin{equation}
    \int\limits_{\partial \Ucal_v}
    \eval{
        \frac{ \dd{\lambda} }{2 \pi \rmi}
        \tr\left\{
            \Phi' (\lambda) \sigma^z \Phi^{-1} (\lambda)
        \right\}
        d_\beta (\lambda) \rme^{- \eta \lambda^2}
    }_{\eta = 0+}
    = \sum\limits_{\mu \in \Scal \cap \Ucal_v}
    \res\limits_{\lambda = \mu}
    \left(
        \tr\left\{
            \Phi' (\lambda) \sigma^z \Phi^{-1} (\lambda)
        \right\}
        d_\beta (\lambda)
    \strut\right)
\end{equation}
for $v \in \{ 2, 3, 5, 6 \}$,
where $\partial \Ucal_v$ is the positively oriented 
boundary of the region $\Ucal_v$, see Figure~\ref{fig:contours-U}.
Then we get
\begin{multline}
    \label{eq:log-der-Fredholm-Phi-modified}
    \partial_\beta \ln \det_{\Ccal_{\lambda_0}}
    \big( \id + \rmV \big)
    = a_\beta (x, \lambda_0)
    - \sum\limits_{\mu \in \Scal}
    \res\limits_{\lambda = \mu}
    \left(
        \tr\left\{
            \Phi' (\lambda) \sigma^z \Phi^{-1} (\lambda)
        \right\}
        d_\beta (\lambda)
    \strut\right)
    \\
    - \int\limits_{\Gamma' (\Ccal_{\lambda_0})}
    \eval{
        \frac{ \dd{\lambda} }{2 \pi \rmi}
        \tr\left\{
            \Phi' (\lambda) \sigma^z \Phi^{-1} (\lambda)
        \right\}
        d_\beta (\lambda) \rme^{- \eta \lambda^2}
    }_{\eta = 0+}
    \\
    + \sum\limits_{v \in \{ 2, 3, 5, 6 \}}
    \int\limits_{\partial \Ucal_v}
    \eval{
        \frac{ \dd{\lambda} }{2 \pi \rmi}
        \tr\left\{
            \Phi' (\lambda) \sigma^z \Phi^{-1} (\lambda)
        \right\}
        d_\beta (\lambda) \rme^{- \eta \lambda^2}
    }_{\eta = 0+}
\end{multline}
\begin{figure}[t]
    \centering
    \inputfig{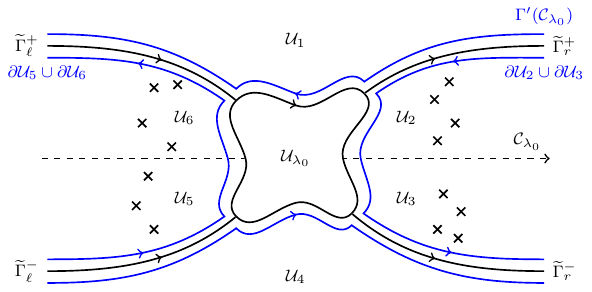}
    \caption{The contours $\partial \Ucal_v$ are the 
        positively oriented contours
        along the boundaries of the corresponding region $\Ucal_v$,
        $v = 2, 3, 5, 6$.}
    \label{fig:contours-U}
\end{figure}

Next we combine the integrals together.
For example, the combination of the integrals
along the contour $\widetilde{\Gamma}_r^+$,
see Figure~\ref{fig:contours-U}, is given by
\begin{equation}
    \int\limits_{\widetilde{\Gamma}_r^+}
    \eval{
        \frac{ \dd{\lambda} }{2 \pi \rmi}
        \left[
            \tr\left\{
                \Phi_-' (\lambda) \sigma^z \Phi_-^{-1} (\lambda)
            \right\}
            - \tr\left\{
                \Phi_+' (\lambda) \sigma^z \Phi_+^{-1} (\lambda)
            \right\}
        \strut\right]
        d_\beta (\lambda) \rme^{- \eta \lambda^2}
    }_{\eta = 0+}.
\end{equation}
Using the jump condition for the matrix $\Phi$
on the jump contour $\widetilde{\Gamma}_r^+$,
see equations~\eqref{eq:jump-matrix-Phi}
and~\eqref{eq:matrices-M-right},
we get
\begin{multline}
    \tr\left\{
        \Phi_-' (\lambda) \sigma^z \Phi_-^{-1} (\lambda)
    \right\}
    - \tr\left\{
        \Phi_+' (\lambda) \sigma^z \Phi_+^{-1} (\lambda)
    \right\}
    \\
    = \tr\left\{
        \Phi_+' (\lambda)
        \left( \mathrm{I}_2 + e^2 (\lambda) Q_r^+ (\lambda) \sigma^- \right)
        \sigma^z
        \left( \mathrm{I}_2 - e^2 (\lambda) Q_r^+ (\lambda) \sigma^- \right)
        \Phi_+^{-1} (\lambda)
    \right\}
    \\
    + \left(e^2 (\lambda) Q_r^+ (\lambda)\right)'
    \tr\left\{
        \sigma^-
        \sigma^z
        \left( \mathrm{I}_2 - e^2 (\lambda) Q_r^+ (\lambda) \sigma^- \right)
    \right\}
    - \tr\left\{
        \Phi_+' (\lambda) \sigma^z \Phi_+^{-1} (\lambda)
    \right\}
    \\
    = 2 e^2 (\lambda) Q_r^+ (\lambda)
    \tr\left\{
        \Phi_+' (\lambda) \sigma^- \Phi_+^{-1} (\lambda)
    \right\}.
\end{multline}
Here the first trace in the third line is zero,
and in the second equation
we used the commutation relations of the Pauli matrices.

Performing similar calculations for $\lambda \in \widetilde{\Gamma}_\ell^\pm$
and $\widetilde{\Gamma}_r^-$,
we obtain
\begin{multline}
\label{eq:difference-of-traces-with-Phi}
    \tr\left\{
        \Phi_-' (\lambda) \sigma^z \Phi_-^{-1} (\lambda)
    \right\}
    - \tr\left\{
        \Phi_+' (\lambda) \sigma^z \Phi_+^{-1} (\lambda)
    \right\}
    \\
    =
    \begin{cases}
        - 2 e^{-2} (\lambda) Q_\ell^+ (\lambda)
        \tr\left\{
            \Phi_+' (\lambda) \sigma^+ \Phi_+^{-1} (\lambda)
        \right\},
        & \quad
        \lambda \in \widetilde{\Gamma}_\ell^+,\\
        2 e^2 (\lambda) Q_r^+ (\lambda)
        \tr\left\{
            \Phi_+' (\lambda) \sigma^- \Phi_+^{-1} (\lambda)
        \right\},
        & \quad
        \lambda \in \widetilde{\Gamma}_r^+,\\
        2 e^2 (\lambda) Q_\ell^- (\lambda)
        \tr\left\{
            \Phi_+' (\lambda) \sigma^- \Phi_+^{-1} (\lambda)
        \right\},
        & \quad
        \lambda \in \widetilde{\Gamma}_\ell^-,\\
        - 2 e^{-2} (\lambda) Q_r^- (\lambda)
        \tr\left\{
            \Phi_+' (\lambda) \sigma^+ \Phi_+^{-1} (\lambda)
        \right\},
        & \quad
        \lambda \in \widetilde{\Gamma}_r^-.\\
    \end{cases}
\end{multline}
Here the expressions on the right-hand side
are given in terms of the boundary value $\Phi_+$
which can be analytically continued
to the region to the positive side
of the corresponding contour $\widetilde{\Gamma}$.
Therefore,
we can deform the integration contours $\widetilde{\Gamma}_\ell^\pm$
and $\widetilde{\Gamma}_r^\pm$
into contours $\gamma_\ell^\pm$ and $\gamma_r^\pm$,
see Figure~\ref{fig:contour-gamma-and-gamma-0},
as long as they do not cross the poles of $Q_\ell^\pm$ and $Q_r^\pm$.
\begin{figure}[t]
    \centering
    \resizebox{\textwidth}{!}{%
        \inputfig{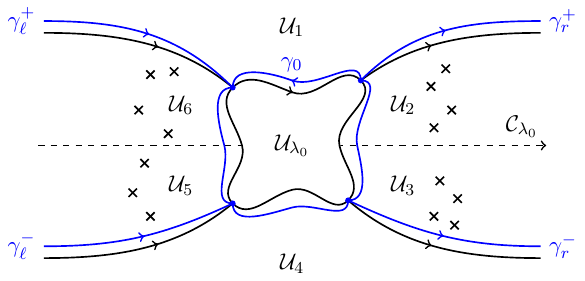}%
        \inputfig{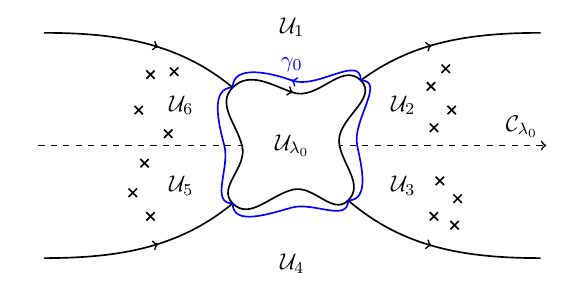}%
    }
    \caption{Contours $\gamma$ (on the left) and $\gamma_0$ (on the right).}
    \label{fig:contour-gamma-and-gamma-0}
\end{figure}

Then the integrals in~\eqref{eq:log-der-Fredholm-Phi-modified}
combine into
\begin{multline}
    \label{eq:contour-integral-combination}
    \int\limits_{\Gamma' (\Ccal_{\lambda_0})}
    \eval{
        \frac{ \dd{\lambda} }{2 \pi \rmi}
        \tr\left\{
            \Phi' (\lambda) \sigma^z \Phi^{-1} (\lambda)
        \right\}
        d_\beta (\lambda) \rme^{- \eta \lambda^2}
    }_{\eta = 0+}
    \\
    - \sum\limits_{v \in \left\{ 2, 3, 5, 6 \right\}}
    \int\limits_{\partial \Ucal_v}
    \eval{
        \frac{ \dd{\lambda} }{2 \pi \rmi}
        \tr\left\{
            \Phi' (\lambda) \sigma^z \Phi^{-1} (\lambda)
        \right\}
        d_\beta (\lambda)
        \rme^{- \eta \lambda^2}
    }_{\eta = 0+}
    \\
    =
    \int\limits_\gamma
    \eval{
        \frac{ \dd{\lambda} }{ 2 \pi \rmi }
        G_\beta (\lambda) \tr\left\{
            \Phi' (\lambda) \sigma(\lambda) \Phi^{-1} (\lambda)
        \right\}
         \rme^{- \eta \lambda^2}
    }_{\eta = 0+},
\end{multline}
where the function $G_\beta (\lambda)$ is given by
\begin{multline}
    G_\beta (\lambda)
    = d_\beta (\lambda)
    \Big[
        \indicator{\gamma_0} (\lambda)
        - 2 e^{-2} (\lambda) Q_\ell^+ (\lambda)
        \indicator{\gamma_\ell^+} (\lambda)
        + 2 e^2 (\lambda) Q_r^+ (\lambda)
        \indicator{\gamma_r^+} (\lambda)
        \\
        + 2 e^2 (\lambda) Q_\ell^- (\lambda)
        \indicator{\gamma_\ell^-} (\lambda)
        - 2 e^{-2} (\lambda) Q_r^- (\lambda)
        \indicator{\gamma_r^-} (\lambda)
    \Big],
\end{multline}
and the matrix $\sigma (\lambda)$ is defined as
\begin{equation}
    \sigma (\lambda)
    =
    \sigma^z \ \indicator{\gamma_0} (\lambda)
    + \sigma^+ \ \indicator{\gamma_\ell^+} (\lambda)
    + \sigma^- \ \indicator{\gamma_r^+} (\lambda)
    + \sigma^- \ \indicator{\gamma_\ell^-} (\lambda)
    + \sigma^+ \ \indicator{\gamma_r^-} (\lambda),
\end{equation}
see equation~\eqref{eq:difference-of-traces-with-Phi}.
The contour
$\gamma = \gamma_0
\cup \gamma_\ell^+ \cup \gamma_r^+
\cup \gamma_\ell^- \cup \gamma_r^-$
is shown in Figure~\ref{fig:contour-gamma-and-gamma-0}.

The integrals along the contours $\gamma_\ell^\pm$
and $\gamma_r^\pm$
are of order $\Ocal \big(\rme^{ - m_\ell^\pm x} \big)$
and~$\Ocal \big(\rme^{ - m_r^\pm x} \big)$, respectively,
due to the factors $e^{\pm 2} (\lambda)$,
where
\begin{equation}
\begin{aligned}
    m_\ell^+
    & = \inf\limits_{\lambda \in \gamma_\ell^+} \Imag( u (\lambda) ),
    \qquad
    & m_r^+
    & = \inf\limits_{\lambda \in \gamma_r^+} \Imag( - u (\lambda) ),
    \\
    m_\ell^-
    & = \inf\limits_{\lambda \in \gamma_\ell^-} \Imag( - u (\lambda) ),
    \qquad
    & m_r^-
    & = \inf\limits_{\lambda \in \gamma_r^-} \Imag( u (\lambda) ).
\end{aligned}
\end{equation}
They can be combined into a remainder of the form 
$\Ocal (\rme^{- m x})$ with $m$ given by
\begin{equation}
    m = \min\{
        m_\ell^+, m_r^+, m_\ell^-, m_r^-
    \}.
\end{equation}
Hence, the leading contributions to the large-$x$ asymptotics
will come from the integral over $\gamma_0$,
see Figure~\ref{fig:contour-gamma-and-gamma-0}.
Moreover, we do not need the regularization parameter $\eta$ any longer,
since the integration contour $\gamma_0$ is compact.

Finally, we have the following expression
for the logarithmic derivative of the Fredholm determinant
\begin{multline}
    \label{eq:log-der-Fredholm-prefinal}
    \partial_\beta \ln \det_{\Ccal_{\lambda_0}}
    \big( \id + \rmV \big)
    = a_\beta (x, \lambda_0)
    - \int\limits_{\gamma_0}
    \frac{ \dd{\lambda} }{ 2 \pi \rmi }
    \tr\left\{
        \Phi' (\lambda) \sigma^z \Phi^{-1} (\lambda)
    \right\}
    d_\beta (\lambda)
    \\
    - \sum\limits_{\mu \in \Scal}
    \res\limits_{\lambda = \mu}
    \left(
        \tr\left\{
            \Phi' (\lambda) \sigma^z \Phi^{-1} (\lambda)
        \right\}
        d_\beta (\lambda)
    \strut\right)
    + \Ocal(\rme^{- m x}).
\end{multline}
Substituting $\Phi (\lambda) = \rmS (\lambda) \Pi (\lambda)$,
we obtain
\begin{equation}
    \tr\{ \Phi' (\lambda) \sigma^z \Phi^{-1} (\lambda)\}
    = \tr\{
        \rmS' (\lambda) \Pi (\lambda)
        \sigma^z
        \Pi^{-1} (\lambda) \rmS^{-1} (\lambda)
    \}
    + \tr\{
        \Pi' (\lambda) \sigma^z \Pi^{-1} (\lambda)
    \},
\end{equation}
from which we derive the expression~\eqref{eq:prop:log-der-Fredholm-for-AE}
in Proposition~\ref{prop:log-der-Fredholm-for-AE}.

\subsection{Analysis of the pole contributions}
\label{sec:analysis-of-pole-contribution}
Next, we derive an explicit expression
for the direct contribution of the poles
to the Fredholm determinant asymptotics,
see the last term on the right-hand side
of equation~\eqref{eq:prop:log-der-Fredholm-for-AE}
in Proposition~\ref{prop:log-der-Fredholm-for-AE}
or equation~\eqref{eq:log-der-Fredholm-prefinal}.
The result is shown
in Proposition~\ref{prop:contribution-of-poles}.

First, we derive an expression for
\begin{equation}
\label{eq:trace-with-S}
    \tr\{
        \rmS^{-1} (\lambda) \rmS' (\lambda)
        \Pi (\lambda) \sigma^z \Pi^{-1} (\lambda)
    \},
\end{equation}
Since $\det \rmS (\lambda) = 1$,
we have $\rmS^{-1} (\lambda) = \sigma^y \rmS^\intercal (\lambda) \sigma^y$,
and it follows from the expression~\eqref{eq:matrix-S-new} for $\rmS (\lambda)$
and from the fact that $\det \Pi (\lambda) = 1$ that
\begin{equation}
\label{eq:matrix-S-new-inverse}
    \rmS^{-1} (\lambda)
    = \mathrm{I}_2
    + \sum\limits_{j = 1}^{n^+}
    \frac{ \sigma_j^+ }{\lambda - s_j^+}
    \Pi (s_j^+)
    \sigma^y
    (\mathbf{0}, \mathbf{y}_j)^\intercal
    \sigma^y
    + \sum\limits_{j = 1}^{n^-}
    \frac{ \sigma_j^- }{\lambda - s_j^-}
    \Pi (s_j^-)
    \sigma^y
    (\mathbf{x}_j, \mathbf{0})^\intercal
    \sigma^y.
\end{equation}
Then the matrix $\rmS^{-1} (\lambda) \rmS' (\lambda)$
can be expressed as
\begin{multline}
\label{eq:S-inv-S-prime}
    \rmS^{-1} (\lambda)
    \rmS' (\lambda)
    = - \sum\limits_{j = 1}^{n^+}
    \frac{ \sigma_j^+ }{ (\lambda - s_j^+)^2 }
    \mathbf{y}_j (0, 1) \Pi^{-1} (s_j^+)
    - \sum\limits_{j = 1}^{n^-}
    \frac{ \sigma_j^- }{ (\lambda - s_j^-)^2 }
    \mathbf{x}_j (1, 0) \Pi^{-1} (s_j^-)
    \\
    -
    \Bigg[
        \sum\limits_{j = 1}^{n^+}
        \frac{ \sigma_j^+ }{\lambda - s_j^+}
        \Pi (s_j^+)
        \sigma^y
        (0, 1)^\intercal \mathbf{y}_j^\intercal
        \sigma^y
        + \sum\limits_{j = 1}^{n^-}
        \frac{ \sigma_j^- }{\lambda - s_j^-}
        \Pi (s_j^-)
        \sigma^y
        (1, 0)^\intercal \mathbf{x}_j^\intercal
        \sigma^y
    \Bigg]
    \\
    \times
    \Bigg[
        \sum\limits_{j = 1}^{n^+}
        \frac{ \sigma_j^+ }{ (\lambda - s_j^+)^2 }
        \mathbf{y}_j (0, 1) \Pi^{-1} (s_j^+)
        + \sum\limits_{j = 1}^{n^-}
        \frac{ \sigma_j^- }{ (\lambda - s_j^-)^2 }
        \mathbf{x}_j (1, 0) \Pi^{-1} (s_j^-)
    \Bigg],
\end{multline}
where we also decomposed the matrices as
\begin{equation}
    (\mathbf{0}, \mathbf{y}_j)
    = \mathbf{y}_j \cdot (0, 1),
    \qquad
    (\mathbf{x}_j, 0)
    = \mathbf{x}_j \cdot (1, 0).
\end{equation}

Let us analyse the pole structure
of $\rmS^{-1} (\lambda) \rmS' (\lambda)$.
The coefficient in front of $(\lambda - s_j^+)^{-3}$ is
\begin{equation}
    - (\sigma_j^+)^2
    \Pi (s_j^+)
    \sigma^y
    (0, 1)^\intercal
    \mathbf{y}_j^\intercal
    \sigma^y
    \mathbf{y}_j
    (0, 1)
    \Pi^{-1} (s_j^+)
    = 0,
\end{equation}
since $\mathbf{y}_j^\intercal \sigma^y \mathbf{y}_j = 0$, and
a similar argument holds for the other potential third order pole.
Hence, the poles in the expression~\eqref{eq:S-inv-S-prime}
are at most of second order.
Accordingly, there are two types of contributions
to the sum of the residues at the poles
in the expression for the logarithmic derivative
of the Fredholm determinant~\eqref{eq:prop:log-der-Fredholm-for-AE},
stemming from first and second order poles, respectively.

Before we derive an expression for the sum over the poles,
we shall make a useful statement concerning the contributions
from the second order poles.
\begin{proposition}
\label{prop:residue-form}
    Let $\rmS (\lambda)$ be given by~\eqref{eq:matrix-S-new},
    where the vectors $\mathbf{y}_j$, $j = 1, \dots, n^+$,
    and $\mathbf{x}_j$, $j = 1, \dots, n^-$
    are the solutions of the linear system~\eqref{eq:SLE-new}.
    Then, for $\mu \in \Scal$, we have the identity
    \begin{multline}
    \label{eq:prop:residue-form}
        \tr\Big\{
            \lim\limits_{\lambda \to \mu}
            \left(
                \rmS^{-1} (\lambda) \rmS' (\lambda) (\lambda - \mu)^2
            \right)
            \eval{
                \partial_{\lambda}
                \left(
                    \Pi (\lambda) \sigma^z \Pi^{-1} (\lambda)
                    d_\beta (\lambda)
                 \right)
            }_{\lambda = \mu}
        \Big\}
        \\
        = 
        \tr\Big\{
            \lim\limits_{\lambda \to \mu}
            \left(
                \rmS^{-1} (\lambda) \rmS' (\lambda) (\lambda - \mu)^2
            \right)
            \eval{
                \partial_{\lambda}
                \left(
                    \Pi (\lambda) \sigma^z \Pi^{-1} (\lambda)
                 \right)
            }_{\lambda = \mu}
        \Big\}
        d_\beta (\mu).
    \end{multline}
\end{proposition}
\begin{proof}
    We prove this identity only for
    $\mu = s_j^+$, $j = 1, \dots, n^+$,
    since the proof for $\mu = s_j^-$, $j = 1, \dots, n^-$
    is similar. In particular,
    we show that the trace in front of the term with $d_\beta' (\mu)$
    in~\eqref{eq:prop:residue-form} is zero.

    Consider the coefficient
    in front of the second order pole at $s_j^+$
    in the expression~\eqref{eq:S-inv-S-prime},
    \begin{multline}
        - \sigma_j^+
        \mathbf{y}_j (0, 1) \Pi^{-1} (s_j^+)
        - \sigma_j^+
        \sum\limits_{\substack{k = 1\\ k \neq j}}^{n^+}
        \frac{ \sigma_k^+ }{s_j^+ - s_k^+}
        \Pi (s_k^+)
        \sigma^y
        (0, 1)^\intercal \mathbf{y}_k^\intercal
        \sigma^y
        \mathbf{y}_j (0, 1) \Pi^{-1} (s_j^+)
        \\
        - \sigma_j^+
        \sum\limits_{k = 1}^{n^-}
        \frac{ \sigma_k^- }{s_j^+ - s_k^-}
        \Pi (s_k^-)
        \sigma^y
        (1, 0)^\intercal \mathbf{x}_k^\intercal
        \sigma^y
        \mathbf{y}_j (0, 1) \Pi^{-1} (s_j^+).
    \end{multline}
    Multiplying this expression by
    $\Pi (s_j^+) \sigma^z \Pi^{-1} (s_j^+)$
    from the right and taking the trace, we get
    \begin{multline}
    \label{eq:prop-residue-form-expression}
        \tr\Big\{
            \Big(
                \lim\limits_{\lambda \to \mu}
                \rmS^{-1} (\lambda) \rmS' (\lambda)
                (\lambda - \mu)^2
            \Big)
            \Pi (\mu) \sigma^z \Pi^{-1} (\mu)
        \Big\}
        =
        - \sigma_j^+
        \tr\Big\{
            \mathbf{y}_j (0, 1) \sigma^z \Pi^{-1} (s_j^+)
        \Big\}
        \\
        - \sigma_j^+
        \sum\limits_{\substack{k = 1\\ k \neq j}}^{n^+}
        \frac{ \sigma_k^+ }{s_j^+ - s_k^+}
        \tr\Big\{
            \Pi (s_k^+)
            \sigma^y
            (0, 1)^\intercal \mathbf{y}_k^\intercal
            \sigma^y
            \mathbf{y}_j (0, 1) \sigma^z \Pi^{-1} (s_j^+)
        \Big\}
        \\
        - \sigma_j^+
        \sum\limits_{k = 1}^{n^-}
        \frac{ \sigma_k^- }{s_j^+ - s_k^-}
        \tr\Big\{
            \Pi (s_k^-)
            \sigma^y
            (1, 0)^\intercal \mathbf{x}_k^\intercal
            \sigma^y
            \mathbf{y}_j (0, 1) \sigma^z \Pi^{-1} (s_j^+)
        \Big\}.
    \end{multline}
    Using the cyclicity of the trace we obtain
    \begin{equation}
    \label{eq:cyclicity-of-trace}
    \begin{aligned}
        \tr\Big\{
            \Pi (s_k^+)
            \sigma^y
            (0, 1)^\intercal \mathbf{y}_k^\intercal
            \sigma^y
            \mathbf{y}_j (0, 1) \sigma^z \Pi^{-1} (s_j^+)
        \Big\}
        & = \rmi \big( P_{jk}^{++} \big)_{21}
        \mathbf{y}_k^\intercal
        \sigma^y
        \mathbf{y}_j,
        \\
        \tr\Big\{
            \Pi (s_k^-)
            \sigma^y
            (1, 0)^\intercal \mathbf{x}_k^\intercal
            \sigma^y
            \mathbf{y}_j (0, 1) \sigma^z \Pi^{-1} (s_j^+)
        \Big\}
        & = - \rmi \big( P_{jk}^{+-} \big)_{22}
        \mathbf{x}_k^\intercal
        \sigma^y
        \mathbf{y}_j,
    \end{aligned}
    \end{equation}
    and therefore expression~\eqref{eq:prop-residue-form-expression} reads
    \begin{multline}
        \tr\Big\{
            \Big(
                \lim\limits_{\lambda \to \mu}
                \rmS^{-1} (\lambda) \rmS' (\lambda)
                (\lambda - \mu)^2
            \Big)
            \Pi (\mu) \sigma^z \Pi^{-1} (\mu)
        \Big\}
        =
        \sigma_j^+
        (0, 1) \Pi^{-1} (s_j^+) \mathbf{y}_j
        \\
        - \rmi \sigma_j^+
        \sum\limits_{\substack{k = 1\\ k \neq j}}^{n^+}
        \frac{
            \sigma_k^+
            \big( P_{jk}^{++} \big)_{21} }{
            s_j^+ - s_k^+ }
        \mathbf{y}_k^\intercal
        \sigma^y
        \mathbf{y}_j
        + \rmi \sigma_j^+
        \sum\limits_{k = 1}^{n^-}
        \frac{
            \sigma_k^-
            \big( P_{jk}^{+-} \big)_{22} }{
            s_j^+ - s_k^- }
        \mathbf{x}_k^\intercal
        \sigma^y
        \mathbf{y}_j.
    \end{multline}
    Now, the coefficients $P$,
    cf.\ equation~\eqref{eq:matrix-elements-P-and-Q},
    exhibit the symmetry $P_{jk}^{\epsilon \epsilon'} = 
    \sigma^y (P_{kj}^{\epsilon' \epsilon})^\intercal \sigma^y$,
    and the two sums on the right-hand side are exactly the same
    as those in the first set of linear equations~\eqref{eq:SLE-new}.
    Moreover, the first term is nothing but
    \begin{equation}
    \label{eq:trace-with-W}
        (0, 1) \Pi^{-1} (s_j^+) \mathbf{y}_j
        = (- \Pi_{21} (s_j^+), \Pi_{11} (s_j^+))
        \mathbf{y}_j
        = \rmi \mathbf{w}_j^\intercal
        \sigma^y
        \mathbf{y}_j.
    \end{equation}
    Hence,
    \begin{multline}
        \tr\Big\{
            \Big(
                \lim\limits_{\lambda \to \mu}
                \rmS^{-1} (\lambda) \rmS' (\lambda)
                (\lambda - \mu)^2
            \Big)
            \Pi (\mu) \sigma^z \Pi^{-1} (\mu)
        \Big\}
        \\
        = \rmi \sigma_j^+
        \Bigg[
            \mathbf{w}_j^\intercal
            + \sum\limits_{\substack{k = 1\\ k \neq j}}^{n^+}
            \frac{
                \sigma_k^+
                \big( P_{kj}^{++} \big)_{21} }{
                s_j^+ - s_k^+ }
            \mathbf{y}_k^\intercal
            + \sum\limits_{k = 1}^{n^-}
            \frac{
                \sigma_k^-
                \big( P_{kj}^{-+} \big)_{11} }{
                s_j^+ - s_k^- }
            \mathbf{x}_k^\intercal
        \Bigg]
        \sigma^y
        \mathbf{y}_j
        = 
        \rmi \sigma_j^+
        \mathbf{y}_j^\intercal
        \sigma^y
        \mathbf{y}_j
        = 0,
    \end{multline}
    which completes the proof.
\end{proof}

It is now straightforward to derive the direct contribution of
the poles to the logarithmic derivative of the Fredholm determinant
in terms of the vectors $\mathbf{y}_j$ and $\mathbf{x}_j$,
which are the solutions of the linear system~\eqref{eq:SLE-new}.
\begin{proposition}
    \label{prop:contribution-of-poles}
    Let $\rmS (\lambda)$ be given by~\eqref{eq:matrix-S-new},
    where the vectors $\mathbf{y}_j$, $j = 1, \dots, n^+$,
    and $\mathbf{x}_j$, $j = 1, \dots, n^-$
    are the solutions of the linear system~\eqref{eq:SLE-new}.
    Then
    \begin{multline}
    \label{eq:prop:contribution-of-poles}
        \sum\limits_{\mu \in \mathcal{S}}
        \res\limits_{\lambda = \mu}
        \Big(
            \tr\{
                \rmS' (\lambda) \Pi (\lambda)
                \sigma^z
                \Pi^{-1} (\lambda) \rmS^{-1} (\lambda)
            \}
            d_\beta (\lambda)
        \Big)
        \\
        =
        2 \rmi \sum\limits_{j = 1}^{n^+}
        \sigma_j^+
        d_\beta (s_j^+)
        \left(\mathbf{w}'_j\right)^\intercal
        \sigma^y
        \mathbf{y}_j
        + 2 \rmi \sum\limits_{j = 1}^{n^-}
        \sigma_j^-
        d_\beta (s_j^-)
        \left(\mathbf{v}'_j\right)^\intercal
        \sigma^y
        \mathbf{x}_j
        \\
        - 2 \rmi
        \sum\limits_{\substack{j, k = 1\\ j \neq k}}^{n^+}
        \frac{
            \sigma_j^+ \sigma_k^+
            \big( Q_{kj}^{++} \big)_{21}
            d_\beta (s_j^+)
        }{
            s_j^+ - s_k^+ }
        \mathbf{y}_j^\intercal
        \sigma^y
        \mathbf{y}_k
        - 2 \rmi
        \sum\limits_{\substack{j, k = 1\\ j \neq k}}^{n^-}
        \frac{
            \sigma_j^- \sigma_k^-
            \big( Q_{kj}^{--} \big)_{12}
            d_\beta (s_j^-)
        }{
            s_j^- - s_k^- }
        \mathbf{x}_j^\intercal
        \sigma^y
        \mathbf{x}_k
        \\
        - 2 \rmi
        \sum\limits_{j = 1}^{n^+}
        \sum\limits_{k = 1}^{n^-}
        \frac{
            \sigma_j^+ \sigma_k^- }{
            s_j^+ - s_k^- }
        \Big[
            \big( Q_{kj}^{-+} \big)_{11}
            d_\beta (s_j^+)
            + \big( Q_{jk}^{+-} \big)_{22}
            d_\beta (s_k^-)
        \Big]
        \mathbf{y}_j^\intercal
        \sigma^y
        \mathbf{x}_k
        \\
        + 2 \rmi \sum\limits_{\substack{j, k = 1\\ j \neq k}}^{n^+}
        \frac{
            \sigma_j^+ \sigma_k^+
            \big( P_{kj}^{++} \big)_{21}
            d_\beta (s_j^+) }{
            (s_j^+ - s_k^+)^2 }
        \mathbf{y}_j^\intercal \sigma^y \mathbf{y}_k
        + 2 \rmi
        \sum\limits_{\substack{j, k = 1\\j \neq k}}^{n^-}
        \frac{
            \sigma_j^- \sigma_k^-
            \big( P_{kj}^{--} \big)_{12}
            d_\beta (s_j^-) }{
            (s_j^- - s_k^-)^2 }
        \mathbf{x}_j^\intercal \sigma^y \mathbf{x}_k
        \\
        + 2 \rmi
        \sum\limits_{j = 1}^{n^+}
        \sum\limits_{k = 1}^{n^-}
        \frac{
            \sigma_j^+ \sigma_k^-
            \big( P_{jk}^{+-} \big)_{22}
        }{
            (s_j^+ - s_k^-)^2 }
        \big[
            d_\beta (s_j^+) - d_\beta (s_k^-)
        \big]
        \mathbf{y}_j^\intercal
        \sigma^y
        \mathbf{x}_k.
    \end{multline}
    Here the matrix $\rmS(\lambda)$ is given by~\eqref{eq:matrix-S-new}.
\end{proposition}
The proof of the proposition
is presented in Appendix~\ref{app:contribution-of-poles}.

The expression~\eqref{eq:prop:contribution-of-poles}
for the contribution of the poles can be simplified in
the so-called static case $\lambda_0 \to + \infty$,
when there is no saddle point,
meaning that $\Pi (\lambda) \simeq \mathrm{I}_2$.
We consider this case in detail in Section~\ref{sec:static-case}.

\section{Asymptotic expansion of the Fredholm determinant
    in the case of no poles on the real axis}
\label{sec:asymptotics-no-poles}

We now turn to the actual calculation of the large-$x$ asymptotics
of our Fredholm determinant starting from
equation~\eqref{eq:prop:log-der-Fredholm-for-AE}. We first consider
the important special case when there are no poles of the function
$\Lcal'(\lambda|\lambda_0)$ on the real axis, i.e.\ when
$\mathcal{S} \cap \mathbb{R} = 0$.
Then it follows from the system of linear equations,
see~\eqref{eq:matrix-S-new} and~\eqref{eq:SLE-new},
that
\begin{equation} \label{eq:S_when_no_poles}
    \rmS (\lambda) =
    \mathrm{I}_2
    + \Ocal\big(
        x^{-\infty}
        \big(
            \begin{smallmatrix}
                1 & 1
                \\
                1 & 1
            \end{smallmatrix}
        \big)
    \big),
\end{equation}
because all the coefficients $\sigma$ are exponentially small,
due to the factors $e^{\pm 2} (\lambda)$ for $\lambda$ evaluated
at the poles in the corresponding coefficients $h$,
see equations~\eqref{eq:coefficients-sigma} and~\eqref{eq:residues-h}.

Inserting \eqref{eq:S_when_no_poles} into the right-hand side
of equation \eqref{eq:prop:log-der-Fredholm-for-AE} we remain with
the following simpler expression for the logarithmic derivatives
of the Fredholm determinant,
\begin{equation}
\label{eq:prop:log-der-Fredholm-for-AE_when_no_poles}
    \partial_\beta \ln \det_{\Ccal_{\lambda_0}}
    \big( \id + \mathrm{V} \big)
    = a_\beta (x, \lambda_0)
    - \int\limits_{\gamma_0}
    \frac{\dd{\lambda}}{ 2 \pi \rmi }
    \tr\{ \Pi' (\lambda) \sigma^z \Pi^{-1} (\lambda)\}
    \partial_\beta d (\lambda)
    + \Ocal\big( x^{-\infty} \big).
\end{equation}
We shall integrate this equation with respect to $\beta$ for
$\beta = x, \lambda_0$, while using the large-$x$ asymptotic
expansion for $\Pi$ that is obtained 
from~\eqref{eq:singular-integral-equation}. This will
result in a proof of Theorem~\ref{thm:Fredholm-det-AE-no-poles}.

\subsection{Asymptotic expansion of the parametrix}
\label{sec:AE_parametrix}
The $x$-dependence enters the singular integral following from
equation~\eqref{eq:singular-integral-equation} through the inverse
para\-metrix $\Pcal^{-1}$. Using the asymptotic expansion of
the parabolic cylinder functions for large arguments, see
Appendix~\ref{app:parabolic-cylinder-functions}
equations~\eqref{eq:AE-parabolic-cylinder-function}
and~\eqref{eq:parabolic-cylinder-function-relation},
we obtain the asymptotic series
\begin{equation}
\label{eq:AE-parametrix}
    \Pcal^{-1} (\lambda)
    = \mathrm{I}_2
    + \sum\limits_{k = 1}^\infty
    \frac{ \Pcal_k (\lambda) }{
        x^{k / 2}
        (\lambda - \lambda_0)^{k} },
\end{equation}
where for even indices the coefficients $\Pcal_k (\lambda)$ are given by
\begin{subequations}
\label{eq:parametrix-coefficients}
\begin{equation}
    \Pcal_{2 n} (\lambda)
    =
    \frac{ (- \rmi)^n }{n! \ 2^{2 n}}
    \Bigg[
        \dfrac{
            \lambda - \lambda_0}{
            \omega (\lambda - \lambda_0 | \lambda_0) }
    \Bigg]^{2 n}
    \begin{pmatrix}
        \big( -\tau (\lambda) \big)_{2 n} & 0
        \\
        0 & (-1)^n \big( \tau (\lambda) \big)_{2 n}
    \end{pmatrix}
\end{equation}
and for odd indices by
\begin{multline}
    \Pcal_{2 n + 1} (\lambda)
    =
    - \frac{ (- \rmi)^n }{n! \ 2^{2 n}}
    \Bigg[
        \dfrac{
            \lambda - \lambda_0 }{
            \omega (\lambda - \lambda_0 | \lambda_0) }
    \Bigg]^{2 n + 1}
    \\
    \times
    \begin{pmatrix}
        0 & b_{12} (\lambda)
        \big(1 - \tau (\lambda) \big)_{2 n}
        \\
        (-1)^{n} b_{21} (\lambda)
        \big(1 + \tau (\lambda) \big)_{2 n} & 0
    \end{pmatrix}.
\end{multline}
\end{subequations}
Here $(a)_n$ denotes the Pochhammer symbol,
see equation~\eqref{eq:Pochhammer-symbol}.
We stress that $b_{12} (\lambda)$ and $b_{21} (\lambda)$ depend on $x$
as well, see~\eqref{eq:b12-and-b21} and~\eqref{eq:m-and-n}.

\subsection[Asymptotic expansion of \texorpdfstring{$\Pi$}{Pi}]
{\boldmath Asymptotic expansion of $\Pi$}
\label{sec:asymptotics-Pi}
The integral representation~\eqref{eq:singular-integral-equation} of
the function $\Pi (\lambda)$ in terms of is $+$ boundary value
can be used to obtain its large-$x$ asymptotic expansion.
We note that the jump matrix $\rmG_\Phi$ occurring under the integral
in~\eqref{eq:singular-integral-equation} is exponentially close to the
identity matrix on the contours $\Gamma_\ell^\pm$ and $\Gamma_r^\pm$.
Hence,
a restriction of the integration contour to $\partial \mathcal{U}_{\lambda_0}$,
where $\rmG_\Phi = \Pcal$, means to neglect only terms that are
exponentially small in $x$. Therefore
\begin{equation}
\label{eq:AE-singular-integral-equation}
    \Pi (\lambda)
    = \mathrm{I}_2
    + \int\limits_{\partial \mathcal{U}_{\lambda_0}}
    \frac{\dd \mu}{2 \pi \rmi}
    \frac{
        \Pi_+ (\mu)
        \left( \Pcal^{-1} (\mu) - \mathrm{I}_2\right) }{
        \mu - \lambda }
    + \Ocal \big(
        \rme^{ - m x }
        \big(
            \begin{smallmatrix}
                1 & 1
                \\
                1 & 1
            \end{smallmatrix}
        \big)
    \big).
\end{equation}
The boundary value $\Pi_+$ (from outside the negatively oriented
contour $\partial \mathcal{U}_{\lambda_0}$), in particular, satisfies
the singular linear integral equation
\begin{equation}
\label{eq:AE-singular-integral-equation-Pi-plus}
    \Pi_+ (\lambda)
    = \mathrm{I}_2
    + \int\limits_{\partial \mathcal{U}_{\lambda_0}}
    \frac{\dd \mu}{2 \pi \rmi}
    \frac{
        \Pi_+ (\mu)
        \left( \Pcal^{-1} (\mu) - \mathrm{I}_2\right) }{
        \mu - \lambda }
    + \Ocal \big(
        \rme^{ - m x }
        \big(
            \begin{smallmatrix}
                1 & 1
                \\
                1 & 1
            \end{smallmatrix}
        \big)
    \big).
\end{equation}

It is clear from the structure of
equation~\eqref{eq:AE-singular-integral-equation-Pi-plus} and the
large-$x$ asymptotic expansion~\eqref{eq:AE-parametrix},
\eqref{eq:parametrix-coefficients}
of $\Pcal^{-1}$ that $\Pi_+ - \mathrm{I}_2$ can be represented
as a series in negative powers of $x^{\flatfrac12}$.
Thus, for $\lambda \in
\Ext(\mathcal{U}_{\lambda_0})$, the function $\Pi$ will
be of the asymptotic form
\begin{equation}
\label{eq:Pi-AE-Neumann-series}
    \Pi (\lambda)
    = \sum\limits_{n = 0}^\infty \frac{\Pi_n (\lambda)}{x^{n / 2}}, \qquad
    \Pi_0 (\lambda) = \mathrm{I}_2.
    \qquad
\end{equation}
Inserting \eqref{eq:Pi-AE-Neumann-series} together with
\eqref{eq:AE-parametrix} into \eqref{eq:AE-singular-integral-equation}
and comparing the coefficients in front of equal powers of $x^{-\flatfrac12}$
we obtain a recursion relation for the $\Pi_n$,
\begin{equation} \label{eq:pi_n_first_recur}
    \Pi_n (\lambda) = 
       - \sum_{k=1}^n \int\limits_{\partial \mathcal{U}_{\lambda_0}}
         \frac{\dd \mu}{2 \pi \rmi}
         \frac{
            \Pi_{n-k} (\mu) \Pcal_k (\mu) }{
            (\lambda - \mu)(\mu - \lambda_0)^k }.
\end{equation}
This can be directly used to obtain the first few functions $\Pi_n$.
However, we get a better understanding of the structure of the
$\Pi_n$, if we first show inductively that, for $\lambda \in
\Ext(\mathcal{U}_{\lambda_0})$, they are of the form
\begin{equation}
\label{eq:Pi-coefficients-n}
    \Pi_n (\lambda) = \sum\limits_{j = 1}^n
    \frac{ \Pi_{n, j} (\lambda_0) }{ (\lambda - \lambda_0)^j }.
\end{equation}
Inserting this into \eqref{eq:pi_n_first_recur}
we obtain another recursion relation for the coefficients
$\Pi_{n, j} (\lambda_0)$ which involves no integration anymore,
\begin{equation}
\label{eq:Pi-coefficients-n-m}
    \Pi_{n, j} (\lambda_0) =
    - \frac{\Pcal_n^{(n - j)} (\lambda_0)}{ (n - j)! }
    - \sum\limits_{k = j}^{n}
    \sum\limits_{\ell = 1}^{k - 1}
    \frac{\Pi_{n - \ell, k - \ell} (\lambda_0)
    \Pcal_\ell^{(k - j)} (\lambda_0)}{(k - j)!},
    \qquad
    j = 1, \dots, n.
\end{equation}
Here we employ the usual convention that a sum is zero, if
the lower summation boundary exceeds the upper one. With this
convention \eqref{eq:Pi-coefficients-n-m} is valid for all
positive $n, j \in \mathbb{N}$. We note as well that the
coefficients $\Pi_n$ inherit the property of the coefficients
$\Pcal_n$ to be diagonal for even $n$ and off-diagonal for odd $n$.

In particular, we see from~\eqref{eq:Pi-coefficients-n-m}
that the first coefficients $\Pi_n (\lambda)$
in~\eqref{eq:Pi-AE-Neumann-series}
are given by
\begin{align}
\label{eq:Pi-coefficients-1-and-2}
    & \Pi_1 (\lambda)
    = - \frac{
        \Pcal_1 (\lambda_0) }{
        \lambda - \lambda_0 },
    \qquad
    \Pi_2 (\lambda)
    = \pdv{\mu} \left.\left(
        \frac{
            \Pcal_2 (\mu)
            - \Pcal_1 (\lambda_0) \Pcal_1 (\mu) }{
            \mu - \lambda }
    \right)\right|_{\mu = \lambda_0}, \\[2ex] \notag
\label{eq:Pi-coefficient-3}
    & \Pi_3 (\lambda)
    = \frac1{2!} \pdv[2]{\mu} \left. \left(
        \frac{
            \Pcal_3 (\mu)
            - \Pcal_1 (\lambda_0) \Pcal_2 (\mu)
            - \Pcal_2 (\lambda_0) \Pcal_1 (\mu)
            + \Pcal_1^2 (\lambda_0) \Pcal_1 (\mu) }{
            \mu - \lambda }
    \right) \right|_{\mu = \lambda_0}
    \\ & \mspace{209.mu}
    - \pdv{\mu} \left. \left(
        \frac{
            \Pcal_2' (\lambda_0) \Pcal_1 (\mu)
            - \Pcal_1 (\lambda_0) \Pcal_1' (\lambda_0)
            \Pcal_1 (\mu) }{
            \mu - \lambda }
    \right)\right|_{\mu = \lambda_0}.
\end{align}
Moreover, from~\eqref{eq:Pi-AE-Neumann-series}
and~\eqref{eq:Pi-coefficients-n-m}
it is straightforward to see that
\begin{equation}
\label{eq:order-of-Pi-3}
    \frac{\Pi_3 (\lambda)}{x^{3 / 2}} =
    \Ocal \left(
        \frac{
            (\ln x)^2
            \rme^{\rmi x u (\lambda_0)} }{
            x^{\flatfrac{3}{2} - \tau (\lambda_0)} } \cdot \sigma^+
    +
           \frac{
            (\ln x)^2
            \rme^{- \rmi x u (\lambda_0)} }{
            x^{\flatfrac{3}{2} + \tau (\lambda_0)} } \cdot \sigma^-
    \right)
\end{equation}
and
\begin{equation}
\label{eq:order-of-Pi-4}
    \frac{\Pi_4 (\lambda)}{x^2}
    = \Ocal\left( \frac{(\ln x)^2}{x^2}  \cdot \mathrm{I}_2 \right).
    \end{equation}
The key observation here is
that the leading order in $x$ propagates from the terms,
where the derivative with respect to $\lambda$ in $\Pcal_\ell^{(k - j)}$
acts on the functions $b_{12}$ and $b_{21}$
in the coefficients with odd indices $\ell$,
see the expression~\eqref{eq:parametrix-coefficients}.
Each such action of the derivative produces a factor $\ln x$,
however, the coefficient in front of $(\ln x)^3 / x^2$ in $\Pi_4$
happens to be zero.

Explicitly,
the first two coefficients in~\eqref{eq:Pi-AE-Neumann-series},
see equations~\eqref{eq:Pi-coefficients-1-and-2}
and~\eqref{eq:parametrix-coefficients},
are given by
\begin{equation}
\label{eq:coefficient-Pi-1-explicit}
    \Pi_1 (\lambda)
    = \frac{1}{\lambda - \lambda_0}
    \frac{1}{\omega' (0 | \lambda_0)}
    \begin{pmatrix}
        0 & b_{12} (\lambda_0)
        \\
        b_{21} (\lambda_0) & 0
    \end{pmatrix}
\end{equation}
and
\begin{multline}
\label{eq:coefficient-Pi-2-explicit}
    \Pi_2 (\lambda)
    = \frac{1}{(\lambda - \lambda_0)^2}
    \frac{\rmi \tau (\lambda_0)}{4 (\omega' (0 | \lambda_0))^2}
    \left[\mathrm{I}_2 + \tau (\lambda_0) \sigma^z\right]
    - \frac{1}{ (\lambda - \lambda_0) }
    \frac{ \rmi \tau (\lambda_0) }{4 (\omega' (0 | \lambda_0))^2}
    \Bigg[
        \frac{
            \tau (\lambda_0)
            \omega'' (0 | \lambda_0) }{
            \omega' (0 | \lambda_0) }
        \\
        - 2 \tau' (\lambda_0)
        + \frac{2}{\rmi \tau (\lambda_0)} \big[
            b_{12}' (\lambda_0) b_{21} (\lambda_0) 
            - b_{12} (\lambda_0) b_{21}' (\lambda_0) 
        \big]
    \Bigg]  \sigma^z.
\end{multline}
Here and in the following we denote
$\omega^{(n)} (0 | \lambda_0)
:= \eval{\partial^n_\lambda
\omega (\lambda - \lambda_0 | \lambda_0)}_{\lambda = \lambda_0}$.

\subsection[Asymptotic expansion of the integral involving \texorpdfstring{{$\Pi$}}{Pi}]
{\boldmath Asymptotic expansion of the integral involving $\Pi$}
With the above results we can now calculate the leading orders
in $x$ of the expression
$\tr\big\{ \Pi' (\lambda) \sigma^z \Pi^{-1} (\lambda) \big\}$
occurring under the integral on the right-hand side of
\eqref{eq:prop:log-der-Fredholm-for-AE_when_no_poles},
while keeping control of the corrections. First of all,
\begin{multline}
    \label{eq:tr-incomplete-AE-part-1}
    \tr\big\{ \Pi' (\lambda) \sigma^z \Pi^{-1} (\lambda) \big\}
    =
    \frac1{x^{\flatfrac12}}
    \tr\big\{ \Pi_1' (\lambda) \sigma^z \big\}
    + \frac1x
    \tr\big\{  
        \Pi_2' (\lambda) \sigma^z
        + \Pi_1' (\lambda) \sigma^z \Pi_1^{\rm ad} (\lambda)
    \big\}
    \\
    + \frac1{x^{\flatfrac32}}
    \tr\big\{
        \Pi_3' (\lambda) \sigma^z
        + \Pi_2' (\lambda) \sigma^z \Pi_1^{\rm ad} (\lambda)
        + \Pi_1' (\lambda) \sigma^z \Pi_2^{\rm ad} (\lambda)
    \big\}
    +
    \frac1{x^2}
    \tr\big\{
        \Pi_4' (\lambda) \sigma^z
        \\
        + \Pi_3' (\lambda) \sigma^z \Pi_1^{\rm ad} (\lambda)
        + \Pi_2' (\lambda) \sigma^z \Pi_2^{\rm ad} (\lambda)
        + \Pi_1' (\lambda) \sigma^z \Pi_3^{\rm ad} (\lambda)
    \big\}
    + \ocal \left( x^{-2} \right),
\end{multline}
where we have denoted the
adjugate of $\Pi_n$ by $\sigma^y \Pi_n^\intercal \sigma^y = \Pi^{\rm ad}_n$.
Using the fact that the coefficients $\Pi_n$ are diagonal for even $n$
and off-diagonal for odd $n$ and the observations~\eqref{eq:order-of-Pi-3},
\eqref{eq:order-of-Pi-4} this further reduces to
\begin{equation}
\label{eq:integrand-gamma-0}
    \tr\big\{ \Pi' (\lambda) \sigma^z \Pi^{-1} (\lambda) \big\}
    =
    \frac1x
    \tr\big\{
        \Pi_2' (\lambda) \sigma^z
        + \Pi_1' (\lambda) \sigma^z \Pi_1^{\rm ad} (\lambda)
    \big\}
    + \Ocal\left( \frac{ (\ln x)^2 }{x^2} \right).
\end{equation}
Moreover, from the diagonal/off-diagonal form of the coefficients $\Pi_n$,
see expressions~\eqref{eq:Pi-coefficients-n} and~\eqref{eq:Pi-coefficients-n-m},
and the form of the coefficients $\Pcal_n$,
see equations~\eqref{eq:parametrix-coefficients},
one can prove by induction that
\begin{equation}
    \tr\big\{ \Pi' (\lambda) \sigma^z \Pi^{-1} (\lambda) \big\}
    = \sum\limits_{n = 0}^\infty
    \frac{p_n (\ln x)}{x^{n}},
\end{equation}
where $p_n (\ln x)$ are some polynomials in $\ln x$ of
degree $\leq 2 n - 1$, and there are no factors that
are oscillating in $x$ in this asymptotic expansion.

For the calculation of the two trace terms on the right-hand side
we insert the explicit expressions~\eqref{eq:coefficient-Pi-1-explicit},
\eqref{eq:coefficient-Pi-2-explicit}. Then
\begin{equation} \label{eq:trace_first_order_coefficient}
    \tr\bigl\{\Pi_1' (\lambda) \sigma^z \Pi_1^{\rm ad} (\lambda)\bigr\} = 0
\end{equation}
and
\begin{multline}
\label{eq:tr-with-Pi-2-prime-explicit}
    \tr\bigl\{\Pi_2' (\lambda) \sigma^z\bigr\}
    = - \frac{1}{(\lambda - \lambda_0)^3}
    \frac{\rmi \tau^2 (\lambda_0)}{(\omega' (0 | \lambda_0))^2} \\
    + \frac{1}{(\lambda - \lambda_0)^2}
    \frac{\rmi \tau(\lambda_0)}{2 (\omega' (0 | \lambda_0))^2}
    \Biggl[
        \frac{
            \tau (\lambda_0)
            \omega'' (0 | \lambda_0) }{
            \omega' (0 | \lambda_0) }
        - 2 \tau' (\lambda_0)
        + 2 \frac{b_{12}' (\lambda_0)}{b_{12} (\lambda_0)}
        - \frac{\tau'(\lambda_0)}{\tau(\lambda_0)}
    \Biggr].
\end{multline}
In the derivation of \eqref{eq:tr-with-Pi-2-prime-explicit} we have
used the identity
\begin{equation}
    2 \big[
        b_{12}' (\lambda_0) b_{21} (\lambda_0)
        - b_{12} (\lambda_0) b_{21}' (\lambda_0)
    \big]
    = 2 \rmi \tau (\lambda_0)
    \frac{ b_{12}' (\lambda_0) }{ b_{12} (\lambda_0) }
    - \rmi \tau' (\lambda_0)
\end{equation}
that follows from~\eqref{eq:product-of-b12-and-b21}. At this point
we can make the $x$-dependence hidden in $b_{12} (\lambda_0)$ explicit.
Using equations~\eqref{eq:b12-and-b21} and~\eqref{eq:m-and-n} we
see that the last two terms in the square bracket on the right-hand
side of \eqref{eq:tr-with-Pi-2-prime-explicit} take the form
\begin{equation}
\label{eq:definition-B}
    2 \frac{ b_{12}' (\lambda_0) }{ b_{12} (\lambda_0) }
    - \frac{\tau'(\lambda_0)}{\tau(\lambda_0)}
    = 2 \tau' (\lambda_0) \ln x + B(\lambda_0),
\end{equation}
where we introduced a function $B (\lambda_0)$ that does not depend on $x$.
Its explicit form, which does not matter at this point, will be shown
in equation~\eqref{eq:explicit_B} below.

Substituting~\eqref{eq:trace_first_order_coefficient},
\eqref{eq:tr-with-Pi-2-prime-explicit}, and~\eqref{eq:definition-B}
into~\eqref{eq:integrand-gamma-0},
we can now evaluate the integral
over $\gamma_0$ in~\eqref{eq:prop:log-der-Fredholm-for-AE_when_no_poles},
\begin{equation}
\label{eq:integral-over-gamma-0-no-poles-final}
    \int\limits_{\gamma_0}
    \frac{\dd{\lambda}}{ 2 \pi \rmi }
    \tr\left\{ \Pi' (\lambda) \sigma^z \Pi^{-1} (\lambda) \right\}
    d_\beta (\lambda)
    =
    - \frac{a_\beta^{(2)}(x, \lambda_0)}{x}
        + \Ocal \left( \frac{ (\ln x)^2 }{x^2} \right),
\end{equation}
where
\begin{multline} \label{eq:def_a_upper_two}
    a_\beta^{(2)}(x, \lambda_0) = \\
    \frac{ \rmi \tau^2 (\lambda_0) }{ 2 (\omega' (0 | \lambda_0))^2 }
    \Biggl\{
        d_\beta'' (\lambda_0)
        - \left(
            \frac{\omega'' (0 | \lambda_0) }{\omega' (0 | \lambda_0) }
            + \frac{2 \tau' (\lambda_0)}{\tau(\lambda_0)} ( \ln x - 1)
            + \frac{B (\lambda_0)}{\tau(\lambda_0)}
        \right)
        d_\beta' (\lambda_0)
    \Biggr\}.
\end{multline}

\subsection{Fredholm determinant asymptotics}
\label{sec:asymptotics-of-Fredholm-determinant-no-poles}
Substituting~\eqref{eq:integral-over-gamma-0-no-poles-final}
into~\eqref{eq:prop:log-der-Fredholm-for-AE_when_no_poles},
we have obtained an explicit formula for the leading large-$x$
asymptotics of the logarithmic derivative of the Fredholm determinant
with respect to the parameter $\beta$:
\begin{equation}
    \label{eq:Fred-det-log-der-series-in-x-no-poles}
    \partial_\beta \ln \det_{\Ccal_{\lambda_0}}
    (\id + \rmV)
    = a_\beta (x, \lambda_0)
    + \frac{a_\beta^{(2)} (x,\lambda_0)}{x}
    + \Ocal \left( \frac{ (\ln x)^2 }{x^2} \right),
\end{equation}
where $a_\beta (x, \lambda_0)$ was defined in \eqref{eq:prop:function-a}
and $a_\beta^{(2)} (x, \lambda_0)$ in \eqref{eq:def_a_upper_two}.

Let us introduce the function
\begin{equation}
    \label{eq:def_fun_a}
        a (x, \lambda_0)
        = 2 \int\limits_{\Ccal_{\lambda_0}}
        \dd{\lambda}
        \Lcal (\lambda | \lambda_0) d' (\lambda).
    \end{equation}
Then
\begin{equation}
    \label{eq:relation-for-a-x}
    a_x (x, \lambda_0) = \partial_x a(x, \lambda_0).
\end{equation}
Taking into account that $d_x' (\lambda_0) = 0$
and $d_x'' (\lambda_0) = \rmi (\omega' (0 | \lambda_0))^2$,
the expression~\eqref{eq:def_a_upper_two}
for $a_\beta^{(2)} (x, \lambda_0)$ with
$\beta = x$ turns into
\begin{equation}
\label{eq:a-2-no-poles-log}
    a_x^{(2)} (x, \lambda_0) = - \frac{\tau^2 (\lambda_0)}{2}.
\end{equation}
Using now \eqref{eq:relation-for-a-x}, \eqref{eq:a-2-no-poles-log} in 
\eqref{eq:Fred-det-log-der-series-in-x-no-poles} for $\beta = x$
we can easily take the anti-derivative with respect to $x$ and
obtain
\begin{equation}
    \label{eq:log-Fred-det-no-poles}
    \ln \det_{\Ccal_{\lambda_0}} (\id + \rmV)
    = 
    C (\lambda_0)
    + a (x, \lambda_0)
    - \frac{\tau^2 (\lambda_0)}{2} \ln x
    + \Ocal \left( \frac{ (\ln x)^2 }{x} \right),
\end{equation}
where $C(\lambda_0)$ is an integration constant that does not
depend on $x$ and that remains to be determined.

\subsection{The integration constant}
\label{sec:integration-constant}
In order to fix the constant $C(\lambda_0)$
in equation~\eqref{eq:log-Fred-det-no-poles},
we consider equation~\eqref{eq:Fred-det-log-der-series-in-x-no-poles}
with $\beta = \lambda_0$ and compare with what we obtain by taking the
$\lambda_0$ derivative of \eqref{eq:log-Fred-det-no-poles}. We shall need
the relations
\begin{align}
\label{eq:der_lambdazero_a}
    a_{\lambda_0} (x,\lambda_0)
    & = \partial_{\lambda_0} a(x, \lambda_0)
    + \tau(\lambda_0) g' (\lambda_0),
    \\[.5ex]
    d_{\lambda_0}' (\lambda_0)
    & = - \rmi x\, (\omega' (0 | \lambda_0))^2, \label{eq:d_der_lambdazero}
    \\
    \label{eq:d_dder_lambdazero}
    d_{\lambda_0}'' (\lambda_0)
    & = - 3 \rmi x\, \omega' (0 | \lambda_0) \omega'' (0 | \lambda_0)
    + 2 \rmi x\, \omega' (0 | \lambda_0)
    \partial_{\lambda_0} \omega' (0 | \lambda_0)
\end{align}
that follow directly from the definitions of the respective functions.
Substituting \eqref{eq:d_der_lambdazero} and \eqref{eq:d_dder_lambdazero}
into \eqref{eq:def_a_upper_two} for $\beta = \lambda_0$ we see that
\begin{equation} \label{eq:a2_lambdazero_explicit}
    a_{\lambda_0}^{(2)} (x, \lambda_0)
    = x\, \tau^2 (\lambda_0)
    \biggl(
        \frac{\omega'' (0 | \lambda_0) }{\omega' (0 | \lambda_0)}
        - \partial_{\lambda_0} \ln \omega' (0 | \lambda_0)
        - \frac{\tau' (\lambda_0)}{\tau (\lambda_0)} ( \ln x - 1)
        - \frac{B (\lambda_0)}{2 \tau (\lambda_0)}
    \biggr).
\end{equation}
Inserting then \eqref{eq:der_lambdazero_a}
and \eqref{eq:a2_lambdazero_explicit}
into the right-hand side
of equation~\eqref{eq:Fred-det-log-der-series-in-x-no-poles}
we obtain an explicit expression
for $\partial_{\lambda_0} \ln \det_{\Ccal_{\lambda_0}}
(\id + \rmV)$. Taking the $\lambda_0$ derivative
of \eqref{eq:log-Fred-det-no-poles},
on the other hand, yields
\begin{equation}
    \label{eq:der_lambdazero_log-Fred-det-no-poles}
    \partial_{\lambda_0} \ln \det_{\Ccal_{\lambda_0}} (\id + \rmV)
    = 
    C' (\lambda_0)
    + \partial_{\lambda_0} a (x, \lambda_0)
    - \tau (\lambda_0) \tau' (\lambda_0) \ln x
    + \Ocal \left( \frac{ (\ln x)^2 }{x} \right).
\end{equation}
Comparing the two expressions
for $\partial_{\lambda_0} \ln \det_{\Ccal_{\lambda_0}}
(\id + \rmV)$ we see that
\begin{equation} \label{eq:der_c_lambdazero}
    C' (\lambda_0)
    = \tau^2 (\lambda_0)
    \biggl(
        \frac{\omega'' (0 | \lambda_0) }{\omega' (0 | \lambda_0)}
        - \partial_{\lambda_0} \ln \omega' (0 | \lambda_0)
        + \frac{\tau' (\lambda_0)}{\tau (\lambda_0)}
        - \frac{B (\lambda_0)}{2 \tau (\lambda_0)}
        + \frac{g' (\lambda_0)}{\tau(\lambda_0)}
    \biggr).
\end{equation}
Finally, we insert the explicit expression for $B (\lambda_0)$
from its definition~\eqref{eq:definition-B},
\begin{multline} \label{eq:explicit_B}
    B (\lambda_0)
    = - \partial_{\lambda_0} \ln \tau (\lambda_0)
    + 2 \tau' (\lambda_0) \left( \ln 2 - \frac{\pi \rmi}{2} + 2 \right)
    + 4 \tau' (\lambda_0) \ln \omega' (0 | \lambda_0)
    + 2 \tau (\lambda_0)
    \frac{ \omega'' ( 0 | \lambda_0 ) }{ \omega' ( 0 | \lambda_0 ) }
        \\
    - 2 \tau' (\lambda_0) \psi (\tau (\lambda_0))
    - 2 \partial_{\lambda_0}
    \ln \big[ \vartheta (\lambda_0) \sin^2 (\pi \nu (\lambda_0)) \big]
    - 4 \partial_{\lambda_0} \ln \varkappa ( \lambda_0 | \lambda_0 )
    + 2 g' (\lambda_0),
\end{multline}
to obtain
\begin{multline}
\label{eq:integration-constant-derivative}
    C' (\lambda_0)
    =
    - \frac12 \partial_{\lambda_0} \left(
        \tau^2 (\lambda_0)
        \ln \left[ 2 (\omega' (0 | \lambda_0))^2 \right]
    \right)
    + \frac{1}{2}
    \Big(
        \frac{\pi \rmi}{2} - 1
    \Big)
    \partial_{\lambda_0} \tau^2 (\lambda_0)
    + \frac12 \partial_{\lambda_0} \tau (\lambda_0)
    \\
    \mspace{-1mu}
    + \tau (\lambda_0) \tau' (\lambda_0) \psi (\tau (\lambda_0))
    + \tau (\lambda_0) \partial_{\lambda_0}
    \ln \big[ \vartheta (\lambda_0) \sin^2 (\pi \nu (\lambda_0)) \big]
    + \tau (\lambda_0)
    \partial_{\lambda_0} \ln \varkappa^2 ( \lambda_0 | \lambda_0 ).
\end{multline}

Integrating the latter equation from $\lambda_0$ to $ + \infty$
along the contour $\Ccal_{\lambda_0}$,
which in the case of no poles on the real axis is the contour along $\mathbb{R}$,
we obtain the following result for the integration constant,
\begin{multline}
\label{eq:integration-constant-derivation-plus}
    \lim_{\lambda_0 \rightarrow + \infty}
    \big( C (\lambda_0) \big)
    - C (\lambda_0)
    \\
    = 
    \Big(
        \ln \left[ 2 (\omega' (0 | \lambda_0))^2 \right]
        - \frac{\pi \rmi}{2}
        + 1
    \Big)
    \frac{\tau^2 (\lambda_0)}{2}
    - \frac{ \tau (\lambda_0) }{2}
    + \int\limits_{\lambda_0}^{+ \infty}
    \dd{\lambda}
    \tau (\lambda) \tau' (\lambda) \psi (\tau (\lambda))
        \\
    + \int\limits_{\lambda_0}^{+ \infty}
    \dd{\lambda}
    \tau (\lambda) \partial_{\lambda}
    \ln \big[ \vartheta (\lambda) \sin^2 (\pi \nu (\lambda)) \big]
    + 2 \int\limits_{\lambda_0}^{+ \infty}
    \dd{\lambda}
    \tau (\lambda)
    \partial_{\lambda} \ln \varkappa ( \lambda | \lambda ).
\end{multline}
The remaining tasks are the calculation of the limit on the left-hand side
of this equation and the simplification of the integrals
on the right-hand side.

The first integral on the right-hand side
can be expressed in terms of the Barnes $G$-function,
see its integral representation~\eqref{eq:Barnes-G-function}.
Indeed,
changing the integration variable, $\lambda \to \tau (\lambda)$,
we transform the contour from $\lambda_0$ to $+ \infty$
into a contour $\Ccal$ going in the complex $\tau$ plane
from $\tau (\lambda_0)$ to $0$.
Then,
\begin{multline}
\label{eq:Barnes-G-function-integral}
    \int\limits_{\lambda_0}^{+ \infty}
    \dd{\lambda}
    \tau (\lambda) \tau' (\lambda) \psi (\tau (\lambda))
    = \int\limits_{\Ccal}
    \dd{\tau}
    \tau \psi (\tau)
    = - \int\limits_0^{\tau (\lambda_0)}
    \dd{\tau}
    \tau \psi (\tau)
    \\
    =
    - \log G (\tau (\lambda_0) + 1)
    + \frac{\tau (\lambda_0) (1 - \tau (\lambda_0))}{2}
    + \frac{\tau (\lambda_0)}{2} \log (2 \pi).
\end{multline}
Here, in the second equation,
we have deformed the contour $\Ccal$ from $\tau (\lambda_0)$ to $0$
into the straight contour from $\tau (\lambda_0)$ to $0$.
Since $\abs{\Real \tau (\lambda)} < 1 / 2$
for $\lambda \in \Ccal_{\lambda_0}$
due to the Assumption~\ref{item:assumption-5},
we do not cross any pole
when deforming the contour.

The last integral on the right-hand side
of equation~\eqref{eq:integration-constant-derivation-plus}
involves the function $\varkappa (\lambda | \lambda)$,
which in itself is defined by an integral~\eqref{eq:varkappa}.
We cannot evaluate it in terms of known special functions, but
we can transform it into a more canonical form. Integrating by parts
and then sending the regularization parameter $\varepsilon$ to zero,
we see that
\begin{equation}
    \label{eq:varkappa(lambda|lambda)}
    \ln \varkappa (\lambda | \lambda)
        = - \int\limits_{- \infty}^\lambda
        \dd{\mu}
        \Lcal_\ell' (\mu)
        \ln (\lambda - \mu)
        - \int\limits_{\lambda}^{+ \infty}
        \dd{\mu}
        \Lcal_r' (\mu)
        \ln (\mu - \lambda).
\end{equation}
We first integrate the last integral on the right-hand side
of equation~\eqref{eq:integration-constant-derivation-plus} by parts with
respect to $\lambda$,
\begin{equation}
    \int\limits_{\lambda_0}^{+ \infty}
    \dd{\lambda}
    \tau (\lambda)
    \partial_{\lambda} \ln \varkappa ( \lambda | \lambda )
    = - \tau (\lambda_0) \ln \varkappa ( \lambda_0 | \lambda_0 )
    - \int\limits_{\lambda_0}^{+ \infty}
    \dd{\lambda}
    \tau' (\lambda)
    \ln \varkappa ( \lambda | \lambda ).
\end{equation}
Then we substitute
$\tau (\lambda) = \Lcal_\ell (\lambda) - \Lcal_r (\lambda)$
and the expression~\eqref{eq:varkappa(lambda|lambda)}
and obtain
\begin{multline}
    \int\limits_{\lambda_0}^{+ \infty}
    \dd{\lambda}
    \tau' (\lambda)
    \ln \varkappa ( \lambda | \lambda )
    \\[-2ex]
    = \int\limits_{\lambda_0}^{+ \infty}
    \dd{\lambda}
    \int\limits_{- \infty}^\lambda
    \dd{\mu}
    \Lcal_\ell' (\lambda)
    \Lcal_\ell' (\mu)
    \ln (\lambda - \mu)
    - \int\limits_{\lambda_0}^{+ \infty}
    \dd{\lambda}
    \int\limits_{- \infty}^\lambda
    \dd{\mu}
    \Lcal_r' (\lambda)
    \Lcal_\ell' (\mu)
    \ln (\lambda - \mu)
    \\
    + \int\limits_{\lambda_0}^{+ \infty}
    \dd{\lambda}
    \int\limits_{\lambda}^{+ \infty}
    \dd{\mu}
    \Lcal_\ell' (\lambda)
    \Lcal_r' (\mu)
    \ln (\mu - \lambda)
    - \int\limits_{\lambda_0}^{+ \infty}
    \dd{\lambda}
    \int\limits_{\lambda}^{+ \infty}
    \dd{\mu}
    \Lcal_r' (\lambda)
    \Lcal_r' (\mu)
    \ln (\mu - \lambda).
\end{multline}
Symmetrising each of the double integrals above with respect to the
two integrations involved, we express them as
\begin{multline}
    \int\limits_{\lambda_0}^{+ \infty}
    \dd{\lambda}
    \tau' (\lambda)
    \ln \varkappa(\lambda | \lambda)
    =
    \frac12
    \int\limits_{- \infty}^{+ \infty}
    \dd{\lambda}
    \int\limits_{- \infty}^{+ \infty}
    \dd{\mu}
    \Lcal' (\lambda | \lambda_0)
    \Lcal' (\mu | \lambda_0)
    \ln |\lambda - \mu|
    \\
    - \frac12
    \int\limits_{- \infty}^{+ \infty}
    \dd{\lambda}
    \int\limits_{- \infty}^{+ \infty}
    \dd{\mu}
    \Lcal_\ell' (\lambda)
    \Lcal_\ell' (\mu)
    \ln |\lambda - \mu|,
\end{multline}
where $\Lcal' (\lambda | \lambda_0)$
was introduced in~\eqref{eq:definition-of-Lcal-prime}.

Integrating once more by parts, we obtain
\begin{multline}
\label{eq:integral-with-ln-varkappa-plus}
    2 \int\limits_{\lambda_0}^{+ \infty}
    \dd{\lambda}
    \tau' (\lambda) \ln \varkappa (\lambda | \lambda)
    =
    - \frac12
    \int\limits_{- \infty}^{+ \infty}
    \dd{\lambda}
    \int\limits_{- \infty}^{+ \infty}
    \dd{\mu}
    \frac{
        \Lcal_\ell' (\lambda) \Lcal_\ell (\mu)
        - \Lcal_\ell (\lambda) \Lcal_\ell' (\mu) }{
        \lambda - \mu}
    \\
    - \tau (\lambda_0) \ln \varkappa(\lambda_0 | \lambda_0)
    + \frac12
    \int\limits_{- \infty}^{+ \infty}
    \dd{\lambda}
    \int\limits_{- \infty}^{+ \infty}
    \dd{\mu}
    \frac{
        \Lcal' (\lambda | \lambda_0) \Lcal (\mu | \lambda_0)
        - \Lcal (\lambda | \lambda_0) \Lcal' (\mu | \lambda_0) }{
        \lambda - \mu}.
\end{multline}
Here the boundary term stems from the point $\lambda_0$
in the integrals involving $\Lcal (\lambda | \lambda_0)$.
Eventually, the last integral
in~\eqref{eq:integration-constant-derivation-plus},
involving $\ln \varkappa (\lambda | \lambda)$, can be expressed as
\begin{multline}
\label{eq:integral-with-kappa}
    2 \int\limits_{\lambda_0}^{+ \infty}
    \dd{\lambda}
    \tau (\lambda)
    \partial_{\lambda} \ln \varkappa ( \lambda | \lambda )
    \\[-2ex]
    = - \tau (\lambda_0) \ln \varkappa ( \lambda_0 | \lambda_0 )
    + \frac12
    \int\limits_{- \infty}^{+ \infty}
    \dd{\lambda}
    \int\limits_{- \infty}^{+ \infty}
    \dd{\mu}
    \frac{
        \Lcal_\ell' (\lambda) \Lcal_\ell (\mu)
        - \Lcal_\ell (\lambda) \Lcal_\ell' (\mu) }{
        \lambda - \mu}
    \\
    - \frac12
    \int\limits_{- \infty}^{+ \infty}
    \dd{\lambda}
    \int\limits_{- \infty}^{+ \infty}
    \dd{\mu}
    \frac{
        \Lcal' (\lambda | \lambda_0) \Lcal (\mu | \lambda_0)
        - \Lcal (\lambda | \lambda_0) \Lcal' (\mu | \lambda_0) }{
        \lambda - \mu}.
\end{multline}

Finally, we need to fix the integration constant
$\lim_{\lambda_0 \rightarrow + \infty} C (\lambda_0)$
on the left-hand side
of the expression~\eqref{eq:integration-constant-derivation-plus}.
For this purpose we go back to \eqref{eq:log-Fred-det-no-poles}
and note that
\begin{equation} \label{eq:c_infinity_from_det}
    \lim_{\lambda_0 \rightarrow + \infty} C (\lambda_0)
    =
    \lim_{\lambda_0 \rightarrow + \infty}
    \biggl\{
    \ln \det_{\mathbb{R}} (\id + \rmV)
    - a (x, \lambda_0)
    + \Ocal \left( \frac{ (\ln x)^2 }{x} \right)
    \biggr\}.
\end{equation}
We show in Appendix~\ref{app:existence_of_Fred} that
\begin{equation} \label{eq:Fred_lambda0_infty}
    \lim_{\lambda_0 \rightarrow + \infty}
    \det_{\mathbb{R}} (\id + \rmV)
    = \det_{\mathbb R} \big(\id + \rmV_0 \big) \big( 1 + \Ocal (x^{- \infty}) \big).
\end{equation}
where
\begin{equation}
\label{eq:V-GSK}
    V_0 (\lambda, \mu)
    = \frac{
        \sqrt{ F (\lambda) F (\mu) } }{
        2 \pi \rmi (\lambda - \mu) }
    \left(
        e^{-1}_0 (\lambda) e_0 (\mu)
        - e_0 (\lambda) e^{-1}_0 (\mu)
    \right)
\end{equation}
with
\begin{align} \label{eq:def_F}
    F(\lambda)
    & = \vartheta (\lambda) \bigl(\rme^{2\pi \rmi \nu(\lambda)} - 1\bigr),
    \\[1ex] \label{eq:def_ezero}
    e_0 (\lambda)
    & = \exp \biggl\{- \frac{\rmi x p(\lambda)}{2}
    - \frac12 \Bigl(g(\lambda)
    - \ln\bigl(\rme^{- 2\pi \rmi \nu(\lambda)} - 1\bigr)\Bigr)\biggr\}.
\end{align}

The Fredholm determinant on the right hand side of
equation~\eqref{eq:Fred_lambda0_infty} corresponds to the `static case'
considered in~\cite{KKMST-09-RHp,S-10}. In~\cite{KKMST-09-RHp} the
authors obtain the large-$x$ asymptotic expansion of the Fredholm
determinant $\det_{[-q, q]} \big(\id + \rmV_0 \big)$ (cf.\ Theorem
2.1 of that paper). Since this expansion is uniform in $q$, we
may take the limit $q \to + \infty$, which results in the following
\begin{proposition} \cite{KKMST-09-RHp}
\label{prop:GSK}
    \begin{multline}
    \label{eq:prop:GSK-asymptotics}
    \ln \det_{\mathbb R} \big(\id + \rmV_0 \big) \\
        =
        2
        \int\limits_{- \infty}^{+ \infty}
        \dd{\lambda}
        \Lcal_\ell (\lambda)
        \partial_\lambda \ln e_0 (\lambda)
        + \frac12
        \int\limits_{- \infty}^{+ \infty}
        \dd{\lambda}
        \int\limits_{- \infty}^{+ \infty}
        \dd{\mu}
        \frac{
            \Lcal_\ell' (\lambda) \Lcal_\ell (\mu)
            - \Lcal_\ell (\lambda) \Lcal_\ell' (\mu) }{
            \lambda - \mu }
        + \ocal(1).
    \end{multline}
    \end{proposition}
For the proof we note that $\Lcal_\ell (\lambda)
= - \frac{1}{ 2 \pi \rmi} \ln \left(1 + F (\lambda) \right)$
and that $e_0$ and $\Lcal_\ell$ satisfy the hypothesis of Theorem 2.1
in \cite{KKMST-09-RHp}.

Substituting~\eqref{eq:Fred_lambda0_infty} into \eqref{eq:c_infinity_from_det}
and then \eqref{eq:prop:GSK-asymptotics} into the resulting equation
we obtain the limit
\begin{multline}
\label{eq:constant-at-plus-infty}
    \lim_{\lambda_0 \rightarrow + \infty} C (\lambda_0)
    =
    \frac12
    \int\limits_{- \infty}^{+ \infty}
    \dd{\lambda}
    \int\limits_{- \infty}^{+ \infty}
    \dd{\mu}
    \frac{
        \Lcal_\ell' (\lambda) \Lcal_\ell (\mu)
        - \Lcal_\ell (\lambda) \Lcal_\ell' (\mu) }{
        \lambda - \mu }
    \\[-1ex]
    +
    \int\limits_{- \infty}^{+ \infty}
    \dd{\lambda}
    \Lcal_\ell (\lambda) \partial_\lambda
    \ln\big( \rme^{-2 \pi \rmi \nu (\lambda)} - 1 \big).
\end{multline}

With equations~\eqref{eq:Barnes-G-function-integral},
\eqref{eq:integral-with-kappa},
and~\eqref{eq:constant-at-plus-infty}
into~\eqref{eq:integration-constant-derivation-plus}
we have arrived at the following representation of the integration constant,
\begin{multline}
\label{eq:expression-for-integration-constant-plus}
    C (\lambda_0)
    =
    \ln \left[ G (\tau (\lambda_0) + 1) \right]
    - \frac{\tau (\lambda_0)}{2} \ln (2 \pi)
    + \frac{\tau^2 (\lambda_0)}{2}
    \Big(
        \frac{\pi \rmi}{2}
        -  \ln \left[ 2 (\omega' (0 | \lambda_0))^2 \right]
    \Big)
    \\[1ex]
    + \tau (\lambda_0) \ln \varkappa (\lambda_0 | \lambda_0)
    + \frac12
    \int\limits_{- \infty}^{+ \infty}
    \dd{\lambda}
    \int\limits_{- \infty}^{+ \infty}
    \dd{\mu}
    \frac{
        \Lcal' (\lambda | \lambda_0) \Lcal (\mu | \lambda_0)
        - \Lcal (\lambda | \lambda_0) \Lcal' (\mu | \lambda_0) }{
        \lambda - \mu}
    \\
    + \int\limits_{- \infty}^{+ \infty}
    \dd{\lambda}
    \Lcal_\ell (\lambda)
    \partial_\lambda
    \ln \big(\rme^{- 2 \pi \rmi \nu (\lambda)} - 1\big)
    - \int\limits_{\lambda_0}^{+ \infty}
    \dd{\lambda}
    \tau (\lambda)
    \partial_\lambda
    \ln \big[ \vartheta(\lambda) \sin^2 (\pi \nu (\lambda)) \big].
\end{multline}

Here the last two integrals on the right-hand side can be
slightly simplified by means of the identity
\begin{equation} \label{eq:a_vanishing_formula}
    \int\limits_{- \infty}^{+ \infty}
    \dd{\lambda}
    \Lcal_\ell (\lambda)
    \partial_\lambda
    \ln \big[ \vartheta(\lambda)(\rme^{2 \pi \rmi \nu (\lambda)} - 1) \big]
    = 0.
\end{equation}
The latter is a consequence of
\begin{multline}
\label{eq:dilogarithm-integrals}
    \int\limits_{- \infty}^{\lambda_0}
    \dd{\lambda}
    \Lcal_\ell (\lambda)
    \partial_\lambda \ln \big[
        \vartheta (\lambda)
        (\rme^{2 \pi \rmi \nu (\lambda)} - 1)
    \big] \\
    =
    - \int\limits_{\lambda_0}^{+ \infty}
    \dd{\lambda}
    \Lcal_\ell (\lambda)
    \partial_\lambda \ln \big[
        \vartheta (\lambda)
        (\rme^{2 \pi \rmi \nu (\lambda)} - 1)
    \big]
    =
    \frac{1}{2 \pi \rmi}
    \Li_2 \big(
        \vartheta (\lambda_0)
        (1 - \rme^{2 \pi \rmi \nu(\lambda_0)})
    \big),
\end{multline}
where $\Li_2$ is the dilogarithm. Both integrals reduce to the integral
representation~\eqref{eq:dilogarithm} of the dilogarithm by means of the
substitution
\begin{equation}
    t = - \vartheta (\lambda) (\rme^{2 \pi \rmi \nu (\lambda)} - 1).
\end{equation}
It is important that the integrals are along the contour $\Ccal_{\lambda_0}$,
since due to the Assumption~\ref{as:number6}
and $\Scal \cap \mathbb{R} = \emptyset$,
\begin{equation}
    \vartheta (\lambda)
    (1 - \rme^{\pm 2 \pi \rmi \nu (\lambda)}) \notin [1, {+ \infty}),
    \qquad
    \lambda \in \mathbb{R},
\end{equation}
and therefore both integrals have a trivial monodromy.

Using \eqref{eq:a_vanishing_formula} and the definition
\eqref{eq:definition-of-tau} of $\tau$ we see that
\begin{multline}
    \int\limits_{- \infty}^{+ \infty}
    \dd{\lambda}
    \Lcal_\ell (\lambda)
    \partial_\lambda
    \ln \big(\rme^{- 2 \pi \rmi \nu (\lambda)} - 1\big)
    - \int\limits_{\lambda_0}^{+ \infty}
    \dd{\lambda}
    \tau (\lambda)
    \partial_\lambda
    \ln \big[ \vartheta(\lambda) \sin^2 (\pi \nu (\lambda) \big]
    \\[-2ex] 
    =
    \int\limits_{- \infty}^{+ \infty}
    \dd{\lambda}
    \Lcal (\lambda | \lambda_0)
    \partial_\lambda
    \ln \big[
        \vartheta (\lambda) \sin^2 (\pi \nu (\lambda))
    \big].
\end{multline}
With this identity the integration constant can finally be written as
\begin{multline}
\label{eq:expression-for-integration-constant-final}
    C (\lambda_0)
    =
    \ln \left[ G (\tau (\lambda_0) + 1) \right]
    - \frac{\tau (\lambda_0)}{2} \ln (2 \pi)
    + \frac{\tau^2 (\lambda_0)}{2}
    \Big(
        \frac{\pi \rmi}{2}
        -  \ln (- u''(\lambda_0))
    \Big)
    \\
    + \tau (\lambda_0) \ln \varkappa (\lambda_0 | \lambda_0)
    + \int\limits_{- \infty}^{+ \infty}
    \dd{\lambda}
    \Lcal (\lambda | \lambda_0)
    \partial_\lambda
    \ln \big[
        \vartheta (\lambda) \sin^2 (\pi \nu (\lambda))
    \big]
    \\
    + \frac12
    \int\limits_{- \infty}^{+ \infty}
    \dd{\lambda}
    \int\limits_{- \infty}^{+ \infty}
    \dd{\mu}
    \frac{
        \Lcal' (\lambda | \lambda_0) \Lcal (\mu | \lambda_0)
        - \Lcal (\lambda | \lambda_0) \Lcal' (\mu | \lambda_0) }{
        \lambda - \mu}.
\end{multline}
Here we used that $2 (\omega' (0 | \lambda_0))^2 = - u''(\lambda_0)$,
see equation~\eqref{eq:local-parametrization-omega}.
This concludes the proof of the Theorem~\ref{thm:Fredholm-det-AE-no-poles}.

\subsection{A Fredholm minor}
\label{sec:Fredholm-minor-no-poles}
For the applications, for example, to the impenetrable Bose gas,
we often need to evaluate a minor of the Fredholm determinant,
which has the form
\begin{equation}
    \Bigg[
        \int\limits_{\Ccal_{\lambda_0}}
        \frac{ \dd{\mu} }{ 2 \pi }
        e^{-2} (\mu)
        + \partial_\gamma
    \Bigg]
    \eval{
        \det\limits_{\Ccal_{\lambda_0}}
        \big(
            \id + \rmV
            + \gamma \rmP
        \big)
    }_{\gamma = 0}
\end{equation}
with $\rmP$ being a projector.

If $\det_{\Ccal_{\lambda_0}} \big( \id + \rmV \big) \neq 0$,
one can show that
\begin{equation}
    \partial_\gamma
    \eval{
        \det_{\Ccal_{\lambda_0}}
        \big(
            \id + \rmV
            + \gamma \rmP
        \big)
    }_{\gamma = 0}
    = \tr\big\{
        \big( \id - \rmR \big) \rmP
    \big\}
    \cdot
    \det\limits_{\Ccal_{\lambda_0}}
    \big(
        \id + \rmV
    \big),
\end{equation}
where $\rmR$ is the resolvent, $(\id + \rmV) (\id - \rmR) = \id$.
Hence,
\begin{multline}
\label{eq:minor-of-Fredholm-det}
    \Bigg[
        \int\limits_{\Ccal_{\lambda_0}}
        \frac{ \dd{\mu} }{ 2 \pi }
        e^{-2} (\mu)
        + \partial_\gamma
    \Bigg]
    \eval{
        \det\limits_{\Ccal_{\lambda_0}}
        \big(
            \id + \rmV
            + \gamma \rmP
        \big)
    }_{\gamma = 0}
    \\
    = \Bigg[
        \int\limits_{\Ccal_{\lambda_0}}
        \frac{ \dd{\mu} }{ 2 \pi }
        e^{-2} (\mu)
        + \tr\big\{
            \big( \id - \rmR \big) \rmP
        \big\}
    \Bigg]
    \det\limits_{\Ccal_{\lambda_0}}
    \big(
        \id + \rmV
    \big).
\end{multline}

In the special case
when the kernel of the projector $\rmP$ is given by
\begin{equation}
\label{eq:P-kernel}
    P (\lambda, \mu)
    = \frac{2}{\pi} \vartheta (\mu)
    \sin[ \pi \nu (\lambda) ]
    \sin[ \pi \nu (\mu) ]
    E (\lambda) E (\mu),
\end{equation}
the prefactor in equation~\eqref{eq:minor-of-Fredholm-det}
can be expressed in terms of the solution of the Riemann--Hilbert problem,
which we formulate in the following proposition.
\begin{proposition}
    \label{prop:minor-of-Fredholm-det}
    If $\det_{\Ccal_{\lambda_0}} \big( \id + \rmV \big) \neq 0$,
    then
    \begin{equation}
    \label{eq:prop:minor-of-Fredholm-det}
        \Bigg[
            \int\limits_{\Ccal_{\lambda_0}}
            \frac{ \dd{\mu} }{ 2 \pi }
            e^{-2} (\mu)
            + \partial_\gamma
        \Bigg]
        \eval{
            \det\limits_{\Ccal_{\lambda_0}}
            \big(
                \id + \rmV
                + \gamma \rmP
            \big)
        }_{\gamma = 0}
        \\
        =
        \rmi
        \lim\limits_{\lambda \to \infty}
        \big[
            \lambda \cdot \widetilde{\chi}_{12} (\lambda)
        \big]
        \cdot
        \det\limits_{\Ccal_{\lambda_0}}
        \big( \id + \rmV \big).
    \end{equation}
    Here $\rmP$ is the projector given by equation~\eqref{eq:P-kernel},
    and $\widetilde{\chi}_{12} (\lambda)$ is the matrix element
    of the solution $\widetilde{\chi} (\lambda)$
    of the Riemann--Hilbert Problem~\ref{RHp:chi-tilde}.
\end{proposition}
\begin{proof}
    Consider the matrix $\chi (\lambda)$
    given by equation~\eqref{eq:chi-and-chi-inverse}.
    Since $\det \chi (\lambda) = 1$,
    \begin{equation}
    \label{eq:limit-chi-12}
        \lim\limits_{\lambda \to \infty}
        \big[ \lambda \cdot \chi_{12} (\lambda) \big]
        = - \lim\limits_{\lambda \to \infty}
        \big[ \lambda \cdot (\chi^{-1} (\lambda))_{12} \big]
        = \int\limits_{\Ccal_{\lambda_0}} \dd{\mu}
        \big( \mathbf{E}_R (\mu) \big)_1
        \cdot \big( \mathbf{F}_L^\intercal (\mu) \big)_2,
    \end{equation}
    Next, we use the integral equation
    for the vector~\eqref{eq:vector-F-L} to get
    \begin{equation}
        \big( \mathbf{F}_R (\lambda) \big)_2
        = \big( \id - \rmR \big)
        \big( \mathbf{E}_R (\lambda) \big)_2.
    \end{equation}
    Substituting this expression
    and~\eqref{eq:vectors-E} into~\eqref{eq:limit-chi-12},
    we obtain
    \begin{multline}
        \lim\limits_{\lambda \to \infty}
        \big[ \lambda \cdot \chi_{12} (\lambda) \big]
        \\
        = \frac{2}{\pi \rmi}
        \int\limits_{\Ccal_{\lambda_0}} \dd{\mu}
        \vartheta (\mu)
        \sin[ \pi \nu (\mu) ] E (\mu)
        \int\limits_{\Ccal_{\lambda_0}} \dd{\lambda}
        \big( \delta (\mu - \lambda) - R (\mu, \lambda) \big)
        \sin[ \pi \nu (\lambda) \big] E (\lambda),
    \end{multline}
    which is $- \rmi \tr\{ (\id - \rmR) \rmP \}$,
    see equations~\eqref{eq:P-kernel} and~\eqref{eq:minor-of-Fredholm-det}.
    Finally,
    the first term on the right-hand side of~\eqref{eq:minor-of-Fredholm-det}
    can be absorbed into the matrix element
    of $\widetilde{\chi}_{12} (\lambda)$,
    due to~\eqref{eq:chi-tilde},
    \begin{multline}
        \lim\limits_{\lambda \to \infty}
        \big[
            \lambda \cdot \chi_{12} (\lambda)
        \big]
        = \lim\limits_{\lambda \to \infty}
        \big[
            \lambda \cdot \widetilde{\chi}_{12} (\lambda)
        \big]
        + \lim\limits_{\lambda \to \infty}
        \big[
            \lambda \cdot \widetilde{\chi}_{11} (\lambda)
            \rmC (\lambda)
        \big]
        \\
        = \lim\limits_{\lambda \to \infty}
        \big[
            \lambda \cdot \widetilde{\chi}_{12} (\lambda)
        \big]
        -
        \int\limits_{\Ccal_{\lambda_0}}
        \frac{ \dd{\mu} }{2 \pi  \rmi}
        e^{-2} (\mu),
    \end{multline}
    which completes the proof.
\end{proof}

In the case of no poles in the vicinity of the real axis,
the asymptotic expansion
of the matrix element $\widetilde{\chi}_{12} (\lambda)$ takes the form
\begin{equation}
    \lim\limits_{\lambda \to \infty}
    \big[
        \lambda \cdot \widetilde{\chi}_{12} (\lambda)
    \big]
    = \lim\limits_{\lambda \to \infty}
    \big[
        \lambda \cdot \Pi_{12} (\lambda)
    \big]
    =
    \frac{ b_{12} (\lambda_0) }{ x^{1/2} \omega' (0 | \lambda_0) }
    + \Ocal \Bigg(
        \frac{
            (\ln x)^2 \rme^{\rmi u (\lambda_0)} }{
            x^{3 / 2 - \tau (\lambda_0)} }
    \Bigg).
\end{equation}
Here we used the explicit expansion of
$\Pi (\lambda)$ in the second equation,
see equations~\eqref{eq:Pi-AE-Neumann-series},
\eqref{eq:coefficient-Pi-1-explicit},
and~\eqref{eq:order-of-Pi-3}.

Then it is straightforward to derive the following proposition.
\begin{proposition}
    \label{prop:minor-of-Fredholm-det-no-poles}
    Let Assumptions \ref{item:assumption-1}--\ref{as:number6}
    be fulfilled, and let $\Scal \cap {\mathbb R} = \emptyset$.
    Then the Fredholm minor defined by
    the integrable integral operator $\rmV$
    with kernel~\eqref{eq:kernel-V-as-scalar-product}
    and by the projector $\rmP$ with kernel~\eqref{eq:P-kernel}
    has the large-$x$ asymptotic expansion
    \begin{multline}
        \label{eq:prop:Fredholm-det-minor-AE-no-poles}
        \Bigg[
            \int\limits_{\Ccal_{\lambda_0}}
            \frac{ \dd{\mu} }{ 2 \pi }
            e^{-2} (\mu)
            + \partial_\gamma
        \Bigg]
        \eval{
            \det\limits_{\Ccal_{\lambda_0}}
            \big(
                \id + \rmV
                + \gamma \rmP
            \big)
        }_{\gamma = 0}
        =
        \frac{
            \sqrt{2} \rmi b_{12} (\lambda_0) }{
            \sqrt{- u'' (\lambda_0)} }
        \cdot \Acal (\lambda_0)
        \\
        \times
        x^{ - \flatfrac12 - \flatfrac{\tau^2 (\lambda_0)}{2}}
        \exp\Big\{
            - \int\limits_{\Ccal_{\lambda_0}}
            \dd{\lambda} \Lcal (\lambda | \lambda_0)
            \left( \rmi x u' (\lambda) + g' (\lambda) \right)
        \Big\}
        \Big(
            1
            + \Ocal \Big( \frac{ (\ln x)^2 }{x} \Big)
        \Big),
    \end{multline}
    where the functions $\tau$,
    $\Lcal (\cdot | \lambda_0)$
    were defined in~\eqref{eq:definition-of-tau-Lcal-and-Lcal-prime},
    the amplitude $\Acal (\lambda_0)$ in~\eqref{eq:thm:integration-constant},
    and the coefficient $b_{12} (\lambda_0)$
    in~\eqref{eq:definition-of-b12-and-b21}.
\end{proposition}

\section{The case of \texorpdfstring{{\boldmath$2$}}{2}
    poles on the real axis}
\label{sec:2-poles-on-the-real-axis}
In this section we consider the case of two poles on the real axis,
when $n^+ = n^- = 1$, see equation~\eqref{eq:definition-Scal-pm-set}.
First we prove the statement of Theorem~\ref{thm:Fredholm-det-AE-two-poles}.
Then we derive an expression for the Fredholm minor,
relevant for physical applications,
which is formulated in Proposition~\ref{prop:minor-of-Fredholm-det-two-poles}.

\subsection{Solution for the system for two poles}
We first consider the physically relevant case,
when there are exactly two poles on the real axis,
such that $n^+ = n^- = 1$, see equations~\eqref{eq:definition-Scal-pm}
and \eqref{eq:definition-Scal-pm-set},
which correspond to the configurations of the poles
shown in Figure~\ref{fig:possible-pole-configurations},
see also Figure~\ref{fig:4-configurations-of-the-poles}
and Section~\ref{sec:assumptions}.
We denote the poles by $\{ s_1^\pm \} = \Scal^\pm \cap \mathbb{R}$.
\begin{figure}[ht]
    \centering
    \begin{subfigure}[]{0.5\textwidth}
        \centering
        \inputfig{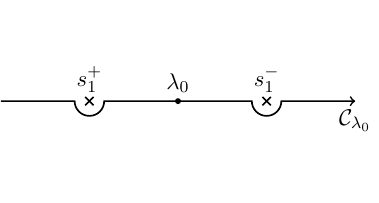}
    \end{subfigure}%
    \begin{subfigure}[]{0.5\textwidth}
        \centering
        \inputfig{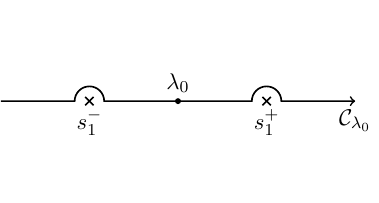}
    \end{subfigure}%
    \newline%
    \begin{subfigure}[]{0.5\textwidth}
        \centering
        \inputfig{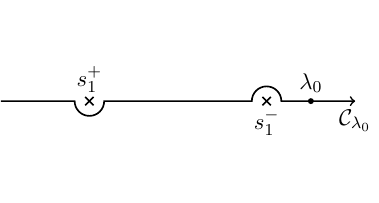}
    \end{subfigure}%
    \begin{subfigure}[]{0.5\textwidth}
        \centering
        \inputfig{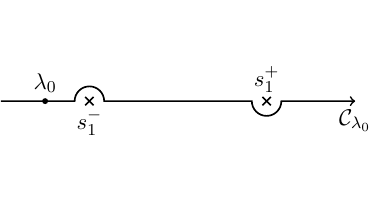}
    \end{subfigure}%
    \caption{
        Possible pole configurations
        corresponding to $n^+ = n^- = 1$.
        For the two cases at the bottom the relative position
        of the poles does not matter, i.e.,
        both possibilities $s_1^+ > s_1^-$
        and $s_1^+ < s_1^-$ are allowed.}
    \label{fig:possible-pole-configurations}%
\end{figure}

The linear system~\eqref{eq:SLE-new}
then takes the form
\begin{equation}
\label{eq:SLE-two-poles}
\left\{
\begin{aligned}
    \mathbf{y}_1
    &= \mathbf{w}_1
    + \frac{
        \sigma_1^- \big( P_{11}^{-+} \big)_{11} }{
        s_1^+ - s_1^- }
    \mathbf{x}_1,
    \\
    \mathbf{x}_1
    &= \mathbf{v}_1
    + \frac{
        \sigma_1^+ \big( P_{11}^{+-} \big)_{22} }{
        s_1^- - s_1^+ }
    \mathbf{y}_1.
\end{aligned}
\right.
\end{equation}
The solution is straightforwardly given by
\begin{equation}
\label{eq:solution-system-for-two-poles}
\begin{aligned}
    \mathbf{y}_1
    &=
    C \cdot \Bigg[
        \mathbf{w}_1
        + \frac{
            \sigma_1^- \big( P_{11}^{+-} \big)_{22} }{
            s_1^+ - s_1^- }
        \mathbf{v}_1
    \Bigg],
    \\
    \mathbf{x}_1
    &=
    C \cdot \Bigg[
        \mathbf{v}_1
        - \frac{
            \sigma_1^+ \big( P_{11}^{+-} \big)_{22} }{
            s_1^+ - s_1^- }
        \mathbf{w}_1
    \Bigg],
\end{aligned}
\end{equation}
where we introduced $C$
\begin{equation}
\label{eq:coefficient-C}
    C
    = \Bigg[
        1 + \frac{
            \sigma_1^+ \sigma_1^-
            \big( P_{11}^{+-} \big)_{22}^2 }{
            (s_1^+ - s_1^-)^2 }
    \Bigg]^{-1}.
\end{equation}
We recall that
$\big( P_{11}^{-+} \big)_{11} = \big( P_{11}^{+-} \big)_{22}$.

Substituting the solution~\eqref{eq:solution-system-for-two-poles}
into the expression~\eqref{eq:prop:contribution-of-poles}
and using the identities
\begin{equation}
\begin{aligned}
    \big( \mathbf{w}'_1 \big)^\intercal
    \sigma^y \mathbf{w}_1
    & = \rmi \big( Q_{11}^{++} \big)_{21},
    \qquad
    & \big( \mathbf{w}'_1 \big)^\intercal
    \sigma^y \mathbf{v}_1
    & = - \rmi \big( Q_{11}^{-+} \big)_{11},
    \\
    \big( \mathbf{v}'_1 \big)^\intercal
    \sigma^y \mathbf{w}_1
    & = \rmi \big( Q_{11}^{+-} \big)_{22},
    \qquad
    & \big( \mathbf{v}'_1 \big)^\intercal
    \sigma^y \mathbf{v}_1
     &= - \rmi \big( Q_{11}^{--} \big)_{12},
\end{aligned}
\end{equation}
and
\begin{equation}
    \mathbf{y}^\intercal \sigma^y \mathbf{x}
    = - \rmi C \big( P_{11}^{+-} \big)_{22},
\end{equation}
we obtain, for $\beta = x$,
\begin{multline}
\label{eq:contribution-from-two-poles}
    \sum\limits_{\mu \in \{s_1^+, s_1^-\}}
    \res\limits_{\lambda = \mu}
    \Big(
        \tr\{
            \rmS' (\lambda) \Pi (\lambda)
            \sigma^z
            \Pi^{-1} (\lambda) \rmS^{-1} (\lambda)
        \}
        d_x (\lambda)
    \Big)
    =
    - 2 C \sigma_1^+
    \big( Q_{11}^{++} \big)_{21}
    d_x (s_1^+)
    \\
    + 2 C \sigma_1^-
    \big( Q_{11}^{--} \big)_{12}
    d_x (s_1^-)
    + \frac{
        2 C \sigma_1^+ \sigma_1^-
    }{
        (s_1^+ - s_1^-)^2 }
    \big( P_{11}^{+-} \big)_{22}^2
    \left[
        d_x (s_1^+) - d_x (s_1^-)
    \right].
\end{multline}
In the following we shall expand this expression in $x^{-1/2}$.

\subsection{Asymptotic expansion of the pole contribution}
We substitute the explicit asymptotic expansion
for $\Pi (\lambda)$,
see equations~\eqref{eq:Pi-AE-Neumann-series},
\eqref{eq:order-of-Pi-3},
\eqref{eq:coefficient-Pi-1-explicit},
and~\eqref{eq:coefficient-Pi-2-explicit},
into the definition~\eqref{eq:matrix-elements-P-and-Q}
of the coefficients $P$ and $Q$ and obtain
\begin{subequations}
\label{eq:AE-for-P-and-Q}
\begin{align}
\label{eq:AE-for-P22}
    \big( P_{11}^{+ -} \big)_{22}
    & = 1 + \Ocal \left( \frac{\ln x}{x} \right),
    \\
\label{eq:AE-for-Q12}
    \big( Q_{11}^{--} \big)_{12}
    &= \frac1{ x^{1 / 2} }
    \big( \Pi_1' (s_1^-) \big)_{12}
    + \Ocal \left(
        \frac{
            (\ln x)^2 \rme^{\rmi x u (\lambda_0)} }{
            x^{3 / 2 - \tau (\lambda_0)} }
    \right),
    \\
\label{eq:AE-for-Q21}
    \big( Q_{11}^{++} \big)_{21}
    &= \frac1{ x^{1 / 2} }
    \big( \Pi_1' (s_1^+) \big)_{21}
    + \Ocal \left(
        \frac{
            (\ln x)^2 \rme^{- \rmi x u (\lambda_0)} }{
            x^{3 / 2 + \tau (\lambda_0)} }
    \right).
\end{align}
\end{subequations}

Subsequently substituting~\eqref{eq:AE-for-Q12} and~\eqref{eq:AE-for-Q21}
into the definition~\eqref{eq:coefficients-sigma}
of the coefficients $\sigma_1^\pm$,
we see that
\begin{subequations}
\label{eq:AE-for-sigma-pm}
\begin{align}
    \sigma_1^+
    = h_1^+
    \Bigg[
        1 + \frac{1}{x^{1/2}} h_1^+ \big( \Pi_1' (s_1^+) \big)_{21}
        + \Ocal
        \Bigg(
            \frac{ (h_1^+ b_{21})^2 }{x}
        \Bigg)
    \Bigg],
    \\
    \sigma_1^-
    = h_1^-
    \Bigg[
        1 + \frac{1}{x^{1/2}} h_1^- \big( \Pi_1' (s_1^-) \big)_{12}
        + \Ocal \Bigg(
            \frac{ (h_1^- b_{12})^2 }{x}
        \Bigg)
    \Bigg].
\end{align}
\end{subequations}
Here we write the corrections in a form that allows us 
to keep track of $x$-dependence, albeit the
dependence of the coefficients $h_1^\pm$, $b_{12}$ and $b_{21}$
is simple, see equations~\eqref{eq:residues-h},
\eqref{eq:b12-and-b21}, and~\eqref{eq:m-and-n},
\begin{equation}
    h_1^\pm \propto \rme^{\pm \rmi x u (s_1^\pm)},
    \qquad
    b_{12} \propto x^{\tau (\lambda_0)} \rme^{\rmi x u (\lambda_0)},
    \qquad
    b_{21} \propto x^{- \tau (\lambda_0)} \rme^{- \rmi x u (\lambda_0)}.
\end{equation}
Exactly such $x$-dependence, originating
from the coefficients $b_{12}$ and $b_{21}$,
appeared in the corrections~\eqref{eq:AE-for-P-and-Q}.

It follows that the asymptotic expansion of $C$,
see the definition~\eqref{eq:coefficient-C},
takes the form
\begin{multline}
\label{eq:AE-for-C}
    C
    = \frac{ \Delta^2 }{ \Delta^2 + h_1^+ h_1^-}
        - \frac{1}{x^{1/2}}
        \frac{ \Delta^2 h_1^+ h_1^- }{ \big[\Delta^2 + h_1^+ h_1^-\big]^2 }
        \Big[
            h_1^+ \big( \Pi_1' (s_1^+) \big)_{21}
            + h_1^- \big( \Pi_1' (s_1^-) \big)_{12}
        \Big]
        \\
        + \frac1x
        \frac{h_1^+ h_1^-}{ \big[ \Delta^2 + h_1^+ h_1^- \big]^2 }
        \Big[
            \Ocal\big( \ln x \big)
            + \Ocal\big( (h_1^+ b_{21})^2 \big)
            + \Ocal\big( (h_1^- b_{12})^2 \big)
        \Big]
        \\
        + \frac1x
        \frac{ (h_1^+ h_1^-)^2 }{ \big[ \Delta^2 + h_1^+ h_1^- \big]^3 }
        \Big[
            \Ocal\big( (h_1^+ b_{21})^2 \big)
            + \Ocal\big( h_1^+ h_1^- \big)
            + \Ocal\big( (h_1^- b_{12})^2 \big)
        \Big],
\end{multline}
where we introduced the abbreviation $\Delta := s_1^+ - s_1^-$.
Here and in what follows we assume
that $\Delta^2 + h_1^+ h_1^- \neq 0$,
i.e., the determinant of the matrix
encoding the linear system is not zero to
leading order in the large parameter $x$,
$\det[\rmI_{n^+} - \rmA^- \rmA^+] \neq 0$
with $n^+ = 1$,
see equation~\eqref{eq:matrices-A-pm}.
The correction of order $\ln x / x$
propagates from~\eqref{eq:AE-for-P22}
and the remaining corrections stem from $\sigma_1^\pm$,
see the asymptotic expansions~\eqref{eq:AE-for-sigma-pm}.

Now, we substitute~\eqref{eq:AE-for-P-and-Q},
\eqref{eq:AE-for-sigma-pm},
and~\eqref{eq:AE-for-C}
into~\eqref{eq:contribution-from-two-poles}.
We also use relation
\begin{equation}
    \frac{
        C \sigma_1^+ \sigma_1^-
    }{
        (s_1^+ - s_1^-)^2 }
    \left[ \big( P_{11}^{+-} \big)_{22} \right]^2
    = 1 - C
\end{equation}
to simplify the asymptotic expansion
of the last term in~\eqref{eq:contribution-from-two-poles}.
We obtain
\begin{multline}
\label{eq:contribution-from-two-poles-in-terms-of-Pi}
    \sum\limits_{\mu \in \{s_1^+, s_1^-\}}
    \res\limits_{\lambda = \mu}
    \Big(
        \tr\{
            \rmS' (\lambda) \Pi (\lambda)
            \sigma^z
            \Pi^{-1} (\lambda) \rmS^{-1} (\lambda) 
        \}
        d_x (\lambda)
    \Big)
    =
    \frac{
        2 h_1^+ h_1^-
        \big[ d_x (s_1^+) - d_x (s_1^-) \big] }{
        \Delta^2 + h_1^+ h_1^-}
    \\
    - \frac1{ x^{1 / 2} }
    \frac{ 2 \Delta^2 }{ \Delta^2 + h_1^+ h_1^-}
    \Big[
        h_1^+
        \big( \Pi_1' (s_1^+) \big)_{21}
        d_x (s_1^+)
        - h_1^-
        \big( \Pi_1' (s_1^-) \big)_{12}
        d_x (s_1^-)
    \Big]
    \\
    + \frac{1}{x^{1/2}}
    \frac{
        2 \Delta^2 h_1^+ h_1^-
        \big[ d_x (s_1^+) - d_x (s_1^-) \big] }{
        \big[\Delta^2 + h_1^+ h_1^-\big]^2 }
    \Big[
        h_1^+ \big( \Pi_1' (s_1^+) \big)_{21}
        + h_1^- \big( \Pi_1' (s_1^-) \big)_{12}
    \Big]
    \\
    + \frac{1}{x}
    \frac{ 1 }{ [\Delta^2 + h_1^+ h_1^-] }
    \Big[
        \Ocal\big( (h_1^+ b_{21})^2 \big)
        + \Ocal\big( (h_1^- b_{12})^2 \big)
    \Big]
    \\
    + \frac{1}{x}
    \frac{ h_1^+ h_1^- }{ [\Delta^2 + h_1^+ h_1^-]^2 }
    \Big[
        \Ocal\big( \ln x \big)
        + \Ocal\big( (h_1^+ b_{21})^2 \big)
        + \Ocal\big( h_1^+ h_1^- \big)
        + \Ocal\big( (h_1^- b_{12})^2 \big)
    \Big]
    \\
    + \frac{1}{x}
    \frac{ (h_1^+ h_1^-)^2 }{ \big[ \Delta^2 + h_1^+ h_1^- \big]^3 }
    \Big[
        \Ocal\big( (h_1^+ b_{21})^2 \big)
        + \Ocal\big( h_1^+ h_1^- \big)
        + \Ocal\big( (h_1^- b_{12})^2 \big)
    \Big].
\end{multline}
Here we also used that $d_x (\lambda) = \Ocal (1)$.

Substituting the coefficient $\Pi_1' (\lambda)$,
see equation~\eqref{eq:coefficient-Pi-1-explicit},
we end up with
\begin{multline}
\label{eq:pole-contribution-two-poles-final}
    \sum\limits_{\mu \in \{s_1^+, s_1^-\}}
    \res\limits_{\lambda = \mu}
    \Big(
        \tr\{
            \rmS' (\lambda) \Pi (\lambda)
            \sigma^z
            \Pi^{-1} (\lambda) \rmS^{-1} (\lambda) 
        \}
        d_x (\lambda)
    \Big)
    =
    \frac{
        2 h_1^+ h_1^-
        \big[ d_x (s_1^+) - d_x (s_1^-) \big] }{
        \Delta^2 + h_1^+ h_1^-}
    \\
    + \frac1{ x^{1 / 2} }
    \frac{1}{\omega' (0 | \lambda_0)}
    \frac{ 2 \Delta^2 }{ \Delta^2 + h_1^+ h_1^-}
    \Bigg[
        \frac{
            h_1^+ b_{21}
            d_x (s_1^+) }{
            (s_1^+ - \lambda_0)^2 }
        -
        \frac{
            h_1^- b_{12}
            d_x (s_1^-) }{
            (s_1^- - \lambda_0)^2 }
    \Bigg]
    \\
    - \frac{1}{x^{1/2}}
    \frac{
        2 \Delta^2 h_1^+ h_1^-
        \big[ d_x (s_1^+) - d_x (s_1^-) \big] }{
        \omega' (0 | \lambda_0)
        \big[\Delta^2 + h_1^+ h_1^-\big]^2 }
    \Bigg[
        \frac{ h_1^+ b_{21} }{ (s_1^+ - \lambda_0)^2 }
        + \frac{ h_1^- b_{12} }{ (s_1^- - \lambda_0)^2 }
    \Bigg]
    \\
    + \frac{1}{x}
    \frac{ 1 }{ [\Delta^2 + h_1^+ h_1^-] }
    \Big[
        \Ocal\big( (h_1^+ b_{21})^2 \big)
        + \Ocal\big( (h_1^- b_{12})^2 \big)
    \Big]
    \\
    + \frac{1}{x}
    \frac{ h_1^+ h_1^- }{ [\Delta^2 + h_1^+ h_1^-]^2 }
    \Big[
        \Ocal\big( \ln x \big)
        + \Ocal\big( (h_1^+ b_{21})^2 \big)
        + \Ocal\big( h_1^+ h_1^- \big)
        + \Ocal\big( (h_1^- b_{12})^2 \big)
    \Big]
    \\
    + \frac{1}{x}
    \frac{ (h_1^+ h_1^-)^2 }{ \big[ \Delta^2 + h_1^+ h_1^- \big]^3 }
    \Big[
        \Ocal\big( (h_1^+ b_{21})^2 \big)
        + \Ocal\big( h_1^+ h_1^- \big)
        + \Ocal\big( (h_1^- b_{12})^2 \big)
    \Big].
\end{multline}

\subsection[Asymptotic expansion
    of the integral over \texorpdfstring{$\gamma_0$}{gamma0}]{\boldmath Asymptotic expansion
    of the integral over $\gamma_0$}
We already calculated the contribution
of the first integral in~\eqref{eq:prop:log-der-Fredholm-for-AE},
see~\eqref{eq:integral-over-gamma-0-no-poles-final}.
Now we calculate the second one,
\begin{equation}
\label{eq:integral-over-gamma-0-two-poles}
    \int\limits_{\gamma_0}
    \frac{\dd{\lambda}}{ 2 \pi \rmi }
    \tr\{
        \rmS' (\lambda) \Pi (\lambda)
        \sigma^z
        \Pi^{-1} (\lambda) \rmS^{-1} (\lambda)
    \}
    d_\beta (\lambda)
\end{equation}
for $\beta = x$.
Substituting the asymptotic expansion for $\Pi (\lambda)$
under the trace, we obtain
\begin{multline}
    \tr\{
        \rmS' (\lambda) \Pi (\lambda)
        \sigma^z
        \Pi^{-1} (\lambda) \rmS^{-1} (\lambda)
    \}
    =
    \tr\{
        \rmS^{-1} (\lambda) \rmS' (\lambda) \sigma^z
    \}
    \\
    + \frac{1}{ x^{1/2} }
    \tr\bigl\{
        \rmS^{-1} (\lambda) \rmS' (\lambda)
        \bigl(
            \Pi_1 (\lambda) \sigma^z + \sigma^z \Pi^{\rm ad}_1 (\lambda)
        \bigr)
    \bigr\}
    \\
    + \frac{1}{ x }
    \tr\bigl\{
        \rmS^{-1} (\lambda) \rmS' (\lambda)
        \bigl(
            \Pi_2 (\lambda) \sigma^z
            - \Pi_1 (\lambda) \sigma^z \Pi_1^{\rm ad} (\lambda)
            + \sigma^z \Pi_2^{\rm ad} (\lambda)
        \bigr)
    \bigr\}
    \\
    + \tr\{
        \rmS^{-1} (\lambda) \rmS' (\lambda)
        \sigma^+
    \}
    \Ocal\biggl(
        \frac{
            (\ln x)^2 \rme^{\rmi x u (\lambda_0)} }{
            x^{3 / 2 - \tau (\lambda_0)} }
    \biggr)
    + \tr\{
        \rmS^{-1} (\lambda) \rmS' (\lambda)
        \sigma^-
    \}
    \Ocal\biggl(
        \frac{
            (\ln x)^2 \rme^{- \rmi x u (\lambda_0)} }{
            x^{3 / 2 + \tau (\lambda_0)} }
    \biggr).
\end{multline}

Now
\begin{equation}
    \Pi (\lambda) \Pi^{-1} (\lambda) = 1
    \quad \Rightarrow \quad
    \Pi^{\rm ad}_{1} (\lambda) = - \Pi_{1} (\lambda),
    \quad
    \Pi^{\rm ad}_{2} (\lambda) = \Pi_1^2 (\lambda) - \Pi_2 (\lambda).
\end{equation}
We further note
that $\comm{\Pi_2 (\lambda)}{\sigma^z} = 0$,
since the matrix $\Pi_2 (\lambda)$ is diagonal, and that
\begin{equation}
    \Pi_1 (\lambda) \sigma^z = - \sigma^z \Pi_1 (\lambda).
\end{equation}
Then
\begin{multline}
    \tr\{
        \rmS' (\lambda) \Pi (\lambda)
        \sigma^z
        \Pi^{-1} (\lambda) \rmS^{-1} (\lambda)
    \}
    =
    \tr\{
        \rmS^{-1} (\lambda) \rmS' (\lambda) \sigma^z
    \}
    \\
    + \frac{2}{ x^{1/2} }
    \tr\{
        \rmS^{-1} (\lambda) \rmS' (\lambda)
        \Pi_1 (\lambda) \sigma^z
    \}
    + \frac{2}{ x }
    \tr\{
        \rmS^{-1} (\lambda) \rmS' (\lambda)
        \Pi_1^2 (\lambda) \sigma^z
    \}
    \\
    + \Ocal\biggl(
        \frac{ (\ln x)^2 \rme^{\rmi x u (\lambda_0)} }{
            x^{3 / 2 - \tau (\lambda_0)} }
    \biggr)
    + \Ocal\biggl(
        \frac{ (\ln x)^2 \rme^{- \rmi x u (\lambda_0)} }{
            x^{3 / 2 + \tau (\lambda_0)} }
    \biggr).
\end{multline}
It follows
from the expression~\eqref{eq:coefficient-Pi-1-explicit} for $\Pi_1 (\lambda)$
and~\eqref{eq:product-of-b12-and-b21} that
\begin{equation}
    \Pi_1^2 (\lambda)
    = \frac{1}{(\lambda - \lambda_0)^2}
    \frac{
        \rmi \tau (\lambda_0) }{
        2 (\omega' (0 | \lambda_0))^2 } \mathrm{I}_2.
\end{equation}
Hence,
\begin{multline}
\label{eq:tr-S-and-Pi-for-integral-over-gamma-0}
    \tr\{
        \rmS' (\lambda) \Pi (\lambda)
        \sigma^z
        \Pi^{-1} (\lambda) \rmS^{-1} (\lambda)
    \}
    =
    \tr\{
        \rmS^{-1} (\lambda) \rmS' (\lambda) \sigma^z
    \}
    \\
    + \frac{2}{ x^{1/2} }
    \tr\{
        \rmS^{-1} (\lambda) \rmS' (\lambda)
        \Pi_1 (\lambda) \sigma^z
    \}
    + \frac{1}{ x }
    \frac{1}{(\lambda - \lambda_0)^2}
    \frac{ \rmi \tau (\lambda_0) }{ (\omega' (0 | \lambda_0))^2 }
    \tr\{
        \rmS^{-1} (\lambda) \rmS' (\lambda)
        \sigma^z
    \}
    \\
    + \Ocal\biggl(
        \frac{
            (\ln x)^2 \rme^{\rmi x u (\lambda_0)} }{
            x^{3 / 2 - \tau (\lambda_0)} }
    \biggr)
    + \Ocal\biggl(
        \frac{
            (\ln x)^2 \rme^{- \rmi x u (\lambda_0)} }{
            x^{3 / 2 + \tau (\lambda_0)} }
    \biggr).
\end{multline}

We emphasize that we have expanded so far
only the matrix $\Pi (\lambda)$, whereas
the matrix $\rmS (\lambda)$ has been kept
exact. To proceed any further,
we need an explicit expression for $\rmS^{-1} (\lambda) \rmS' (\lambda)$
and its asymptotic expansion.
Substituting~\eqref{eq:solution-system-for-two-poles}
into~\eqref{eq:S-inv-S-prime}
we obtain after some algebra
\begin{multline}
    \rmS^{-1} (\lambda)
    \rmS' (\lambda)
    = - \frac{ \rmi C \sigma_1^+ }{ (\lambda - s_1^+) }
    \mathbf{w}_1 \mathbf{w}_1^\intercal
    \sigma^y
    + \frac{ \rmi C \sigma_1^- }{ (\lambda - s_1^-) }
    \mathbf{v}_1 \mathbf{v}_1^\intercal
    \sigma^y
    \\
    - \frac{
        \rmi C \sigma_1^+ \sigma_1^- \big( P_{11}^{+-} \big)_{22} }{
        (\lambda - s_1^+)
        (\lambda - s_1^-)
        (s_1^+ - s_1^-) }
    \Big[
        \mathbf{v}_1 \mathbf{w}_1^\intercal
        + \mathbf{w}_1 \mathbf{v}_1^\intercal
    \Big]
    \sigma^y,
\end{multline}
where
\begin{align}
    \mathbf{w}_1 \mathbf{w}_1^\intercal
    &= \begin{pmatrix}
        \Pi_{11}^2 (s_1^+)
        & \Pi_{11} (s_1^+) \Pi_{21} (s_1^+)
        \\
        \Pi_{11} (s_1^+) \Pi_{21} (s_1^+)
        & \Pi_{21}^2 (s_1^+)
    \end{pmatrix},
    \\
    \mathbf{v}_1 \mathbf{v}_1^\intercal
    &= \begin{pmatrix}
        \Pi_{12}^2 (s_1^-)
        & \Pi_{12} (s_1^-) \Pi_{22} (s_1^-)
        \\
        \Pi_{12} (s_1^-) \Pi_{22} (s_1^-)
        & \Pi_{22}^2 (s_1^-)
    \end{pmatrix},
    \\
    \mathbf{v}_1 \mathbf{w}_1^\intercal
    = \left( \mathbf{w}_1 \mathbf{v}_1^\intercal \right)^\intercal
    &= \begin{pmatrix}
        \Pi_{12} (s_1^-) \Pi_{11} (s_1^+)
        & \Pi_{12} (s_1^-) \Pi_{21} (s_1^+)
        \\
        \Pi_{22} (s_1^-) \Pi_{11} (s_1^+)
        & \Pi_{22} (s_1^-) \Pi_{21} (s_1^+)
    \end{pmatrix},
\end{align}
Then the diagonal matrix elements of $\rmS^{-1} (\lambda) \rmS' (\lambda)$
are given by
\begin{multline}
    \big( \rmS^{-1} (\lambda) \rmS' (\lambda) \big)_{11}
    = \frac{
        C \sigma_1^+ \Pi_{11} (s_1^+) \Pi_{21} (s_1^+) }{
        (\lambda - s_1^+)^2 }
    - \frac{
        C \sigma_1^- \Pi_{12} (s_1^-) \Pi_{22} (s_1^-) }{
        (\lambda - s_1^-)^2 }
    \\
    + \frac{
        C \sigma_1^+ \sigma_1^-
        \big[
            \Pi_{11}^2 (s_1^+) \Pi_{22}^2 (s_1^-)
            - \Pi_{21}^2 (s_1^+) \Pi_{12}^2 (s_1^-)
        \big]
    }{
        (\lambda - s_1^+)
        (\lambda - s_1^-)
        (s_1^+ - s_1^-) },
\end{multline}
\begin{multline}
    \big( \rmS^{-1} (\lambda) \rmS' (\lambda) \big)_{22}
    = - \frac{
        C \sigma_1^+ \Pi_{11} (s_1^+) \Pi_{21} (s_1^+) }{
        (\lambda - s_1^+)^2 }
    + \frac{
        C \sigma_1^- \Pi_{12} (s_1^-) \Pi_{22} (s_1^-) }{
        (\lambda - s_1^-)^2 }
    \\
    - \frac{
        C \sigma_1^+ \sigma_1^-
        \big[
            \Pi_{11}^2 (s_1^+) \Pi_{22}^2 (s_1^-)
            - \Pi_{21}^2 (s_1^+) \Pi_{12}^2 (s_1^-)
        \big]
    }{
        (\lambda - s_1^+)
        (\lambda - s_1^-)
        (s_1^+ - s_1^-) },
\end{multline}
and the off-diagonal ones by
\begin{multline}
    \big( \rmS^{-1} (\lambda) \rmS' (\lambda) \big)_{12}
    = - \frac{ C \sigma_1^+ \Pi_{11}^2 (s_1^+) }{
        (\lambda - s_1^+)^2 }
    + \frac{ C \sigma_1^- \Pi_{12}^2 (s_1^-) }{
        (\lambda - s_1^-)^2 }
    \\
    - \frac{
        2 C \sigma_1^+ \sigma_1^-
        \Pi_{11} (s_1^+) \Pi_{12} (s_1^-)
        \big[
            \Pi_{11} (s_1^+) \Pi_{22} (s_1^-)
            - \Pi_{21} (s_1^+) \Pi_{12} (s_1^-)
        \big] }{
        (\lambda - s_1^+)
        (\lambda - s_1^-)
        (s_1^+ - s_1^-) },
\end{multline}
\begin{multline}
    \big( \rmS^{-1} (\lambda) \rmS' (\lambda) \big)_{21}
    = \frac{ C \sigma_1^+ \Pi_{21}^2 (s_1^+) }{
        (\lambda - s_1^+)^2 }
    - \frac{ C \sigma_1^- \Pi_{22}^2 (s_1^-) }{
        (\lambda - s_1^-)^2 }
    \\
    + \frac{
        2 C \sigma_1^+ \sigma_1^-
        \Pi_{21} (s_1^+) \Pi_{22} (s_1^-)
        \big[
            \Pi_{11} (s_1^+) \Pi_{22} (s_1^-)
            - \Pi_{21} (s_1^+) \Pi_{12} (s_1^-)
        \big] }{
        (\lambda - s_1^+)
        (\lambda - s_1^-)
        (s_1^+ - s_1^-) }.
\end{multline}
We note that
$(\rmS^{-1} (\lambda) \rmS' (\lambda) )_{22}
= - (\rmS^{-1} (\lambda) \rmS' (\lambda) )_{11}$.

The asymptotic expansion for the matrix elements is thus given by
\begin{equation}
    \big( \rmS^{-1} (\lambda) \rmS' (\lambda) \big)_{11}
    = \frac{ h_1^+ h_1^- }{ \Delta^2 + h_1^+ h_1^- }
    \Ocal(1)
    = - \big( \rmS^{-1} (\lambda) \rmS' (\lambda) \big)_{22},
\end{equation}
for the diagonal matrix elements,
and
\begin{multline}
    \big( \rmS^{-1} (\lambda) \rmS' (\lambda) \big)_{12}
    = - \frac{ h_1^+ }{ \Delta^2 + h_1^+ h_1^- }
    \frac{ \Delta^2 }{ (\lambda - s_1^+)^2 }
    + \frac{
        \Ocal\big( (h_1^+)^2 b_{21} \big)
        + \Ocal\big( h_1^+ h_1^- b_{12} \big)
    }{
        x^{1 / 2} \big[ \Delta^2 + h_1^+ h_1^- \big] }
    \\
    + \frac{
        \Ocal\big( (h_1^+)^3 h_1^- b_{21} \big)
        + \Ocal\big( (h_1^+ h_1^-)^2 b_{12} \big)
    }{
        x^{1 / 2} \big[ \Delta^2 + h_1^+ h_1^- \big]^2 }
    ,
\end{multline}
\begin{multline}
    \big( \rmS^{-1} (\lambda) \rmS' (\lambda) \big)_{21}
    = - \frac{ h_1^- }{ \Delta^2 + h_1^+ h_1^- }
    \frac{ \Delta^2 }{ (\lambda - s_1^-)^2 }
    + \frac{
        \Ocal\big( (h_1^-)^2 b_{12} \big)
        + \Ocal\big( h_1^+ h_1^- b_{21} \big)
    }{
        x^{1 / 2} \big[ \Delta^2 + h_1^+ h_1^- \big] }
    \\
    + \frac{
        \Ocal\big( h_1^+ (h_1^-)^3 b_{12} \big)
        + \Ocal\big( (h_1^+ h_1^-)^2 b_{21} \big)
    }{
        x^{1 / 2} \big[ \Delta^2 + h_1^+ h_1^- \big]^2 }
\end{multline}
for the off-diagonal matrix elements.

Next, we insert
the expression~\eqref{eq:coefficient-Pi-1-explicit} for $\Pi_1 (\lambda)$
and the asymptotic expansions
of the matrix elements of $\rmS^{-1} (\lambda) \rmS' (\lambda)$
into equation~\eqref{eq:tr-S-and-Pi-for-integral-over-gamma-0}.
Subsequently, we substitute the asymptotic expansions of 
$\sigma_1^\pm$, and $C$, see equations~\eqref{eq:AE-for-sigma-pm},
and~\eqref{eq:AE-for-C}.
As a result we obtain the following asymptotic expansion
for the second term
on the right-hand side of~\eqref{eq:tr-S-and-Pi-for-integral-over-gamma-0},
\begin{multline}
    \frac{2}{ x^{1/2} }
    \tr\{
        \rmS^{-1} (\lambda) \rmS' (\lambda)
        \Pi_1 (\lambda) \sigma^z
    \}
    \\
    =
    \frac{1}{ x^{1/2} }
    \frac{1}{ (\lambda - \lambda_0) }
    \frac{2}{ \omega' (0 | \lambda_0) }
    \Big[
        \big( \rmS^{-1} (\lambda) \rmS' (\lambda) \big)_{12} b_{21}
        - \big( \rmS^{-1} (\lambda) \rmS' (\lambda) \big)_{21} b_{12}
    \Big]
    \\
    = - \frac{1}{ x^{1/2} }
    \frac{1}{ (\lambda - \lambda_0)}
    \frac{1}{ \omega' (0 | \lambda_0) }
    \frac{ 2 \Delta^2 }{ \big[ \Delta^2 + h_1^+ h_1^- \big]}
    \Bigg[
        \frac{ h_1^+ b_{21} }{ (\lambda - s_1^+)^2 }
        - \frac{ h_1^- b_{12} }{ (\lambda - s_1^-)^2 }
    \Bigg]
    \\
    + \frac{1}{ x }
    \frac{1}{ (\lambda - \lambda_0) }
    \frac{1}{ \big[ \Delta^2 + h_1^+ h_1^- \big]}
    \Big[
        \Ocal\big( (h_1^+ b_{21})^2 \big)
        + \Ocal\big( h_1^+ h_1^- \big)
        + \Ocal\big( (h_1^- b_{12})^2 \big)
    \Big]
    \\
    + \frac{1}{ x }
    \frac{1}{ (\lambda - \lambda_0) }
    \frac{
        h_1^+ h_1^- }{
        \big[ \Delta^2 + h_1^+ h_1^- \big]^2}
    \Big[
        \Ocal\big( (h_1^+ b_{21})^2 \big)
        + \Ocal\big( h_1^+ h_1^- \big)
        + \Ocal\big( (h_1^- b_{12})^2 \big)
    \Big],
\end{multline}
which has a first order pole at $\lambda = \lambda_0$,
and the asymptotic expansion for the third term
on the right-hand side of~\eqref{eq:tr-S-and-Pi-for-integral-over-gamma-0},
\begin{multline}
    \frac{1}{ x }
    \frac{1}{(\lambda - \lambda_0)^2}
    \frac{ \rmi \tau (\lambda_0) }{ (\omega' (0 | \lambda_0))^2 }
    \Big[
        \big( \rmS^{-1} (\lambda) \rmS' (\lambda) \big)_{11}
        - \big( \rmS^{-1} (\lambda) \rmS' (\lambda) \big)_{22}
    \Big]
    \\
    = \frac{1}{ x }
    \frac{1}{(\lambda - \lambda_0)^2}
    \frac{ 2 \rmi \tau (\lambda_0) }{ (\omega' (0 | \lambda_0))^2 }
    \big( \rmS^{-1} (\lambda) \rmS' (\lambda) \big)_{11}
    =
    \frac{1}{ x }
    \frac{ 1 }{ (\lambda - \lambda_0)^2 }
    \frac{ \Ocal\big( h_1^+ h_1^- \big) }{ \big[ \Delta^2 + h_1^+ h_1^- \big]},
\end{multline}
which has a second order pole at $\lambda = \lambda_0$.

Then the trace~\eqref{eq:tr-S-and-Pi-for-integral-over-gamma-0}
is evaluated as
\begin{multline}
    \tr\{
        \rmS' (\lambda) \Pi (\lambda)
        \sigma^z
        \Pi^{-1} (\lambda) \rmS^{-1} (\lambda)
    \}
    =
    \tr\{
        \rmS^{-1} (\lambda) \rmS' (\lambda) \sigma^z
    \}
    \\
    - \frac{1}{ x^{1/2} }
    \frac{1}{(\lambda - \lambda_0)}
    \frac{ 2 \Delta^2 }{
        \omega' (0 | \lambda_0) \big[ \Delta^2 + h_1^+ h_1^- \big]}
    \Bigg[
        \frac{ h_1^+ b_{21} }{ (\lambda - s_1^+)^2 }
        - \frac{ h_1^- b_{12} }{ (\lambda - s_1^-)^2 }
    \Bigg]
    \\
    + \frac{1}{ x }
    \frac{1}{(\lambda - \lambda_0)}
    \frac{1}{ \big[ \Delta^2 + h_1^+ h_1^- \big]}
    \Big[
        \Ocal\big( (h_1^+ b_{21})^2 \big)
        + \Ocal\big( h_1^+ h_1^- \big)
        + \Ocal\big( (h_1^- b_{12})^2 \big)
    \Big]
    \\
    + \frac{1}{ x }
    \frac{1}{(\lambda - \lambda_0)}
    \frac{ h_1^+ h_1^- }{ \big[ \Delta^2 + h_1^+ h_1^- \big]^2}
    \Big[
        \Ocal\big( (h_1^+ b_{21})^2 \big)
        + \Ocal\big( h_1^+ h_1^- \big)
        + \Ocal\big( (h_1^- b_{12})^2 \big)
    \Big]
    \\
    + \frac{1}{ x }
    \frac{ 1 }{ (\lambda - \lambda_0)^2 }
    \frac{ \Ocal\big( h_1^+ h_1^- \big) }{ \big[ \Delta^2 + h_1^+ h_1^- \big]}.
\end{multline}
Here we keep the factors $(\lambda-\lambda_0)^{-n}$ for $n = 1, 2$ explicit,
since we integrate this expression
over the contour $\gamma_0$ around the pole at $\lambda_0$.
The first term on the right-hand side does not have a pole at $\lambda_0$
and therefore does not contribute
to the integral~\eqref{eq:integral-over-gamma-0-two-poles}.
Hence,
the integral over the contour $\gamma_0$ with $\beta = x$
has the large-$x$ asymptotic expansion
\begin{multline}
\label{eq:integral-over-gamma-0-two-poles-final}
    \int\limits_{\gamma_0}
    \frac{\dd{\lambda}}{ 2 \pi \rmi }
    \tr\{
        \rmS' (\lambda) \Pi (\lambda)
        \sigma^z
        \Pi^{-1} (\lambda) \rmS^{-1} (\lambda)
    \}
    d_x (\lambda)
    \\
    = - \frac{1}{ x^{1/2} }
    \frac{
        2 \Delta^2 d_x (\lambda_0) }{
        \omega' (0 | \lambda_0) \big[ \Delta^2 + h_1^+ h_1^- \big]}
    \Bigg[
        \frac{ h_1^+ b_{21} }{ (\lambda_0 - s_1^+)^2 }
        - \frac{ h_1^- b_{12} }{ (\lambda_0 - s_1^-)^2 }
    \Bigg]
    \\
    + \frac{1}{ x }
    \frac{1}{ \big[ \Delta^2 + h_1^+ h_1^- \big] }
    \Big[
        \Ocal\big( (h_1^+ b_{21})^2 \big)
        + \Ocal\big( h_1^+ h_1^- \big)
        + \Ocal\big( (h_1^- b_{12})^2 \big)
    \Big]
    \\
    + \frac{1}{ x }
    \frac{ h_1^+ h_1^- }{ \big[ \Delta^2 + h_1^+ h_1^- \big]^2 }
    \Big[
        \Ocal\big( (h_1^+ b_{21})^2 \big)
        + \Ocal\big( h_1^+ h_1^- \big)
        + \Ocal\big( (h_1^- b_{12})^2 \big)
    \Big].
\end{multline}
At this point we have obtained asymptotic expansions for
all terms in~\eqref{eq:prop:log-der-Fredholm-for-AE}
and can write the logarithmic derivative of the Fredholm determinant
with respect to $\beta = x$ explicitly.

\subsection{Asymptotic expansion of the Fredholm determinant}
We substitute the asymptotic expansions%
~\eqref{eq:integral-over-gamma-0-no-poles-final},
\eqref{eq:pole-contribution-two-poles-final}
and~\eqref{eq:integral-over-gamma-0-two-poles-final}
into~\eqref{eq:prop:log-der-Fredholm-for-AE}
with $\beta = x$.
We obtain the following asymptotic expansion
of the logarithmic derivative of the Fredholm determinant:
\begin{equation}
    \label{eq:Fred-det-log-der-series-in-x-two-poles}
    \partial_x \ln \det_{\Ccal_{\lambda_0}}
    (\id + \rmV)
    = a_x (x, \lambda_0)
    + b_x^{(0)}
    + \frac{b_x^{(1)}}{x^{1/2}}
    + \frac{ b_x^{(2)} }{x}
    + \frac{ a_x^{(2)}}{x}
    + \Ocal \left( \frac{ (\ln x)^2 }{x^2} \right).
\end{equation}
Here the coefficients $b_x^{(0)}$ $b_x^{(1)}$ and $b_x^{(2)}$
are the contributions of the integral over $\gamma_0$
involving the matrix $\rmS (\lambda)$%
~\eqref{eq:integral-over-gamma-0-two-poles-final}
and of the sum over the poles%
~\eqref{eq:pole-contribution-two-poles-final}
to the orders $x^{0}$, $x^{- 1 / 2}$, and $x^{-1}$, respectively.
The coefficient $b_x^{(0)}$ originates
solely from~\eqref{eq:pole-contribution-two-poles-final},
\begin{equation}
    b_x^{(0)}
    = - \frac{
        2 h_1^+ h_1^-
        \big[ d_x (s_1^+) - d_x (s_1^-) \big] }{
        \Delta^2 + h_1^+ h_1^-},
\end{equation}
while the coefficients $b_x^{(1)}$ and $b_x^{(2)}$ both originate
from~\eqref{eq:pole-contribution-two-poles-final}
and~\eqref{eq:integral-over-gamma-0-two-poles-final},
\begin{multline}
    \frac{ b_x^{(1)} }{ x^{1/2} }
    = \frac{1}{ x^{1/2} }
    \frac{\Delta^2}{\omega' (0 | \lambda_0)}
    \Bigg\{
        \frac{
            2 \big[ d_x (s_1^+) - d_x (s_1^-) \big]
            h_1^+ h_1^- }{
            \big[\Delta^2 + h_1^+ h_1^-\big]^2 }
        \Bigg[
            \frac{ h_1^+ b_{21} }{ (s_1^+ - \lambda_0)^2 }
            + \frac{ h_1^- b_{12} }{ (s_1^- - \lambda_0)^2 }
        \Bigg]
        \\
        + \frac{1}{ \big[ \Delta^2 + h_1^+ h_1^- \big]}
        \Bigg[
            \frac{
                2 \big[ d_x (\lambda_0) - d_x (s_1^+) \big]
                h_1^+ b_{21} }{
                (\lambda_0 - s_1^+)^2 }
            - \frac{
                2 \big[ d_x (\lambda_0) - d_x (s_1^-) \big]
                h_1^- b_{12} }{
                (\lambda_0 - s_1^-)^2 }
        \Bigg]
    \Bigg\}
\end{multline}
and
\begin{multline}
    \frac{
        b_x^{(2)}
    }{x}
    =
    \frac{1}{x}
    \frac{ 1 }{ [\Delta^2 + h_1^+ h_1^-] }
    \Big[
        \Ocal\big( (h_1^+ b_{21})^2 \big)
        + \Ocal\big( h_1^+ h_1^- \big)
        + \Ocal\big( (h_1^- b_{12})^2 \big)
    \Big]
    \\
    + \frac{1}{x}
    \frac{ h_1^+ h_1^- }{ \big[ \Delta^2 + h_1^+ h_1^- \big]^2 }
    \Big[
        \Ocal\big( \ln x \big)
        + \Ocal\big( (h_1^+ b_{21})^2 \big)
        + \Ocal\big( h_1^+ h_1^- \big)
        + \Ocal\big( (h_1^- b_{12})^2 \big)
    \Big]
    \\
    + \frac{1}{x}
    \frac{ (h_1^+ h_1^-)^2 }{ \big[ \Delta^2 + h_1^+ h_1^- \big]^3 }
    \Big[
        \Ocal\big( (h_1^+ b_{21})^2 \big)
        + \Ocal\big( h_1^+ h_1^- \big)
        + \Ocal\big( (h_1^- b_{12})^2 \big)
    \Big].
\end{multline}
The coefficient $a_x^{(2)}$
is the contribution of that part of 
the integral over $\gamma_0$ that does not
involve the matrix $\rmS (\lambda)$,
i.e., the contribution
of the term~\eqref{eq:integral-over-gamma-0-no-poles-final}.
It was calculated explicitly in the previous
section, see equation~\eqref{eq:a-2-no-poles-log}.
The correction $\Ocal( (\ln x)^2 / x^2)$ in the asymptotic expansion
is also from~\eqref{eq:integral-over-gamma-0-no-poles-final}.

Now, we integrate the logarithmic derivative
of the Fredholm determinant with respect to $x$.
First, it is straightforward to check that
\begin{equation}
    b_x^{(0)} = \partial_x \ln \big[ \Delta^2 + h_1^+ h_1^- \big]
\end{equation}
and
\begin{multline}
    \frac{b_x^{(1)}}{x^{1 / 2}}
    =
    \frac{1}{ x^{1/2} }
    \frac{\Delta^2}{\omega' (0 | \lambda_0)}
    \partial_x \Bigg\{
        \frac{1}{ \big[ \Delta^2 + h_1^+ h_1^- \big]}
        \Bigg[
            \frac{ h_1^+ b_{21} }{ (\lambda_0 - s_1^+)^2 }
            + \frac{ h_1^- b_{12} }{ (\lambda_0 - s_1^-)^2 }
        \Bigg]
    \Bigg\}
    \\
    =
    \frac{
        \Delta^2 \big( 1 + \Ocal \left( x^{-1} \right) \big) }{
        \omega' (0 | \lambda_0)}
    \partial_x \Bigg\{
        \frac{1}{ x^{1/2} }
        \frac{1}{ \big[ \Delta^2 + h_1^+ h_1^- \big]}
        \Bigg[
            \frac{ h_1^+ b_{21} }{ (\lambda_0 - s_1^+)^2 }
            + \frac{ h_1^- b_{12} }{ (\lambda_0 - s_1^-)^2 }
        \Bigg]
    \Bigg\}.
\end{multline}
To evaluate the anti-derivatives of $b_x^{(0)}$ and $b_x^{(1)}$,
we used that by definition $d_x (\lambda) = \partial_x \ln e (\lambda)$,
see equation~\eqref{eq:function-d},
and, therefore,
\begin{equation}
    \partial_x h_1^+
    = - 2 d_x (s_1^+) \cdot h_1^+,
    \qquad
    \partial_x h_1^-
    = 2 d_x (s_1^-) \cdot h_1^-,
\end{equation}
\begin{equation}
    \partial_x b_{12}
    = - 2 d_x (\lambda_0) \cdot b_{12}
    \big[ 1 + \Ocal(x^{-1}) \big],
    \qquad
    \partial_x b_{21}
    = 2 d_x (\lambda_0) \cdot b_{21}
    \big[ 1 + \Ocal(x^{-1}) \big].
\end{equation}

Similarly, we argue that the anti-derivative of $b_x^{(2)}$
gives rise to higher-order corrections to the Fredholm determinant.
This can be seen by the following observation.
Consider, for $a > 1$, $b \geq 0$, $c \in \mathbb{R} \backslash \{0\}$,
\begin{multline}
\label{eq:how-to-integrate}
    \partial_x \Bigg\{
        \frac{1}{ \big[\Delta^2 + h_1^+ h_1^-\big]^a}
        \frac{ (\ln x)^b \rme^{\rmi x c} }{ x }
    \Bigg\}
    =
    \frac{1}{ \big[\Delta^2 + h_1^+ h_1^-\big]^a}
    \frac{ (\ln x)^b \rme^{\rmi x c} }{ x }
    \Big(
        \rmi c + \Ocal(x^{-1})
    \Big)
    \\
    - \frac{
        2 a h_1^+ h_1^- (d_x (s_1^-) - d_x (s_1^+)) }{
        \big[ \Delta^2 + h_1^+ h_1^- \big]^{a + 1}}
    \frac{ (\ln x)^b \rme^{\rmi x c} }{ x }
\end{multline}
which suggests that we should first integrate
$[ \Delta^2 + h_1^+ h_1^- ]$
in the denominators
as well as the oscillating factors
(in the coefficients $h_1^+$, $h_1^-$, $b_{12}$, and $b_{21}$)
in the expression for $b_x^{(2)}$.

In the end of the day, the only term
in the expression for the logarithmic derivative of the Fredholm determinant
proportional to $x^{-1}$,
that contributes to the Fredholm determinant asymptotic expansion
after integration over $x$,
is the term $a_x^{(2)}$.
This term is responsible for the logarithmic correction
in the asymptotic expansion of the Fredholm determinant
as we showed in the case of no poles on the real axis
\eqref{eq:a-2-no-poles-log}.

Finally, we obtain
\begin{multline}
    \ln \det_{\Ccal_{\lambda_0}}
    (\id + \rmV)
    =
    C (\lambda_0)
    + \ln \big[ \Delta^2 + h_1^+ h_1^- \big]
    + a (x, \lambda_0)
    - \frac{\tau^2 (\lambda_0)}{2} \ln x
    \\
    + \frac{ 1 + \ocal(1) }{ x^{1/2} }
    \frac{\Delta^2}{\omega' (0 | \lambda_0)}
    \frac{1}{ \big[ \Delta^2 + h_1^+ h_1^- \big]}
    \Bigg[
        \frac{ h_1^+ b_{21} }{ (\lambda_0 - s_1^+)^2 }
        + \frac{ h_1^- b_{12} }{ (\lambda_0 - s_1^-)^2 }
    \Bigg]
\end{multline}
and hence
\begin{multline}
    \det_{\Ccal_{\lambda_0}}
    (\id + \rmV)
    =
    \exp\{C (\lambda_0) \}
    x^{ - \frac{\tau^2 (\lambda_0)}{2} }
    \exp\{a (x, \lambda_0) \}
    \\
    \times
    \Bigg\{
        \Delta^2 + h_1^+ h_1^-
        + \frac{ 1 + \ocal(1) }{ x^{1/2} }
        \frac{ \sqrt2 \Delta^2 }{ \sqrt{- u'' (\lambda_0)} }
        \Bigg[
            \frac{ h_1^+ b_{21} }{ (\lambda_0 - s_1^+)^2 }
            + \frac{ h_1^- b_{12} }{ (\lambda_0 - s_1^-)^2 }
        \Bigg]
    \Bigg\}.
\end{multline}
Here we also substituted
$u'' (\lambda_0) = - 2 (\omega' (0 | \lambda_0))^2$.

\subsection{The integration constant}
In order to fix the integration constant,
we note that we can slightly deform
the function $\nu (\lambda)$ to some function $\widetilde{\nu} (\lambda)$,
such that for this new function the poles have moved away from the real axis,
$\Scal \cap \mathbb{R} = \emptyset$,
and the integration contour can be deformed back
as in Figure~\ref{fig:configuration-of-the-poles}.
We denote the constants for the Fredholm determinants with $\nu$
and $\widetilde{\nu}$ as $C (\lambda_0)$ and $\widetilde{C} (\lambda_0)$,
respectively.
Since
\begin{equation}
    \ln \big[ \Delta^2 + h_1^+ h_1^- \big]
    \xrightarrow[\nu \to \widetilde{\nu}]{}
    \ln \Delta^2 + \Ocal(x^{-\infty}),
\end{equation}
we have that
\begin{equation}
    C (\lambda_0)
    \xrightarrow[\nu \to \widetilde{\nu}]{}
    \widetilde{C} (\lambda_0)
    - \ln \Delta^2,
\end{equation}
where
$\exp\{ \widetilde{C} (\lambda_0) \}
    = \eval{\Acal (\lambda_0)}_{\nu \to \tilde{\nu}}$
with $\Acal$ given by~\eqref{eq:thm:integration-constant}
due to Theorem~\ref{thm:Fredholm-det-AE-no-poles}.

Therefore, the asymptotic expansion of the Fredholm
determinant reads
\begin{multline}
    \det\limits_{\Ccal_{\lambda_0}} (\id + \rmV)
    =
    \exp\{C (\lambda_0) \}
    x^{ - \frac{\tau^2 (\lambda_0)}{2} }
    \exp\{a (x, \lambda_0) \}
    \\
    \times
    \Bigg\{
        1 + \frac{ h_1^+ h_1^- }{ (s_1^+ - s_1^-)^2 }
        + \frac{ 1 + \ocal(1) }{ x^{1/2} }
        \frac{\sqrt 2}{ \sqrt{- u'' (\lambda_0)} }
        \Bigg[
            \frac{ h_1^+ b_{21} (\lambda_0) }{ (\lambda_0 - s_1^+)^2 }
            + \frac{ h_1^- b_{12} (\lambda_0) }{ (\lambda_0 - s_1^-)^2 }
        \Bigg]
    \Bigg\},
\end{multline}
where the integration constant is given by the same expression
as in the case of no poles,
see equation~\eqref{eq:expression-for-integration-constant-final},
but with an integration contour $\Ccal_{\lambda_0}$ deformed
in the vicinity of the poles according to Assumption~\ref{as:number6}.
This completes the proof of Theorem~\ref{thm:Fredholm-det-AE-two-poles}.

\subsection{A Fredholm minor}
\label{sec:Fredholm-minor-two-poles}
Similarly as in the case of no poles on the real axis,
described in Section~\ref{sec:Fredholm-minor-no-poles},
we derive the asymptotic expansion of the Fredholm minor
considered in Proposition~\ref{prop:minor-of-Fredholm-det}
in the case of two poles.

As in the proof of Proposition~\ref{prop:minor-of-Fredholm-det-no-poles},
we calculate first the leading large-$|\lambda|$
asymptotics of $\widetilde{\chi}_{12} (\lambda)$,
\begin{equation}
\label{eq:chi-12-two-poles}
    \lim\limits_{\lambda \to \infty}
    \big[
        \lambda \cdot \widetilde{\chi}_{12} (\lambda)
    \big]
    =
    \lim\limits_{\lambda \to \infty}
    \big[
        \lambda \cdot \rmS_{12} (\lambda)
    \big]
    + \frac{ b_{12} (\lambda_0) }{ x^{1/2}\omega' (0 | \lambda_0) }
    + \Ocal \Bigg(
        \frac{
            (\ln x)^2 \rme^{\rmi u (\lambda_0)} }{
            x^{3 / 2 - \tau (\lambda_0)} }
    \Bigg).
\end{equation}
From the expression~\eqref{eq:matrix-S-new} for the matrix $\rmS (\lambda)$
it follows that
\begin{equation} \label{eq:S_large_lambda}
    \lim\limits_{\lambda \to \infty} 
    \big[
        \lambda \cdot \rmS_{12} (\lambda)
    \big]
    = \sum\limits_{j = 1}^{n^+}
    \sigma_j^+
    \cdot \big( \mathbf{y}_j \big)_1
    \cdot \Pi_{11} (s_j^+)
    - \sum\limits_{j = 1}^{n^-}
    \sigma_j^-
    \cdot \big( \mathbf{x}_j \big)_1
    \cdot \Pi_{12} (s_j^-).
\end{equation}
The right-hand side can be evaluated explicitly
in the case of two poles considered above.
Substituting expression~\eqref{eq:solution-system-for-two-poles}
into \eqref{eq:S_large_lambda}, we obtain
\begin{multline}
    \lim\limits_{\lambda \to \infty} 
    \big[
        \lambda \cdot \rmS_{12} (\lambda)
    \big]
    \\
    =
    C \big[
        \sigma_1^+
        \Pi_{11}^2 (s_1^+)
        -
        \sigma_1^-
        \Pi_{12}^2 (s_1^-)
    \big]
    + \frac{
        2 C
        \sigma_1^+ \sigma_1^-
        \big( P_{11}^{+-} \big)_{22} }{
        s_1^+ - s_1^- }
    \Pi_{11} (s_1^+) \Pi_{12} (s_1^-)
\end{multline}
with $C$ given by~\eqref{eq:coefficient-C}.
Here we insert the asymptotic expansions 
\eqref{eq:AE-for-P22},
\eqref{eq:AE-for-sigma-pm},
\eqref{eq:AE-for-C},
and~\eqref{eq:coefficient-Pi-1-explicit}.
We substitute the resulting expression back into
equation~\eqref{eq:chi-12-two-poles}
and multiply it by the asymptotic expansion
of the Fredholm determinant from Theorem~\ref{thm:Fredholm-det-AE-two-poles}.
Using Proposition~\ref{prop:minor-of-Fredholm-det},
we obtain the following result for the Fredholm minor.
\begin{proposition}
    \label{prop:minor-of-Fredholm-det-two-poles}
    Let Assumptions \ref{item:assumption-1}--\ref{as:number6}
    be satisfied.
    Let the sets $\Scal^\pm \cap {\mathbb R}$ both have cardinality one
    and set $\{s_1^\pm\} = \Scal^\pm \cap {\mathbb R}$,
    see Figure~\ref{fig:4-configurations-of-the-poles},
    and $(s_1^+ - s_1^-)^2 + h_1^+ h_1^- \neq 0$.
    Then the Fredholm minor composed of the integrable integral operator
    $\rmV$ with kernel \eqref{eq:kernel-V-as-scalar-product}
    and the projector $\rmP$ with kernel~\eqref{eq:P-kernel}
    has the large-$x$ asymptotic expansion
    \begin{multline}
        \label{eq:prop:Fredholm-det-minor-AE-two-poles}
        \Bigg[
            \int\limits_{\Ccal_{\lambda_0}}
            \frac{ \dd{\mu} }{ 2 \pi }
            e^{-2} (\mu)
            + \partial_\gamma
        \Bigg]
        \eval{
            \det\limits_{\Ccal_{\lambda_0}}
            \big(
                \id + \rmV
                + \gamma \rmP
            \big)
        }_{\gamma = 0}
        \\
        =
        \rmi
        h_1^+
        \cdot \Acal (\lambda_0)
        \ x^{ - \flatfrac{\tau^2 (\lambda_0)}{2}}
        \exp\Bigg\{
            - \int\limits_{\Ccal_{\lambda_0}}
            \dd{\lambda} \Lcal (\lambda | \lambda_0)
            \left( \rmi x u' (\lambda) + g' (\lambda) \right)
        \Bigg\}
        \\
        \times
        \Bigg\{
            1
            + \frac{ 1 }{ x^{1/2} }
            \frac{
                \sqrt{2} b_{12} (\lambda_0) }{
                \sqrt{- u'' (\lambda_0)} h_1^+ }
            \Bigg[
                1
                + \frac{ h_1^+ h_1^- (\lambda_0 - s_1^+)^2 }{
                    (\lambda_0 - s_1^-)^2 (s_1^+ - s_1^-)^2 }
            \Bigg]
            + \frac{
                \ocal\big(b_{12} (\lambda_0) \big)
                +
                \ocal\big(b_{21} (\lambda_0) \big) }{
                x^{1/2}}
        \Bigg\},
    \end{multline}
    where the functions $\tau$,
    $\Lcal (\cdot | \lambda_0)$
    were defined in~\eqref{eq:definition-of-tau-Lcal-and-Lcal-prime},
    the amplitude $\Acal (\lambda_0)$
    in~\eqref{eq:thm:integration-constant},
    and the coefficients $h_1^\pm$, $b_{12} (\lambda_0)$,
    and $b_{21} (\lambda_0)$
    in~\eqref{eq:definition-of-h-plus-and-h-minus}
    and~\eqref{eq:definition-of-b12-and-b21}.
\end{proposition}

\section[The case of \texorpdfstring{{$\mathbf{n}$}}{n} poles]%
    {The case of $\mathbf{n}$ poles}
\label{sec:n-poles-on-the-real-axis}
In this section we consider the case of $n$ poles in $\Omega$,
an arbitrary number of which may be located on the real axis.
First we prove Theorem~\ref{thm:Fredholm-det-AE-static}
for `the static case', i.e., for $\lambda_0 \to + \infty$.
Then we formulate a conjecture for the case of arbitrary
$\lambda_0 \in \mathbb{R}$.

\subsection{Static case}
\label{sec:static-case}
The static case corresponds to $\lambda_0 \to + \infty$.
Then, by Assumption~\ref{item:assumption-2},
$\lim_{\lambda_0 \rightarrow + \infty} u(\lambda|\lambda_0)
= p(\lambda)$, where $p'(\lambda) > 0$ for all 
$\lambda \in {\mathbb R}$. Hence, there is no saddle-point
in this limit, and what used to be the saddle-point
contribution reduces to 
\begin{equation}
    \Pi (\lambda)
    = \mathrm{I}_2
    + \Ocal
    \big(
        \rme^{- m x}
        \big(
            \begin{smallmatrix}
                1 & 1
                \\
                1 & 1
            \end{smallmatrix}
        \big)
    \big)
\end{equation}
for some $m > 0$.
Consequently, the coefficients $P$ and $Q$ simplify to
\begin{equation}
    \big( P_{jk}^{\epsilon \epsilon'} \big)_{11}
    = \big( P_{jk}^{\epsilon \epsilon'} \big)_{22} = 1 + \Ocal (\rme^{- m x}),
    \quad \big( P_{jk}^{\epsilon \epsilon'} \big)_{12}
    = \big( P_{jk}^{\epsilon \epsilon'} \big)_{21} = \Ocal (\rme^{- m x}),
\end{equation}
and
\begin{equation}
    \big( Q_{jk}^{\epsilon \epsilon'} \big)_{ln} = \Ocal (\rme^{- m x}),
    \qquad l, n = 1, 2,
\end{equation}
for $\epsilon, \epsilon' = \pm$. Then
the linear system~\eqref{eq:SLE-new} for $\mathbf{y}_j$ and $\mathbf{x}_j$
has the following off-diagonal block structure,
\begin{equation}
\label{eq:prop:SLE}
\left\{
\begin{aligned}
    \mathbf{x}_j
    - \sum\limits_{k = 1}^{n^+}
    \frac{ \eval{h_k^+}_{\lambda_0 = + \infty} }{ s_j^- - s_k^+ }
    \mathbf{y}_k
    &= \begin{pmatrix} 0 \\ 1 \end{pmatrix},
    \\
    \mathbf{y}_j
    - \sum\limits_{k = 1}^{n^-}
    \frac{ \eval{h_k^-}_{\lambda_0 = + \infty} }{ s_j^+ - s_k^- }
    \mathbf{x}_k
    &= \begin{pmatrix} 1 \\ 0 \end{pmatrix},
\end{aligned}
\right.
\end{equation}
and the contribution of the poles~\eqref{eq:prop:contribution-of-poles}
for $\beta = x$ significantly simplifies as well and reads
\begin{multline}
\label{eq:contribution-of-poles-static-case}
    \sum\limits_{\mu \in \mathcal{S}}
    \res\limits_{\lambda = \mu}
    \Big(
        \tr\{
            \rmS' (\lambda) \Pi (\lambda)
            \sigma^z
            \Pi^{-1} (\lambda) \rmS^{-1} (\lambda)
        \}
        \eval{ d_x (\lambda) }_{\lambda_0 = + \infty}
    \Big)
    \\
    = - \rmi
    \sum\limits_{j = 1}^{n^+}
    \sum\limits_{k = 1}^{n^-}
    \frac{
        \eval{ \partial_x \big( h_j^+ h_k^- \big) }_{\lambda_0 = + \infty} }{
        (s_j^+ - s_k^-)^2 }
    \mathbf{x}_k^\intercal \sigma^y \mathbf{y}_j 
    + \Ocal (\rme^{- m x}).
\end{multline}
Here we used that
\begin{equation}
    d_x = \partial_x \ln e (\lambda)
    = \frac{ \partial_x e (\lambda) }{ e (\lambda) }
    \quad
    \Rightarrow
    \quad
    \partial_x h_j^\pm (\lambda)
    = - 2 d_x (\lambda) h_j^\pm (\lambda).
\end{equation}

The expression~\eqref{eq:contribution-of-poles-static-case}
can be expressed as the derivative of the determinant of
the matrix encoding the linear system~\eqref{eq:SLE-static},
a result that was first obtained in \cite{S-10} in the
less general context of the static case with explicit
momentum function $p (\lambda) = \lambda$.
\begin{proposition}
\label{prop:Slavnov-formula-for-general-momentum}
    Let $\mathbf{y}_j$, $j = 1, \dots, n^+$,
    and $\mathbf{x}_j$, $j = 1, \dots, n^-$ be the two-dimensional vectors
    which satisfy the linear system
    \begin{equation}
    \label{eq:SLE-static}
    \left\{
    \begin{aligned}
        \mathbf{x}_j
        - \sum\limits_{k = 1}^{n^+}
        A_{jk}^+
        \mathbf{y}_k
        &= \begin{pmatrix} 0 \\ 1 \end{pmatrix},
        \\
        \mathbf{y}_j
        - \sum\limits_{k = 1}^{n^-}
        A_{jk}^-
        \mathbf{x}_k
        &= \begin{pmatrix} 1 \\ 0 \end{pmatrix},
    \end{aligned}
    \right.
    \end{equation}
    where
    the elements of the matrices $\rmA^\pm$
    are given by~\eqref{eq:matrices-A-pm}.
    Assume that $\rmA^\pm$ are $C^1$ functions of an auxiliary variable $x$
    and that the matrix $\rmI_{n^+} - \rmA$
    with $\rmA = \rmA^- \rmA^+$ is invertible.
    Then $\mathbf{x}_j$ and $\mathbf{y}_k$ are well-defined, and
    \begin{equation}
    \label{eq:prop:Slavnov-formula-for-general-momentum}
        \rmi
        \sum\limits_{j = 1}^{n^+}
        \sum\limits_{k = 1}^{n^-}
        \partial_x \big[ A_{kj}^+ A_{jk}^- \big]
        \cdot
        \mathbf{x}_k^\intercal \sigma^y \mathbf{y}_j
        = \partial_x \ln \det\big[ \mathrm{I}_{n^+} - \rmA \big].
    \end{equation}
\end{proposition}
The proof of this proposition involves combinatorial arguments
and is presented in Appendix~\ref{app:combinatorial-identity}.
We stress that the validity of the proposition is not restricted
to the static case. We shall use it again
in the next section,
when we formulate a conjecture on the general dynamical case
with $n$ poles on the real axis
that contribute to the large-$x$ asymptotics of the Fredholm determinant.

Due to Proposition~\ref{prop:Slavnov-formula-for-general-momentum},
in the static case, the logarithmic derivative
of the Fredholm determinant~\eqref{eq:prop:log-der-Fredholm-for-AE} reads
\begin{equation}
\label{eq:log-der-Fredholm-for-AE-static}
    \partial_x \ln \det_{\Ccal_\infty}
    (\id + \mathrm{V})
    = a_x (x, + \infty)
    + \partial_x \ln \det\big[ \mathrm{I}_{n^+} - \rmA_\infty \big]
    + \Ocal (\rme^{- m x}),
\end{equation}
where $a_x (x, + \infty)$ is given by~\eqref{eq:prop:function-a}
and where $\rmA_\infty = \eval{\rmA^- \rmA^+}_{\lambda_0 = + \infty}$
with $\rmA^\pm$ defined in~\eqref{eq:matrices-A-pm}.
Therefore, the Fredholm determinant can be expressed as
\begin{equation}
\label{eq:Fredholm-det-AE-n-poles-static}
    \det_{\Ccal_\infty}
    (\id + \mathrm{V})
    =
    \Acal_\infty \cdot \det\big[ \mathrm{I}_{n^+} - \rmA_\infty \big]
    \exp\Bigg\{
        - \int\limits_{\Ccal_\infty}
        \dd{\lambda} \Lcal_{\ell} (\lambda)
        \big[ \rmi x p' (\lambda) + g' (\lambda) \big]
    \Bigg\} \big( 1 + \Ocal (\rme^{-m x}) \big),
\end{equation}
where the integration constant $\Acal_\infty$ can be determined
as in the case of two poles.
Namely, we can slightly deform the function $\nu (\lambda)$
to some $\widetilde{\nu} (\lambda)$
such that all the poles move away from the real axis.
Then the determinant of the matrix
on the right-hand side of~\eqref{eq:Fredholm-det-AE-n-poles-static}
turns into $1 + \Ocal(x^{-\infty})$,
and we should recover the asymptotic expansion
in Theorem~\ref{thm:Fredholm-det-AE-no-poles}
as $\lambda_0 \to + \infty$ see,
e.g., equations~\eqref{eq:log-Fred-det-no-poles}
and~\eqref{eq:constant-at-plus-infty}.
Then
\begin{multline}
    \Acal_\infty = \lim\limits_{\lambda_0 \to + \infty} \Acal (\lambda_0)
    = \exp\Bigg\{
        \int\limits_{\Ccal_\infty}
        \dd{\lambda}
        \Lcal_\ell (\lambda) \partial_\lambda
        \ln\big( \rme^{-2 \pi \rmi \nu (\lambda)} - 1 \big)
    \Bigg\}
    \\
    \times
    \exp\Bigg\{
        \frac12
        \int\limits_{\Ccal_\infty}
        \dd{\lambda}
        \int\limits_{\Ccal_\infty}
        \dd{\mu}
        \frac{
            \Lcal_\ell' (\lambda) \Lcal_\ell (\mu)
            - \Lcal_\ell (\lambda) \Lcal_\ell' (\mu) }{
            \lambda - \mu}
    \Bigg\}.
\end{multline}
This concludes the proof of Theorem~\ref{thm:Fredholm-det-AE-static}.

\subsection{Conjecture}
\label{sec:conjecture}
In the general case of $n$ poles in $\Omega$ with $\lambda_0 \in \mathbb{R}$,
we have, to leading order, the same asymptotic expansion
for the contribution of the poles to the logarithmic
derivative with respect to $x$ of the Fredholm determinant
as in the static case,
\begin{multline}
\label{eq:contribution-of-poles-more-poles-on-the-real-axis}
    \sum\limits_{\mu \in \mathcal{S}}
    \res\limits_{\lambda = \mu}
    \Big(
        \tr\{
            \rmS' (\lambda) \Pi (\lambda)
            \sigma^z
            \Pi^{-1} (\lambda) \rmS^{-1} (\lambda)
        \}
        d_x (\lambda)
    \Big)
    \\
    =
    - 2 \rmi
    \sum\limits_{j = 1}^{n^+}
    \sum\limits_{k = 1}^{n^-}
    \frac{ h_j^+ h_k^- }{ (s_j^+ - s_k^-)^2 }
    \big[
        d_x (s_j^+) - d_x (s_k^-)
    \big]
    \mathbf{x}_k^\intercal
    \sigma^y
    \mathbf{y}_j
    + \ocal(1).
\end{multline}
Compare this expression
with the formula~\eqref{eq:contribution-of-poles-static-case}.
This expression follows from~\eqref{eq:prop:contribution-of-poles},
if we substitute the asymptotic expansions for the coefficients $\sigma_j^\pm$,
see equation~\eqref{eq:AE-for-sigma-pm},
and the asymptotic expansions for coefficients $P$ and $Q$
similar to expressions~\eqref{eq:AE-for-P-and-Q}, namely,
\begin{subequations}
\label{eq:AE-for-P-and-Q-more-poles}
\begin{align}
    \label{eq:AE-for-P12-more-poles}
    \big( P_{kj}^{--} \big)_{12}
    &= \frac1{ x^{1 / 2} }
    \big( \Pi_1' (s_j^-) \big)_{12}
    - \frac1{ x^{1 / 2} }
    \big( \Pi_1' (s_k^-) \big)_{12}
    + \Ocal \left(
        \frac{
            (\ln x)^2 \rme^{\rmi x u (\lambda_0)} }{
            x^{3 / 2 - \tau (\lambda_0)} }
    \right),
    \\
    \label{eq:AE-for-P21-more-poles}
    \big( P_{kj}^{++} \big)_{21}
    &= \frac1{ x^{1 / 2} }
    \big( \Pi_1' (s_j^+) \big)_{21}
    - \frac1{ x^{1 / 2} }
    \big( \Pi_1' (s_k^+) \big)_{21}
    + \Ocal \left(
        \frac{
            (\ln x)^2 \rme^{- \rmi x u (\lambda_0)} }{
            x^{3 / 2 + \tau (\lambda_0)} }
    \right),
    \\
    \label{eq:AE-for-P22-more-poles}
    \big( P_{jk}^{+ -} \big)_{22}
    & = 1 + \Ocal \left( \frac{\ln x}{x} \right),
    \\
    \label{eq:AE-for-Q12-more-poles}
    \big( Q_{kj}^{--} \big)_{12}
    &= \frac1{ x^{1 / 2} }
    \big( \Pi_1' (s_j^-) \big)_{12}
    + \Ocal \left(
        \frac{
            (\ln x)^2 \rme^{\rmi x u (\lambda_0)} }{
            x^{3 / 2 - \tau (\lambda_0)} }
    \right),
    \\
    \label{eq:AE-for-Q21-more-poles}
    \big( Q_{kj}^{++} \big)_{21}
    &= \frac1{ x^{1 / 2} }
    \big( \Pi_1' (s_j^+) \big)_{21}
    + \Ocal \left(
        \frac{
            (\ln x)^2 \rme^{- \rmi x u (\lambda_0)} }{
            x^{3 / 2 + \tau (\lambda_0)} }
    \right),
    \\
    \label{eq:AE-for-Q11-and-Q22-more-poles}
    \big( Q_{kj}^{-+} \big)_{11}
    & = \Ocal \left( \frac{\ln x}{x} \right)
    = \big( Q_{jk}^{+-} \big)_{22}.
\end{align}
\end{subequations}

Also, as we have shown above in the case of two poles,
the integral over the contour $\gamma_0$
involving the matrix $\rmS (\lambda)$
contributes only to the next-order correction
of the large-$x$ asymptotic expansion of the Fredholm determinant,
see equation~\eqref{eq:tr-S-and-Pi-for-integral-over-gamma-0}.
Thus, we straightforwardly obtain
\begin{multline}
    \label{eq:Fred-det-log-der-series-in-x-n-poles}
    \partial_x \ln \det_{\Ccal_{\lambda_0}}
    (\id + \rmV)
    \\
    = a_x (x, \lambda_0)
    - 2 \rmi
    \sum\limits_{j = 1}^{n^+}
    \sum\limits_{k = 1}^{n^-}
    \frac{ h_j^+ h_k^- }{ (s_j^+ - s_k^-)^2 }
    \big[
        d_x (s_j^+) - d_x (s_k^-)
    \big]
    \mathbf{x}_k^\intercal
    \sigma^y
    \mathbf{y}_j
    + \frac{ a_x^{(2)} }{x}
    + \ocal(1),
\end{multline}
where $a_x^{(2)}$ is given by~\eqref{eq:a-2-no-poles-log}.
The essence of the conjecture that we formulate below is
that we obtain the leading order asymptotics
\begin{equation}
\label{eq:Fredholm-det-AE-n-poles}
    \det_{\Ccal_{\lambda_0}}
    (\id + \rmV)
    = \exp\{ C (\lambda_0) \}
    \det\big[ \mathrm{I}_{n^+} - \rmA \big]
    x^{- \frac{\tau^2 (\lambda_0)}{2}}
    \exp\{ a (x, \lambda_0) \}
    (1 + \ocal(1)),
\end{equation}
where $\rmA = \rmA^- \rmA^+$
with the matrices $\rmA^\pm$
given by~\eqref{eq:matrices-A-pm},
employing again Proposition~\ref{prop:Slavnov-formula-for-general-momentum}.
The sub-leading corrections $\ocal(1)$ remain sub-leading
after the integration over $x$,
due to a relations like~\eqref{eq:how-to-integrate},
except for the term $a_x^{(2)} / x$,
which is responsible for the logarithmic corrections
$x^{- \flatfrac{\tau^2 (\lambda_0)}{2}}$.
Finally, the constant $C (\lambda_0)$ can be fixed,
as in the case of two poles.

\begin{conjecture}
    \label{conj:Fredholm-det-AE-n-poles}
    Let Assumptions \ref{item:assumption-1}--\ref{as:number6}
    be fulfilled. Then the Fredholm determinant
    of the integrable integral operator $\rmV$
    with kernel~\eqref{eq:kernel-V-as-scalar-product}
    has the large-$x$ asymptotic expansion
    \begin{multline}
        \label{eq:conj:Fredholm-det-AE-n-poles}
        \det_{\Ccal_{\lambda_0}} (\id + \rmV)
        =
        \Acal (\lambda_0)
        \det\big[ \mathrm{I}_{n^+} - \rmA \big]
        \ x^{ - \flatfrac{\tau^2 (\lambda_0)}{2}}
        \\
        \times
        \exp\Bigg\{
            - \int\limits_{\Ccal_{\lambda_0}}
            \dd{\lambda} \Lcal (\lambda | \lambda_0)
            \left( \rmi x u' (\lambda) + g' (\lambda) \right)
        \Bigg\}
        \big( 1 + \ocal (1) \big),
    \end{multline}
    where the constant $\Acal (\lambda_0)$
    is given by~\eqref{eq:thm:integration-constant} and
    the function $\Lcal (\cdot | \lambda_0)$
    by~\eqref{eq:definition-of-tau-Lcal-and-Lcal-prime}. The leading
    order contribution from the poles in $\Scal^\pm$,
    cf.~\eqref{eq:definition-Scal-pm}, \eqref{eq:definition-Scal-pm-set},
    is comprised in the matrix $\rmA = \rmA^- \rmA^+$,
    see expression~\eqref{eq:matrices-A-pm}.
    The matrix $\rmI_{n^+} - \rmA$ is assumed to be invertible.
    If $n^+ = 0$ or $n^- = 0$
    then the determinant on the right-hand side of
    \eqref{eq:conj:Fredholm-det-AE-n-poles} is equal to $1$ by
    definition.
\end{conjecture}

\section{Discussion}
\label{sec:discussion}
We have analyzed the large-$x$ asymptotic behaviour of the
Fredholm determinant of an integrable integral operator
whose kernel depends on four functional parameters (subject to
Assumptions~\ref{item:assumption-1}-\ref{as:number6}) and is
a generalization of the famous sine kernel
\cite{DeiftItsZhouSineKernelOnUnionOfIntervals} that appears
in various branches of mathematical physics. Our results are
stated in three theorems in Section~\ref{sec:main_results}.
They include explicit expressions for the constant terms
of the asymptotics in the cases of no pole or two poles on
the real axis. In the case of two poles on the real axis we 
have also calculated the first sub-leading corrections in
explicit form. For the most general case of an arbitrary
number of poles on the real axis we have presented a 
conjecture for the explicit form of the constant term 
in Section~\ref{sec:conjecture}.

We were mainly motivated by the perspective to apply our results
to the asymptotic analysis of certain dynamical two-point
correlation functions of the Lieb--Liniger Bose gas at arbitrary
positive coupling in thermal and non-thermal equilibrium using the
method developed in~\cite{KoTe11,Kozlowski15b,KMS11a}. This
will be the subject of subsequent work.

The first and immediate application, however, will be to the
impenetrable Bose gas \cite{M-25}. An asymptotic analysis of the long-time,
large-distance asymptotic behaviour of the two-point functions
of fields of this model in thermal equilibrium was accomplished
in~\cite{IIKV-92}. In a separate publication we shall reconsider
this subject in a slightly more general context. Our functional
parameter $\vartheta$ will play the role of the so-called
filling fraction. In thermal equilibrium the filling fraction
is equal to the Fermi function~\eqref{eq:Fermi_distribution}.
In our analysis the choice of the function $\vartheta$ is only
restricted by Assumptions~\ref{item:assumption-1}--\ref{as:number6}
which allows us to consider a large class of non-thermal equilibria
characterized, e.g., by generalized Gibbs ensembles. Moreover, as
compared to~\cite{IIKV-92} we have obtained the constant term in
the asymptotics completely and in a nice and compact form.

Our results may also be applied to improve the existing
asymptotic results for a number of generalizations of the
impenetrable Bose gas including the gas of impenetrable
anyons~\cite{OKA-10,P-15} or the two-component gases of
impenetrable Bosons or Fermions~\cite{IP-98}. The asymptotic
analysis of the two-component Fermi gas was carried out
in~\cite{GIK-98} for the case of negative chemical potential
and in~\cite{CZ-04} for the case of positive chemical
potential. Our improvements of the representation of the
constant term and the sub-leading corrections will be
relevant in both cases.

A related problem is that of the asymptotic analysis of
the transverse correlation functions of the XX model at
finite temperature~\cite{CIKT92,CIKT93,IIKS93b}. The
space-like asymptotics, including the constant term has
been obtained in~\cite{GKS20b} directly from a thermal
form factor series. In the time-like
regime the constant term is only known for non-critical
external fields~\cite{Jie98}. The case of dynamical
correlation functions for critical values of the
field strength is still partially open. It should be treatable
by Riemann--Hilbert techniques similar to those in the
present work.

\section*{Acknowledgement}
The authors would like to thank Sergei Adler and Oleksandr Gamayun for
stimulating discussions. They are very grateful to Alexander Wei{\ss}e
for providing his numerical data for the Fredholm minor related to
the impenetrable Bose gas for comparison. FG and MM thank DFG and ANR
for financial support under grant numbers Go825/12-1 and Go825/13-1.
KKK gratefully acknowledges funding by an ANR-DFG TFS24 grant number
ANR-24-CE-92-0033-0.

\clearpage

 \clearpage

\appendix
\section{Special functions}
\label{app:special-functions}

In this appendix
we collect the definitions of the special functions
needed in the main text and list some of their properties.

First, we introduce
the dilogarithm function, $\Li_2 (z)$,
and the Barnes $G$-function, $G(z)$,
which admit the following integral representations:
\begin{equation}
\label{eq:dilogarithm}
    \Li_2 (z)
    = - \int\limits_0^z
    \dd{t}
    \frac{ \ln (1 - t) }{t},
    \qquad
    z \notin [1, \infty)
\end{equation}
and
\begin{equation}
\label{eq:Barnes-G-function}
    \log G (z + 1)
    = \frac{z (1 - z)}{2}
    + \frac{z}{2} \log (2 \pi)
    + \int\limits_0^z
    \dd{t}
    t \psi (t),
    \qquad
    \Real (z) > - 1,
\end{equation}
where $\psi$ is the digamma functions.

\subsection{Parabolic cylinder functions}
\label{app:parabolic-cylinder-functions}
The solutions of the differential equation
\begin{equation}
\label{eq:diff-equation}
	y'' (z) + \left(\nu + \frac12 - \frac{ z^2 }{4} \right) y (z) = 0
\end{equation}
are called \textit{parabolic cylinder functions},
see~\cite{WhWa63ch16}, for example.
We define the parabolic cylinder function $D_\nu$
in terms of confluent hypergeometric function
\begin{equation}
\label{eq:parabolic-cylinder-function-def}
	D_\nu (z)
	= 2^{ \flatfrac{ (\nu - 1) }{2} }
	\rme^{- \flatfrac{z^2}4}
	z
	\
	\Psi \left(
		\frac{(1 - \nu)}{2}, \frac32; \frac{z^2}{2}
	\right).
\end{equation}
There are four (linearly dependent) functions
satisfying equation~\eqref{eq:diff-equation},
namely,
$D_\nu (z)$, $D_\nu (- z)$, $D_{- \nu - 1} ( \rmi z)$,
and $D_{- \nu - 1} ( - \rmi z)$.

The asymptotic series for the parabolic cylinder functions $D_\nu$
for large values of $\abs{z}$ is given by
\cite{WhWa63ch16}
\begin{equation}
\label{eq:AE-parabolic-cylinder-function}
    D_\nu (z)
    = z^\nu \mathrm{e}^{- \frac{z^2}4}
    \left[
        \sum\limits_{n = 0}^N
        \frac{
            (- \nu)_{2n} }{
            n! (- 2 z^2)^n}
        +\mathrm{O} ( \abs{z}^{- 2 (N + 1)} )
    \right],
    \quad
    \abs{\arg z} < \flatfrac{3 \pi}{4},
\end{equation}
where $(a)_n$ denotes the Pochhammer symbol (the rising factorial),
\begin{equation}
\label{eq:Pochhammer-symbol}
	(a)_n
	= \begin{cases}
		a \cdot (a + 1) \cdots (a + n - 1),
		& n \in \mathbb{N},
		\\
		1, & n = 0.
	\end{cases}
\end{equation}
There is also the useful relation
\begin{equation}
\label{eq:parabolic-cylinder-function-relation}
    D_\nu (z)
    = \mathrm{e}^{\pi \mathrm{i} \nu}
    D_\nu (-z)
    + \frac{ \sqrt{2 \pi} }{\Gamma (- \nu)}
    \mathrm{e}^{ \frac{ \pi \mathrm{i} (\nu + 1) }{ 2 }}
    D_{- \nu - 1} (- \mathrm{i} z),
\end{equation}
which allows us to consider the large-$|z|$ asymptotic behaviour of the
parabolic cylinder functions for $z$ outside of the region $\abs{\arg z}
< \flatfrac{3 \pi}{4}$.

\section{Fredholm determinant -- convergence and static limit}
\label{app:existence_of_Fred}
In this appendix we show that the series
\begin{equation}
    1 + \sum_{n=1}^\infty \frac{1}{n!}
    \int\limits_{{\mathbb R}^n} \dd^n \lambda \:
       \det_{n} \bigl\{V(\lambda_j, \lambda_k)\bigr\}
\end{equation}
which defines the Fredholm determinant~\eqref{eq:Fred} with
integration kernel~\eqref{eq:kernel-V-as-scalar-product} is
absolutely convergent, pointwise in $\lambda_0$, once our
basic assumptions listed in Section~\ref{sec:main_results}
are satisfied. We then consider the static limit
$\lambda_0 \rightarrow + \infty$.

Upon inserting the explicit expressions \eqref{eq:vectors-E} into
the defining equation~\eqref{eq:kernel-V-as-scalar-product}, we
obtain the following representation of the kernel,
\begin{equation} \label{eq:fred_series}
    V(\lambda,\mu)
    = \vartheta^\frac{1}{2} (\mu) \vartheta^{- \frac{1}{2}} (\lambda)
    \bigl
        (V_{S1} (\lambda,\mu) + V_{S2} (\lambda,\mu)
    \bigr),
\end{equation}
where
\begin{align} \label{eq:def_vs1}
    V_{S1} (\lambda, \mu) & = \rme^{\rmi \pi(\nu(\lambda) + \nu(\mu))}
      \vartheta^\frac{1}{2} (\mu) \vartheta^\frac{1}{2} (\lambda) 
      e^{-1} (\lambda) e^{-1} (\mu) \notag \\ & \mspace{180.mu}
      \frac{e^2 (\lambda) \bigl(\rme^{- 2 \pi \rmi \nu (\lambda)} - 1\bigr)
            - e^2 (\mu) \bigl(\rme^{- 2 \pi \rmi \nu (\mu)} - 1\bigr)}
           {2 \pi \rmi(\lambda - \mu)}, \\[1ex] \label{eq:def_vs2}
    V_{S2} (\lambda, \mu) & = \frac{2 \rmi}{\pi} 
       \sin \bigl(\pi(\nu(\lambda)\bigr) \sin \bigl(\pi(\nu(\mu)\bigr)
       \vartheta^\frac{1}{2} (\mu) \vartheta^\frac{1}{2} (\lambda) e(\lambda) e(\mu)
       \frac{\rmC_- (\lambda) - \rmC_- (\mu)}{\lambda - \mu}
\end{align}
are symmetric.

We want to establish bounds on
\begin{equation}
    \det_n \bigl(V(\lambda_j, \lambda_k)\bigr)
       = \det_n \bigl(V_{S1} (\lambda_j, \lambda_k)
           + V_{S2} (\lambda_j, \lambda_k)\bigr).
\end{equation}
For this purpose we start with
\begin{lemma} \label{lem:u_esti}
Fix $\sigma > 0$. Then, for some $K > 0$, we have the estimate
\begin{multline}
    \biggr|\frac{\rmC_- (\lambda) - \rmC_- (\mu)}{\lambda - \mu}\biggl| \leq
       K \biggl[\biggl\{1 + \sum_{\eta \in \{\lambda, \mu\}}
                           \sup_{\nu \in [- \sigma/4, \sigma/4]}
                           \bigl|u'(\eta + \nu|\lambda_0)\bigr|\biggr\}
        \times \indicator{|\lambda - \mu| \geq \sigma} \\
       + \sup_{\nu \in [- 2 \sigma, 2 \sigma]}
          \bigl\{1 + \bigl|u' (\lambda + \nu|\lambda_0)\bigr|
                   + \bigl|u'(\lambda + \nu|\lambda_0)^2\bigr|
                   + \bigl|u'' (\lambda + \nu|\lambda_0)\bigr|\bigr\}
          \times \indicator{|\lambda - \mu| < \sigma} \biggr].
\end{multline}
\end{lemma}
\begin{proof}
Fix $\sigma > 0$ and let $|\lambda - \mu| < \sigma$. We rewrite $\rmC_-$
in such a way that we can easily take its derivative,
\begin{multline} \label{eq:decompose_cminus}
    \rmC_- (\lambda) = - \frac{e^{-2} (\lambda)}{2}
                      + \Pcal \int\limits_{\mathbb R} \frac{\dd \mu}{2 \pi \rmi}
                           \frac{e^{-2} (\mu)}{\mu - \lambda} \\[1ex]
                     = - \frac{e^{-2} (\lambda)}{2}
                       + \int\limits_{{\mathbb R} \setminus [- \delta,\delta]}
                           \frac{\dd \mu}{2 \pi \rmi}
                           \frac{e^{-2} (\mu + \lambda)}{\mu}
                       + \int\limits_{- \delta}^\delta \frac{\dd \mu}{2 \pi \rmi}
                           \frac{e^{-2} (\mu + \lambda) - e^{-2} (\lambda)}{\mu},
\end{multline}
where $\delta > 0$. In order to see that the second equation
holds, replace the last (regular) integral in the second
line by its principal value and observe that the principal
value of the second term under the integral sign is equal to
zero. Equation~\eqref{eq:decompose_cminus} implies that
\begin{multline}
    \rmC_-' (\lambda) = - \frac{g' (\lambda) e^{-2} (\lambda)}{2} 
       - \frac{e^{-2} (\lambda + \delta) + e^{-2} (\lambda - \delta)}{\delta}
       + \int\limits_{{\mathbb R} \setminus [- \delta,\delta]}
                           \frac{\dd \mu}{2 \pi \rmi}
                           \frac{e^{-2} (\mu + \lambda)}{\mu^2} \\[1ex]
       - \frac{\rmi x u' (\lambda|\lambda_0) e^{-2} (\lambda)}{2}
       + \int\limits_{- \delta}^\delta \frac{\dd \mu}{2 \pi \rmi} \partial_\lambda
                           \frac{e^{-2} (\mu + \lambda) - e^{-2} (\lambda)}{\mu}.
\end{multline}
Hence, since $g'$ and $e^{-2}$ are bounded on $\mathbb R$,
\begin{equation}
     \bigl|\rmC_-' (\lambda)\bigr| \leq c_0 + \frac{x |u' (\lambda|\lambda_0)|}{2}
       + \int\limits_{- \delta}^\delta \frac{\dd \mu}{2 \pi} \biggl|\partial_\lambda
                           \frac{e^{-2} (\mu + \lambda) - e^{-2} (\lambda)}{\mu}\biggr|
\end{equation}
for some $c_0 > 0$, which we can further estimate as
\begin{multline}
    \bigl| \rmC_-' (\lambda) \bigr|
    \leq c_0 + \frac{x |u' (\lambda|\lambda_0)|}{2}
    + \frac{\delta}{\pi}
    \sup_{\nu \in [- \delta, \delta]}
    \bigl| \partial_\nu^2 e^{-2} (\lambda + \nu) \bigr|
    \\
    \leq c_0 + \frac{x |u' (\lambda|\lambda_0)|}{2}
    + c_1 \sup_{\nu \in [- \delta, \delta]}
    \bigl\{
        \bigl| u'(\lambda + \nu | \lambda_0) \bigr|
        + \bigl| u''(\lambda + \nu | \lambda_0) \bigr|
        + \bigl| u'(\lambda + \nu | \lambda_0)^2 \bigr|
    \bigr\}\\
    \leq c \sup_{\nu \in [- \delta, \delta]}
    \bigl\{
    1 + \bigl|u' (\lambda + \nu|\lambda_0)\bigr| + \bigl|u'(\lambda + \nu|\lambda_0)^2\bigr|
    + \bigl|u'' (\lambda + \nu|\lambda_0)\bigr|\bigr\},
\end{multline}
where $c, c_1 > 0$ and independent of $\lambda_0$. Choosing
now $\delta = \sigma$ we see that
\begin{multline}
    \biggr|\frac{\rmC_- (\lambda) - \rmC_- (\mu)}{\lambda - \mu}\biggl| \leq
    \sup_{\nu \in [-\sigma, \sigma]} \bigl|\rmC_-' (\lambda + \nu)\bigr|
    \\
    \leq
    c \sup_{\nu \in [- 2 \sigma, 2 \sigma]}
    \bigl\{
        1
        + \bigl|u' (\lambda + \nu|\lambda_0)\bigr|
        + \bigl|u'(\lambda + \nu|\lambda_0)^2\bigr|
        + \bigl|u'' (\lambda + \nu|\lambda_0)\bigr|
    \bigr\}
\end{multline}
which entails the claim for $|\lambda - \mu| < \sigma$.

If $|\lambda - \mu| \geq \sigma$, we start with
\begin{equation} \label{eq:ct_ext}
    \frac{\rmC_- (\lambda) - \rmC_- (\mu)}{\lambda - \mu} =
       - \frac{e^{-2} (\lambda) - e^{-2} (\mu)}{2(\lambda - \mu)}
       + \Pcal \int\limits_{\mathbb R} \frac{\dd \nu}{2 \pi \rmi}
                  \frac{e^{-2} (\nu)}{(\nu - \lambda)(\nu - \mu)}.
\end{equation}
The first term on the right hand
side of \eqref{eq:ct_ext} is bounded by $c/\sigma$ for
some $c > 0$. Let $\delta = \sigma/4$. In order to estimate
the principal value integral we split it as
\begin{equation}
    \Pcal \int\limits_{\mathbb R} \frac{\dd \nu}{2 \pi \rmi}
                  \frac{e^{-2} (\nu)}{(\nu - \lambda)(\nu - \mu)} =
       \Delta C_{\rm in} (\lambda, \mu) + \Delta C_{\rm in} (\mu, \lambda) +
       \Delta C_{\rm out} (\lambda, \mu),
\end{equation}
where
\begin{subequations}
\begin{align}
    \Delta C_{\rm in} (\lambda, \mu) & =
       \Pcal \int\limits_{- \delta}^\delta \frac{\dd \nu}{2 \pi \rmi}
                \frac{e^{-2} (\nu + \lambda)}{\nu(\nu + \lambda - \mu)}, \\[1ex]
    \Delta C_{\rm out} (\lambda, \mu) & =
       \int\limits_{{\mathbb R} \setminus \{[-\delta,\delta] + \lambda\}
             \setminus \{[-\delta,\delta] + \mu\}} \frac{\dd \nu}{2 \pi \rmi}
             \frac{e^{-2} (\nu)}{(\nu - \lambda)(\nu - \mu)}.
\end{align}
\end{subequations}
Then
\begin{multline}
    \bigl|\Delta C_{\rm in} (\lambda, \mu)\bigr|
    = \biggl|
        \int\limits_{- \delta}^\delta
        \frac{\dd \nu}{2 \pi \rmi} \frac{1}{\nu}
        \biggl\{
            \frac{e^{-2} (\nu + \lambda)}{\nu + \lambda - \mu}
            - \frac{e^{-2} (\lambda)}{\lambda - \mu}
        \biggr\}
    \biggr|
    \\[1ex]
    \leq \frac{\delta}{\pi} \sup_{\nu \in [-\delta,\delta]}
    \biggl|
        \partial_\nu \frac{e^{-2} (\lambda + \nu)}{\lambda - \mu + \nu}
    \biggr|
    \leq c \biggl\{
        1 
        + \sup_{\nu \in [- \sigma/4, \sigma/4]}
        \bigl|u' (\lambda + \nu|\lambda_0) \bigr|
    \biggr\}
\end{multline}
for some $c > 0$.

For the remaining contribution $\Delta C_{\rm out}$ we
note that $\{[- \delta, \delta] + \lambda\} \cap
\{[- \delta,\delta] + \mu\} = \emptyset$,
since $|\lambda - \mu| \geq \sigma = 4 \delta$.
Thus, assuming (w.l.o.g.) that $\lambda < \mu$, we obtain
\begin{multline} \label{eq:delta_c_out}
    \bigl|\Delta C_{\rm out} (\lambda, \mu)\bigr| \leq
       c_0 \biggl\{\int\limits_{- \infty}^{\lambda - \delta}
                 - \int\limits_{\lambda + \delta}^{\mu - \delta}
                 + \int\limits_{\mu + \delta}^\infty\biggr\} \dd \nu \:
                 \frac{1}{(\nu - \lambda)(\nu - \mu)} \\
       = \frac{2 c_0}{\sigma} \frac{1}{(\mu - \lambda)/\sigma}
         \ln \bigl(
            \bigl( 4 (\mu - \lambda)/\sigma \bigr)^2 - 1
        \bigr)
         \leq c
\end{multline}
for some $c_0, c > 0$. The last inequality follows from
the fact that the function $\ln(4 x^2 - 1)/x$ is bounded
for $x > 1$. Finally, equations \eqref{eq:ct_ext}-\eqref{eq:delta_c_out}
imply the claim for $|\lambda - \mu| \ge \sigma$.
\end{proof}

\begin{lemma} \label{lem:detn_esti}
If the functions $u$, $g$, $\nu$, and $\vartheta$ fulfill
Assumptions \ref{item:assumption-1}, \ref{item:assumption-on-nu-and-g},
and \ref{item:assumption-on-vartheta} in Section~\ref{sec:main_results},
then the integrands in \eqref{eq:fred_series} are bounded as
\begin{equation}
    \bigl|\det_n \bigl(V(\lambda_j, \lambda_k)\bigr)\bigr|
       \leq c^n n^\frac{n}{2} \prod_{j=1}^n \vartheta^\frac{1}{4} (\lambda_j)
\end{equation}
for some $c > 0$.
\end{lemma}
\begin{proof}
We separately estimate the kernels $V_{S1}$ and $V_{S2}$
(cf.~\eqref{eq:def_vs1}, \eqref{eq:def_vs2}). For this
purpose fix $\sigma > 0$. Let
\begin{equation}
    f(\lambda) = e^2 (\lambda) \bigl(\rme^{- 2 \pi \rmi \nu (\lambda)} - 1\bigr)
\end{equation}
For $|\lambda - \mu| < \sigma$ we have the estimate
\begin{equation} \label{eq:diff_quot_f_estimate}
    \biggl|\frac{f(\lambda) - f(\mu)}{\lambda - \mu}\biggr|
    \leq \sup_{\nu \in [-\sigma, \sigma]} \bigl|f'(\lambda + \nu)\bigr|
    \leq c \sup_{\nu \in [-\sigma, \sigma]}
    \bigl\{1 + \bigl|u'(\lambda + \nu | \lambda_0)\bigr|\bigr\}
\end{equation}
for some $c > 0$, since $e$, $g'$, $\nu$, and $\nu'$ are
bounded. For $|\lambda - \mu| \geq \sigma$ the boundedness
of $e$ and $\nu$ implies that
\begin{equation} \label{eq:vs1_greater_sigma_estimate}
    \bigl|V_{S1} (\lambda, \mu)\bigr| \leq \frac{c}{\sigma} 
       \vartheta^\frac{1}{2} (\lambda) \vartheta^\frac{1}{2} (\mu)
\end{equation}
for some $c > 0$.

Then, using  \eqref{eq:diff_quot_f_estimate},
\eqref{eq:vs1_greater_sigma_estimate}, and 
Assumption~\ref{item:assumption-on-vartheta} we see that there
exist $c, c_0, c_1 > 0$ such that
\begin{multline}
    \bigl|V_{S1} (\lambda, \mu)\bigr| \leq \\ 
       \vartheta^\frac{1}{2} (\lambda) \vartheta^\frac{1}{2} (\mu) \biggl\{
       \frac{c_0}{\sigma}
       \indicator{|\lambda - \mu| \geq \sigma}
    + c_1 \sup_{\nu \in [-\sigma, \sigma]}
    \bigl\{1 + \bigl|u'(\lambda + \nu | \lambda_0)\bigr|\bigr\}
      \indicator{|\lambda - \mu| < \sigma}\biggr\}
    \leq c \vartheta^\frac{1}{2} (\lambda).
\end{multline}
Due to the boundedness of $\vartheta$ on $\mathbb R$, there
are $c, c' > 0$ such that
\begin{equation} \label{eq:vs1_esti}
    \bigl|V_{S1} (\lambda, \mu)\bigr|
       \leq c \vartheta^\frac{1}{4} (\lambda) \vartheta^\frac{1}{4} (\lambda)
       \leq c' \vartheta^\frac{1}{4} (\lambda)
\end{equation}
for all $(\lambda, \mu) \in {\mathbb R}^2$.

Using Lemma~\ref{lem:u_esti} and the boundedness of $e$ and
$\nu$ we obtain an estimate for $V_{S2}$,
\begin{multline}
    \bigl|V_{S2} (\lambda, \mu)\bigr|
    \leq
    \vartheta^\frac{1}{2} (\lambda) \vartheta^\frac{1}{2} (\mu)
    K \biggl[
        \sup_{\nu \in [- \sigma/4, \sigma/4]}
        \biggl\{
            1 + \sum_{\eta \in \{\lambda, \mu\}}
            \bigl|u'(\eta + \nu|\lambda_0)\bigr|
        \biggr\}
        \times \indicator{|\lambda - \mu| \geq \sigma} \\
        + \sup_{\nu \in [- 2 \sigma, 2 \sigma]}
        \bigl\{
            1 + \bigl|u' (\lambda + \nu|\lambda_0)\bigr|
            + \bigl|u'(\lambda + \nu|\lambda_0)^2\bigr|
            + \bigl|u'' (\lambda + \nu|\lambda_0)\bigr|
        \bigr\}
        \times \indicator{|\lambda - \mu| < \sigma}
    \biggr],
\end{multline}
holding for a fixed $K > 0$. Using once more
Assumption~\ref{item:assumption-on-vartheta} we see that
\begin{equation} \label{eq:vs2_esti}
    \bigl|V_{S2} (\lambda, \mu)\bigr| 
        \leq c' \vartheta^\frac{1}{4} (\lambda) \vartheta^\frac{1}{4} (\mu)
        \leq c \vartheta^\frac{1}{4} (\lambda)
\end{equation}
for some $c', c > 0$.

Employing \eqref{eq:vs1_esti}, \eqref{eq:vs2_esti}
we infer that there exists a $c > 0$ such that
\begin{multline}
    \bigl|\det_n \bigl(V(\lambda_j, \lambda_k)\bigr)\bigr|
       = \bigl|\det_n \bigl(V_{S1} (\lambda_j, \lambda_k) +
                            V_{S2} (\lambda_j, \lambda_k) \bigr)\bigr| \\
       \leq \prod_{j=1}^n \biggl(\sum_{k=1}^n \bigl|V_{S1} (\lambda_j, \lambda_k) +
                            V_{S2} (\lambda_j, \lambda_k) \bigr|^2\biggr)^\frac{1}{2}
       \leq c^n n^\frac{n}{2} \prod_{j=1}^n \vartheta^\frac{1}{4} (\lambda_j)
\end{multline}
which concludes the proof.
\end{proof}
\begin{proposition}
\begin{enumerate}
\item 
If the functions $u$, $g$, $\nu$, and $\vartheta$ satisfy
Assumptions \ref{item:assumption-1}, \ref{item:assumption-on-nu-and-g},
and \ref{item:assumption-on-vartheta}, the series \eqref{eq:fred_series}
is absolutely convergent.
\item
If, in addition, Assumption~\ref{item:assumption-2} holds true, it
further follows that
\begin{equation} \label{eq:prop_stat_lim}
    \lim_{\lambda_0 \rightarrow + \infty} \det_{\mathbb R} (\id + \rmV)
    = \det_{\mathbb R} (\id + \rmV_0 + \delta \rmV_0)
    = \det_{\mathbb R} (\id + \rmV_0)
    \bigl(1 + \Ocal(x^{- \infty})\bigr),
\end{equation}
where $\rmV_0$ is the integral operator with kernel $V_0 (\lambda, \mu)$
defined in \eqref{eq:V-GSK} and where the integral operator
$\delta V_0$ is defined by its kernel
\begin{equation} \label{eq:delta_v_zero}
    \delta V_0 (\lambda, \mu) = - \frac{1}{2\pi \rmi}
       F^\frac{1}{2} (\lambda) F^\frac{1}{2} (\mu) e_0 (\lambda) e_0 (\mu)
       \int\limits_{{\mathbb R} + \rmi \varepsilon} \frac{\dd \nu}{2 \pi \rmi} 
       \frac{\rme^{\rmi x p (\nu) + g (\nu)}}{(\nu - \lambda)(\nu - \mu)}
\end{equation}
with $F$ and $e_0$ as in \eqref{eq:def_F}, \eqref{eq:def_ezero}
and $\varepsilon > 0$ small.
\end{enumerate}
\end{proposition}
\begin{proof}
(i) Employing Lemma~\ref{lem:detn_esti} and
Assumption~\ref{item:assumption-on-vartheta} we see that there are
a $c_0, c > 0$ such that
\begin{multline}
    \biggl|1 + \sum_{n=1}^\infty \frac{1}{n!}
    \int\limits_{{\mathbb R}^n} \dd^n \lambda \:
       \det_{n} \bigl\{V(\lambda_j, \lambda_k)\bigr\}\biggr|
    \leq
    1 + \sum_{n=1}^\infty \frac{1}{n!}
    \int\limits_{{\mathbb R}^n} \dd^n \lambda \:
       \bigl|\det_{n} \bigl\{V(\lambda_j, \lambda_k)\bigr\}\bigr| \\[1ex]
    \leq 
    1 + \sum_{n=1}^\infty \frac{c_0^n n^\frac{n}{2}}{n!}
    \biggl(\int\limits_{\mathbb R} \dd \lambda \:
       \vartheta^\frac{1}{4} (\lambda) \biggr)^n
     \leq 
    1 + \sum_{n=1}^\infty c^n n^{- \frac{n}{2}}.
    \end{multline}
The last series on the right-hand side is a convergent majorant,
which entails the claim.

(ii) Then, by the dominated convergence theorem on
$L^1 \bigl(\oplus_n {\mathbb R}^n, \sum_n \dd^n \lambda\bigr)$,
\begin{equation}
    \lim_{\lambda_0 \rightarrow + \infty} \det_{\mathbb R}
    (\id + \rmV)
    = \det_{\mathbb R} (\id + \rmV_0 + \delta \rmV_0),
\end{equation}
where $\rmV_0$ and $\delta \rmV_0$ are the integral operators with
kernels
\begin{equation}
    V_0 (\lambda, \mu)
    = \lim_{\lambda_0 \rightarrow + \infty}
    V_{S1} (\lambda, \mu),
    \quad
    \delta V_0 (\lambda, \mu)
    = \lim_{\lambda_0 \rightarrow + \infty} V_{S2} (\lambda, \mu).
\end{equation}

According to Assumption~\ref{item:assumption-2} we have the pointwise
limit $\lim_{\lambda_0 \rightarrow + \infty} u(\lambda|\lambda_0)
= p(\lambda)$. Then the continuity of the involved functions implies
that $V_0 (\lambda, \mu)$ is given by \eqref{eq:V-GSK}.

The evaluation of the $\lambda_0 \rightarrow + \infty$ limit of
$V_{S2}$ is a bit less direct. First of all, with a little rewriting,
\begin{equation}
    \delta V_0 (\lambda, \mu) = 
       - \frac{1}{2 \pi \rmi} F^\frac{1}{2} (\lambda) F^\frac{1}{2} (\mu)
         e_0 (\lambda) e_0 (\mu)
       \lim_{\lambda_0 \rightarrow + \infty}
       \frac{\rmC_- (\lambda) - \rmC_- (\mu)}{\lambda - \mu}.
\end{equation}
Then the claim is that one can pull the limit under the integral
that is defining the Cauchy transform. In order to see this one
introduces
\begin{equation}
    \rmC^{(0)} (\lambda)
    = \int\limits_{\mathbb R} \frac{\dd \nu}{2 \pi \rmi}
    \frac{\rme^{\rmi x p(\nu) + g(\nu)}}{\nu - \lambda}, \quad
    \delta \rmC (\lambda) = \rmC (\lambda) - \rmC^{(0)} (\lambda)
\end{equation}
and estimates $\bigl(\delta C_- (\lambda) - \delta C_- (\mu)\bigr)
/(\lambda - \mu)$ by the same techniques as in the proof of
Lemma~\ref{lem:u_esti}, distinguishing the cases $|\lambda - \mu| <
\sigma$ and $|\lambda - \mu| \geq \sigma$ for some fixed small
$\sigma > 0$ etc. We leave the details to the reader. Eventually,
\begin{equation}
    \lim_{\lambda_0 \rightarrow + \infty}
    \frac{\rmC_- (\lambda) - \rmC_- (\mu)}{\lambda - \mu}
    = \frac{\rmC_-^{(0)} (\lambda) - \rmC_-^{(0)} (\mu)}{\lambda - \mu}
    = \pv \int\limits_{\mathbb R} \frac{\dd \nu}{2 \pi \rmi} 
    \frac{\rme^{\rmi x p (\nu) + g (\nu)}}{(\nu - \lambda)(\nu - \mu)}.
\end{equation}
Assumption~\ref{item:assumption-2} tells us that $\Im p(\lambda)$
is positive immediately above the real axis. We may therefore shift
the integration contour upward by a small amount $\varepsilon > 0$ and,
after this regularization, drop the principle-value symbol. With this
we have arrived at equation~\eqref{eq:delta_v_zero}.

It remains to prove the second equation~\eqref{eq:prop_stat_lim}.
We note that
\begin{equation}
    \delta V_0 (\lambda, \mu) = \Ocal \bigl(x^{- \infty}\bigr),
\end{equation}
because $\Im p(\lambda) > c \varepsilon$ if $\Im \lambda =
\varepsilon$ (cf.\ Assumption~\ref{item:assumption-2}). Let
$\rmR_0$ denote the resolvent of $\rmV_0$. Then
\begin{equation}
    \det_{\mathbb R} (\id + (\id - \rmR_0) \delta \rmV_0)
    = \exp\biggl\{- \sum_{n=1}^\infty \frac{(-1)^n}{n}
    \tr\bigl\{ \bigl( (\id - \rmR_0) \delta \rmV_0\bigr)^n \bigr\} \biggl\}
    = 1 + \Ocal\bigl(x^{-\infty}\bigr),
\end{equation}
implying that
\begin{multline}
    \det_{\mathbb R} (\id + \rmV_0 + \delta \rmV_0)
    = \det_{\mathbb R} (\id + \rmV_0)
    \det_{\mathbb R} (\id + (\id - \rmR_0) \delta \rmV_0)
    \\
    = \det_{\mathbb R} (\id + \rmV_0)
    \bigl( 1 + \Ocal\bigl( x^{-\infty} \bigr) \bigr).
\end{multline}
\end{proof}

\section{The linear system}
\label{app:linear-system}

In this appendix, starting with the regularity
conditions~\eqref{eq:regularity-condition-for-Phi}
for $\Phi$, we derive the
representation~\eqref{eq:matrix-S-new} for the matrix
$\rmS (\lambda)$ and the system~\eqref{eq:SLE-new}
of linear algebraic equations whose solutions determine
the coefficients in~\eqref{eq:matrix-S-new}. We also
provide a proof of Proposition~\ref{prop:contribution-of-poles}
which describes the pole contributions to the logarithmic
derivatives of the Fredholm determinant
in~\eqref{eq:prop:log-der-Fredholm-for-AE}.

\subsection{Derivation of the linear system}
\label{app:derivation-of-the-linear-system}
We denote the residua $\rmS_\mu$ in~\eqref{eq:matrix-S}
at the poles in the sets $\Lcal^\pm$ and $\Rcal^\pm$
by $\rmC_j^\pm$ and $\rmD_j^\pm$. Then
\begin{equation}
    \label{eq:matrix-S-app}
    \rmS (\lambda)
    = \mathrm{I}_2
    + \sum\limits_{j = 1}^{n_\ell^+} \frac{\rmC_j^+}{\lambda - \ell_j^+}
    + \sum\limits_{j = 1}^{n_r^+} \frac{\rmD_j^+}{\lambda - r_j^+}
    + \sum\limits_{j = 1}^{n_\ell^-} \frac{\rmC_j^-}{\lambda - \ell_j^-}
    + \sum\limits_{j = 1}^{n_r^-} \frac{\rmD_j^-}{\lambda - r_j^-}.
\end{equation}
For later convenience we also introduce
the matrices $\rmS_{\ell, j}^\pm$ and $\rmS_{r, j}^\pm$
which are obtained from $\rmS$, see~\eqref{eq:matrix-S-app},
by subtracting the contribution of the poles at
$\ell_j^\pm$ or $r_j^\pm$, respectively,
\begin{subequations}
\label{eq:def-S-reduced}
\begin{align}
    \rmS_{\ell, j}^\pm (\lambda)
    = \rmS (\lambda)
    - \frac{\rmC_j^\pm}{\lambda - \ell_j^\pm},
    \qquad
    j = 1, \dots, n_\ell^\pm,
    \\
    \rmS_{r, j}^\pm (\lambda)
    = \rmS (\lambda)
    - \frac{\rmD_j^\pm}{\lambda - r_j^\pm},
    \qquad
    j = 1, \dots, n_r^\pm.
\end{align}
\end{subequations}

The regularity condition~\eqref{eq:regularity-condition-for-Phi}
for the matrix $\Phi$ at the pole $\ell_j^+$ for $j = 1, \dots, n_\ell^+$,
requires
\begin{equation}
\label{eq:regularity-condition-for-Phi-at-ell-+}
    \Phi (\lambda) \big( \rmM_\ell^+ (\lambda) \big)^{-1}
    = \rmS (\lambda) \Pi (\lambda) \big( \rmM_\ell^{+} (\lambda) \big)^{-1}
\end{equation}
to be regular, implying that
\begin{multline}
\label{eq:expansion-at-ell-+}
    \Bigg(
        \rmS_{\ell, j}^+ (\ell_j^+) + \frac{\rmC_j^+}{\lambda - \ell_j^+}
    \Bigg)
    \Big(
        \Pi (\ell_j^+) + (\lambda - \ell_j^+) \Pi^\prime (\ell_j^+)
    \Big)
    \Big( \mathrm{I}_2 - e^{-2} (\lambda) Q_\ell^+ (\lambda) \sigma^+ \Big)
    \\
    =
    - \frac{
        \rmC_j^+ \Pi (\ell_j^+)
        h_{\ell, j}^+ \sigma^+
    }{ (\lambda - \ell_j^+)^2 }
    - \frac{
        \rmS_{\ell, j}^+ (\ell_j^+) \Pi (\ell_j^+)
        h_{\ell, j}^+ \sigma^+
    }{\lambda - \ell_j^+}
    + \frac{ \rmC_j^+ \Pi (\ell_j^+) }{\lambda - \ell_j^+}
    \\
    - \frac{1}{ \lambda - \ell_j^+ }
    \rmC_j^+
    \cdot
    \eval{
        \partial_\lambda
        \big\{
            \Pi (\lambda) e^{-2} (\lambda)
            Q_\ell^+ (\lambda) (\lambda - \ell_j^+)
        \big\}
    }_{\lambda = \ell_j^+}
    \cdot \sigma^+
    + \Ocal(1)
\end{multline}
must be regular. Here we substituted
the expression~\eqref{eq:matrices-M-left} for the matrix $\rmM_\ell^+$
and denoted the residue at $\ell_j^+$ by $h_{\ell, j}^+ \sigma^+$,
\begin{equation}
\label{eq:residues-h-ell-+}
    h_{\ell, j}^+
    = \res\limits_{\lambda = \ell_j^+}
    \big( \rmM_\ell^+ (\lambda) \big)_{12}
    = e^{-2} (\ell_j^+)
    \cdot
    \res\limits_{\lambda = \ell_j^+}
    \big( Q_\ell^+ (\lambda) \big),
    \qquad
    j = 1, \dots, n_\ell^+.
\end{equation}
Thus, the regularity condition implies
that the coefficients in front of the second and
the first order pole at $\ell_j^+$ on the right
hand side of~\eqref{eq:expansion-at-ell-+} must
be zero, i.e.,
\begin{equation}
\label{eq:condition-second-order-pole-ell-+}
    \rmC_j^+ \Pi (\ell_j^+) \sigma^+ = 0
\end{equation}
and
\begin{equation}
\label{eq:condition-first-order-pole-ell-+}
    \rmC_j^+ \Pi (\ell_j^+)
    =
    h_{\ell, j}^+ \rmS_{\ell, j}^+ (\ell_j^+) \Pi (\ell_j^+) \sigma^+
    + \rmC_j^+ \Pi^\prime (\ell_j^+)
    h_{\ell, j}^+
    \sigma^+
\end{equation}
for $j = 1, \dots, n_\ell^+$.
Here in the second equation we used the fact that the terms
in which the derivative in~\eqref{eq:expansion-at-ell-+}
does not act on the matrix $\Pi$ are zero,
due to the first condition~\eqref{eq:condition-second-order-pole-ell-+}.

The first condition~\eqref{eq:condition-second-order-pole-ell-+} implies
that the matrix $\rmC_j^+$ is of the form
\begin{equation}
    \rmC_j^+
    = \begin{pmatrix}
        0 & *\\ 0 & *
    \end{pmatrix}
    \Pi^{-1} (\ell_j^+),
\end{equation}
and the second condition~\eqref{eq:condition-first-order-pole-ell-+}
implies that
\begin{equation}
    \rmC_j^+ \Pi (\ell_j^+)
    \Big(
        \mathrm{I}_2 -
        \Pi^{-1} (\ell_j^+) \Pi^\prime (\ell_j^+) h_{\ell, j}^+ \sigma^+
    \Big)
    = h_{\ell, j}^+
    \rmS_{\ell, j}^+ (\ell_j^+) \Pi (\ell_j^+) \sigma^+.
\end{equation}
Multiplying from the right by
$(\rmI_2
    - \Pi^{-1} (\ell_j^+) \Pi' (\ell_j^+) h_{\ell, j}^+ \sigma^+)^{-1}$
and using that
\begin{equation}
    \sigma^+ (
        \mathrm{I}_2
        - \Pi^{-1} (\ell_j^+) \Pi' (\ell_j^+) h_{\ell, j}^+ \sigma^+)^{-1}
    = \sigma_{\ell, j}^+ \sigma^+
\end{equation}
with
\begin{equation}
    \sigma_{\ell, j}^+
    = \frac{
        h_{\ell, j}^+
    }{
        1 - h_{\ell, j}^+
        \big[ \Pi^{-1} (\ell_j^+) \Pi^\prime (\ell_j^+) \big]_{21} },
\end{equation}
we obtain the first set of equations,
\begin{equation}
    \rmC_j^+ \Pi (\ell_j^+) 
    = \sigma_{\ell, j}^+
    \rmS_{\ell, j}^+ (\ell_j^+) \Pi (\ell_j^+) \sigma^+.
\end{equation}
In a similar way we analyse the regularity condition
at $\ell_j^-$ for $j = 1, \dots, n_\ell^-$
and at $r_j^\pm$ for $j = 1, \dots, n_r^\pm$.

Altogether, the conditions from the second order
poles imply that
\begin{equation}
\label{eq:conditions-at-second-order-poles}
\begin{aligned}
    \rmC_j^\pm \Pi (\ell_j^\pm) \sigma^\pm
    &= 0,
    & j = 1, \dots, n_\ell^\pm,
    \\
    \rmD_j^\pm \Pi (r_j^\pm) \sigma^\mp
    &= 0,
    & j = 1, \dots, n_r^\pm,
\end{aligned}
\end{equation}
while the conditions from the first order poles
imply the linear system itself,
\begin{equation}
\label{eq:linear-system-short-form}
\begin{aligned} 
    \rmC_j^\pm \Pi (\ell_j^\pm)
    & = \sigma_{\ell, j}^\pm \rmS_{\ell, j}^\pm (\ell^\pm)
    \Pi (\ell_j^\pm) \sigma^\pm,
    & j = 1, \dots, n_\ell^\pm,
    \\[1ex]
    \rmD_j^\pm \Pi (r_j^\pm)
    & = \sigma_{r, j}^\pm \rmS_{r, j}^\pm (r^\pm)
    \Pi (r_j^\pm) \sigma^\mp,
    & j = 1, \dots, n_r^\pm.
\end{aligned}
\end{equation}
The coefficients $\sigma$ in the latter equations
are given by
\begin{subequations}
\label{eq:coefficients-sigma-app}
\begin{equation}
    \sigma_{\ell, j}^+
    = \frac{
        h_{\ell, j}^+ }{
        1 - h_{\ell, j}^+
        \big[ \Pi^{-1} (\ell_j^+) \Pi ' (\ell_j^+) \big]_{21} },
    \qquad
    \sigma_{r, j}^+
    = \frac{
        h_{r, j}^+ }{
        1 - h_{r, j}^+
        \big[ \Pi^{-1} (r_j^+) \Pi ' (r_j^+) \big]_{12} },
\end{equation}
\begin{equation}
    \sigma_{\ell, j}^-
    = \frac{
        h_{\ell, j}^- }{
        1 - h_{\ell, j}^-
        \big[ \Pi^{-1} (\ell_j^-) \Pi ' (\ell_j^-) \big]_{12} },
    \qquad
    \sigma_{r, j}^-
    = \frac{
        h_{r, j}^- }{
        1 - h_{r, j}^-
        \big[ \Pi^{-1} (r_j^-) \Pi ' (r_j^-) \big]_{21} },
\end{equation}
\end{subequations}
where, in a similar way as in~\eqref{eq:residues-h-ell-+},
we introduced (up to their sign) the residues $h_\ell^\pm$
and $h_r^\pm$ of the off-diagonal matrix elements of
$\rmM_\ell^-$ and $\rmM_r^\pm$,
see equations~\eqref{eq:matrices-M-left}
and~\eqref{eq:matrices-M-right},
\begin{equation}
\label{eq:residues-h-app}
\begin{aligned}
    h_{\ell, j}^+
    & = \res\limits_{\lambda = \ell_j^+}
    \big( \rmM_\ell^+ (\lambda) \big)_{12}
    = e^{-2} (\ell_j^+)
    \cdot
    \res\limits_{\lambda = \ell_j^+}
    \big( Q_\ell^+ (\lambda) \big),
    \qquad
    &&
    j = 1, \dots, n_\ell^+,
    \\
    - h_{\ell, j}^-
    &= \res\limits_{\lambda = \ell_j^-}
    \big({\rmM_\ell^-} (\lambda)\big)_{21}
    = e^2 (\ell_j^-)
    \cdot
    \res\limits_{\lambda = \ell_j^-}
    \big( Q_\ell^- (\lambda) \big),
    \qquad
    &&
    j = 1, \dots, n_\ell^-,
    \\
    h_{r, j}^+
    &= \res\limits_{\lambda = r_j^+}
    \big({\rmM_r^+} (\lambda)\big)_{21}
    = e^2 (r_j^+)
    \cdot
    \res\limits_{\lambda = r_j^+}
    \big( Q_r^+ (\lambda) \big),
    \qquad
    &&
    j = 1, \dots, n_r^+,
    \\
    - h_{r, j}^-
    &= \res\limits_{\lambda = r_j^-}
    \big({\rmM_r^-} (\lambda)\big)_{12}
    = e^{-2} (r_j^-)
    \cdot
    \res\limits_{\lambda = r_j^-}
    \big( Q_r^- (\lambda) \big),
    \qquad
    &&
    j = 1, \dots, n_r^-.
\end{aligned}
\end{equation}

Now we substitute the explicit form of 
the matrices $\rmS_{\ell, j}^\pm$ and $\rmS_{r, j}^\pm$
and rewrite the linear system~\eqref{eq:linear-system-short-form}.
We do it again only for the first set of equations in the system,
since the derivation of the remaining equations is along the same lines.

The conditions~\eqref{eq:conditions-at-second-order-poles}
from the second order poles imply the following form
of the matrices $\rmC_j^\pm$ and $\rmD_j^\pm$,
\begin{align}
    & \rmC_j^+ = \begin{pmatrix}
        0 & *\\
        0 & *
    \end{pmatrix}
    \Pi^{-1} (\ell_j^+),
    \quad
    & \rmD_j^+ = \begin{pmatrix}
        * & 0\\
        * & 0
    \end{pmatrix}
    \Pi^{-1} (r_j^+),
    \\
    & \rmC_j^- = \begin{pmatrix}
        * & 0\\
        * & 0
    \end{pmatrix}
    \Pi^{-1} (\ell_j^-),
    \quad
    & \rmD_j^- = \begin{pmatrix}
        0 & *\\
        0 & *
    \end{pmatrix}
    \Pi^{-1} (r_j^-).
\end{align}
We rescale these expressions by the corresponding coefficients $\sigma$
and denote the unknown entries of the matrices
by column vectors $\mathbf{x}^\pm$ and $\mathbf{y}^\pm$
\begin{align}
    \rmC_j^+ &= \sigma_{\ell, j}^+
    \ (\mathbf{0}, \mathbf{y}_j^+) \ \Pi^{-1} (\ell_j^+),
    \qquad
    & \rmD_j^+ &= \sigma_{r, j}^+
    \ (\mathbf{x}_j^+, \mathbf{0}) \ \Pi^{-1} (r_j^+),
    \\
    \rmC_j^- &= \sigma_{\ell, j}^-
    \ (\mathbf{x}_j^-, \mathbf{0}) \ \Pi^{-1} (\ell_j^-),
    \qquad
    & \rmD_j^- &= \sigma_{r, j}^-
    \ (\mathbf{0}, \mathbf{y}_j^-) \ \Pi^{-1} (r_j^-).
\end{align}
Now we substitute the expressions
for the matrices $\rmS_{\ell, j}^\pm$ and $\rmS_{r, j}^\pm$
from~\eqref{eq:def-S-reduced}
into the system~\eqref{eq:linear-system-short-form}.
For example, for $\rmC_j^+$ we obtain
\begin{multline}
    \frac{\rmC_j^+ \Pi (\ell_j^+)}{\sigma_{\ell, j}^+}
    = \big(\mathbf{0}, \mathbf{y}_j^+\big)
    = \Pi (\ell_j^+) \sigma^+
    \\
    + \sum\limits_{\substack{k = 1\\ k \neq j}}^{n_\ell^+}
    \frac{
        \rmC_k^+ \Pi (\ell_k^+)
    }{
        \ell_j^+ - \ell_k^+ }
    + \sum\limits_{s = 1}^{n_r^+}
    \frac{
        \rmD_k^+ \Pi (r_k^+)
    }{
        \ell_j^+ - r_k^+ }
    + \sum\limits_{k = 1}^{n_\ell^-}
    \frac{
        \rmC_k^- \Pi (\ell_k^-)
    }{
        \ell_j^+ - \ell_k^- }
    + \sum\limits_{k = 1}^{n_r^-}
    \frac{
        \rmD_k^- \Pi (r_k^-)
    }{
        \ell_j^+ - r_k^- }.
\end{multline}
Next we also insert the expressions
for $\rmC_j^\pm$ and $\rmD_j^\pm$
in terms of $\mathbf{x}_j^\pm$ and $\mathbf{y}_j^\pm$
and use the following identities
\begin{align}
\label{eq:identities-for-C-plus-and-D-minus}
    \big( \mathbf{x}_j^\pm, \mathbf{0} \big)
    \Pi^{-1} (\lambda) \Pi (\mu)
    \sigma^+
    = \big[ \Pi^{-1} (\lambda) \Pi (\mu) \big]_{11}
    \big( \mathbf{0}, \mathbf{x}_j^\pm \big),
    \\
    \big( \mathbf{0}, \mathbf{y}_j^\pm \big)
    \Pi^{-1} (\lambda) \Pi (\mu)
    \sigma^+
    = \big[ \Pi^{-1} (\lambda) \Pi (\mu) \big]_{21}
    \big( \mathbf{0}, \mathbf{y}_j^\pm\big).
\end{align}
This way we derive the equation
\begin{multline}
    \big( \mathbf{0}, \mathbf{y}_j^+ \big)
    = \Pi (\ell_j^+) \sigma^+
    \\
    + \sum\limits_{\substack{k = 1\\ k \neq j}}^{n_\ell^+}
    \frac{
        \sigma_{\ell, k}^+ \big[ \Pi^{-1} (\ell_k^+) \Pi (\ell_j^+)\big]_{21}
    }{
        \ell_j^+ - \ell_k^+ }
    \big( \mathbf{0}, \mathbf{y}_k^+ \big)
    + \sum\limits_{k = 1}^{n_r^+}
    \frac{
        \sigma_{r, k}^+ \big[ \Pi^{-1} (r_k^+) \Pi (\ell_j^+)\big]_{11}
    }{
        \ell_j^+ - r_k^+ }
    \big( \mathbf{0}, \mathbf{x}_k^+ \big)
    \\
    + \sum\limits_{k = 1}^{n_\ell^-}
    \frac{
        \sigma_{\ell, k}^- \big[ \Pi^{-1} (\ell_k^-) \Pi (\ell_j^+)\big]_{11}
    }{
        \ell_j^+ - \ell_k^- }
    \big( \mathbf{0}, \mathbf{x}_k^- \big)
    + \sum\limits_{k = 1}^{n_r^-}
    \frac{
        \sigma_{r, k}^- \big[ \Pi^{-1} (r_k^-) \Pi (\ell_j^+)\big]_{21}
    }{
        \ell_j^+ - r_k^- }
    \big( \mathbf{0}, \mathbf{y}_k^- \big).
\end{multline}
Further setting
\begin{equation}
    \Pi (\ell_j^+) \sigma^+
    = \begin{pmatrix}
        0 & \Pi_{11} (\ell_j^+)\\ 0 & \Pi_{21} (\ell_j^+)
    \end{pmatrix}
    = (0, \mathbf{w}_j^+),
    \qquad
    j = 1, \dots, n_\ell^+,
\end{equation}
we end up with the first set of equations in our linear
system,
\begin{multline}
    \mathbf{y}_j^+
    = \mathbf{w}_j^+
    + \sum\limits_{\substack{k = 1\\k \neq j}}^{n_\ell^+}
    \frac{
        \sigma_{\ell, k}^+
        \big[ \Pi^{-1} (\ell_k^+) \Pi (\ell_j^+)\big]_{21} }{
        \ell_j^+ - \ell_k^+ }
    \mathbf{y}_k^+
    + \sum\limits_{k = 1}^{n_r^+}
    \frac{
        \sigma_{r, k}^+
        \big[ \Pi^{-1} (r_k^+) \Pi (\ell_j^+)\big]_{11} }{
        \ell_j^+ - r_k^+ }
    \mathbf{x}_k^+
    \\
    + \sum\limits_{k = 1}^{n_\ell^-}
    \frac{
        \sigma_{\ell, k}^-
        \big[ \Pi^{-1} (\ell_k^-) \Pi (\ell_j^+)\big]_{11} }{
        \ell_j^+ - \ell_k^- }
    \mathbf{x}_k^-
    + \sum\limits_{k = 1}^{n_r^-}
    \frac{
        \sigma_{r, k}^-
        \big[ \Pi^{-1} (r_k^-) \Pi (\ell_j^+)\big]_{21} }{
        \ell_j^+ - r_k^- }
    \mathbf{y}_k^-.
\end{multline}

Similarly, one can derive the set of equations for $\rmD_j^-$.
For $\rmC_j^-$ and $\rmD_j^+$, we need the second pair of identities
\begin{align}
    \big( \mathbf{x}_j^\pm, \mathbf{0} \big)
    \Pi^{-1} (\lambda) \Pi (\mu)
    \sigma^-
    & = \big[ \Pi^{-1} (\lambda) \Pi (\mu) \big]_{12}
    \big( \mathbf{x}_j^\pm, \mathbf{0} \big),
    \\
    \big( \mathbf{0}, \mathbf{y}_j^\pm \big)
    \Pi^{-1} (\lambda) \Pi (\mu)
    \sigma^-
    & = \big[ \Pi^{-1} (\lambda) \Pi (\mu) \big]_{22}
    \big( \mathbf{y}_j^\pm, \mathbf{0} \big)
\end{align}
instead of~\eqref{eq:identities-for-C-plus-and-D-minus},
but the rest of the derivation is exactly along the same lines.

At the end of the day,
one derives the expression~\eqref{eq:matrix-S-app}
for the matrix $\rmS (\lambda)$
with the matrices $\rmC_j^\pm$ and $\rmD_j^\pm$ given by
\begin{equation}
    \label{eq:C-and-D-via-x-and-y}
    \begin{aligned}
        \rmC_j^+ &=
        \sigma_{\ell, j}^+
        \cdot \big(\mathbf{0}, \mathbf{y}_j^+\big)
        \cdot \Pi^{-1} (\ell_j^+),
        \qquad
        & \rmD_j^+ &=
        \sigma_{r, j}^+
        \cdot \big(\mathbf{x}_j^+, \mathbf{0}\big)
        \cdot \Pi^{-1} (r_j^+),
        \\
        \rmC_j^- &=
        \sigma_{\ell, j}^-
        \cdot \big(\mathbf{x}_j^-, \mathbf{0}\big)
        \cdot \Pi^{-1} (\ell_j^-),
        \qquad
        & \rmD_j^- &=
        \sigma_{r, j}^-
        \cdot \big(\mathbf{0}, \mathbf{y}_j^-\big)
        \cdot \Pi^{-1} (r_j^-),
    \end{aligned}
\end{equation}
where the two-dimensional vectors $\mathbf{x}_j^\pm$ and $\mathbf{y}_j^\pm$
solve the following system of linear equations
\begin{subequations}
\label{eq:SLE}
\begin{multline}
\label{eq:SLE-left-plus}
    \mathbf{y}_j^+
    = \mathbf{w}_j^+
    + \sum\limits_{\substack{k = 1\\k \neq j}}^{n_\ell^+}
    \frac{
        \sigma_{\ell, k}^+
        \big[ \Pi^{-1} (\ell_k^+) \Pi (\ell_j^+)\big]_{21} }{
        \ell_j^+ - \ell_k^+ }
    \mathbf{y}_k^+
    + \sum\limits_{k = 1}^{n_r^+}
    \frac{
        \sigma_{r, k}^+
        \big[ \Pi^{-1} (r_k^+) \Pi (\ell_j^+)\big]_{11} }{
        \ell_j^+ - r_k^+ }
    \mathbf{x}_k^+
    \\
    + \sum\limits_{k = 1}^{n_\ell^-}
    \frac{
        \sigma_{\ell, k}^-
        \big[ \Pi^{-1} (\ell_k^-) \Pi (\ell_j^+)\big]_{11} }{
        \ell_j^+ - \ell_k^- }
    \mathbf{x}_k^-
    + \sum\limits_{k = 1}^{n_r^-}
    \frac{
        \sigma_{r, k}^-
        \big[ \Pi^{-1} (r_k^-) \Pi (\ell_j^+)\big]_{21} }{
        \ell_j^+ - r_k^- }
    \mathbf{y}_k^-,
\end{multline}
\begin{multline}
\label{eq:SLE-right-plus}
    \mathbf{x}_j^+
    = \mathbf{v}_j^+
    + \sum\limits_{k = 1}^{n_\ell^+}
    \frac{
        \sigma_{\ell, k}^+
        \big[ \Pi^{-1} (\ell_k^+) \Pi (r_j^+)\big]_{22} }{
        r_j^+ - \ell_k^+ }
    \mathbf{y}_k^+
    + \sum\limits_{\substack{k = 1\\k \neq j}}^{n_r^+}
    \frac{
        \sigma_{r, k}^+
        \big[ \Pi^{-1} (r_k^+) \Pi (r_j^+)\big]_{12} }{
        r_j^+ - r_k^+ }
    \mathbf{x}_k^+
    \\
    + \sum\limits_{k = 1}^{n_\ell^-}
    \frac{
        \sigma_{\ell, k}^-
        \big[ \Pi^{-1} (\ell_k^-) \Pi (r_j^+)\big]_{12} }{
        r_j^+ - \ell_k^- }
    \mathbf{x}_k^-
    + \sum\limits_{k = 1}^{n_r^-}
    \frac{
        \sigma_{r, k}^-
        \big[ \Pi^{-1} (r_k^-) \Pi (r_j^+)\big]_{22} }{
        r_j^+ - r_k^- }
    \mathbf{y}_k^-,
\end{multline}
\begin{multline}
\label{eq:SLE-left-minus}
    \mathbf{x}_j^-
    = \mathbf{v}_j^-
    + \sum\limits_{k = 1}^{n_\ell^+}
    \frac{
        \sigma_{\ell, k}^+
        \big[ \Pi^{-1} (\ell_k^+) \Pi (\ell_j^-)\big]_{22} }{
        \ell_j^- - \ell_k^+ }
    \mathbf{y}_k^+
    + \sum\limits_{k = 1}^{n_r^+}
    \frac{
        \sigma_{r, k}^+
        \big[ \Pi^{-1} (r_k^+) \Pi (\ell_j^-)\big]_{12} }{
        \ell_j^- - r_k^+ }
    \mathbf{x}_k^+
    \\
    + \sum\limits_{\substack{k = 1\\k \neq j}}^{n_\ell^-}
    \frac{
        \sigma_{\ell, k}^-
        \big[ \Pi^{-1} (\ell_k^-) \Pi (\ell_j^-)\big]_{12} }{
        \ell_j^- - \ell_k^- }
    \mathbf{x}_k^-
    + \sum\limits_{k = 1}^{n_r^-}
    \frac{
        \sigma_{r, k}^-
        \big[ \Pi^{-1} (r_k^-) \Pi (\ell_j^-)\big]_{22} }{
        \ell_j^- - r_k^- }
    \mathbf{y}_k^-,
\end{multline}
\begin{multline}
\label{eq:SLE-right-minus}
    \mathbf{y}_j^-
    = \mathbf{w}_j^-
    + \sum\limits_{k = 1}^{n_\ell^+}
    \frac{
        \sigma_{\ell, k}^+
        \big[ \Pi^{-1} (\ell_k^+) \Pi (r_j^-)\big]_{21} }{
        r_j^- - \ell_k^+ }
    \mathbf{y}_k^+
    + \sum\limits_{k = 1}^{n_r^+}
    \frac{
        \sigma_{r, k}^+
        \big[ \Pi^{-1} (r_k^+) \Pi (r_j^-)\big]_{11} }{
        r_j^- - r_k^+ }
    \mathbf{x}_k^+
    \\
    + \sum\limits_{k = 1}^{n_\ell^-}
    \frac{
        \sigma_{\ell, k}^-
        \big[ \Pi^{-1} (\ell_k^-) \Pi (r_j^-)\big]_{11} }{
        r_j^- - \ell_k^- }
    \mathbf{x}_k^-
    + \sum\limits_{\substack{k = 1\\k \neq j}}^{n_r^-}
    \frac{
        \sigma_{r, k}^-
        \big[ \Pi^{-1} (r_k^-) \Pi (r_j^-)\big]_{21} }{
        r_j^- - r_k^- }
    \mathbf{y}_k^-.
\end{multline}
\end{subequations}
The vectors $\mathbf{v}_j^\pm$ and $\mathbf{w}_j^\pm$
are defined by
\begin{equation}
\label{eq:vectors-V-and-W}
    \begin{aligned}
        \mathbf{w}_j^+
        & = \begin{pmatrix}
            \Pi_{11} (\ell_j^+) \\ \Pi_{21} (\ell_j^+)
        \end{pmatrix},
        \quad
        & \mathbf{v}_j^+
        & = \begin{pmatrix}
            \Pi_{12} (r_j^+) \\ \Pi_{22} (r_j^+)
        \end{pmatrix},
        \\
        \mathbf{v}_j^-
        & = \begin{pmatrix}
            \Pi_{12} (\ell_j^-) \\ \Pi_{22} (\ell_j^-)
        \end{pmatrix},
        \quad
        & \mathbf{w}_j^-
        & = \begin{pmatrix}
            \Pi_{11} (r_j^-) \\ \Pi_{21} (r_j^-)
        \end{pmatrix},
    \end{aligned}
\end{equation}
and the coefficients $\sigma^\pm$
by expressions~\eqref{eq:coefficients-sigma-app}.

We note that the equations in the linear system
related to the poles $\Lcal^+$ and $\Rcal_-$,
compare equations~\eqref{eq:SLE-left-plus} and~\eqref{eq:SLE-right-minus},
and to the poles $\Lcal^-$ and $\Rcal^+$,
compare equations~\eqref{eq:SLE-left-minus} and~\eqref{eq:SLE-right-plus},
as well as the form of the corresponding matrices
$\rmC^\pm$ and $\rmD^\pm$ in~\eqref{eq:C-and-D-via-x-and-y}
have the same mathematical structure.
Thus, we combine these sets of poles into two sets
\begin{equation}
    \Scal^\pm
    := \Lcal^\pm \cup \Rcal^\mp
    = \{ s_j^\pm \}_{j = 1}^{n^\pm},
\end{equation}
where the number of poles
in the sets is $n^\pm = n_\ell^\pm + n_r^\mp$, respectively.
The order of the poles in these sets is, of course,
not important, but let us fix it as follows:
\begin{align}
    s_j^+
    & = \begin{cases}
        \ell_j^+,
        & j = 1, \dots, n_\ell^+,
        \\
        r_{j - n_\ell^+}^-,
        & j = n_\ell^+ + 1, \dots n_\ell^+ + n_r^-,
    \end{cases}
    \\
    s_j^-
    & = \begin{cases}
        \ell_j^-,
        & j = 1, \dots, n_\ell^-,
        \\
        r_{j - n_\ell^-}^+,
        & j = n_\ell^- + 1, \dots n_\ell^- + n_r^+.
    \end{cases}
\end{align}
Next we introduce
\begin{align}
    \mathbf{y}_j
    & = \begin{cases}
        \mathbf{y}_j^+,
        & j = 1, \dots, n_\ell^+,
        \\
        \mathbf{y}_{j - n_\ell^+}^-,
        & j = n_\ell^+ + 1, \dots, n_\ell^+ + n_r^-,
    \end{cases}
    \\
    \mathbf{x}_j
    & = \begin{cases}
        \mathbf{x}_j^-,
        & j = 1, \dots, n_\ell^-,
        \\
        \mathbf{x}_{j - n_\ell^-}^+,
        & j = n_\ell^- + 1, \dots, n_\ell^- + n_r^+,
    \end{cases}
\end{align}
and, similarly, the coefficients $\sigma_j^\pm$,
and vectors $\mathbf{v}_j$, $\mathbf{w}_j$,
see equations\eqref{eq:coefficients-sigma-app}
and~\eqref{eq:vectors-V-and-W}.
Lastly, we employ the short-hand notations
\eqref{eq:matrix-elements-P-and-Q} for
matrix elements involving the matrices
$\Pi^{-1} (\lambda)$, $\Pi (\lambda)$ and
$\Pi' (\lambda)$,

Using all of these notations,
the linear system~\eqref{eq:SLE}
can be formulated in a more compact $2 \times 2$ block form,
see the linear system~\eqref{eq:SLE-new} in the main text.
Then the corresponding matrix $\rmS (\lambda)$,
see equations~\eqref{eq:matrix-S-app}
and~\eqref{eq:C-and-D-via-x-and-y},
can be expressed as~\eqref{eq:matrix-S-new}.
In particular,
the expressions~\eqref{eq:coefficients-sigma-app}
for the coefficients $\sigma^\pm$
and~\eqref{eq:residues-h-app} for the coefficients $h^\pm$
take the form~\eqref{eq:coefficients-sigma} and~\eqref{eq:residues-h},
respectively.

\subsection{Contribution of the poles}
\label{app:contribution-of-poles}
In this subsection we prove Proposition~\ref{prop:contribution-of-poles}.

We start with the expression for $\rmS^{-1} (\lambda) \rmS' (\lambda)$,
see equation~\eqref{eq:S-inv-S-prime} in the main text,
\begin{multline}
\label{eq:S-inv-S-prime-app}
    \rmS^{-1} (\lambda)
    \rmS' (\lambda)
    = - \sum\limits_{j = 1}^{n^+}
    \frac{ \sigma_j^+ }{ (\lambda - s_j^+)^2 }
    \mathbf{y}_j (0, 1) \Pi^{-1} (s_j^+)
    - \sum\limits_{j = 1}^{n^-}
    \frac{ \sigma_j^- }{ (\lambda - s_j^-)^2 }
    \mathbf{x}_j (1, 0) \Pi^{-1} (s_j^-)
    \\
    -
    \Bigg[
        \sum\limits_{j = 1}^{n^+}
        \frac{ \sigma_j^+ }{\lambda - s_j^+}
        \Pi (s_j^+)
        \sigma^y
        (0, 1)^\intercal \mathbf{y}_j^\intercal
        \sigma^y
        + \sum\limits_{j = 1}^{n^-}
        \frac{ \sigma_j^- }{\lambda - s_j^-}
        \Pi (s_j^-)
        \sigma^y
        (1, 0)^\intercal \mathbf{x}_j^\intercal
        \sigma^y
    \Bigg]
    \\
    \times
    \Bigg[
        \sum\limits_{j = 1}^{n^+}
        \frac{ \sigma_j^+ }{ (\lambda - s_j^+)^2 }
        \mathbf{y}_j (0, 1) \Pi^{-1} (s_j^+)
        + \sum\limits_{j = 1}^{n^-}
        \frac{ \sigma_j^- }{ (\lambda - s_j^-)^2 }
        \mathbf{x}_j (1, 0) \Pi^{-1} (s_j^-)
    \Bigg],
\end{multline}
and calculate the following contribution of the poles,
\begin{equation}
    \sum\limits_{\mu \in \Scal}
    \res\limits_{\lambda = \mu}
    \Big(
        \tr\big\{
            \rmS' (\lambda)
            \Pi (\lambda)
            \sigma^z
            \Pi^{-1} (\lambda)
            \rmS^{-1} (\lambda)
        \big\}
        d_\beta (\lambda)
    \Big),
\end{equation}
where
$\Scal = \{s_j^+\}_{j = 1}^{n^+} \cup \{s_j^-\}_{j = 1}^{n^-}$
is the set of all poles.
We will only provide the details of the calculation
for the residue at $\lambda = s_j^+$ for some
$j = 1, \dots, n^+$, since the derivation for the
residue at $\lambda = s_j^-$ for $j = 1, \dots, n^-$
is along the same lines.

As we discussed in the main text in Section~\ref{sec:contribution-of-poles},
there are no third order poles,
then the Laurent expansion of $\rmS^{-1} (\lambda) \rmS' (\lambda)$
at $\lambda = s_j^+$ has the following form
\begin{equation}
\label{eq:S-inv-S-prime-Laurent}
    \rmS^{-1} (\lambda) \rmS' (\lambda)
    = \frac{\rmS_2 (\lambda)}{(\lambda - s_j^+)^2}
    + \frac{\rmS_1 (\lambda)}{\lambda - s_j^+}
    + \Ocal(1).
\end{equation}

In the next sections
we first evaluate the contribution of the second order poles
and then the contribution of the first order poles.
Finally, we combine the results
and get the final expression for the contribution of the poles,
which is shown in Proposition~\ref{prop:contribution-of-poles}
in the main body of the text.

\subsubsection{Contribution of the second order poles}
We start with the contribution of the second order poles.
The matrix $\rmS_2 (\lambda)$ in~\eqref{eq:S-inv-S-prime-Laurent}
accounting for the second order pole at $\lambda = \ell_j^+$
propagates from expression~\eqref{eq:S-inv-S-prime-app}
\begin{multline}
\label{eq:expression-for-S-2}
    \rmS_2 (\lambda)
    = - \sigma_j^+
    \mathbf{y}_j (0, 1) \Pi^{-1} (s_j^+)
    -
    \sum\limits_{\substack{k = 1\\ k \neq j}}^{n^+}
    \frac{ \sigma_j^+ \sigma_k^+ }{\lambda - s_k^+}
    \Pi (s_k^+)
    \sigma^y
    (0, 1)^\intercal \mathbf{y}_k^\intercal
    \sigma^y
    \mathbf{y}_j (0, 1) \Pi^{-1} (s_j^+)
    \\
    - \sum\limits_{k = 1}^{n^-}
    \frac{ \sigma_j^+ \sigma_k^- }{\lambda - s_k^-}
    \Pi (s_k^-)
    \sigma^y
    (1, 0)^\intercal \mathbf{x}_k^\intercal
    \sigma^y
    \mathbf{y}_j (0, 1) \Pi^{-1} (s_j^+).
\end{multline}

Then the corresponding residue is given by
\begin{multline}
\label{eq:residue-with-S-2}
    \res\limits_{\lambda = s_j^+}
    \Bigg(
        \frac{d_\beta (\lambda)}{(\lambda - s_j^+)^2}
        \tr\big\{
            \rmS_2 (\lambda) 
            \Pi (\lambda)
            \sigma^z
            \Pi^{-1} (\lambda)
        \big\}
    \Bigg)
    =
    \tr\big\{
        \rmS_2' (s_j^+)
        \Pi (s_j^+)
        \sigma^z
        \Pi^{-1} (s_j^+)
    \big\}
    d_\beta (s_j^+)
    \\
    +
    \eval{\tr\big\{
        \rmS_2 (s_j^+)
        \partial_\lambda
        \big(
            \Pi (\lambda)
            \sigma^z
            \Pi^{-1} (\lambda)
        \big)
    \big\}}_{\lambda = s_j^+}
    d_\beta (s_j^+).
\end{multline}
Here we used Proposition~\ref{prop:residue-form} from the main text.
It guarantees that there is no term involving $d' (s_j^+)$.
In the following we separately consider the terms
on the right-hand side of the expression for the residue.

We start with the first term
on the right-hand side of~\eqref{eq:residue-with-S-2}.
Substituting the derivative of the expression~\eqref{eq:expression-for-S-2},
we obtain
\begin{multline}
    \tr\big\{
        \rmS_2' (s_j^+)
        \Pi (s_j^+)
        \sigma^z
        \Pi^{-1} (s_j^+)
    \big\}
    \\
    = \sum\limits_{\substack{k = 1\\ k \neq j}}^{n^+}
    \frac{ \sigma_j^+ \sigma_k^+ }{ (s_j^+ - s_k^+)^2 }
    \tr\big\{
        \Pi (s_k^+)
        \sigma^y
        (0, 1)^\intercal \mathbf{y}_k^\intercal
        \sigma^y
        \mathbf{y}_j (0, 1)
        \sigma^z
        \Pi^{-1} (s_j^+)
    \big\}
    \\
    + \sum\limits_{k = 1}^{n^-}
    \frac{ \sigma_j^+ \sigma_k^- }{ (s_j^+ - s_k^-)^2 }
    \tr\big\{
        \Pi (s_k^-)
        \sigma^y
        (1, 0)^\intercal \mathbf{x}_k^\intercal
        \sigma^y
        \mathbf{y}_j (0, 1)
        \sigma^z
        \Pi^{-1} (s_j^+)
    \big\}.
\end{multline}
Using the cyclicity of the trace
as in the proof of Proposition~\ref{prop:residue-form},
see equation~\eqref{eq:cyclicity-of-trace} in the main text,
we get
\begin{multline}
\label{eq:residue-with-S-2-term-1}
    \tr\big\{
        \rmS_2' (s_j^+)
        \Pi (s_j^+)
        \sigma^z
        \Pi^{-1} (s_j^+)
    \big\}
    d_\beta (s_j^+)
    \\
    = \rmi \sum\limits_{\substack{k = 1\\ k \neq j}}^{n^+}
    \frac{
        \sigma_j^+ \sigma_k^+
        \big( P_{jk}^{++} \big)_{21} }{
        (s_j^+ - s_k^+)^2 }
        d_\beta (s_j^+)
    \mathbf{y}_k^\intercal \sigma^y \mathbf{y}_j
    - \rmi \sum\limits_{k = 1}^{n^-}
    \frac{
        \sigma_j^+ \sigma_k^-
        \big( P_{jk}^{+-} \big)_{22} }{
        (s_j^+ - s_k^-)^2 }
    d_\beta (s_j^+)
    \mathbf{x}_k^\intercal \sigma^y \mathbf{y}_j.
\end{multline}

Substituting expression~\eqref{eq:expression-for-S-2}
into the second term on the right-hand side of~\eqref{eq:residue-with-S-2},
we obtain
\begin{multline}
\label{eq:residue-with-S-2-second-term}
    \eval{\tr\big\{
        \rmS_2 (s_j^+)
        \partial_\lambda
        \big(
            \Pi (\lambda)
            \sigma^z
            \Pi^{-1} (\lambda)
        \big)
    \big\}}_{\lambda = s_j^+}
    \\
    =
    - \sigma_j^+
    \eval{\tr\big\{
        \mathbf{y}_j (0, 1) \Pi^{-1} (s_j^+)
        \partial_\lambda
        \big(
            \Pi (\lambda)
            \sigma^z
            \Pi^{-1} (\lambda)
        \big)
    \big\}}_{\lambda = s_j^+}
    \\
    - \sum\limits_{\substack{k = 1\\ k \neq j}}^{n^+}
    \frac{ \sigma_j^+ \sigma_k^+ }{s_j^+ - s_k^+}
    \eval{\tr\big\{
        \Pi (s_k^+)
        \sigma^y
        (0, 1)^\intercal \mathbf{y}_k^\intercal
        \sigma^y
        \mathbf{y}_j (0, 1) \Pi^{-1} (s_j^+)
        \partial_\lambda
        \big(
            \Pi (\lambda)
            \sigma^z
            \Pi^{-1} (\lambda)
        \big)
    \big\}}_{\lambda = s_j^+}
    \\
    \mspace{-2mu}
    - \sum\limits_{k = 1}^{n^-}
    \frac{ \sigma_j^+ \sigma_k^- }{s_j^+ - s_k^-}
    \eval{\tr\big\{
        \Pi (s_k^-)
        \sigma^y
        (1, 0)^\intercal \mathbf{x}_k^\intercal
        \sigma^y
        \mathbf{y}_j (0, 1) \Pi^{-1} (s_j^+)
        \partial_\lambda
        \big(
            \Pi (\lambda)
            \sigma^z
            \Pi^{-1} (\lambda)
        \big)
    \big\}}_{\lambda = s_j^+}.
\end{multline}
Next we use the following identity
\begin{equation}
    \Pi (\lambda) \Pi^{-1} (\lambda) = \mathrm{I}_2
    \quad
    \Rightarrow
    \quad
    \big( \Pi^{-1} (\lambda) \big)'
    = - \Pi^{-1} (\lambda) \Pi' (\lambda) \Pi^{-1} (\lambda)
\end{equation}
to get
\begin{equation}
\label{eq:derivative-of-Pi-Sigma-z-Pi-inv}
    \partial_\lambda
    \big(
        \Pi (\lambda)
        \sigma^z
        \Pi^{-1} (\lambda)
    \big)
    =
    \Pi' (\lambda) \sigma^z \Pi^{-1} (\lambda)
    - \Pi (\lambda) \sigma^z
    \Pi^{-1} (\lambda) \Pi' (\lambda) \Pi^{-1} (\lambda).
\end{equation}
This allows us to evaluate
all the traces in terms of the coefficients $P$ and $Q$,
see definition~\eqref{eq:matrix-elements-P-and-Q} in the main text.
We substitute~\eqref{eq:derivative-of-Pi-Sigma-z-Pi-inv}
into the first trace in~\eqref{eq:residue-with-S-2-second-term},
use the cyclicity of the trace
and derive the following expression
\begin{multline}
    \eval{\tr\big\{
        \mathbf{y}_j (0, 1) \Pi^{-1} (s_j^+)
        \partial_\lambda
        \big(
            \Pi (\lambda)
            \sigma^z
            \Pi^{-1} (\lambda)
        \big)
    \big\}}_{\lambda = s_j^+}
    \\
    =
    (0, 1) \Pi^{-1} (s_j^+)
    \Pi' (s_j^+)
    \sigma^z
    \Pi^{-1} (s_j^+)
    \mathbf{y}_j
    - (0, 1)
    \sigma^z
    \Pi^{-1} (s_j^+)
    \Pi' (s_j^+)
    \Pi^{-1} (s_j^+)
    \mathbf{y}_j
    \\
    =
    \Big[
        \big( Q_{jj}^{++} \big)_{21}
        (1, 0)
        -
        \big( Q_{jj}^{++} \big)_{22}
        (0, 1)
    \Big]
    \Pi^{-1} (s_j^+)
    \mathbf{y}_j
    + \Big[
        \big( Q_{jj}^{++} \big)_{21}
        (1, 0)
        +
        \big( Q_{jj}^{++} \big)_{22}
        (0, 1)
    \Big]
    \Pi^{-1} (s_j^+)
    \mathbf{y}_j
    \\
    = 2 \big( Q_{jj}^{++} \big)_{21}
    (1, 0)
    \Pi^{-1} (s_j^+)
    \mathbf{y}_j
    = - 2 \rmi \big( Q_{jj}^{++} \big)_{21}
    \mathbf{v}_j^\intercal
    \sigma^y
    \mathbf{y}_j.
\end{multline}
Here in the second equality
we inserted the identity after $\Pi' (s_j^+)$,
\begin{equation}
    \mathrm{I}_2
    = (1, 0)^\intercal \cdot (1, 0)
    + (0, 1)^\intercal \cdot (0, 1).
\end{equation}
Similarly, we obtain
\begin{multline}
    \eval{\tr\big\{
        \Pi (s_k^+)
        \sigma^y
        (0, 1)^\intercal \mathbf{y}_k^\intercal
        \sigma^y
        \mathbf{y}_j (0, 1) \Pi^{-1} (s_j^+)
        \partial_\lambda
        \big(
            \Pi (\lambda)
            \sigma^z
            \Pi^{-1} (\lambda)
        \big)
    \big\}}_{\lambda = s_j^+}
    \\
    = - 2 \rmi
    \big( P_{jk}^{++} \big)_{11}
    \big( Q_{jj}^{++} \big)_{21}
    \mathbf{y}_k^\intercal
    \sigma^y
    \mathbf{y}_j
\end{multline}
and
\begin{multline}
    \eval{\tr\big\{
        \Pi (s_k^-)
        \sigma^y
        (1, 0)^\intercal \mathbf{x}_k^\intercal
        \sigma^y
        \mathbf{y}_j (0, 1) \Pi^{-1} (s_j^+)
        \partial_\lambda
        \big(
            \Pi (\lambda)
            \sigma^z
            \Pi^{-1} (\lambda)
        \big)
    \big\}}_{\lambda = s_j^+}
    \\
    = 2 \rmi
    \big( P_{jk}^{+-} \big)_{12}
    \big( Q_{jj}^{++} \big)_{21}
    \mathbf{x}_k^\intercal
    \sigma^y
    \mathbf{y}_j.
\end{multline}
Substituting all the traces back
into the expression~\eqref{eq:residue-with-S-2-second-term}
and multiplying by $d_\beta (s_j^+)$,
we get
\begin{multline}
\label{eq:residue-with-S-2-term-2-preliminary}
    \eval{\tr\big\{
        \rmS_2 (s_j^+)
        \partial_\lambda
        \big(
            \Pi (\lambda)
            \sigma^z
            \Pi^{-1} (\lambda)
        \big)
    \big\}}_{\lambda = s_j^+}
    d_\beta (s_j^+)
    =
    - 2 \sigma_j^+ \big( Q_{jj}^{++} \big)_{21}
    d_\beta (s_j^+)
    (1, 0)
    \Pi^{-1} (s_j^+)
    \mathbf{y}_j
    \\
    + 2 \rmi
    \sum\limits_{\substack{k = 1\\ k \neq j}}^{n^+}
    \frac{
        \sigma_j^+ \sigma_k^+ }{
        s_j^+ - s_k^+ }
    \big( P_{jk}^{++} \big)_{11}
    \big( Q_{jj}^{++} \big)_{21}
    d_\beta (s_j^+)
    \mathbf{y}_k^\intercal \sigma^y \mathbf{y}_j
    \\
    - 2 \rmi
    \sum\limits_{k = 1}^{n^-}
    \frac{
        \sigma_j^+ \sigma_k^- }{
        s_j^+ - s_k^- }
    \big( P_{jk}^{+-} \big)_{12}
    \big( Q_{jj}^{++} \big)_{21}
    d_\beta (s_j^+)
    \mathbf{x}_k^\intercal \sigma^y \mathbf{y}_j.
\end{multline}

Before we move on to the contribution of the first order poles,
let us slightly rewrite the last two terms with the sums,
using the following identities
\begin{equation}
    \big( Q_{kj}^{-+} \big)_{11}
    = \big( P_{kj}^{-+} \big)_{11}
    \big( Q_{jj}^{++} \big)_{11}
    - \big( P_{jk}^{+-} \big)_{12}
    \big( Q_{jj}^{++} \big)_{21},
\end{equation}
and the linear system~\eqref{eq:SLE-new}.
Then, up to some factors, the last term on the right-hand side
of the expression~\eqref{eq:residue-with-S-2-term-2-preliminary}
becomes
\begin{multline}
    - \sum\limits_{k = 1}^{n^-}
    \frac{ \sigma_k^- }{ s_j^+ - s_k^- }
    \big( P_{jk}^{+-} \big)_{12}
    \big( Q_{jj}^{++} \big)_{21}
    \mathbf{x}_k^\intercal
    = 
    \sum\limits_{k = 1}^{n^-}
    \frac{
        \sigma_k^-
        \big( Q_{kj}^{-+} \big)_{11} }{
        s_j^+ - s_k^- }
    \mathbf{x}_k^\intercal
    -
    \big( Q_{jj}^{++} \big)_{11}
    \sum\limits_{k = 1}^{n^-}
    \frac{
        \sigma_k^-
        \big( P_{kj}^{-+} \big)_{11} }{
        s_j^+ - s_k^- }
    \mathbf{x}_k^\intercal
    \\
    = 
    \big( \mathbf{w}_j^\intercal - \mathbf{y}_j^\intercal \big)
    \big( Q_{jj}^{++} \big)_{11}
    +
    \sum\limits_{k = 1}^{n^-}
    \frac{
        \sigma_k^-
        \big( Q_{kj}^{-+} \big)_{11} }{
        s_j^+ - s_k^- }
    \mathbf{x}_k^\intercal
    +  \big( Q_{jj}^{++} \big)_{11} 
    \sum\limits_{\substack{k = 1\\k \neq j}}^{n^+}
    \frac{
        \sigma_k^+ \big( P_{kj}^{++} \big)_{21} }{
        s_j^+ - s_k^+ }
    \mathbf{y}_k^\intercal.
\end{multline}
Here in the second equality
we used the first set of the equations
in the linear system~\eqref{eq:SLE-new}.
Substituting this expression
into~\eqref{eq:residue-with-S-2-term-2-preliminary}
and using one more identity to combine the terms
with $\mathbf{y}_k^\intercal \sigma^y \mathbf{y}_j$,
\begin{equation}
    \big( P_{jk}^{++} \big)_{11} \big( Q_{jj}^{++} \big)_{21}
    + \big( P_{kj}^{++} \big)_{21} \big( Q_{jj}^{++} \big)_{11}
    = \big( Q_{kj}^{++} \big)_{21},
\end{equation}
we obtain
\begin{multline}
    \eval{\tr\big\{
        \rmS_2 (s_j^+)
        \partial_\lambda
        \big(
            \Pi (\lambda)
            \sigma^z
            \Pi^{-1} (\lambda)
        \big)
    \big\}}_{\lambda = s_j^+}
    d_\beta (s_j^+)
    \\
    =
    2 \sigma_j^+
    d_\beta (s_j^+)
    \Big[
        \rmi \big( Q_{jj}^{++} \big)_{11}
        \mathbf{w}_j^\intercal
        \sigma^y
        - \big( Q_{jj}^{++} \big)_{21}
        (1, 0)
        \Pi^{-1} (s_j^+)
    \Big]
    \mathbf{y}_j
    \\
    +
    2 \rmi
    \sum\limits_{\substack{k = 1\\k \neq j}}^{n^+}
    \frac{
        \sigma_j^+ \sigma_k^+
        \big( Q_{kj}^{++} \big)_{21} }{
        s_j^+ - s_k^+ }
    d_\beta (s_j^+)
    \mathbf{y}_k^\intercal \sigma^y \mathbf{y}_j
    +
    2 \rmi
    \sum\limits_{k = 1}^{n^-}
    \frac{
        \sigma_j^+ \sigma_k^-
        \big( Q_{kj}^{-+} \big)_{11} }{
        s_j^+ - s_k^- }
        d_\beta (s_j^+)
    \mathbf{x}_k^\intercal \sigma^y \mathbf{y}_j.
\end{multline}
Here we used again
that $\mathbf{y}_j^\intercal \sigma^y \mathbf{y}_j = 0$.
Finally, the first term on the right-hand side can be written as
\begin{multline}
    \rmi \big( Q_{jj}^{++} \big)_{11}
    \mathbf{w}_j^\intercal
    \sigma^y
    - \big( Q_{jj}^{++} \big)_{21}
    (1, 0)
    \Pi^{-1} (s_j^+)
    \\
    =
    \big[
        \Pi_{22} (s_j^+)
        \Pi_{11}' (s_j^+)
        - \Pi_{12} (s_j^+)
        \Pi_{21}' (s_j^+)
    \big]
    \cdot
    \big(- \Pi_{21} (s_j^+), \Pi_{11} (s_j^+) \big)
    \\
    - \big[
        \Pi_{11} (s_j^+)
        \Pi_{21}' (s_j^+)
        - \Pi_{21} (s_j^+)
        \Pi_{11}' (s_j^+)
    \big]
    \cdot
    \big(\Pi_{22} (s_j^+), - \Pi_{12} (s_j^+) \big)
    \\
    = 
    \big(- \Pi_{21}' (s_j^+), \Pi_{11}' (s_j^+) \big)
    = \rmi \big( \mathbf{w}_j' \big)^\intercal \sigma^y.
\end{multline}
and therefore
\begin{multline}
\label{eq:residue-with-S-2-term-2}
    \eval{\tr\big\{
        \rmS_2 (s_j^+)
        \partial_\lambda
        \big(
            \Pi (\lambda)
            \sigma^z
            \Pi^{-1} (\lambda)
        \big)
    \big\}}_{\lambda = s_j^+}
    d_\beta (s_j^+)
    =
    2 \rmi \sigma_j^+
    d_\beta (s_j^+)
    \big( \mathbf{w}_j' \big)^\intercal \sigma^y \mathbf{y}_j
    \\
    + 2 \rmi
    \sum\limits_{\substack{k = 1\\k \neq j}}^{n^+}
    \frac{
        \sigma_j^+ \sigma_k^+
        \big( Q_{kj}^{++} \big)_{21} }{
        s_j^+ - s_k^+ }
    d_\beta (s_j^+)
    \mathbf{y}_k^\intercal \sigma^y \mathbf{y}_j
    + 2 \rmi
    \sum\limits_{k = 1}^{n^-}
    \frac{
        \sigma_j^+ \sigma_k^-
        \big( Q_{kj}^{-+} \big)_{11} }{
        s_j^+ - s_k^- }
        d_\beta (s_j^+)
    \mathbf{x}_k^\intercal \sigma^y \mathbf{y}_j.
\end{multline}
Next we evaluate the contribution of the first order poles.

\subsubsection{Contribution of the first order poles}
It follows from~\eqref{eq:S-inv-S-prime-app}
that the matrix $\rmS_1 (\lambda)$
in~\eqref{eq:S-inv-S-prime-Laurent}
is given by
\begin{multline}
    \rmS_1 (\lambda)
    = - \sum\limits_{\substack{k = 1\\ k \neq j}}^{n^+}
    \frac{ \sigma_j^+ \sigma_k^+ }{ (\lambda - s_k^+)^2 }
    \Pi (s_j^+)
    \sigma^y
    (0, 1)^\intercal \mathbf{y}_j^\intercal
    \sigma^y
    \mathbf{y}_k (0, 1) \Pi^{-1} (s_k^+)
    \\
    - \sum\limits_{k = 1}^{n^-}
    \frac{ \sigma_j^+ \sigma_k^- }{ (s_j^+ - s_k^-)^2 }
    \Pi (s_j^+)
    \sigma^y
    (0, 1)^\intercal \mathbf{y}_j^\intercal
    \sigma^y
    \mathbf{x}_k (1, 0) \Pi^{-1} (s_k^-).
\end{multline}
The corresponding residue at $s_j^+$ is
\begin{multline}
    \res\limits_{\lambda = s_j^+}
    \Bigg(
        \frac{ d_\beta (\lambda) }{\lambda - s_j^+}
        \tr\big\{
            \rmS_1 (\lambda)
            \Pi (\lambda)
            \sigma^z
            \Pi^{-1} (\lambda)
        \big\}
    \Bigg)
    =
    \tr\big\{
        \rmS_1 (s_j^+)
        \Pi (s_j^+)
        \sigma^z
        \Pi^{-1} (s_j^+)
    \big\}
    d_\beta (s_j^+)
    \\
    =
    - \sum\limits_{\substack{k = 1\\ k \neq j}}^{n^+}
    \frac{ \sigma_j^+ \sigma_k^+ }{ (s_j^+ - s_k^+)^2 }
    \tr\big\{
        \sigma^y
        (0, 1)^\intercal \mathbf{y}_j^\intercal
        \sigma^y
        \mathbf{y}_k (0, 1) \Pi^{-1} (s_k^+)
        \Pi (s_j^+)
        \sigma^z
    \big\}
    d_\beta (s_j^+)
    \\
    - \sum\limits_{k = 1}^{n^-}
    \frac{ \sigma_j^+ \sigma_k^- }{ (s_j^+ - s_k^-)^2 }
    \tr\big\{
        \sigma^y
        (0, 1)^\intercal \mathbf{y}_j^\intercal
        \sigma^y
        \mathbf{x}_k (1, 0) \Pi^{-1} (s_k^-)
        \Pi (s_j^+)
        \sigma^z
    \big\}
    d_\beta (s_j^+).
\end{multline}
The traces under the sums can be written in terms of the coefficients $P$,
\begin{equation}
\begin{aligned}
    \tr\big\{
        \sigma^y
        (0, 1)^\intercal \mathbf{y}_j^\intercal
        \sigma^y
        \mathbf{y}_k (0, 1) \Pi^{-1} (s_k^+)
        \Pi (s_j^+)
        \sigma^z
    \big\}
    & = - \rmi
    \big( P_{kj}^{++} \big)_{21}
    \mathbf{y}_j^\intercal \sigma^y \mathbf{y}_k,
    \\
    \tr\big\{
        \sigma^y
        (0, 1)^\intercal \mathbf{y}_j^\intercal
        \sigma^y
        \mathbf{x}_k (1, 0) \Pi^{-1} (s_k^-)
        \Pi (s_j^+)
        \sigma^z
    \big\}
    & = - \rmi
    \big( P_{kj}^{-+} \big)_{11}
    \mathbf{y}_j^\intercal \sigma^y \mathbf{x}_k.
\end{aligned}
\end{equation}
Then
\begin{multline}
\label{eq:residue-with-S-1}
    \res\limits_{\lambda = s_j^+}
    \Bigg(
        \frac{ d_\beta (\lambda) }{\lambda - s_j^+}
        \tr\big\{
            \rmS_1 (\lambda)
            \Pi (\lambda)
            \sigma^z
            \Pi^{-1} (\lambda)
        \big\}
    \Bigg)
    \\
    =
    \rmi \sum\limits_{\substack{k = 1\\ k \neq j}}^{n^+}
    \frac{
        \sigma_j^+ \sigma_k^+
        \big( P_{kj}^{++} \big)_{21} }{
        (s_j^+ - s_k^+)^2 }
    d_\beta (s_j^+)
    \mathbf{y}_j^\intercal \sigma^y \mathbf{y}_k
    + \rmi
    \sum\limits_{k = 1}^{n^-}
    \frac{
        \sigma_j^+ \sigma_k^-
        \big( P_{kj}^{-+} \big)_{11} }{
        (s_j^+ - s_k^-)^2 }
    d_\beta (s_j^+)
    \mathbf{y}_j^\intercal \sigma^y \mathbf{x}_k.
\end{multline}
We note that this expression coincides
with the expression~\eqref{eq:residue-with-S-2-term-1}.

\subsubsection{Complete contribution of the residue}
Now we combine the contributions of the poles,
see expressions~\eqref{eq:residue-with-S-2-term-1},
\eqref{eq:residue-with-S-2-term-2},
and~\eqref{eq:residue-with-S-1}.
Then we obtain
\begin{multline}
\label{eq:residue-at-lambda-j-plus}
    \res\limits_{\lambda = s_j^+}
    \Big(
        \tr\big\{
            \rmS' (\lambda)
            \Pi (\lambda)
            \sigma^z
            \Pi^{-1} (\lambda)
            \rmS^{-1} (\lambda)
        \big\}
        d_\beta (\lambda)
    \Big)
    =
    2 \rmi \sigma_j^+
    d_\beta (s_j^+)
    \big( \mathbf{w}_j' \big)^\intercal \sigma^y \mathbf{y}_j
    \\
    - 2 \rmi
    \sum\limits_{\substack{k = 1\\k \neq j}}^{n^+}
    \frac{
        \sigma_j^+ \sigma_k^+
        \big( Q_{kj}^{++} \big)_{21} }{
        s_j^+ - s_k^+ }
    d_\beta (s_j^+)
    \mathbf{y}_j^\intercal \sigma^y \mathbf{y}_k
    - 2 \rmi
    \sum\limits_{k = 1}^{n^-}
    \frac{
        \sigma_j^+ \sigma_k^-
        \big( Q_{kj}^{-+} \big)_{11} }{
        s_j^+ - s_k^- }
        d_\beta (s_j^+)
    \mathbf{y}_j^\intercal \sigma^y \mathbf{x}_k
    \\
    + 2 \rmi \sum\limits_{\substack{k = 1\\ k \neq j}}^{n^+}
    \frac{
        \sigma_j^+ \sigma_k^+
        \big( P_{kj}^{++} \big)_{21} }{
        (s_j^+ - s_k^+)^2 }
    d_\beta (s_j^+)
    \mathbf{y}_j^\intercal \sigma^y \mathbf{y}_k
    + 2 \rmi
    \sum\limits_{k = 1}^{n^-}
    \frac{
        \sigma_j^+ \sigma_k^-
        \big( P_{kj}^{-+} \big)_{11} }{
        (s_j^+ - s_k^-)^2 }
    d_\beta (s_j^+)
    \mathbf{y}_j^\intercal \sigma^y \mathbf{x}_k.
\end{multline}

The evaluation of the residue at $\lambda = s_j^-$ is similar
and leads to
\begin{multline}
\label{eq:residue-at-lambda-j-minus}
    \res\limits_{\lambda = s_j^-}
    \Big(
        \tr\big\{
            \rmS' (\lambda)
            \Pi (\lambda)
            \sigma^z
            \Pi^{-1} (\lambda)
            \rmS^{-1} (\lambda)
        \big\}
        d_\beta (\lambda)
    \Big)
    =
    2 \rmi \sigma_j^-
    d_\beta (s_j^-)
    \big( \mathbf{v}_j' \big)^\intercal \sigma^y \mathbf{x}_j
    \\
    - 2 \rmi
    \sum\limits_{k = 1}^{n^+}
    \frac{
        \sigma_j^- \sigma_k^+
        \big( Q_{kj}^{+-} \big)_{22} }{
        s_j^- - s_k^+ }
    d_\beta (s_j^-)
    \mathbf{x}_j^\intercal \sigma^y \mathbf{y}_k
    - 2 \rmi
    \sum\limits_{\substack{k = 1\\k \neq j}}^{n^-}
    \frac{
        \sigma_j^- \sigma_k^-
        \big( Q_{kj}^{--} \big)_{12} }{
        s_j^- - s_k^- }
        d_\beta (s_j^-)
    \mathbf{x}_j^\intercal \sigma^y \mathbf{x}_k
    \\
    + 2 \rmi
    \sum\limits_{k = 1}^{n^+}
    \frac{
        \sigma_j^- \sigma_k^+
        \big( P_{kj}^{+-} \big)_{22} }{
        (s_j^- - s_k^+)^2 }
    d_\beta (s_j^-)
    \mathbf{x}_j^\intercal \sigma^y \mathbf{y}_k
    + 2 \rmi
    \sum\limits_{\substack{k = 1\\k \neq j}}^{n^-}
    \frac{
        \sigma_j^- \sigma_k^-
        \big( P_{kj}^{--} \big)_{12} }{
        (s_j^- - s_k^-)^2 }
        d_\beta (s_j^-)
    \mathbf{x}_j^\intercal \sigma^y \mathbf{x}_k.
\end{multline}
Summing these two expressions over $j = 1, \dots, n^+$,
and $j = 1, \dots, n^-$, respectively, we obtain the 
following expression for the contribution of the poles:
\begin{multline}
    \sum\limits_{\mu \in \Scal}
    \res\limits_{\lambda = \mu}
    \Big(
        \tr\big\{
            \rmS' (\lambda)
            \Pi (\lambda)
            \sigma^z
            \Pi^{-1} (\lambda)
            \rmS^{-1} (\lambda)
        \big\}
        d_\beta (\lambda)
    \Big)
    \\
    = 2 \rmi
    \sum\limits_{j = 1}^{n^+}\sigma_j^+
    d_\beta (s_j^+)
    \big( \mathbf{w}_j' \big)^\intercal \sigma^y \mathbf{y}_j
    + 2 \rmi
    \sum\limits_{j = 1}^{n^-}
    \sigma_j^-
    d_\beta (s_j^-)
    \big( \mathbf{v}_j' \big)^\intercal \sigma^y \mathbf{x}_j
    \\
    - 2 \rmi
    \sum\limits_{\substack{j, k = 1\\j \neq k}}^{n^+}
    \frac{
        \sigma_j^+ \sigma_k^+
        \big( Q_{kj}^{++} \big)_{21} }{
        s_j^+ - s_k^+ }
    d_\beta (s_j^+)
    \mathbf{y}_j^\intercal \sigma^y \mathbf{y}_k
    - 2 \rmi
    \sum\limits_{\substack{j, k = 1\\j \neq k}}^{n^-}
    \frac{
        \sigma_j^- \sigma_k^-
        \big( Q_{kj}^{--} \big)_{12} }{
        s_j^- - s_k^- }
        d_\beta (s_j^-)
    \mathbf{x}_j^\intercal \sigma^y \mathbf{x}_k
    \\
    - 2 \rmi
    \sum\limits_{j = 1}^{n^+}
    \sum\limits_{k = 1}^{n^-}
    \frac{ \sigma_j^+ \sigma_k^- }{ s_j^+ - s_k^- }
    \Big[
        \big( Q_{kj}^{-+} \big)_{11}
        d_\beta (s_j^+)
        + \big( Q_{jk}^{+-} \big)_{22}
        d_\beta (s_k^-)
    \Big]
    \mathbf{y}_j^\intercal \sigma^y \mathbf{x}_k
    \\
    + 2 \rmi \sum\limits_{\substack{j, k = 1\\j \neq k}}^{n^+}
    \frac{
        \sigma_j^+ \sigma_k^+
        \big( P_{kj}^{++} \big)_{21} }{
        (s_j^+ - s_k^+)^2 }
    d_\beta (s_j^+)
    \mathbf{y}_j^\intercal \sigma^y \mathbf{y}_k
    + 2 \rmi
    \sum\limits_{\substack{j, k = 1\\j \neq k}}^{n^-}
    \frac{
        \sigma_j^- \sigma_k^-
        \big( P_{kj}^{--} \big)_{12} }{
        (s_j^- - s_k^-)^2 }
        d_\beta (s_j^-)
    \mathbf{x}_j^\intercal \sigma^y \mathbf{x}_k
    \\
    + 2 \rmi
    \sum\limits_{j = 1}^{n^+}
    \sum\limits_{k = 1}^{n^-}
    \frac{
        \sigma_j^+ \sigma_k^- \big( P_{jk}^{+-} \big)_{22} }{
        (s_j^+ - s_k^-)^2 }
    \big[
        d_\beta (s_j^+)
        - d_\beta (s_k^-)
    \big]
    \mathbf{y}_j^\intercal \sigma^y \mathbf{x}_k.
\end{multline}
Here we have combined similar terms
and interchanged summation indices in some
of the sums, which can be tracked back
thanks to the factors $d_\beta (s_j^\pm)$
and $d_\beta (s_k^\pm)$.
This expression is exactly the same
as in Proposition~\ref{prop:contribution-of-poles},
see expression~\eqref{eq:prop:contribution-of-poles},
which completes the proof of the proposition.

\section{Proof of Proposition~\ref{prop:Slavnov-formula-for-general-momentum}}
\label{app:combinatorial-identity}

In this appendix,
we prove Proposition~\ref{prop:Slavnov-formula-for-general-momentum}.

First, we introduce matrices
$\rmA^+ \in \mathbb{C}^{n^- \times n^+}$
and $\rmA^- \in \mathbb{C}^{n^+ \times n^-}$ as
\begin{equation}
\label{eq:matrices-A-plus-and-minus-app}
    A_{jk}^- = \frac{h_k^-}{s_j^+ - s_k^-},
    \qquad
    A_{kj}^+ = \frac{h_j^+}{s_k^- - s_j^+},
    \qquad
    j = 1, \dots, n^+, \quad k = 1, \dots, n^-,
\end{equation}
and matrices $\rmA$ and $\widetilde{\rmA}$ as
\begin{equation}
    \rmA = \rmA^- \rmA^+
    \in \mathbb{C}^{n^+ \times n^+},
    \qquad
    \widetilde{\rmA} = \rmA^+ \rmA^-
    \in \mathbb{C}^{n^- \times n^-}.
\end{equation}
Next, we introduce the components
\begin{equation}
\begin{aligned}
    c_j^+
    = \big( \mathbf{y}_j \big)_{1},
    \qquad
    c_j^-
    = \big( \mathbf{x}_j \big)_{2},
    \\
    d_j^+
    = \big( \mathbf{x}_j \big)_{1},
    \qquad
    d_j^-
    = \big( \mathbf{y}_j \big)_{2}
\end{aligned}
\end{equation}
of the vectors $\mathbf{x}_j$ and $\mathbf{y}_j$
which are solutions of the linear system of equations
\begin{equation}
\label{eq:SLE-static-app}
\left\{
\begin{aligned}
    \mathbf{x}_j
    - \sum\limits_{k = 1}^{n^+}
    \frac{ h_k^+}{ s_j^- - s_k^+ }
    \mathbf{y}_k
    &= \begin{pmatrix} 0 \\ 1 \end{pmatrix},
    \\
    \mathbf{y}_j
    - \sum\limits_{k = 1}^{n^-}
    \frac{ h_k^- }{ s_j^+ - s_k^- }
    \mathbf{x}_k
    &= \begin{pmatrix} 1 \\ 0 \end{pmatrix}.
\end{aligned}
\right.
\end{equation}

From the linear system~\eqref{eq:SLE-static-app}
for the vectors $\mathbf{x}_j$ and $\mathbf{y}_j$,
we get the following system
for the vectors $\mathbf{c}^+ = (c_1^+, \dots, c_{n^+}^+)^\intercal$
and $\mathbf{c}_j^- = (c_1^-, \dots, c_{n^-}^-)^\intercal$,
\begin{equation}
    \Big[
        \mathrm{I}_{n^+} - \rmA
    \Big]
    \mathbf{c}^+ = \mathbf{e}_{n^+},
    \qquad
    \Big[
        \mathrm{I}_{n^-} - \widetilde{\rmA}
    \Big]
    \mathbf{c}^- = \mathbf{e}_{n^-},
\end{equation}
where $\mathbf{e}_{n^\pm}$ denotes the vector of length $n^\pm$
with all components being equal to one.
Since $\det \big( \rmI_{n^+} - \rmA \big)
    = \det \big(\rmI_{n^-} - \widetilde{\rmA} \big) \neq 0$,
$\mathbf{c}^\pm$ are well defined.
The vectors $\mathbf{d}^+ = (d_1^+, \dots, d_{n^-}^+)^\intercal$
and $\mathbf{d}^- = (d_1^-, \dots, d_{n^+}^-)^\intercal$ are then given by
\begin{equation}
    \mathbf{d}^+ = \rmA^+ \mathbf{c}^+,
    \qquad
    \mathbf{d}^- = \rmA^- \mathbf{c}^-.
\end{equation}

Finally, we introduce
\begin{equation}
\label{eq:contribution-of-poles-static-case-appendix}
    P
    = \rmi \sum\limits_{j = 1}^{n^+}
    \sum\limits_{k = 1}^{n^-}
    \partial_x \big[ A_{jk}^- A_{kj}^+ \big]
    \cdot
    \mathbf{x}_k^\intercal \sigma^y \mathbf{y}_j
\end{equation}
which in our new notations takes the form
\begin{equation}
    P
    = - \sum\limits_{j = 1}^{n^+}
    \sum\limits_{k = 1}^{n^-}
    \partial_x \big[ A_{jk}^- A_{kj}^+ \big]
    \cdot 
    \big[
        c_j^+ c_k^- - d_j^- d_k^+
    \big].
\end{equation}
The aim of this appendix is to show that it can be expressed as
\begin{equation}
    P = \partial_x \ln \det\big[ \mathrm{I}_{n^+} - \rmA \big].
\end{equation}

We first assume that $\lVert \rmA \rVert < 1$.
In this case,
one obtains the closed convergent expressions
for the vectors $\mathbf{c}^+$ and $\mathbf{c}^-$
by using the Neumann series representations
for $\big( \mathrm{I}_{n^+} - \rmA \big)^{-1}$
and $\big( \mathrm{I}_{n^-} - \widetilde{\rmA} \big)^{-1}$,
\begin{equation}
\begin{aligned}
    \mathbf{c}^+
    & = \big( \mathrm{I}_{n^+} - \rmA \big)^{-1} \mathbf{e}_{n^+}
    = \sum\limits_{n \geq 0} \rmA^n \mathbf{e}_{n^+},
    \\
    \mathbf{c}^-
    & = \big( \mathrm{I}_{n^-} - \widetilde{\rmA} \big)^{-1} \mathbf{e}_{n^-}
    = \sum\limits_{n \geq 0} \widetilde{\rmA}^n \mathbf{e}_{n^-}.
\end{aligned}
\end{equation}
Hence, the vector elements of $\mathbf{c}^+$ and $\mathbf{c}^-$
can be written as
\begin{equation}
\begin{aligned}
    c_{j_0}^+
    & = \sum\limits_{n \geq 0}
    \sum\limits_{\mathbf{j}_n \in (\Jcal^+)^n}
    \sum\limits_{\mathbf{k}_n \in (\Jcal^-)^n}
    \prod\limits_{\ell = 1}^{n}
    \big[
        A_{j_{\ell - 1} k_{\ell}}^-
        A_{k_{\ell} j_{\ell}}^+
    \big],
    \\
    c_{k_0}^-
    & = \sum\limits_{n \geq 0}
    \sum\limits_{\mathbf{j}_n \in (\Jcal^+)^n}
    \sum\limits_{\mathbf{k}_n \in (\Jcal^-)^n}
    \prod\limits_{\ell = 1}^{n}
    \big[
        A_{k_{\ell - 1} j_{\ell}}^+
        A_{j_{\ell} k_{\ell}}^-
    \big],
\end{aligned}
\end{equation}
and
\begin{equation}
\begin{aligned}
    d_{k_0}^+
    & = \sum\limits_{n \geq 0}
    \sum\limits_{\mathbf{j}_{n + 1} \in (\Jcal^+)^{n + 1}}
    \sum\limits_{\mathbf{k}_n \in (\Jcal^-)^n}
    \prod\limits_{\ell = 1}^{n + 1}
    A_{k_{\ell - 1} j_{\ell}}^+
    \prod\limits_{\ell = 1}^{n}
    A_{j_{\ell} k_{\ell}}^-,
    \\
    d_{j_0}^-
    & = \sum\limits_{n \geq 0}
    \sum\limits_{\mathbf{j}_n \in (\Jcal^+)^n}
    \sum\limits_{\mathbf{k}_{n + 1} \in (\Jcal^-)^{n + 1}}
    \prod\limits_{\ell = 1}^{n + 1}
    A_{j_{\ell - 1} k_{\ell}}^-
    \prod\limits_{\ell = 1}^{n}
    A_{k_{\ell} j_{\ell}}^+,
\end{aligned}
\end{equation}
where we denote $\Jcal^\pm = \{1, 2, \dots, n^\pm\}$
and $\mathbf{j}_n = (j_1, j_2, \dots, j_n)^\intercal$,
$\mathbf{k}_n = (k_1, k_2, \dots, k_n)^\intercal$.

First we evaluate $c_{j_0}^+ c_{k_0}^- - d_{j_0}^- d_{k_0}^+$.
The product $c_{j_0}^+ c_{k_0}^-$ can be written as
\begin{multline}
    c_{j_0}^+ c_{k_0}^-
    = 1
    + \sum\limits_{p \geq 0}
    \sum\limits_{\mathbf{j}_{p + 1} \in (\Jcal^+)^{p + 1}}
    \sum\limits_{\mathbf{k}_{p + 1} \in (\Jcal^-)^{p + 1}}
    \Bigg\{
        \prod\limits_{\ell = 1}^{p + 1}
        \big[
            A_{k_{\ell} j_{\ell}}^+
            A_{j_{\ell - 1} k_{\ell}}^-
        \big]
        \\
        +
        \sum\limits_{m = 0}^p
        A_{k_0 j_{m + 1}}^+
        \prod\limits_{\ell = 1}^{m}
        A_{k_{\ell} j_{\ell}}^+
        \prod\limits_{\ell = m + 1}^{p}
        A_{k_{\ell} j_{\ell + 1}}^+
        \prod\limits_{\ell = 1}^{m}
        A_{j_{\ell - 1} k_{\ell}}^-
        \prod\limits_{\ell = m + 1}^{p + 1}
        A_{j_{\ell} k_{\ell}}^-
    \Bigg\}.
\end{multline}
where we have collected the coefficients in front of the terms
which symbolically can be represented as $(\rmA^+ \rmA^-)^{p + 1}$.
The sums over $j_1, \dots, j_m$ and $k_1, \dots, k_m$
correspond to the matrix elements originating from $c_{k_0}^+$
and the sums over $j_{m + 1}, \dots, j_{p + 1}$
and $k_{m + 1}, \dots, k_{p + 1}$
correspond to the matrix elements from $c_{j_0}^-$.
The first term under the sum over $p$ `represents'
the term with $m = p + 1$ in the sum over $m$.

Similarly, the product $d_{j_0}^- d_{k_0}^+$
can be written as
\begin{equation}
    d_{j_0}^- d_{k_0}^+
    = \sum\limits_{p \geq 0}
    \sum\limits_{m = 0}^p
    \sum\limits_{\mathbf{j}_{p + 1} \in (\Jcal^+)^{p + 1}}
    \sum\limits_{\mathbf{k}_{p + 1} \in (\Jcal^-)^{p + 1}}
    A_{k_0 j_{m + 1}}^+
    \prod\limits_{\substack{\ell = 1 \\ \ell \neq m + 1}}^{p + 1}
    A_{k_{\ell} j_{\ell}}^+
    \prod\limits_{\ell = 1}^{p + 1}
    A_{j_{\ell - 1} k_{\ell}}^-.
\end{equation}
Here we collected the coefficients
in front of the terms $(\rmA^+ \rmA^-)^{p + 1}$,
where the terms with indices $\mathbf{j}_m$ and $\mathbf{k}_{m + 1}$
propagate from $d_{k_0}^-$ and those
with indices $j_{m + 1}, \dots, j_{p + 1}$ and $k_{m + 2}, \dots, k_{p + 1}$
from $d_{j_0}^+$, respectively.

Combining these two expressions together,
we obtain
\begin{equation}
    P
    = \sum\limits_{j_0 = 1}^{n^+}
    \sum\limits_{k_0 = 1}^{n^-}
    \partial_x \big[ A_{j_0 k_0}^- A_{k_0 j_0}^+ \big]
    \Big[
        1
        + \sum\limits_{p \geq 0} K_{j_0 k_0}^{(p)}
    \Big],
\end{equation}
where the coefficients $K_{j_0 k_0}^{(p)}$ are given by
\begin{multline}
    K_{j_0 k_0}^{(p)}
    =
    \sum\limits_{\mathbf{j}_{p + 1} \in (\Jcal^+)^{p + 1}}
    \sum\limits_{\mathbf{k}_{p + 1} \in (\Jcal^-)^{p + 1}}
    \Bigg\{
        \prod\limits_{\ell = 1}^{p + 1}
        \big[
            A_{k_{\ell} j_{\ell}}^+
            A_{j_{\ell - 1} k_{\ell}}^-
        \big]
        \\
        +
        \sum\limits_{m = 0}^p
        \Bigg[
            A_{k_0 j_{m + 1}}^+
            \prod\limits_{\ell = 1}^{m}
            A_{k_{\ell} j_{\ell}}^+
            \prod\limits_{\ell = m + 1}^{p}
            A_{k_{\ell} j_{\ell + 1}}^+
            \prod\limits_{\ell = 1}^{m}
            A_{j_{\ell - 1} k_{\ell}}^-
            \prod\limits_{\ell = m + 1}^{p + 1}
            A_{j_{\ell} k_{\ell}}^-
            \\
            - A_{k_0 j_{m + 1}}^+
            \prod\limits_{\substack{\ell = 1 \\
                \ell \neq m + 1}}^{p + 1}
            A_{k_{\ell} j_{\ell}}^+
            \prod\limits_{\ell = 1}^{p + 1}
            A_{j_{\ell - 1} k_{\ell}}^-
        \Bigg]
    \Bigg\}.
\end{multline}

From now on,
it will be useful to introduce certain shorthand notations:
for a given $j_\ell = 1, \dots, n^+$ and $k_\ell = 1, \dots, n^-$,
we denote the corresponding poles as
\begin{equation}
    y_\ell = s_{j_\ell}^+,
    \qquad
    z_\ell = s_{k_\ell}^-.
\end{equation}
We note that every matrix element $A_{jk}^\pm \propto h_k^\pm$,
see definition~\eqref{eq:matrices-A-plus-and-minus-app},
therefore
\begin{equation}
\label{eq:K-p-coefficient}
    K_{j_0 k_0}^{(p)}
    =
    \sum\limits_{\mathbf{j}_{p + 1} \in (\Jcal^+)^{p + 1}}
    \sum\limits_{\mathbf{k}_{p + 1} \in (\Jcal^-)^{p + 1}}
    \prod\limits_{\ell = 1}^{p + 1}
    \big[
        h_{j_\ell}^+
        h_{k_\ell}^-
    \big]
    \cdot
    R ( \mathbf{y}_{p + 2}, \mathbf{z}_{p + 2}),
\end{equation}
where $\mathbf{y}_{p + 2} = (y_0, \dots, y_{p + 1})^\intercal$
and $\mathbf{z}_{p + 2} = (z_0, \dots, z_{p + 1})^\intercal$
and $R ( \mathbf{y}_{p + 2}, \mathbf{z}_{p + 2})$ is given by
\begin{multline}
    R ( \mathbf{y}_{p + 2}, \mathbf{z}_{p + 2})
    =
    \prod\limits_{\ell = 1}^{p + 1}
    \frac{1}{ (y_{\ell - 1} - z_\ell) (z_\ell - y_\ell) }
    \\
    +
    \sum\limits_{m = 0}^p
    \Bigg[
        \frac{1}{ z_0 - y_{m + 1} }
        \prod\limits_{\ell = 1}^{m}
        \frac{1}{ z_\ell - y_\ell}
        \prod\limits_{\ell = m + 1}^{p}
        \frac{1}{ z_\ell - y_{\ell + 1} }
        \prod\limits_{\ell = 1}^{m}
        \frac{1}{ y_{\ell - 1} - z_\ell}
        \prod\limits_{\ell = m + 1}^{p + 1}
        \frac{1}{ y_\ell - z_\ell}
        \\
        - \frac{1}{ z_0 - y_{m + 1} }
        \prod\limits_{\substack{\ell = 1 \\ \ell \neq m + 1}}^{p + 1}
        \frac{1}{ z_\ell - y_\ell}
        \prod\limits_{\ell = 1}^{p + 1}
        \frac{1}{ y_{\ell - 1} - z_\ell}
    \Bigg].
\end{multline}
The products of variables $h^+$
and $h^-$ in~\eqref{eq:K-p-coefficient}
are symmetric with respect to permutations of the indices $j$ and $k$,
therefore we can symmetrise the coefficient $R$
with respect to permutations which leads to the following expression
\begin{equation}
\label{eq:R-symmetrised}
    R ( \mathbf{y}_{p + 2}, \mathbf{z}_{p + 2})
    = \frac{1}{ \big[ (p + 1)! \big]^2 }
    \sum\limits_{\sigma, \tau \in \mathfrak{S}_{p + 1}}
    \prod\limits_{\ell = 1}^{p + 1}
    \frac{1}{ z_{\tau (\ell)} - y_{\sigma (\ell)} }
    V ( \mathbf{y}_{p + 2}^\sigma, \mathbf{z}_{p + 2}^\tau),
\end{equation}
where $\mathbf{y}_{p + 2}^\sigma
    = (y_{\sigma(0)}, \dots, y_{\sigma(p + 1)})^\intercal$
and $\mathbf{z}_{p + 2}^\tau
    = (z_{\tau(0)}, \dots, z_{\tau(p + 1)})^\intercal$,
by which we mean that $\sigma (0) = 0 = \tau (0)$,
and $V ( \mathbf{y}_{p + 2}^\sigma, \mathbf{z}_{p + 2}^\tau)$ is given by
\begin{multline}
    V ( \mathbf{y}_{p + 2}, \mathbf{z}_{p + 2} )
    = \prod\limits_{\ell = 1}^{p + 1}
    \frac{1}{ y_{\ell - 1} - z_\ell}
    \\
    -
    \sum\limits_{m = 0}^p
    \Bigg[
        \frac{1}{ z_0 - y_{m + 1} }
        \prod\limits_{\ell = 1}^{m}
        \frac{1}{ y_{\ell - 1} - z_\ell}
        \prod\limits_{\ell = m + 1}^{p}
        \frac{1}{ y_{\ell + 1} - z_\ell }
        +
        \frac{z_{m + 1} - y_{m + 1}}{ z_0 - y_{m + 1} }
        \prod\limits_{\ell = 1}^{p + 1}
        \frac{1}{ y_{\ell - 1} - z_\ell}
    \Bigg].
\end{multline}

We split $V$ into two parts:
\begin{equation}
    V ( \mathbf{y}_{p + 2}, \mathbf{z}_{p + 2} )
    = V_1 ( \mathbf{y}_{p + 2}, \mathbf{z}_{p + 2} )
    + V_2 ( \mathbf{y}_{p + 2}, \mathbf{z}_{p + 2} ),
\end{equation}
where
\begin{equation}
    V_1 ( \mathbf{y}_{p + 2}, \mathbf{z}_{p + 2})
    = \prod\limits_{\ell = 0}^{p}
    \frac{1}{ y_{\ell} - z_{\ell + 1} }
    \cdot
    \Bigg[
        1
        - \sum\limits_{m = 0}^p
        \frac{ z_{m + 1} - y_{m + 1} }{ z_0 - y_{m + 1} }
    \Bigg].
\end{equation}
and
\begin{equation}
    V_2 ( \mathbf{y}_{p + 2}, \mathbf{z}_{p + 2} )
    = - \sum\limits_{m = 0}^p
    \frac{1}{ z_0 - y_{m + 1} }
    \prod\limits_{\ell = 0}^{m - 1}
    \frac{1}{ y_{\ell} - z_{\ell + 1}}
    \prod\limits_{\ell = m + 1}^{p}
    \frac{1}{ y_{\ell + 1} - z_\ell },
\end{equation}

First,
we add and subtract one term in $V_1$ as follows
\begin{multline}
    V_1 ( \mathbf{y}_{p + 2}, \mathbf{z}_{p + 2} )
    =
    \frac{ z_0 - y_{0}}{ z_0 - y_{p + 1} }
    \prod\limits_{\ell = 0}^{p}
    \frac{1}{ y_{\ell} - z_{\ell + 1} }
    \\
    + \prod\limits_{\ell = 0}^{p}
    \frac{1}{ y_{\ell} - z_{\ell + 1} }
    \cdot
    \Bigg[
        1
        - \frac{z_{0} - y_{0}}{ z_0 - y_{p + 1} }
        - \sum\limits_{m = 1}^{p}
        \frac{z_{m} - y_{m}}{ z_0 - y_{m} }
    \Bigg].
\end{multline}
The first term in this expression, which we denote as $V_0$,
\begin{equation}
    V_0 ( \mathbf{y}_{p + 2}, \mathbf{z}_{p + 2} )
    = \frac{y_0 - z_0}{ y_{p + 1} - z_{0}}
    \prod\limits_{\ell = 0}^{p}
    \frac{1}{ y_{\ell} - z_{\ell + 1}},
\end{equation}
is the `wanted' term,
which will reproduce the final result in the end.
In what follows,
we show that the second term $\delta V_1$,
\begin{equation}
    \delta V_1 ( \mathbf{y}_{p + 2}, \mathbf{z}_{p + 2} )
    = \prod\limits_{\ell = 0}^{p}
    \frac{1}{ y_{\ell} - z_{\ell + 1} }
    \cdot
    \Bigg[
        1
        - \frac{z_{0} - y_{0}}{ z_0 - y_{p + 1} }
        - \sum\limits_{m = 1}^{p}
        \frac{z_{m} - y_{m}}{ z_0 - y_{m} }
    \Bigg]
\end{equation}
together with $V_2$ will cancel out
after implementing another auxiliary symmetrisation.

At this stage,
we leave $\delta V_1$ and proceed to recast $V_2$.
We  implement another auxiliary symmetrisation of $V_2$
by performing a right shift in the symmetric group
$(\sigma, \tau) \hookrightarrow (\sigma \circ \pi, \tau \circ \pi)$
and summing over that shift,
\begin{multline}
    \sum\limits_{\sigma, \tau \in \mathfrak{S}_{p + 1}}
    \prod\limits_{\ell = 1}^{p + 1}
    \frac{1}{ z_{\tau (\ell)} - y_{\sigma (\ell)} }
    V_2 ( \mathbf{y}_{p + 2}^\sigma, \mathbf{z}_{p + 2}^\tau)
    \\
    = \frac{1}{ (p + 1)! }
    \sum\limits_{\pi \in \mathfrak{S}_{p + 1}}
    \sum\limits_{\substack{\sigma\circ \pi, \tau \circ \pi \\
        \sigma, \tau \in \mathfrak{S}_{p + 1}}}
    \prod\limits_{\ell = 1}^{p + 1}
    \frac{1}{ z_{\tau ( \pi (\ell) )} - y_{\sigma ( \pi (\ell) )} }
    V_2 ( \mathbf{y}_{p + 2}^{\sigma \circ \pi},
        \mathbf{z}_{p + 2}^{\tau \circ \pi}).
\end{multline}

We leave the summand in $V_2$ at $m = p$ as it is,
while for $m < p$ we implement the right shift
$\pi \hookrightarrow \pi \circ \sigma_m$,
where
\begin{equation}
    \begin{cases}
        \sigma_m (j) = j,
        & \text{for } j = 1, \dots, m,
        \\
        \sigma_m (j) = p + 2 + m - j,
        & \text{for } j = m + 1, \dots, p + 1.
    \end{cases}
\end{equation}
Then we get
\begin{multline}
    V_2 ( \mathbf{y}_{p + 2}^{\pi}, \mathbf{z}_{p + 2}^{\pi})
    \\
    \hookrightarrow- \frac{1}{ (p + 1)! }
    \sum\limits_{\pi \in \mathfrak{S}_{p + 1}}
    \frac{1}{ z_0 - y_{\pi(p + 1)} }
    \prod\limits_{\ell = 0}^{p}
    \frac{1}{ y_{\pi(\ell)} - z_{\pi(\ell + 1)}}
    \sum\limits_{m = 0}^{p}
    \big( y_{\pi (m)} - z_{\pi (m + 1)} \big),
\end{multline}
which allows us to sum over $\pi$ in the original expression for $V_2$ as
\begin{multline}
    \frac{1}{ \big[ (p + 1)! \big]^2 }
    \sum\limits_{\sigma, \tau \in \mathfrak{S}_{p + 1}}
    \prod\limits_{\ell = 1}^{p + 1}
    \frac{1}{ z_{\tau (\ell)} - y_{\sigma (\ell)} }
    V_2 ( \mathbf{y}_{p + 2}^\sigma, \mathbf{z}_{p + 2}^\tau)
    \\
    = \frac{1}{ \big[ (p + 1)! \big]^2 }
    \sum\limits_{\sigma, \tau \in \mathfrak{S}_{p + 1}}
    \prod\limits_{\ell = 1}^{p + 1}
    \frac{1}{ z_{\tau (\ell)} - y_{\sigma (\ell)} }
    \widetilde{V}_2 ( \mathbf{y}_{p + 2}^\sigma, \mathbf{z}_{p + 2}^\tau),
\end{multline}
where
\begin{equation}
    \widetilde{V}_2 ( \mathbf{y}_{p + 2}, \mathbf{z}_{p + 2})
    \\
    = - \frac{1}{ z_0 - y_{p + 1} }
    \prod\limits_{\ell = 0}^{p}
    \frac{1}{ y_{\ell} - z_{\ell + 1}}
    \sum\limits_{m = 0}^{p}
    \big( y_{m} - z_{m + 1} \big).
\end{equation}

Then $\delta V_1$ and $\widetilde{V}_2$ together read
\begin{multline}
    \delta V_1 ( \mathbf{y}_{p + 2}, \mathbf{z}_{p + 2} )
    + \widetilde{V}_2 ( \mathbf{y}_{p + 2}, \mathbf{z}_{p + 2} )
    \\
    = \prod\limits_{\ell = 0}^{p}
    \frac{1}{ y_{\ell} - z_{\ell + 1}}
    \cdot
    \Bigg[
        \sum\limits_{m = 1}^{p + 1}
        \frac{ y_{m} - z_{m} }{ y_{p + 1} - z_{0} }
        - \sum\limits_{m = 1}^{p + 1}
        \frac{ y_{m} - z_{m} }{ y_{m} - z_{0} }
    \Bigg].
\end{multline}
Here we have combined the summations over $m$ in
$\widetilde{V}_2$ and $\delta V_1$ with extra terms
in $\delta V_1$.

Finally, we substitute $V$
back into the expression~\eqref{eq:R-symmetrised} for $R$,
\begin{equation}
    R ( \mathbf{y}_{p + 2}, \mathbf{z}_{p + 2})
    = R_0 ( \mathbf{y}_{p + 2}, \mathbf{z}_{p + 2} )
    + \delta R ( \mathbf{y}_{p + 2}, \mathbf{z}_{p + 2} ),
\end{equation}
in which
\begin{multline}
\label{eq:R-0}
    R_0 ( \mathbf{y}_{p + 2}, \mathbf{z}_{p + 2})
    \\
    = \frac{1}{ \big[ (p + 1)! \big]^2 }
    \sum\limits_{\sigma, \tau \in \mathfrak{S}_{p + 1}}
    \frac{y_0 - z_0}{ y_{\sigma (p + 1)} - z_{0}}
    \prod\limits_{\ell = 1}^{p + 1}
    \frac{1}{
        ( z_{\tau (\ell)} - y_{\sigma (\ell)} )
        ( y_{\sigma (\ell - 1)} - z_{\tau (\ell)} ) }
\end{multline}
and
\begin{multline}
    \delta R ( \mathbf{y}_{p + 2}, \mathbf{z}_{p + 2})
    = \frac{1}{ \big[ (p + 1)! \big]^2 }
    \sum\limits_{\sigma, \tau \in \mathfrak{S}_{p + 1}}
    \prod\limits_{\ell = 1}^{p + 1}
    \frac{1}{
        ( z_{\tau (\ell)} - y_{\sigma (\ell)} )
        ( y_{\sigma (\ell - 1)} - z_{\tau (\ell)} ) }
    \\
    \times
    \Bigg[
        \sum\limits_{m = 1}^{p + 1}
        \frac{ y_{\sigma (m)} - z_{\tau (m)} }{ y_{\sigma (p + 1)} - z_{0} }
        - \sum\limits_{m = 1}^{p + 1}
        \frac{ y_{\sigma (m)} - z_{\tau (m)} }{ y_{\sigma (m)} - z_{0} }
    \Bigg].
\end{multline}

Finally, we show that $\delta R$ is zero.
For convenience, we set
\begin{equation}
    \delta R ( \mathbf{y}_{p + 2}, \mathbf{z}_{p + 2})
    = \frac{ (-1)^{p + 1} }{ \big[ (p + 1)! \big]^2 }
    \sum\limits_{\sigma, \tau \in \mathfrak{S}_{p + 1}}
    \sum\limits_{m = 1}^{p + 1}
    \delta R_m ( \mathbf{y}_{p + 2}^\sigma, \mathbf{z}_{p + 2}^\tau),
\end{equation}
where
\begin{multline}
    \delta R_m ( \mathbf{y}_{p + 2}, \mathbf{z}_{p + 2})
    =
    \Bigg[
        \frac{ y_m - z_m }{ y_{p + 1} - z_{0} }
        - \frac{ y_m - z_m }{ y_m - z_0 }
    \Bigg]
    \prod\limits_{\ell = 1}^{p + 1}
    \prod\limits_{r = \ell - 1}^\ell
    \frac{1}{z_\ell - y_r}
    \\
    =
    \frac{
        (y_m - y_{p + 1}) }{
        (y_{p + 1} - z_{0}) (y_m - z_0) (y_{m - 1} - z_m) }
    \prod\limits_{\ell = 1}^{m - 1}
    \prod\limits_{r = \ell - 1}^\ell
    \frac{1}{z_\ell - y_r}
    \cdot
    \prod\limits_{\ell = m + 1}^{p + 1}
    \prod\limits_{r = \ell - 1}^\ell
    \frac{1}{z_\ell - y_r}
    .
\end{multline}
Now we show that $\delta R_a$ is antisymmetric
with respect to permutation of the form
$(\sigma, \tau) \hookrightarrow
    (\sigma \circ \eta_{m + 1}, \tau \circ \eta_m)$,
where
\begin{equation}
    \begin{cases}
        \eta_m (j) = j,
        & \text{for } j = 1, \dots, m - 1,
        \\
        \eta_m (j) = p + 1 + m - j,
        & \text{for } j = m, \dots, p + 1.
    \end{cases}
\end{equation}
Indeed, under this transformation
\begin{multline}
    \frac{
        y_{\sigma (m)} - y_{\sigma (p + 1)} }{
        (y_{\sigma (p + 1)} - z_{0})
        (y_{\sigma (m)} - z_0)
        (y_{\sigma (m - 1)} - z_{\tau (m)}) }
    \\
    \hookrightarrow
    - \frac{
        y_{\sigma (m)} - y_{\sigma (p + 1)} }{
        (y_{\sigma (p + 1)} - z_{0})
        (y_{\sigma (m)} - z_0)
        (y_{\sigma (m - 1)} - z_{\tau (m)}) },
\end{multline}
and the product of the factors in the denominator in $\delta R_m$ is invariant
\begin{multline}
    \prod\limits_{\ell = 1}^{m - 1}
    \prod\limits_{r = \ell - 1}^\ell
    \frac{1}{z_{\tau (\ell)} - y_{\sigma (r)}}
    \cdot
    \prod\limits_{\ell = m + 1}^{p + 1}
    \prod\limits_{r = \ell - 1}^\ell
    \frac{1}{z_{\tau (\ell)} - y_{\sigma (r)}}
    \\
    \hookrightarrow
    \prod\limits_{\ell = 1}^{m - 1}
    \prod\limits_{r = \ell - 1}^\ell
    \frac{1}{z_{\tau (\ell)} - y_{\sigma (r)}}
    \cdot
    \prod\limits_{\ell = m + 1}^{p + 1}
    \prod\limits_{r = \ell - 1}^\ell
    \frac{1}{z_{\tau (\ell)} - y_{\sigma (r)}}
    .
\end{multline}
Thus, by symmetry,
for every $m$ the double sum over permutations $\sigma$ and $\tau$
of $\delta R_m$ is zero,
and therefore $R ( \mathbf{y}_{p + 2}, \mathbf{z}_{p + 2})$
is completely determined by $R_0 ( \mathbf{y}_{p + 2}, \mathbf{z}_{p + 2})$,
see expression~\eqref{eq:R-0}.

Finally, upon de-symmetrising the sums,
we obtain
\begin{multline}
\label{eq:K-p-coefficient-final}
    K_{j_0 k_0}^{(p)}
    =
    \sum\limits_{\mathbf{j}_{p + 1} \in (\Jcal^+)^{p + 1}}
    \sum\limits_{\mathbf{k}_{p + 1} \in (\Jcal^-)^{p + 1}}
    \frac{
        s_{j_0}^+ - s_{k_0}^- }{
        s_{j_{p + 1}}^+ - s_{k_0}^- }
    \prod\limits_{\ell = 1}^{p + 1}
    \frac{
        h_{j_\ell}^+ }{
        s_{k_\ell}^- - s_{j_\ell}^+ }
    \frac{ h_{k_\ell}^- }{
        s_{j_{\ell - 1}}^+ - s_{k_\ell}^-}
    \\
    = 
    \sum\limits_{j_{p + 1} = 1}^{n^+}
    \frac{
        s_{j_0}^+ - s_{k_0}^- }{
        s_{j_{p + 1}}^+ - s_{k_0}^- }
    \big[ \rmA^{p + 1} \big]_{j_0 j_{p + 1}}.
\end{multline}
This entails that
\begin{multline}
    P
    =
    - \sum\limits_{j_0 = 1}^{n^+}
    \sum\limits_{k_0 = 1}^{n^-}
    \partial_x \big[ A_{j_0 k_0}^- A_{k_0 j_0}^+ \big]
    -
    \sum\limits_{p \geq 0}
    \sum\limits_{j_{p + 1} = 1}^{n^+}
    \sum\limits_{j_0 = 1}^{n^+}
    \sum\limits_{k_0 = 1}^{n^-}
    \partial_x \big[ A_{j_{p + 1} k_0}^- A_{k_0 j_0}^+ \big]
    \big[ \rmA^{p + 1} \big]_{j_0 j_{p + 1}}
    \\
    = - \partial_x \tr\big[ \rmA \big]
    - \sum\limits_{p \geq 0}
    \frac{1}{p + 2} \partial_x \tr\big[ \rmA^{p + 2} \big]
    = \partial_x \tr \ln \big[ \mathrm{I}_{n^+} - \rmA \big]
    = \partial_x \ln \det \big[ \mathrm{I}_{n^+} - \rmA \big].
\end{multline}
Here in the first equation we used that
\begin{equation}
    A_{j_0 k_0}^-
    \frac{
        s_{j_0}^+ - s_{k_0}^- }{
        s_{j_{p + 1}}^+ - s_{k_0}^- }
    = A_{j_{p + 1} k_0}^-,
\end{equation}
due to the definition of $\rmA^-$,
see equation~\eqref{eq:matrices-A-plus-and-minus-app}.

To treat the general case,
observe that $\det \big[ \mathrm{I}_{n^+} - \gamma \rmA \big]$
has at most $n^+$ distinct roots.
Thus, there exists an analytic continuation of
$\partial_x \ln \det \big[ \mathrm{I}_{n^+} - \gamma \rmA \big]$
from a neighbourhood of $\gamma = 0$ to $\gamma = 1$.
Then it is enough to consider
the system with
$\rmA^\pm \to \gamma^{\frac12} \rmA^\pm$,
obtain the identity
\begin{equation}
    P
    = \partial_x \ln \det \big[ \mathrm{I}_{n^+} - \gamma \rmA \big],
\end{equation}
in the neighbourhood of $\gamma = 0$ as done above
and then continue it analytically up to $\gamma = 1$.

This completes the proof
of Proposition~\ref{prop:Slavnov-formula-for-general-momentum}.

\section{Construction of the parametrix}
\label{app:parametrix-construction}
In order to construct the parametrix,
i.e., the solution of the Riemann--Hilbert Problem~\ref{RHp:Pcal},
we first consider a model Riemann--Hilbert problem,
for which the jump matrix is piecewise constant.
Starting from this Riemann--Hilbert problem we obtain
a second order differential equation that determines
the  solution of the Riemann--Hilbert problem in terms
of parabolic cylinder functions.

\subsection{Second order differential equation}
We start with the following matrix Riemann--Hilbert problem on the matrix $D$:
\begin{RHp}
\label{RHp:D}
	Determine $D (\zeta) \in \mathbb{C}^{2 \times 2}$ such that
	\begin{enumerate}
		\item
		$D (\zeta)$ is analytic
		in $\mathbb{C} \backslash \gamma_D$
		and has continuous $\pm$ boundary values on
		$\gamma_D \backslash \{0\}$,
		see Figure~\ref{fig:contour-gamma-D-and-E}.
	    \item
	    On the contour $\gamma_D \backslash \{0\}$
	    the boundary values
	    $D_\pm (\zeta)$
	      are subject to the jump condition
	    $D_- (\zeta) = D_+ (\zeta) \rmG_D (\zeta)$
	    with jump matrix $\rmG_D (\zeta)$ given by
	    \begin{equation}
		\label{eq:jump-matrix-D}
			\rmG_D (\zeta)
			= \begin{cases}
				\mathrm{I}_2 + m
				\rme^{2 \pi \rmi \tau}
				\rme^{- \rmi \zeta^2}
				\zeta^{- 2 \tau}
				\sigma^+,
				\quad
				& \zeta \in \rme^{\flatfrac{3 \pi \rmi}{4}}\mathbb{R}_+,
				\\
				\mathrm{I}_2 + n
				\rme^{2 \pi \rmi \tau}
				\rme^{\rmi \zeta^2}
				\zeta^{2 \tau}
				\sigma^-
				\quad
				& \zeta \in \rme^{- \flatfrac{3 \pi \rmi}{4}}\mathbb{R}_+,
				\\
				\mathrm{I}_2 + n
				\rme^{\rmi \zeta^2}
				\zeta^{2 \tau}
				\sigma^-,
				\quad
				& \zeta \in \rme^{\flatfrac{\pi \rmi}{4}}\mathbb{R}_+,
				\\
				\mathrm{I}_2 + m
				\rme^{- \rmi \zeta^2}
				\zeta^{- 2 \tau}
				\sigma^+,
				\quad
				& \zeta \in \rme^{- \flatfrac{\pi \rmi}{4}}\mathbb{R}_+.
			\end{cases}
			\end{equation}
	    \item
	    $D (\zeta) = \mathrm{I}_2 + \zeta^{-1}
		\Ocal \big(
			\begin{smallmatrix}1 & 1 \\ 1 & 1\end{smallmatrix}
		\big)$
	    as $\zeta \to \infty$ up to tangential direction to $\gamma_D$.
	    \item
	    As $\zeta \to 0$
	    \begin{equation}
		\label{eq:matrix-D-zeta-to-0}
			D (\zeta)
			= \left[
				D_0
				+ \zeta
				\cdot
				\Ocal \big(
                    \begin{smallmatrix}
                        1 & 1
                        \\
                        1 & 1
                    \end{smallmatrix}
				\big)
			\right] \zeta^{\tau \sigma^z}.
		\end{equation}
	    for a piecewise constant matrix $D_0 \in \mathbb{C}^{2 \times 2}$.
	\end{enumerate}
\end{RHp}
\begin{figure}[ht]
	\centering
	\inputfig{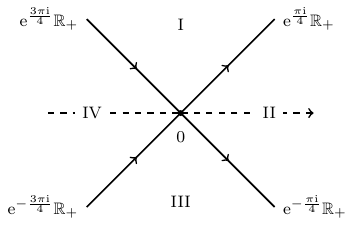}%
	\qquad
	\inputfig{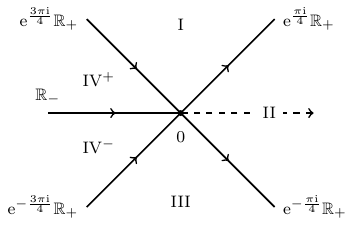}%
	\caption{
		The jump contours $\gamma_D$ on the left
		and $\gamma_E$ on the right.
		Here $\mathbb{R}_+ = (0, \infty)$
		and $\mathbb{R}_- = (-\infty, 0)$.}
	\label{fig:contour-gamma-D-and-E}
\end{figure}
This Riemann--Hilbert problem is a simpler version
of the Riemann--Hilbert Problem~\ref{RHp:Pcal},
where we `forget' about the dependence of the functions
$\tau$, $m$, $n$, and $\zeta$ on $\lambda$.
The expression for the jump matrix follows
from equations~\eqref{eq:jump-matrix-Upsilon}
and~\eqref{eq:matrices-M-local-variables}.
The behaviour at $\zeta = 0$
follows from equation~\eqref{eq:asymptotics-of-Upsilon-at-lambda-0}.

Now we introduce matrix
\begin{equation}
\label{eq:matrix-E}
	E (\zeta)
	= D (\zeta)
	\rme^{- \rmi \zeta^2 \sigma^z / 2} \zeta^{- \tau \sigma^z}
\end{equation}
whose jump matrix,
\begin{equation}
    \rmG_E (\zeta)
    = \zeta^{\tau \sigma^z}
    \rme^{\rmi \zeta^2 \sigma^z / 2}
    \rmG_D (\zeta)
    \rme^{- \rmi \zeta^2 \sigma^z / 2}
    \zeta^{- \tau \sigma^z},
    \qquad
    \zeta \in \gamma_D
\end{equation}
is now a piecewise constant matrix.
We note that $E (\zeta)$
now has an additional jump on $\mathbb{R}_- = (-\infty, 0)$,
see equation~\eqref{eq:matrix-E}.

Since the matrix $D$ has no jump for real negative $\zeta$,
it follows that for $\zeta < 0$
\begin{equation}
\label{eq:matrix-E-jump-condition-on-the-cut}
	D_+ (\zeta) = D_- (\zeta)
	\quad
	\Rightarrow
	\quad
	E_+ (\zeta)
	= E_- (\zeta)
	\frac{
		(\zeta - \rmi 0)^{\tau \sigma^z} }{
		(\zeta + \rmi 0)^{\tau \sigma^z} }
	= E_- (\zeta)
	\rme^{- 2 \pi \rmi \tau \sigma^z},
\end{equation}
i.e., the jump matrix $\rmG_E (\zeta)$
for $\zeta \in \mathbb{R}_-$ reads
\begin{equation}
	\rmG_E (\zeta) = \rme^{- 2 \pi \rmi \tau \sigma^z}.
\end{equation}
We denote the jump contour for the matrix $E$
as $\gamma_E = \gamma_D \cup \mathbb{R}_-$,
see Figure~\ref{fig:contour-gamma-D-and-E}.
Then the jump matrix $\rmG_E (\zeta)$ reads
\begin{equation}
\label{eq:jump-matrix-E}
	\rmG_E (\zeta)
	= \begin{cases}
		\mathrm{I}_2 + m
		\rme^{2 \pi \rmi \tau}
		\sigma^+,
		\quad
		& \zeta \in \rme^{\frac{3 \pi \rmi}{4}}\mathbb{R}_+,
		\\
		\mathrm{I}_2 + n
		\rme^{2 \pi \rmi \tau}
		\sigma^-
		\quad
		& \zeta \in \rme^{- \frac{3 \pi \rmi}{4}}\mathbb{R}_+,
		\\
		\mathrm{I}_2 + n
		\sigma^-,
		\quad
		& \zeta \in \rme^{\frac{\pi \rmi}{4}}\mathbb{R}_+,
		\\
		\mathrm{I}_2 + m
		\sigma^+,
		\quad
		& \zeta \in \rme^{- \frac{\pi \rmi}{4}}\mathbb{R}_+,
        \\
        \rme^{- 2 \pi \rmi \tau \sigma^z},
        \quad
        & \zeta \in \mathbb{R}_-.
	\end{cases}
\end{equation}

Therefore, the matrix $E$ is the unique solution
of the following Riemann--Hilbert problem:
\begin{RHp}
\label{RHp:E}
    Determine $E (\zeta) \in \mathbb{C}^{2 \times 2}$ such that
    \begin{enumerate}
        \item
        $E (\zeta)$ is analytic
        in $\mathbb{C} \backslash \gamma_E$
        and has continuous $\pm$ boundary values on
		$\gamma_E \backslash \{0\}$,
        see Figure~\ref{fig:contour-gamma-D-and-E}.

        \item
        On the contour $\gamma_E \backslash \{0\}$
        the boundary values
        $E_\pm (\zeta)$
        satisfy the jump condition
        $E_- (\zeta) = E_+ (\zeta) \rmG_E (\zeta)$
        with the jump matrix $\rmG_E (\zeta)$
        given by~\eqref{eq:jump-matrix-E}.

        \item
        As $\zeta \to \infty$ up to tangential direction to $\gamma_E$.
        \begin{equation}
            E (\zeta)
            = \big[
				\mathrm{I}_2
				+ \zeta^{-1}
				\Ocal \big(
                    \begin{smallmatrix}
                        1 & 1
                        \\
                        1 & 1
                    \end{smallmatrix}
				\big)
			\big]
            \rme^{- \rmi \zeta^2 \sigma^z / 2}
            \zeta^{\tau \sigma^z}.
        \end{equation}

        \item
        As $\zeta \to 0$
        \begin{equation}
        \label{eq:matrix-E-zeta-to-0}
            E (\zeta)
            = E_0 + \zeta
			\cdot
			\Ocal \big(
                \begin{smallmatrix}1 & 1 \\ 1 & 1\end{smallmatrix}
            \big).
        \end{equation}
        for a piecewise constant matrix $E_0 \in \mathbb{C}^{2 \times 2}$,
        i.e., $E (\zeta)$ is bounded at $\zeta = 0$.
    \end{enumerate}
\end{RHp}

The solution of the Riemann--Hilbert Problem~\ref{RHp:E}
with a piecewise constant jump matrix
can be related to a Fuchsian differential equation.
We introduce a matrix $\varphi (\zeta)$,
\begin{equation}
	\varphi (\zeta)
	= E' (\zeta) E^{-1} (\zeta).
\end{equation}
Then for $\zeta \in \gamma_E \backslash \{0\}$
we have
\begin{multline}
	\varphi_- (\zeta) = E_- ' (\zeta) E_-^{-1} (\zeta)
	= E_+' (\zeta) E_+^{-1} (\zeta)
	+ E_+ (\zeta) \rmG_E' (\zeta) \rmG_E^{-1} (\zeta) E_+^{-1} (\zeta)
	\\
	= E_+ ' (\zeta) E_+^{-1} (\zeta) = \varphi_+ (\zeta),
\end{multline}
where we used the fact that the jump matrix is a piecewise constant matrix,
i.e., $\rmG_E' (\zeta) = 0$.
Hence, $\varphi (\zeta)$ is holomorphic in $\mathbb{C} \backslash \{0\}$.
Moreover, since $\varphi (\zeta)$ is also bounded at $\zeta = 0$,
Riemann's theorem on removable singularities ensures
that $\varphi (\zeta)$ is holomorphic in the vicinity of $\zeta = 0$
and therefore in the whole complex plane.

As $\zeta \to \infty$
\begin{equation}
\label{eq:matrix-E-AE}
	E (\zeta) = 
	\Big(
		\mathrm{I}_2 + E_1 \zeta^{-1}
		+ \zeta^{-2}
		\Ocal \big(
			\begin{smallmatrix}1 & 1 \\ 1 & 1\end{smallmatrix}
		\big)
	\Big)
	\rme^{- \rmi \zeta^2 \sigma^z / 2} \zeta^{- \tau \sigma^z},
\end{equation}
which implies
\begin{equation}
	E^{-1} (\zeta) = 
	\rme^{\rmi \zeta^2 \sigma^z / 2} \zeta^{\tau \sigma^z}
	\Big(
		\mathrm{I}_2
		- E_1 \zeta^{-1}
		+ \zeta^{-2}
		\Ocal \big(
			\begin{smallmatrix}1 & 1 \\ 1 & 1\end{smallmatrix}
		\big)
	\Big)
\end{equation}
and
\begin{multline}
	E' (\zeta)
	= \Big(
		- E_1\zeta^{-2}
		+ \zeta^{-3}
		\Ocal \big(
			\begin{smallmatrix}1 & 1 \\ 1 & 1\end{smallmatrix}
		\big)
	\Big)
	\rme^{- \rmi \zeta^2 \sigma^z / 2} \zeta^{- \tau \sigma^z}
	\\
	+ \Big(
		\mathrm{I}_2 + E_1\zeta^{-1}
		+ \zeta^{-2}
		\Ocal \big(
			\begin{smallmatrix}1 & 1 \\ 1 & 1\end{smallmatrix}
		\big)
	\Big)
	\Big(
		- \rmi \zeta \sigma^z - \tau \zeta^{-1} \sigma^z
	\Big)
	\rme^{- \rmi \zeta^2 \sigma^z / 2} \zeta^{- \tau \sigma^z}
	\\
	= \Big(
		- \rmi \zeta - \rmi E_1
		+ \zeta^{-1}
		\Ocal \big(
			\begin{smallmatrix}1 & 1 \\ 1 & 1\end{smallmatrix}
		\big)
	\Big)
	\sigma^z
	\rme^{- \rmi \zeta^2 \sigma^z / 2} \zeta^{- \tau \sigma^z}.
\end{multline}
Therefore, as $\zeta \to \infty$
\begin{multline}
	\varphi (\zeta)
	= \Big(
		- \rmi \zeta \sigma^z - \rmi E_1 \sigma^z
		+ \zeta^{-1}
		\Ocal \big(
			\begin{smallmatrix}1 & 1 \\ 1 & 1\end{smallmatrix}
		\big)
	\Big)
	\Big(
		\mathrm{I}_2 - E_1 \zeta^{-1}
		+ \zeta^{-2}
		\Ocal \big(
			\begin{smallmatrix}1 & 1 \\ 1 & 1\end{smallmatrix}
		\big)
	\Big)
	\\
	= - \rmi \zeta \sigma^z - \rmi \comm{E_1}{\sigma^z}
	+ \zeta^{-1}
	\Ocal \big(
		\begin{smallmatrix}1 & 1 \\ 1 & 1\end{smallmatrix}
	\big).
\end{multline}
Then the matrix elements of the following expression
\begin{equation}
	\varphi (\zeta)
	+ \rmi \zeta \sigma^z
	+ \rmi \comm{E_1}{\sigma^z}
	= \zeta^{-1}
	\Ocal \big(
		\begin{smallmatrix}1 & 1 \\ 1 & 1\end{smallmatrix}
	\big)
\end{equation}
are entire functions bounded at infinity.
Due to Liouville's theorem, these functions must be constant,
and therefore they are zero.
Then
\begin{equation}
	\varphi (\zeta)
	= E' (\zeta) E^{-1} (\zeta)
	= - \rmi \zeta \sigma^z - \rmi \comm{E_1}{\sigma^z},
\end{equation}
and we end up with the following Fuchsian differential equation
\begin{equation}
\label{eq:matrix-E-DE}
	E' (\zeta)
	= \Big(
		- \rmi \zeta \sigma^z
		- \rmi \comm{E_1}{\sigma^z}
	\Big) E (\zeta).
\end{equation}
Denoting the matrix elements of $E_1$ as follows,
\begin{equation}
\label{eq:matrix-E-coefficient-1}
	E_1 = \begin{pmatrix}
		a & b\\
		c & d
	\end{pmatrix},
\end{equation}
we get
\begin{equation}
	\comm{E_1}{\sigma^z}
	= \begin{pmatrix}
		0 & - 2 b \\
		2 c & 0
	\end{pmatrix}.
\end{equation}
Hence, the differential equation~\eqref{eq:matrix-E-DE}
is equivalent to the following system of differential equations
for the matrix elements of $E (\zeta)$:
\begin{equation}
\left\{
\begin{aligned}
	E_{1 j}' (\zeta)
	&= - \rmi \zeta E_{1 j} (\zeta) + 2 \rmi b E_{2 j} (\zeta),
	\\
	E_{2 j}' (\zeta)
	&= \rmi \zeta E_{2 j} (\zeta) - 2 \rmi c E_{1 j} (\zeta)
\end{aligned}
\right.
\end{equation}
for $j = 1, 2$.

Differentiating both equations with respect to $\zeta$
and substituting the first derivatives from the system,
we obtain the second order differential equations
for the matrix elements of $E (\zeta)$:
\begin{align}
	E_{1 j}'' (\zeta)
	& = - \zeta^2 E_{1 j} (\zeta) + (4 b c - \rmi) E_{1 j} (\zeta),
	\\
	E_{2 j}'' (\zeta)
	& = - \zeta^2 E_{2 j} (\zeta) + (4 b c + \rmi) E_{2 j} (\zeta).
\end{align}

The last step is to write these second order ordinary differential equations
in the canonical form.
We introduce new variable $\xi_1$ and $\xi_2$ for $E_{1 j}$ and $E_{2 j}$,
respectively:
\begin{align}
	\xi_1 = \sqrt 2 \rme^{\frac{ \pi \rmi }{ 4 }} \zeta,
	\quad
	& \Rightarrow
	\quad
	\dv{\zeta} = \sqrt{2} \rme^{\frac{ \pi \rmi }{ 4 }} \dv{\xi_1},
	\\
	\xi_2 = \sqrt 2 \rme^{- \frac{ \pi \rmi }{ 4 }} \zeta,
	\quad
	& \Rightarrow
	\quad
	\dv{\zeta} = \sqrt{2} \rme^{- \frac{ \pi \rmi }{ 4 }} \dv{\xi_2},
\end{align}
and denote $u_{1 j} (\xi_1) := E_{1 j} (\zeta)$
and $u_{2 j} (\xi_2) := E_{2 j} (\zeta)$ for $j = 1, 2$.
Then we obtain the second order differential equations in the canonical form:
\begin{subequations}
\label{eq:diff-equations-for-u}
\begin{align}
	u_{1 j}'' (\xi)
	+ \left(
		2 \rmi b c + \frac12 - \frac{\xi^2}{4}
	\right)
	u_{1 j} (\xi)
	= 0,
	\\
	u_{2 j}'' (\xi)
	+ \Big(
		- 2 \rmi b c + \frac12 - \frac{\xi^2}{4}
	\Big)
	u_{2 j} (\xi)
	= 0
\end{align}
\end{subequations}
for $j = 1, 2$,
which is equivalent to the differential equation~\eqref{eq:diff-equation}.
Hence, the solution can be expressed in terms
of the parabolic cylinder functions
$D_{\rho} (\xi)$, $D_\rho (- \xi)$, $D_{- \rho - 1} ( \rmi \xi)$
and $D_{- \rho - 1} ( - \rmi \xi)$ with $\rho = \pm 2 \rmi b c$.

\subsection{Solution of the Riemann--Hilbert Problem}
One can construct the solution of the Riemann--Hilbert Problem~\ref{RHp:E},
using the following ansatz in terms of the parabolic cylinder functions:
\begin{equation}
	E (\zeta)
	= E_0 (\zeta)
	\cdot \rmL_v.
\end{equation}
Here the matrix $E_0 (\zeta)$ is given by
\begin{equation}
\label{eq:matrix-E-0}
	E_0 (\zeta) =
	\begin{pmatrix}
		D_\rho \big( \sqrt{2} \rme^{\frac{ \pi \rmi }{ 4 }} \zeta \big)
		&
		c_{12}
		\ D_{- \rho - 1} \big(
			\sqrt{2} \rme^{- \frac{ \pi \rmi }{ 4 }} \zeta
		\big)
		\\
		c_{21}
		\ D_{\rho - 1} \big( \sqrt{2} \rme^{\frac{ \pi \rmi }{ 4 }} \zeta \big)
		&
		D_{- \rho} \big( \sqrt{2} \rme^{- \frac{ \pi \rmi }{ 4 }} \zeta \big)
	\end{pmatrix},
\end{equation}
and the matrix $\rmL_v$ is a piecewise constant matrix in regions
$v = \text{I}, \text{II}, \text{III}, \text{IV}^-, \text{IV}^+$,
see Figure~\ref{fig:contour-gamma-D-and-E}.

Starting from the region $\text{II}$,
where the asymptotic expansion~\eqref{eq:AE-parabolic-cylinder-function}
is applicable,
one can satisfy the asymptotic condition for $E (\zeta)$,
see equations~\eqref{eq:matrix-E-AE} and~\eqref{eq:matrix-E-coefficient-1}.
Then one derives that $\rho = 2 \rmi b c = - \tau$,
\begin{equation}
\label{eq:matrix-L-II}
	\rmL_\text{II}
	= \rme^{ \flatfrac{ \pi \rmi \tau }{ 4 }}
	2 ^{\flatfrac{ \tau \sigma^z }{ 2 }},
\end{equation}
and
\begin{equation}
	c_{12} = \sqrt{2} \rme^{- \flatfrac{ \pi \rmi }{4} } b
	\qquad
	c_{21} = \sqrt{2} \rme^{\flatfrac{ \pi \rmi }{4} } c.
\end{equation}
In particular, it follows that
\begin{equation}
	c_{12} c_{21} = \rmi \tau.
\end{equation}

Then from the jump condition on the contour $\rme^{\pi \rmi / 4} \mathbb{R}_+$
it follows that
\begin{equation}
\label{eq:matrix-L-I}
	\rmL_\text{I}
	=
	\begin{pmatrix}
		1 & 0 \\ - n 2^{-\tau} & 1
	\end{pmatrix}
	\cdot
	\rmL_\text{II}.
\end{equation}
Next one can use the relation~\eqref{eq:parabolic-cylinder-function-relation}
to make the asymptotic expansion~\eqref{eq:AE-parabolic-cylinder-function}
applicable in the region $\text{I}$, i.e.,
for $ \flatfrac{\pi}{4} < \arg \zeta < \flatfrac{3 \pi}{4}$
and check that the asymptotic condition for $E (\zeta)$,
see equation~\eqref{eq:matrix-E-AE},
is satisfied in the region $\text{I}$
if one sets
\begin{equation}
	c_{12}
	= \frac{
		\rmi \sqrt{2 \pi}
		2^{\tau} }{
		\rme^{ \flatfrac{ \pi \rmi \tau }{2} } n \Gamma (\tau) },
	\qquad
	c_{21}
	=  \frac{
		\rme^{ \flatfrac{ \pi \rmi \tau }{2} }
		n \Gamma (\tau + 1) }{
		\sqrt{2 \pi} 2^{\tau} }.
\end{equation}

Finally, using the jump condition on the remaining parts of the jump contour,
one can derive that
\begin{equation}
\label{eq:matrix-L-III}
	\rmL_\text{III}
	=
	\rmL_\text{II}
	\cdot
	\begin{pmatrix}
		1 & m \\ 0 & 1
	\end{pmatrix},
\end{equation}
and
\begin{equation}
\label{eq:matrix-L-IV-minus-and-plus}
	\rmL_\text{IV}^-
	= \rmL_\text{III}
	\cdot
	\begin{pmatrix}
		1 & 0 \\ - n \rme^{2 \pi \rmi \tau} & 1
	\end{pmatrix},
	\qquad
	\rmL_\text{IV}^+
	= \rmL_\text{I}
	\cdot
	\begin{pmatrix}
		1 & m \rme^{2 \pi \rmi \tau}
		\\ 0 & 1
	\end{pmatrix}.
\end{equation}
Then combining again relation~\eqref{eq:parabolic-cylinder-function-relation}
and the asymptotic expansion~\eqref{eq:AE-parabolic-cylinder-function},
one can check that the asymptotic condition for $E (\zeta)$,
see equation~\eqref{eq:matrix-E-AE},
is satisfied in the rest of the regions,
i.e., in the regions $\text{III}, \text{IV}^- \text{IV}^+$.

Returning back to the solution
of the matrix Riemann--Hilbert Problem~\ref{RHp:D},
we obtain that
\begin{equation}
	D (\zeta)
	= E_0 (\zeta)
	\rmL_v
	\rme^{\rmi \zeta^2 \sigma^z / 2} \zeta^{\tau \sigma^z},
\end{equation}
where the matrix $E_0 (\zeta)$ is given by
\begin{equation}
	E_0 (\zeta) =
	\begin{pmatrix}
		D_{-\tau} \big( \sqrt{2} \rme^{\frac{ \pi \rmi }{ 4 }} \zeta \big)
		&
		\sqrt{2} \rme^{- \flatfrac{ \pi \rmi }{4} } b
		\ D_{\tau - 1} \big(
			\sqrt{2} \rme^{- \frac{ \pi \rmi }{ 4 }} \zeta
		\big)
		\\
		\sqrt{2} \rme^{\flatfrac{ \pi \rmi }{4} } c
		\ D_{- \tau + 1} \big(
			\sqrt{2} \rme^{\frac{ \pi \rmi }{ 4 }} \zeta
		\big)
		&
		D_{\tau} \big( \sqrt{2} \rme^{- \frac{ \pi \rmi }{ 4 }} \zeta \big)
	\end{pmatrix}
\end{equation}
with $b$ and $c$ given by
\begin{equation}
	b
	= \frac{
		\rmi \sqrt{\pi}
		\rme^{\pi \rmi / 4}
		2^{\tau} }{
		\rme^{ \flatfrac{ \pi \rmi \tau }{2} } n \Gamma (\tau) },
	\qquad
	c
	=  \frac{
		\rme^{ \flatfrac{ \pi \rmi \tau }{2} }
		n \Gamma (\tau + 1) }{
		\rme^{\pi \rmi / 4}
		\sqrt{\pi} 2^{\tau + 1} },
\end{equation}
which we denote $b_{12}$ and $b_{21}$, respectively,
see equation~\eqref{eq:b12-and-b21} in the main text.
The matrix $\rmL_v$
in regions $v = \text{I}, \text{II}, \text{III}, \text{IV}^- \text{IV}^+$
is given by expressions~\eqref{eq:matrix-L-I}, \eqref{eq:matrix-L-II},
\eqref{eq:matrix-L-III}, \eqref{eq:matrix-L-IV-minus-and-plus}, respectively.
In the main text the matrix $\rmL_v$
is given by equation~\eqref{eq:piece-wise-matrix-L}.

We remark
that there is no coefficient $m$ in the expression for $D (\zeta)$,
due to relation \eqref{eq:rel_mn_tau}. Taking this into account
we write $D(\zeta|n, \tau)$ to indicate the dependence of $D$ on
its parameters $n$ and $\tau$. Then we can define
\begin{equation}
    \Pcal (\lambda) = D(\zeta(\lambda)|n(\lambda),\tau(\lambda))
\end{equation}
and straightforwardly verify that $\Pcal$ is the solution of
the Riemann--Hilbert Problem~\ref{RHp:Pcal}. First of all, setting
$\Gamma_\Pcal = \omega^{-1} (D_{0, \epsilon} \cap \gamma_D)$, we
see that $\Gamma_\Pcal$ is the jump contour of the Riemann--Hilbert
Problem~\ref{RHp:Pcal}. Then, since the parabolic cylinder functions
are analytic in their index, $\Pcal$ is analytic in 
$\Ucal_{\lambda_0} \setminus \Gamma_\Pcal$, and, by construction,
satisfies the correct jump condition on $\Gamma_\Pcal$. Likewise
it obviously has the required behaviour close to $\lambda_0$.
For the behaviour on $\partial \Ucal_{\lambda_0}$ (cf.\ condition
(iii) of the Riemann--Hilbert Problem~\ref{RHp:Pcal}) we note
that $\zeta(\partial \Ucal_{\lambda_0}) = \epsilon \sqrt{x}$
which is large for fixed $\epsilon > 0$ and large $x$. Hence, we
can use the asymptotic series \eqref{eq:AE-parabolic-cylinder-function}
for the behaviour of the parabolic cylinder functions for large
modulus of their arguments in order to verify that $\Pcal$ satisfies
condition (iii) of the Riemann--Hilbert Problem~\ref{RHp:Pcal}
(compare Section~\ref{sec:AE_parametrix}).

\bibliographystyle{acm}
\bibliography{bibliography.bib}

\end{document}